\documentclass[12pt,a4paper,reqno]{amsbook}

\usepackage[T1]{fontenc}
\usepackage[utf8]{inputenc}
\usepackage{amscd}
\usepackage{amssymb}
\usepackage{color}
\usepackage{enumerate}
\usepackage{listings}
\usepackage{mathrsfs}
\usepackage{phd-thesis}
\usepackage{pstricks}
\usepackage{rotating}
\usepackage{xspace}

\usepackage{mathpazo}

\usepackage{float}
\usepackage{placeins}

\date{\today}
\title{On-line Chain Partitioning Approach to Scheduling\\
\begin{small}Pokrycie Łańcuchowe On-line w Kolejkowaniu Zadań\end{small}}
\author{Bart\l{}omiej Emil Bosek}
\dedicatory{Wydzia\l{} Matematyki i Informatyki, Uniwersytet Jagiello\'{n}ski\\
PhD Thesis, Supervisor: Professor Paweł M. Idziak}

\begin{document}

\begin{titlepage}
\begin{center}
$ $

\vspace{0.5mm}

\begin{LARGE}
\textbf{On-line Chain Partitioning Approach}

\vspace{2.5mm}

\textbf{to Scheduling}
\end{LARGE}

\vspace{5mm}

\begin{normalsize}
Pokrycie Łańcuchowe On-line w Kolejkowaniu Zadań
\end{normalsize}

\vspace{12mm}

\begin{LARGE}Bartłomiej Emil Bosek\end{LARGE}

\vspace{2.5mm}

\begin{small}
\textit{Algorithmics Research Group\\
Faculty of Mathematics and Computer Science\\
Jagiellonian University}
\end{small}

\vspace{125mm}

\begin{large}
Ph.D. Thesis
\end{large}

\vspace{1.6mm}

\begin{large}
Supervisor: Paweł M. Idziak
\end{large}

\vspace{3mm}

\begin{small}
Kraków, July 2008
\end{small}

\end{center}
\end{titlepage}

%\fontsize{11pt}{13pt}\selectfont

%\fontfamily{ptm}\selectfont

%\maketitle

%\maxtocdepth{ssclause}

\begin{titlepage}
\begin{footnotesize}
\begin{center}
\begin{minipage}{100mm}

\vspace{5mm}

\noindent \textsc{Abstract.}
An on-line chain partitioning algorithm receives the points of the poset from some externally determined list. Being presented with a new point the algorithm learns the comparability status of this new point to all previously presented ones. As each point is received, the algorithm assigns this new point to a chain in an irrevocable manner and  this assignment is made without knowledge of future points. Kierstead presented an algorithm using $(5^w-1)/4$ chains to cover each poset of width $w$. Felsner proved that width $2$ posets can be partitioned on-line into $5$ chains. We present an algorithm using $16$ chains on posets of width $3$. This result significantly narrows down the previous bound of $31$. 
Moreover, we address the on-line chain partitioning problem for interval orders.
Kierstead and Trotter presented an algorithm using $3w-2$ chains. We deal with an up-growing version of an on-line chain partition of interval orders, i.e. we restrict possible inputs by the rule that each new point is maximal at the moment of its arrival. We present an algorithm using $2w-1$ chains and show that there is no better one.
These problems come from a need for better algorithms that can be applied to scheduling.
Each on-line chain partitioning algorithm schedules tasks in a multiprocessor environment, and therefore can be applied  in order to minimize number of processors.

\vspace{1cm}

\noindent \textsc{Streszczenie.} Algorytm pokrycia łańcuchowego on-line otrzymuje na wejściu punkty zbioru częściowo uporządkowanego z zewnętrznie usta\-lonej listy. Podczas prezentowania nowego pun\-ktu algorytm poznaje porównywalności pomiędzy nowo zapre\-zen\-to\-wa\-nym a resztą punktów. Następnie algorytm, bez wiedzy o przyszłych punktach, nieodwoływalnie przydziela nowo przybyły punkt do pewnego łańcucha. Kierstead przedstawił algorytm używający $(5^w-1)/4$ łań\-cu\-chów w pokryciu do\-wol\-ne\-go porządku szerokości $w$.
Felsner z kolei dowiódł, że porządki szerokości $2$ mogą być pokryte on-line $5$ łańcuchami. W pracy przedstawiony jest algorytm używający $16$ łańcuchów na porządkach szero\-kości $3$. Wynik ten poprawia dotychczasowe górne ograniczenie równe $31$. Następnie rozważamy po\-kry\-cie łańcuchowe porządków przedzia\-łowych. \mbox{Kierstead} wraz z \mbox{Trotterem} przedstawili algorytm używający \mbox{$3w-2$} łańcuchy.
Skupiamy się tu na pokryciu on-line porządków przedziałowych w przypadku, gdy każdy nowo przybyły punkt jest maksymalny w chwili swojego pojawienia się. 
Przedstawiamy algorytm używający $2w-1$ łańcuchów i pokazujemy, że nie istnieje lepszy.
Problemy te motywację swą czerpią z faktu, iż algorytmy pokrycia łańcuchowego on-line modelują kolejkowanie zadań w maszynach wieloprocesorowych. Tym samym uzyskiwane algorytmy mogą być wykorzystane w kolejkowaniu w celu zminimalizowania liczby procesorów. 
\end{minipage}
\end{center}
\end{footnotesize}
\end{titlepage}

\begin{titlepage}
$ $

\vspace{15mm}

\begin{center}\textbf{Acknowledgment}\end{center}

\vspace{5mm}

My dissertation is based on my research on on-line partial orders.
I have worked in this area since 2003 when I join with Algorithmics Research Group during my M.Sc. studies at Jagiellonian University.
I was introduced to this subject during the seminar which is led by Paweł M. Idziak.
I would like to thank him for motivation, encouragement, guidance and many hours spent with me.

The group of people, who have created friendly, creative atmosphere and have motivated me to work, has been very helpfull.
I am especially indebted to Kamil Kloch, Tomasz Krawczyk, Grzegorz Matecki and Piotr Micek.

Moreover I want to thank also Izabela for help and  patience.

\end{titlepage}

\begin{titlepage}
\end{titlepage}

\setcounter{page}{4}
\setcounter{tocdepth}{2}
\tableofcontents

\chapter{Introduction}\label{Ch:Introduction}

\section{Preview}
We consider the on-line chain partitioning problem. This problem goes back to  the late 1970s when James Schmerl formulated it in a slightly different setting. For our purposes we assume that an on-line chain partitioning algorithm receives on input a partially ordered set, point by point,  in an on-line way. This means that the elements of the poset are taken one by one from some externally determined, but unknown to the algorithm, list. Being presented with a new element the algorithm learns its comparability status to all previously presented elements. Based on this knowledge the algorithm assigns the new element to a chain in an irrevocable manner. The performance of an on-line chain partitioning algorithm is measured by comparing the number of chains used with the number of chains needed by an optimal off-line algorithm which is always the width of the poset.

Another way to look at on-line problems is to treat them as two-person games. We call the players Algorithm and Spoiler. Algorithm represents an on-line algorithm and Spoiler represents an adaptive adversary who uncover the external list, or simply creates it according to his ``malicious'' will. The game is played in rounds. During each round Spoiler introduces a new point $x$ to the poset and describes comparabilities between $x$ and all points from all previous rounds. Algorithm responds by assigning $x$ to a chain. The most important feature of on-line games is that Algorithm's previous moves (decisions) restrict his present possibilities. The goal of Algorithm is to
minimize the number of chains, while Spoiler tries to force Algorithm to use as many of them as possible.

\smallskip

Now we give a brief overview of the structure of this dissertation. More information are included at the beginning of each chapter.

In this chapter we present a practical motivation for several variants of the on-line chain partitioning problem. We briefly describe problem of scheduling tasks in a multiprocessor environment and its connection with posets. Next we introduce the basic concepts, notation and results on partial orders and interval orders.

In Chapter \ref{Ch:On-line Orders} we deal with on-line algorithms for partial orders. Kierstead \cite{Kierstead} presented an algorithm using $(5^w-1)/4$ chains on posets of width $w$. On the other hand Endre Szemer\'{e}di \cite{Kierstead86} proves that any algorithm has to use at least $\binom{w+1}{2}$ chains. This is already a complicated result, and no progress has been made on this problem in general setting for the last 20 years. However in some special cases better bounds are known. Felsner \cite{Felsner} showed an algorithm using $5$ chains on posets of width $2$. We present an algorithm using $16$ chains on posets of width $3$. This result narrows down the previous bound of $31$.

In Chapter \ref{Ch:On-line Interval} we address the on-line chain partitioning problem for a special, but important, class of posets that can be represented on the real line by intervals. Kierstead and Trotter \cite{KiersteadTrotter} presented an algorithm using $3w-2$ chains and showed that this is the best possible result. We deal with an up-growing version of an on-line chain partition of interval orders, i.e. we restrict possible inputs by assuming that each new point is maximal at the moment of its arrival. We present an algorithm using $2w-1$ chains and show that there is no better one.

\section{Motivation}
The main underlying motivation for on-line chain partitioning comes from scheduling tasks in a multiprocessor environment. We are looking for algorithms which schedule, in an effective way, the consecutive tasks incoming on-line. Some tasks need as an input the outputs from previous tasks. Thus the tasks have to be performed in the order that follows the require data flow. Figure \ref{F:T.example}
\begin{figure}[hbt]
\begin{tabular}{|c|l|}
    \hline
    \begin{minipage}{8mm}task\end{minipage}
    & 
    \begin{small}\begin{minipage}{23mm}\vspace{1mm}needs source\\data from\vspace{1mm} \end{minipage}\end{small}\\
    \hline
    \hline
    $t_1$ & \hspace{1mm} -- \\
    $t_2$ & \hspace{1mm} -- \\
    $t_3$ & \hspace{1mm} $t_1,\ \ t_2$ \\
    $t_4$ & \hspace{1mm} $t_2$ \\
    \hline
\end{tabular}
\caption{Dependencies between tasks $t_1, t_2, t_3, t_4$.}\label{F:T.example}
\end{figure}
presents an example of such situation. Define the relation $<$ between tasks, so that $t_i<t_j$ if and only if $t_j$ needs some data from the output of task $t_i$. In manageable situations there are no dependencies that form cycles so that $<$ is a partial order on the set of tasks. The poset arising from our example is presented on Figure \ref{F:Order dep}.
\begin{figure}[hbt]
\centering\input{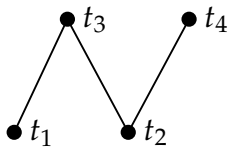}
\caption{The order of dependencies between $t_1, t_2, t_3, t_4$.}\label{F:Order dep}
\end{figure}
At the moment a new task appears, scheduler must immediately assign a processor which will execute this task. It is not known how long the execution of any single job will take and when the next task appears. Thus for each processor the set of tasks, scheduled to it, must form a chain. In other way we would lost optimality of the execution time. To see this, note that tasks which form a chain in the dependency order have to be executed in this order. Thus a task $t_i$ can not start earlier than a total amount of time needed to execute any chain of tasks on which $t_i$ depends. In consequence the fastest solution is obtained when each processor computes jobs from a chain. In our example with tasks $t_1, t_2, t_3, t_4$, one of the possibilities for the scheduler is presented on Figure \ref{F:Sol1}.
\begin{figure}[hbt]
\newcommand{\e}[1]{\begin{small}enqueue $t_{#1}$\end{small}}
\renewcommand{\b}[1]{\begin{small}begin $t_{#1}$\end{small}}
\newcommand{\n}[1]{\begin{small}end $t_{#1}$\end{small}}
\renewcommand{\d}[1]{\begin{small}data for $t_{#1}$\end{small}}
\newcommand{\sch}{\begin{small}scheduler\end{small}}
\newcommand{\p}[1]{\begin{small}$p_{#1}$\end{small}}
\newcommand{\tf}{\begin{small}time flow\end{small}}
\centering\ifx\JPicScale\undefined\def\JPicScale{1}\fi
\psset{unit=\JPicScale mm}
\psset{linewidth=0.3,dotsep=1,hatchwidth=0.3,hatchsep=1.5,shadowsize=1,dimen=middle}
\psset{dotsize=0.7 2.5,dotscale=1 1,fillcolor=black}
\psset{arrowsize=1 2,arrowlength=1,arrowinset=0.25,tbarsize=0.7 5,bracketlength=0.15,rbracketlength=0.15}
\begin{pspicture}(0,0)(113,73.5)
\psline[linewidth=0.15](15,48)(15,0)
\psline[linewidth=0.15](100,18)(100,0)
\psline[linewidth=0.15](100,43.5)(100,28)
\psline[linewidth=0.15,fillcolor=white,fillstyle=solid](15,68)(15,65.5)
\psline[linewidth=0.15,fillcolor=white,fillstyle=solid](60,68)(60,63)
\psline[linewidth=0.15,fillcolor=white,fillstyle=solid](100,68)(100,58)
\pspolygon[linewidth=0.15,fillcolor=white,fillstyle=solid](14,65.5)(16,65.5)(16,63)(14,63)
\psline[linewidth=0.15,fillcolor=white,fillstyle=solid](15,63)(15,60.5)
\pspolygon[linewidth=0.15,fillcolor=white,fillstyle=solid](14,60.5)(16,60.5)(16,58)(14,58)
\psline[linewidth=0.15,fillcolor=white,fillstyle=solid](15,58)(15,55.5)
\pspolygon[linewidth=0.15,fillcolor=white,fillstyle=solid](14,55.5)(16,55.5)(16,53)(14,53)
\psline[linewidth=0.15,fillcolor=white,fillstyle=solid](15,53)(15,50.5)
\pspolygon[linewidth=0.15,fillcolor=white,fillstyle=solid](14,50.5)(16,50.5)(16,48)(14,48)
\psline[linewidth=0.15,fillcolor=white,fillstyle=solid](15,63)(59,63)
\psline[linewidth=0.15,fillcolor=white,fillstyle=solid](15,58)(58,58)
\pspolygon[linewidth=0.15,fillcolor=white,fillstyle=solid](59,63)(61,63)(61,28)(59,28)
\pspolygon[linewidth=0.15,fillcolor=white,fillstyle=solid](99,58)(101,58)(101,43.5)(99,43.5)
\psline[linewidth=0.15,fillcolor=white,fillstyle=solid](62,58)(99,58)
\rput[B](35,64){\e1}
\rput[B](35,59){\e2}
\psline[linewidth=0.15,fillcolor=white,fillstyle=solid](15,53)(58,53)
\psline[linewidth=0.15,fillcolor=white,fillstyle=solid](15,48)(59,48)
\rput[B](35,54){\e3}
\rput[B](35,49){\e4}
\psline[linewidth=0.15,fillcolor=white,fillstyle=solid](62,53)(99,53)
\psline[linewidth=0.15,linestyle=dashed,dash=1 1,fillcolor=white,fillstyle=solid](61,43.5)(99,43.5)
\pspolygon[linewidth=0.15,fillcolor=white,fillstyle=solid](59,28)(61,28)(61,2.5)(59,2.5)
\pspolygon[linewidth=0.15,fillcolor=white,fillstyle=solid](99,28)(101,28)(101,18)(99,18)
\psline[linewidth=0.15,linestyle=dashed,dash=1 1,fillcolor=white,fillstyle=solid](61,28)(99,28)
\psline[linewidth=0.15](60,2.5)(60,0)
\rput[l](102,58){\b2}
\rput[l](102,28){\b3}
\rput[l](62,63){\b1}
\rput[tr](58,27.5){\b4}
\rput[br](58,28.5){\n1}
\rput[l](102,44){\n2}
\rput[l](102,18){\n3}
\rput[l](62,2.5){\n4}
\psline[linewidth=0.15](97,29)(99,28)
\psline[linewidth=0.15](97,27)(99,28)
\psline[linewidth=0.15](61,43.5)(63,42.5)
\psline[linewidth=0.15](61,43.5)(63,44.5)
\psline[linewidth=0.15](59,63)(57,62)
\psline[linewidth=0.15](59,63)(57,64)
\psline[linewidth=0.15](59,48)(57,47)
\psline[linewidth=0.15](59,48)(57,49)
\psline[linewidth=0.15](99,53)(97,52)
\psline[linewidth=0.15](99,53)(97,54)
\psline[linewidth=0.15](99,58)(97,57)
\psline[linewidth=0.15](99,58)(97,59)
\rput[B](80,29){\d3}
\rput[B](80,44.5){\d4}
\rput[B](15,70){\sch}
\rput[B](60,70){\p1}
\rput[B](100,70){\p2}
\psline[linewidth=0.2,linecolor=white,fillcolor=white,fillstyle=solid](20,73.5)(60.5,73.5)
\rput{90}(1,52.5){\tf}
\psline[linewidth=0.15](4,64)(4,4)
\psline[linecolor=white](113,50)(113,43)
\psline[linewidth=0.15](5,6)(4,4)
\psline[linewidth=0.15](3,6)(4,4)
\end{pspicture}
\caption{Data flow diagram -- unoptimal scheduler.}\label{F:Sol1}
\end{figure}
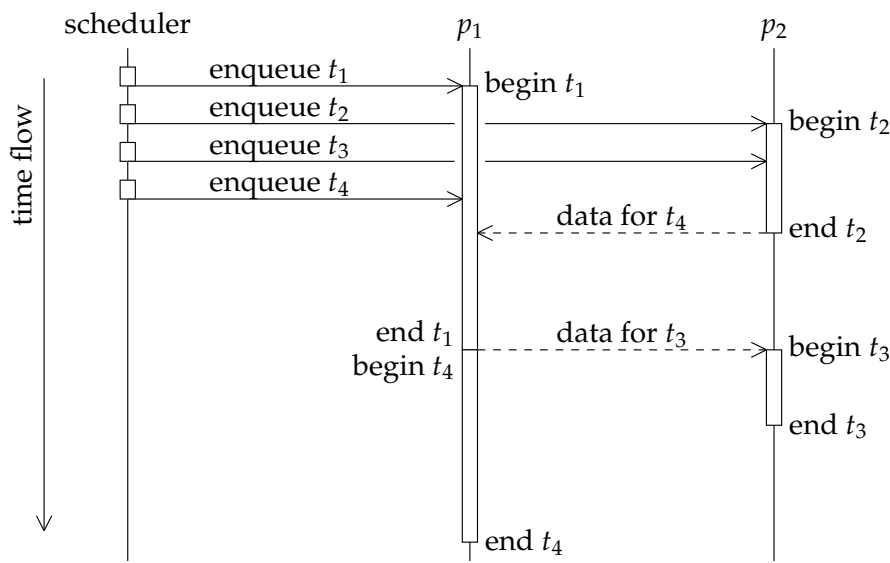
The processor $p_2$ is inactive most of the time. Scheduler could have worked better, e.g. as presented on Figure \ref{F:Sol2}.
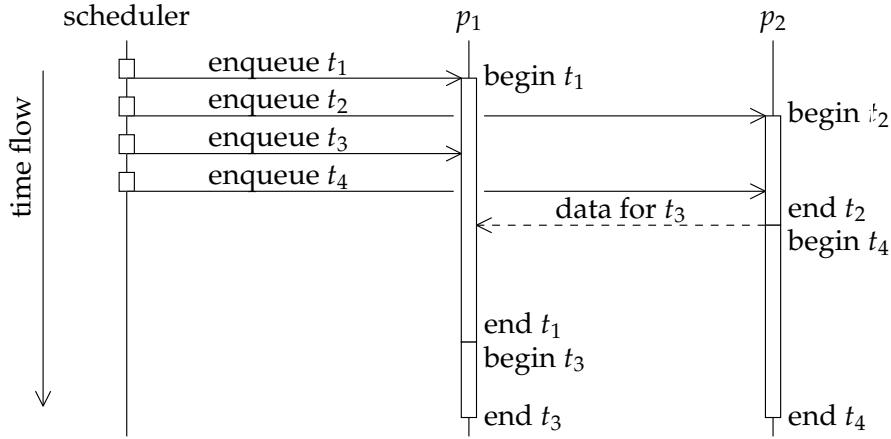
\begin{figure}[hbt]
   \newcommand{\e}[1]{\begin{small}enqueue $t_{#1}$\end{small}}
   \renewcommand{\b}[1]{\begin{small}begin $t_{#1}$\end{small}}
   \newcommand{\n}[1]{\begin{small}end $t_{#1}$\end{small}}
   \renewcommand{\d}[1]{\begin{small}data for $t_{#1}$\end{small}}
   \newcommand{\sch}{\begin{small}scheduler\end{small}}
   \newcommand{\p}[1]{\begin{small}$p_{#1}$\end{small}}
   \newcommand{\tf}{\begin{small}time flow\end{small}}
   \centering\ifx\JPicScale\undefined\def\JPicScale{1}\fi
\psset{unit=\JPicScale mm}
\psset{linewidth=0.3,dotsep=1,hatchwidth=0.3,hatchsep=1.5,shadowsize=1,dimen=middle}
\psset{dotsize=0.7 2.5,dotscale=1 1,fillcolor=black}
\psset{arrowsize=1 2,arrowlength=1,arrowinset=0.25,tbarsize=0.7 5,bracketlength=0.15,rbracketlength=0.15}
\begin{pspicture}(0,0)(113,57.5)
\psline[linewidth=0.15](15,32.5)(15,0)
\psline[linewidth=0.15](100,2.5)(100,0)
\psline[linewidth=0.15,fillcolor=white,fillstyle=solid](15,52.5)(15,50)
\psline[linewidth=0.15,fillcolor=white,fillstyle=solid](60,52.5)(60,47.5)
\psline[linewidth=0.15,fillcolor=white,fillstyle=solid](100,52.5)(100,42.5)
\pspolygon[linewidth=0.15,fillcolor=white,fillstyle=solid](14,50)(16,50)(16,47.5)(14,47.5)
\psline[linewidth=0.15,fillcolor=white,fillstyle=solid](15,47.5)(15,45)
\pspolygon[linewidth=0.15,fillcolor=white,fillstyle=solid](14,45)(16,45)(16,42.5)(14,42.5)
\psline[linewidth=0.15,fillcolor=white,fillstyle=solid](15,42.5)(15,40)
\pspolygon[linewidth=0.15,fillcolor=white,fillstyle=solid](14,40)(16,40)(16,37.5)(14,37.5)
\psline[linewidth=0.15,fillcolor=white,fillstyle=solid](15,37.5)(15,35)
\pspolygon[linewidth=0.15,fillcolor=white,fillstyle=solid](14,35)(16,35)(16,32.5)(14,32.5)
\psline[linewidth=0.15,fillcolor=white,fillstyle=solid](15,47.5)(59,47.5)
\psline[linewidth=0.15,fillcolor=white,fillstyle=solid](15,42.5)(58,42.5)
\pspolygon[linewidth=0.15,fillcolor=white,fillstyle=solid](59,47.5)(61,47.5)(61,12.5)(59,12.5)
\pspolygon[linewidth=0.15,fillcolor=white,fillstyle=solid](99,42.5)(101,42.5)(101,28)(99,28)
\psline[linewidth=0.15,fillcolor=white,fillstyle=solid](62,42.5)(99,42.5)
\rput[B](35,48.5){\e1}
\rput[B](35,43.5){\e2}
\psline[linewidth=0.15,fillcolor=white,fillstyle=solid](15,32.5)(58,32.5)
\psline[linewidth=0.15,fillcolor=white,fillstyle=solid](15,37.5)(59,37.5)
\rput[B](35,38.5){\e3}
\rput[B](35,33.5){\e4}
\psline[linewidth=0.15,fillcolor=white,fillstyle=solid](62,32.5)(99,32.5)
\psline[linewidth=0.15,linestyle=dashed,dash=1 1,fillcolor=white,fillstyle=solid](61,28)(99,28)
\pspolygon[linewidth=0.15,fillcolor=white,fillstyle=solid](99,28)(101,28)(101,2.5)(99,2.5)
\pspolygon[linewidth=0.15,fillcolor=white,fillstyle=solid](59,12.5)(61,12.5)(61,2.5)(59,2.5)
\psline[linewidth=0.15](60,2.5)(60,0)
\rput[l](102,42.5){\b2}
\rput[tl](102,27.5){\b4}
\rput[l](62,47.5){\b1}
\rput[tl](62,12){\b3}
\rput[bl](62,13){\n1}
\rput[bl](102,28.5){\n2}
\rput[l](102,2.5){\n4}
\rput[l](62,2.5){\n3}
\psline[linewidth=0.15](61,28)(63,27)
\psline[linewidth=0.15](61,28)(63,29)
\psline[linewidth=0.15](59,47.5)(57,46.5)
\psline[linewidth=0.15](59,47.5)(57,48.5)
\psline[linewidth=0.15](59,37.5)(57,36.5)
\psline[linewidth=0.15](59,37.5)(57,38.5)
\psline[linewidth=0.15](99,32.5)(97,31.5)
\psline[linewidth=0.15](99,32.5)(97,33.5)
\psline[linewidth=0.15](99,42.5)(97,41.5)
\psline[linewidth=0.15](99,42.5)(97,43.5)
\rput[B](80,29){\d3}
\rput[B](15,54.5){\sch}
\rput[B](60,54.5){\p1}
\rput[B](100,54.5){\p2}
\psline[linewidth=0.15,linecolor=white,fillcolor=white,fillstyle=solid](21,57.5)(70,57.5)
\rput{90}(1,37){\tf}
\psline[linewidth=0.15](4,48.5)(4,4)
\psline[linewidth=0.15](5,6)(4,4)
\psline[linewidth=0.15](3,6)(4,4)
\psline[linecolor=white](113,47)(113,40)
\end{pspicture}
   \caption{Data flow diagram -- an optimal scheduler.}\label{F:Sol2}
\end{figure}
Obviously the task $t_4$ has to wait for the data from $t_2$. However, the first solution forces $t_4$ to unnecessary wait for the processor $p_1$ to finish task $t_1$. The second solution avoids this unnecessary delay. Thus if jobs scheduled to a processor form a chain we do not waste time, reaching an optimal solution. Using order diagrams we present both solutions by Figure \ref{F:both sol}.
\begin{figure}[hbt]
   \centering\input{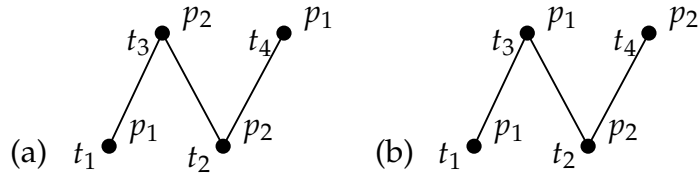}
   \caption{Makespan for $t_1, t_2, t_3, t_4$ for (a) slow and \mbox{(b) fast solutions}.}\label{F:both sol}
\end{figure}

Obviously the number of processors needed to work out all jobs with dependency poset of width $w$ has to be at least $w$. However $w$ may not be achieved if the jobs are incoming on-line, one by one, and they have be assigned to some processors before the next job come in. For example consider dependencies between tasks, presented on Figure \ref{F:N}.a, at the moment task $t_3$  appears.
\begin{figure}[hbt]
   \centering\input{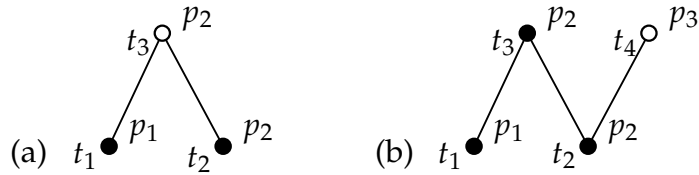}
   \caption{Forcing third processor.}\label{F:N}
\end{figure}
Now, scheduler either assigns $t_3$ to a third processor or, without loss of generality, to the already used $p_2$ (see Figure \ref{F:N}.a).
Next, task $t_4$ depending only on $t_2$  appears. To avoid a delay caused by unnecessary waiting of $p_1$ to perform $t_4$, scheduler has to use a third processor $p_3$ anyway (see Figure \ref{F:N}.b). This means that the fastest solution may be achieved at the cost of a bigger number of processors. On the other hand, in the off-line setting, when scheduler would know the future, a better solution can be applied (as on Figure \ref{F:both sol}.b). This example shows a need of on-line scheduling algorithms that perform as well as possible, i.e. use the number of processors as close, to the off-line optimum (of $w$), as possible.

Unfortunately the only known algorithm has to use exponential number of processors $(5^w-1)/4$ (\cite{Kierstead}). 
This exponential bound is not the best possible. 
In particular by Felsner's result \cite{Felsner}, $5$ processors  suffice for $w=2$. Moreover by our result presented in Chapter \ref{Ch:On-line Orders}, if all the tasks need $3$ processors in the off-line setting then scheduler can correctly assign them on-line to at most $16$ processors. 

One of natural additional restrictions  is that, at the moment of arrival of a new job, all tasks it depends on have been already scheduled. In this (so called up-growing) setting scheduler is in a much better position. Actually Felsner \cite{Felsner} showed that instead of \linebreak \mbox{$(5^w-1)/4$} only $\binom{w+1}{2}$ processors suffice.

In some applications, including real-time systems, each particular task has to be executed in a prescribed interval of time. This obviously restricts the choice for scheduler, but on the other hand gives him some additional knowledge. As all we are interested in, is to compare the effectiveness of on-line and off-line schedulers, it will turn out that linear (in $w$) upper bounds are possible instead of exponential or even quadratic ones. Indeed, if the time intervals (presented e.g. on real line) are known, Kierstead and Trotter \cite{KiersteadTrotter}  presented an algorithm using $3w-2$ processors in general, not necessarily up-growing setting, and Broniek \cite{Broniek} observed that $w$ processors suffice in up-growing setting.

In other applications (e.g. when the incoming tasks come from the source which is not properly time synchronized with the executing multiprocessor environment) we do now how jobs are related in time, i.e. we do know that $t_i$ has to be executed after $t_j$ is completed, but we do not know which exact time slots they have to use. Again we do know that the jobs can be represented by time intervals but this time we do not know any particular interval representation for them. This setting can be modelled by the on-line partition for interval orders without representation. The mentioned algorithm of Kiearsted and Trotter \cite{KiersteadTrotter} appears to be the best possible when using $3w-2$ processors. On the other hand Broniek's algorithm does not survive in this setting when jobs are coming in an  up-growing manner but without precise knowledge about their time frames. Indeed, as we show  in Chapter 3 each scheduler has to use at least $2w-1$ processors. Moreover we show then that $2w-1$ of them actually always suffice.

\clearpage

\section{Partial Orders}

A \emph{partially ordered set} (or \emph{ordered set}, or \emph{poset}) is an ordered pair $\brackets{P,\leq}$, where $P$ is a set and $\leq$ is a binary relation on $P$, which for all $x,y,z\in P$ satisfies:
\begin{align*}
x&\leq x, &\textrm{(reflexivity)}\\
x&\leq y,\ y\leq x \quad\textrm{imply}\quad x=y, &\textrm{(asymmetricity)}\\
x&\leq y,\ y\leq z \quad\textrm{imply}\quad x\leq z. &\textrm{(transitivity)}
\end{align*}
The relation $\leq$ is called a \emph{partial order} (or \emph{order}) on $P$. Sometimes, if the underlying set $P$ is covered by $P=X_1\cup \ldots\cup X_k$ we write $\brackets{X_1,\ldots,X_k,\leq}$ instead $\brackets{P,\leq}$.

We say that $x$ is \emph{above} (or \emph{greater} than) $y$ and $y$ is \emph{below} (or \emph{smaller} than) $x$ if \mbox{$x\leqslant y$} and \mbox{$x\neq y$}. We denote this fact by $y<x$. Sometimes we will use the symbol $<$ in a more general context. We extend the relation $<$ (from single elements) to subsets of elements so that \mbox{$X<Y$} means \mbox{$x<y$} for all $x\in X$ and all $y\in Y$. If one of the sets $X,Y$ is a singleton, for example $Y=\set{y}$, we simply write $X<y$. Points $x,y$ are \emph{comparable} in $P$ if either $x\leqslant y$ or $y\leqslant x$ and we denote it by $x\gtrless y$. Otherwise, we say $x$ and $y$ are \emph{incomparable} and write $x\parallel y$.
Moreover we say that $z$ \emph{lies between} $x$ and $y$ if $x<z<y$. Otherwise, if there is no such $z$ between $x$ and $y$, $x$ is \emph{immediate successor}  of $x<y$ (or $y$ \emph{covers} $x$) which is to be denoted by $x\prec y$. The relation $\prec$ is a subrelation of $\leq$. On the other hand, $\leq$ can be recovered from $\prec$ by taking transitive and reflexive closure of $\prec$.

A partially ordered set is often displayed by its \emph{Hasse diagram} \mbox{(or \emph{diagram})}, i.e. via graph of the relation $\prec$ on the vertically oriented plane. A vertex is drawn for each point of the poset. A line segment (or a curve) between $x$ and $y$ appears only if $x \prec y$ and $y$ lies higher in the plane than $x$.  An example of a Hasse diagram is presented on Figure \ref{F:Ex poset}.
\begin{figure}[hbt]
   \centering\ifx\JPicScale\undefined\def\JPicScale{1}\fi
\psset{unit=\JPicScale mm}
\psset{linewidth=0.3,dotsep=1,hatchwidth=0.3,hatchsep=1.5,shadowsize=1,dimen=middle}
\psset{dotsize=0.7 2.5,dotscale=1 1,fillcolor=black}
\psset{arrowsize=1 2,arrowlength=1,arrowinset=0.25,tbarsize=0.7 5,bracketlength=0.15,rbracketlength=0.15}
\begin{pspicture}(0,0)(28,30)
\rput{0}(6,-1){\psellipse[linestyle=none,fillstyle=solid](0,0)(1,1)}
\psline[linewidth=0.25](1,9)(6,-1)
\rput{0}(1,9){\psellipse[linestyle=none,fillstyle=solid](0,0)(1,1)}
\rput{0}(11,9){\psellipse[linestyle=none,fillstyle=solid](0,0)(1,1)}
\rput{0}(21,-1){\psellipse[linestyle=none,fillstyle=solid](0,0)(1,1)}
\psline[linewidth=0.25](21,-1)(11,19)
\psline[linewidth=0.25](11,9)(6,-1)
\rput{0}(26,9){\psellipse[linestyle=none,fillstyle=solid](0,0)(1,1)}
\rput{0}(11,29){\psellipse[linestyle=none,fillstyle=solid](0,0)(1,1)}
\rput{0}(21,19){\psellipse[linestyle=none,fillstyle=solid](0,0)(1,1)}
\psline[linewidth=0.25](21,-1)(26,9)
\psline[linewidth=0.25](6,-1)(26,9)
\psline[linewidth=0.25](21,19)(26,9)
\psline[linewidth=0.25](11,9)(11,19)
\newrgbcolor{userLineColour}{0.6 0.6 1}
\rput[Bl](13,8){$d$}
\newrgbcolor{userLineColour}{0.6 0.6 1}
\rput[Bl](28,8){$e$}
\newrgbcolor{userLineColour}{0.6 0.6 1}
\rput[Bl](13,18){$f$}
\newrgbcolor{userLineColour}{0.6 0.6 1}
\rput[Bl](23,18){$g$}
\newrgbcolor{userLineColour}{0.6 0.6 1}
\rput[Bl](13,28){$h$}
\rput{0}(11,19){\psellipse[linestyle=none,fillstyle=solid](0,0)(1,1)}
\psline[linewidth=0.25](1,9)(11,19)
\psline[linewidth=0.25](1,9)(21,19)
\psline[linewidth=0.25](11,29)(21,19)
\psline[linewidth=0.25](11,29)(11,19)
\newrgbcolor{userLineColour}{0.6 0.6 1}
\rput[Bl](8,-2){$a$}
\newrgbcolor{userLineColour}{0.6 0.6 1}
\rput[Bl](23,-2){$b$}
\newrgbcolor{userLineColour}{0.6 0.6 1}
\rput[Bl](3,8){$c$}
\end{pspicture}
   \caption{Hasse diagram of \mbox{$\poset{P}_0=\brackets{\set{a,\ldots,h},\leq}$} with \mbox{$a<\set{c,d,e,f,g,h}$},  \mbox{$b<\set{e,f,g,h}$}, \mbox{$c<\set{f,g,h}$}, \mbox{$d<\set{f,h}$}, \mbox{$e<\set{g,h}$}, \mbox{$f<h$}, \mbox{$g<h$} and no other comparabilities.}\label{F:Ex poset}
\end{figure}

\pagebreak

We say that $p$ is a \emph{maximal} (\emph{minimal}) element in a subset $X\subseteq P$ of some poset $\poset{P}=\brackets{P,\leq}$ if in $X$ there is no points above (below) $p$. We denote a set of all maximal (minimal) elements of $X$ by $\max (X)$ ($\min (X)$). By a greatest (smallest) element of $X$ we mean an element which lies above (below) all other points from $X$. By a \emph{supremum} (\emph{infimum}) of $x$ and $y$ we mean the least common upper bound (greatest common lower bound) of $x$ and $y$, i.e. the least element $r \geq x,y$ (\mbox{the greatest} element $r \leq x,y$). If every two points $x,y\in L$ in a poset $\poset{L}=\brackets{L,\leq}$ have the supremum and the infimum, $\poset{L}$ is called to be a lattice. In a lattice the supremum and infimum of $x,y$ is denoted by $x\vee y$ and $x\wedge y$, respectively. We will refer to them as \emph{join} and \emph{meet}, respectively.

Each $X\subseteq P$ determines two subsets of $P$, namely the closed upset of $X$ and the open upset of $X$ defined by 
\begin{eqnarray*}
X\upsetc_{\poset{P}}&:=&\set{p\in P: \textrm{ there is $x\in X$ so that $x\leqslant p$} },\\
X\upseto_{\poset{P}}&:=&\set{p\in P: \textrm{ there is $x\in X$ so that $x<p$} },
\end{eqnarray*}
respectively.
\emph{Downsets} $X\downsetc_{\poset{P}}$ and $X\downseto_{\poset{P}}$ of $X$ in $\poset{P}$ are defined dually. If $X=\set{x}$, we prefer to write $x\upseto_{\poset{P}}$ instead of $\set{x}\upseto_{\poset{P}}$. The reference to the poset $\poset{P}$ is often omitted whenever $\poset{P}$ is clear from the context.

A subset $A\subseteq P$ is an \emph{antichain} in $\poset{P}$ if each two distinct points of $A$ are incomparable. For example sets $\set{b,c,d}$ or $\set{d,e}$ are antichains in $\poset{P}_0$ presented on Figure \ref{F:Ex poset}. The order $\leq$ on elements determines the order $\aleq$ on antichains in such a way that $A\aleq B$ if and only if for any element $a\in A$ there exists an element $b\in B$ such that $a\leq b$. Note that $A\aleq B$ can be also expressed by $A\subseteq B\downsetc$ or $A\downsetc\subseteq B\downsetc$. For example, $\set{a,b}\aleq\set{d,e}$ but $\set{a,b}\not\!\!\aleq\set{c,d}$ in the poset $\poset{P}_0$ of Figure \ref{F:Ex poset}. Moreover, we write $A\al B$ when $A\aleq B$ and $A\neq B$. Of course, if $A<B$ then $A\al B$ but the converse does not hold, as it can be seen in our example $\set{a,b}\nless\set{d,e}$.

The \emph{width} of $X \subseteq P$, denoted by $\fWidth{X}$, is the largest size of all antichains in $X$. An antichain containing exactly $\fWidth{X}$ elements is called a \emph{maximum antichain of $X$}. The family of all maximum antichains in $X$ is denoted by $\fMA{X}$. For maximum antichains the relation $\aleq$ is equivalent to its dual, as stated in the next observation.

\begin{obs}\label{O:another def order ant}
Let $A,B\in\fMA{X}$, where $X$ is a subset of some poset $\poset{P}=\brackets{P,\leq}$. Then $A\subseteq B\downsetc$ if and only if $B \subseteq A\upsetc $.
\end{obs}
\begin{proof}
Assume that $A\subseteq B\downsetc$.
We will show that any $b\in B\setminus A$ is above some point from $A$. If $b$ was incomparable with all points from $A$ then $A\cup \set{b}$ would be an antichain of $\fWidth{X}+1$ elements. Suppose  $b<a$ for some $a\in A$. Since $A\subseteq B\downsetc$, there is $b'\in B$ with $a\leq b'$. This yields to $b<a\leq b'$ contradicting that both $b$ and $b'$ lie in the same antichain $B$.
\end{proof}

\pagebreak

\begin{obs}\label{O: max=sup min=inf}
For any two maximum antichains $A,B$ in a finite poset $\poset{P}=\brackets{P,\leq}$ the following hold:
\begin{enumerate}
\item $A$ and $B$ have the supremum in $\fMA{P}$, namely  $\max \brackets{A\cup B}$,\label{O: max=sup min=inf P:max}
\item $A$ and $B$ have the infimum in $\fMA{P}$, namely  $\min \brackets{A\cup B}$.\label{O: max=sup min=inf P:min}
\end{enumerate} 
\end{obs}
\begin{proof}
Of course both $\max\brackets{A\cup B}$ and $\min\brackets{A\cup B}$ are anti\-chains.
First we will show that $\abs{\max \brackets{A\cup B}}=\abs{\min \brackets{A\cup B}}=w$, where $w=\fWidth{\poset{P}}$. 
Put $v=\abs{A\cap B}$ and let $A_0:=A\setminus B$ and $B_0:=B\setminus A$. Of course, the antichains $A_0$ and $B_0$ have the same cardinality $u=w-v$.
%Because $A,B$ are antichains, if $x\in A\cap B$ then $x\parallel y$ for all $y\in A\cup B \setminus \set{x}$. A set of maximal and minimal points cover $ A_0\cup B_0$.
We will show that
\begin{equation}\label{E:A0=B0=u}
\abs{\max (A_0 \cup B_0)}\ =\ \abs{\min(A_0 \cup B_0 )}\ =\ u.
\end{equation}
This, together with the fact that each point $x\in A_0\cup B_0$ satisfies \mbox{$x\parallel y$} for all $y\in A\cap B$, will imply
\begin{align*}
\max \brackets{A \cup B}\ &=\ \brackets{A \cap B}\ \cup\ \max\brackets{A_0 \cup B_0}\quad \textrm{and}\\
\min \brackets{A \cup B}\ &=\ \brackets{A \cap B}\ \cup\ \min \brackets{A_0 \cup B_0}\!,
\end{align*}
so that  $\abs{\max (A \cup B)} = \abs{\min (A \cup B)} = v + u = w$, as required for $\max(A \cup B)$, $\min (A \cup B) \in \fMA{P}$.
 
To see (\ref{E:A0=B0=u}) first note that $\abs{\max(A_0 \cup B_0)}\leq u$, as otherwise the antichain $\max (A_0 \cup B_0)\cup (A \cap B)$ would have more than $u+v=w$ elements. By the same token $\abs{\min(A_0 \cup B_0)} \leq u$.
Moreover, note that each point $m \in A_0 \cup B_0$ is maximal or minimal in $A_0 \cup B_0$, as otherwise we would have a $3$-element chain $l < m < t$ with \mbox{$l, m, t \in A_0 \cup B_0$}, which contradicts the fact that $A_0 \cup B_0$ is covered by two antichains $A_0$, $B_0$. Now, if one of \mbox{$\max(A_0 \cup B_0)$} or \mbox{$\min(A_0 \cup B_0)$} has fewer than $u$ elements we would have
\begin{multline*}
2u\ >\ \abs{\max(A_0 \cup B_0)}+\abs{\min(A_0 \cup B_0)} \geq\\
\geq \abs{A_0 \cup B_0}\ =\ \abs{A_0} + \abs{B_0}\ =\ 2u,
\end{multline*}
a contradiction that proves (\ref{E:A0=B0=u}).

To see that $\max\brackets{A\cup B}$ is the supremum of $A$ and $B$ in $\fMA{P}$ note that finiteness of $A \cup B$ yields that each $x\in A\cup B$ lies below some maximal point of $A \cup B$. Thus $\max\brackets{A\cup B}\ageq A,B$. For any other upper bound $C\ageq A,B$, the downset $C\downsetc$ contains all points of $A\cup B$, in particular $\max\brackets{A\cup B}\aleq C$, making $\max(A \cup B)$ the lowest upper bound of  $A$ and $B$.

Thanks to Observation \ref{O:another def order ant} a similar argument can be applied to  show (\ref{O: max=sup min=inf P:min}).
\end{proof}

Observation \ref{O: max=sup min=inf} can be generalized to the following:

\begin{thm}[Dilworth \cite{Dilworth2}]\label{T:lattice} The family of all maximum antichains of a poset form a distributive lattice.\end{thm}

From now on  we assume that all posets under consideration are finite. In particular the lattice of maximum antichains is finite in this setting. We say that an antichain $A\subseteq P$ is \begin{em}high\end{em} in $\poset{P}=\brackets{P,\leq}$ if for any antichain $B\subseteq A\upsetc$ we have $A=B$ or $\abs{B}<\abs{A}$.
Of course, the top of the lattice $\fMA{\poset{P}}$ is high in $\poset{P}$. This unique high antichain of maximal size is denoted by $\fHMA{P}$.
The relation of maximum and high antichains are 
described in the following observation.

\begin{obs}\label{O:HM}
For a high antichain $H$ and a maximum antichain $M$ in a poset $\poset{P}=\brackets{P,\leqslant}$ we have $H\subseteq M\upsetc$.
\end{obs}

\begin{proof}
Since $M$ is a maximum antichain in $P$, we know that each point in $P$ is comparable
to some point in $M$. All we have to show is that no element from $H$ lies
strictly below $M$. Suppose to the contrary that the set $H_1=M\downseto\cap H$ is nonempty. Put
\begin{eqnarray*}
H_2&=&H-H_1,\\
M_1&=&H_1\upseto\cap M,\\
M_2&=&M-M_1.
\end{eqnarray*}
Observe that, both $M_1\cup H_2$ as well as $H_1\cup M_2$ are antichains and $M_1\cup H_2\subseteq H\upsetc$ (see Figure \ref{rys zaleznosci miedzy M a H}).
Now, if $\abs{M_1}\geqslant\abs{H_1}$ then \mbox{$M_1\cup H_2$} would have at least the same number of elements as $H$. This is impossible as $H$ is high. Thus $\abs{M_1}<\abs{H_1}$. But then the antichain \mbox{$H_1\cup M_2$} has more elements than the maximum antichain $M$ itself. This contradiction proves the observation.
\end{proof}
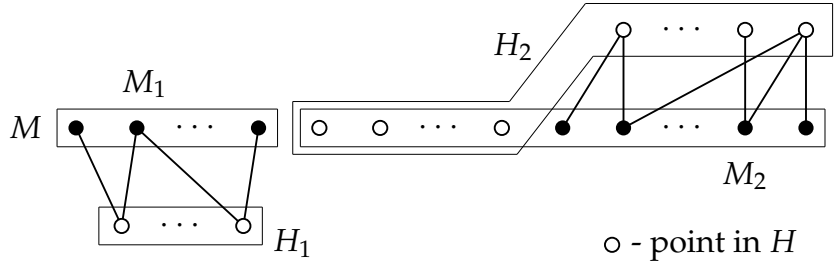
\begin{figure}[hbt]
\ifx\JPicScale\undefined\def\JPicScale{1}\fi
\psset{unit=\JPicScale mm}
\psset{linewidth=0.3,dotsep=1,hatchwidth=0.3,hatchsep=1.5,shadowsize=1,dimen=middle}
\psset{dotsize=0.7 2.5,dotscale=1 1,fillcolor=black}
\psset{arrowsize=1 2,arrowlength=1,arrowinset=0.25,tbarsize=0.7 5,bracketlength=0.15,rbracketlength=0.15}
\begin{pspicture}(0,0)(109,33.5)
\psline[linewidth=0.25,fillcolor=white,fillstyle=solid](15.5,3)(17.5,16)
\psline[linewidth=0.25,fillcolor=white,fillstyle=solid](33.5,16)(31.5,3)
\psline[linewidth=0.25,fillcolor=white,fillstyle=solid](31.5,3)(17.5,16)
\psline[linewidth=0.1](34,5.5)(12.5,5.5)
\psline[linewidth=0.1](34,0.5)(12.5,0.5)
\psline[linewidth=0.1](12.5,5.5)(12.5,0.5)
\psline[linewidth=0.1](34,5.5)(34,0.5)
\psline[linewidth=0.1](36,13.5)(7,13.5)
\psline[linewidth=0.1](36,18.5)(7,18.5)
\psline[linewidth=0.1](36,18.5)(36,13.5)
\psline[linewidth=0.25,fillcolor=white,fillstyle=solid](105.5,29)(81.5,16)
\psline[linewidth=0.1](109,32.5)(76.5,32.5)
\psline[linewidth=0.1](109,25.5)(77.5,25.5)
\psline[linewidth=0.1](38,19.5)(38,12.5)
\psline[linewidth=0.1](77.5,25.5)(67.5,12.5)
\psline[linewidth=0.1](67.5,12.5)(38,12.5)
\psline[linewidth=0.1](66.5,19.5)(38,19.5)
\psline[linewidth=0.1](109,32.5)(109,25.5)
\psline[linewidth=0.1](39,18.5)(39,13.5)
\psline[linewidth=0.25,fillcolor=white,fillstyle=solid](81.5,29)(73.5,16)
\psline[linewidth=0.25,fillcolor=white,fillstyle=solid](97.5,16)(97.5,29)
\psline[linewidth=0.25,fillcolor=white,fillstyle=solid](105.5,29)(97.5,16)
\psline[linewidth=0.25,fillcolor=white,fillstyle=solid](105.5,29)(105.5,16)
\psline[linewidth=0.1](108,13.5)(39,13.5)
\psline[linewidth=0.1](108,18.5)(39,18.5)
\psline[linewidth=0.1](108,18.5)(108,13.5)
\psline[linewidth=0.1](76.5,32.5)(66.5,19.5)
\psline[linewidth=0.25,fillcolor=white,fillstyle=solid](81.5,16)(81.5,29)
\psline[linewidth=0.1](7,18.5)(7,13.5)
\psline[linewidth=0.25,fillcolor=white,fillstyle=solid](15.5,3)(9.5,16)
\rput(25.5,16){$\cdots$}
\rput(23.5,3){$\cdots$}
\rput(57.5,16){$\cdots$}
\rput(89.5,16){$\cdots$}
\rput(89.5,29){$\cdots$}
\rput{0}(105.5,16){\psellipse[linewidth=0.25,linestyle=none,fillstyle=solid](0,0)(1,1)}
\rput{0}(97.5,16){\psellipse[linewidth=0.25,linestyle=none,fillstyle=solid](0,0)(1,1)}
\rput{0}(81.5,16){\psellipse[linewidth=0.25,linestyle=none,fillstyle=solid](0,0)(1,1)}
\rput{0}(73.5,16){\psellipse[linewidth=0.25,linestyle=none,fillstyle=solid](0,0)(1,1)}
\rput{0}(33.5,16){\psellipse[linewidth=0.25,linestyle=none,fillstyle=solid](0,0)(1,1)}
\rput{0}(17.5,16){\psellipse[linewidth=0.25,linestyle=none,fillstyle=solid](0,0)(1,1)}
\rput{0}(9.5,16){\psellipse[linewidth=0.25,linestyle=none,fillstyle=solid](0,0)(1,1)}
\rput[r](5,16){$M$}
\rput[Bl](35.38,-0.38){$H_1$}
\rput[br](21.5,20){$M_1$}
\rput[br](69.5,25){$H_2$}
\rput[tl](94.5,11.5){$M_2$}
\rput[Bl](82.5,-0.5){- point in $H$}
\psline[linecolor=white](0,24.5)(0,19.5)
\psline[linewidth=0.1,linecolor=white](96.5,33.5)(64,33.5)
\rput{0}(15.5,3){\psellipse[linewidth=0.25,fillcolor=white,fillstyle=solid](0,0)(0.94,0.94)}
\rput{0}(31.5,3){\psellipse[linewidth=0.25,fillcolor=white,fillstyle=solid](0,0)(0.94,0.94)}
\rput{0}(41.5,16){\psellipse[linewidth=0.25,fillcolor=white,fillstyle=solid](0,0)(0.94,0.94)}
\rput{0}(49.5,16){\psellipse[linewidth=0.25,fillcolor=white,fillstyle=solid](0,0)(0.94,0.94)}
\rput{0}(65.5,16){\psellipse[linewidth=0.25,fillcolor=white,fillstyle=solid](0,0)(0.94,0.94)}
\rput{0}(81.5,29){\psellipse[linewidth=0.25,fillcolor=white,fillstyle=solid](0,0)(0.94,0.94)}
\rput{0}(97.5,29){\psellipse[linewidth=0.25,fillcolor=white,fillstyle=solid](0,0)(0.94,0.94)}
\rput{0}(105.5,29){\psellipse[linewidth=0.25,fillcolor=white,fillstyle=solid](0,0)(0.94,0.94)}
\rput{0}(80,0.5){\psellipse[linewidth=0.25,fillcolor=white,fillstyle=solid](0,0)(0.94,0.94)}
\end{pspicture}
\caption{Interaction of antichains $H$ and $M$.}\label{rys zaleznosci miedzy M a H}
\end{figure}
\bigskip

A subset $\alpha\subseteq P$ is a \emph{chain} in $\poset{P}$ if each two points in $\alpha$ are comparable in $\poset{P}$.
We say that $\alpha_1,\ldots,\alpha_n$ is a \emph{chain partition} of a poset $\poset{P}=\brackets{P,\leq}$ if all $\alpha_i$ are chains, $\alpha_1\cup\ldots\cup\alpha_n=P$ and $\alpha_i\cap\alpha_j=\emptyset$ for $i\neq j$.

It is useful to identify chains forming a partition of a poset \linebreak \mbox{$\poset{P}=\brackets{P,\leq}$} with colors.
We say that a function $\chains:P\tto\Gamma$ is a \emph{coloring} of $\poset{P}$ if all $\fC{-1}{\gamma}$, with $\gamma\in\Gamma$, are chains.

\begin{thm}[Dilworth \cite{Dilworth}]\label{T:Dilworth}
Let $\poset{P}=\brackets{P,\leq}$ be a poset of width $w$.
Then $P=\alpha_1\cup\ldots\cup \alpha_w$ for some chains $\alpha_1,\ldots,\alpha_w$ in $\poset{P}$.
\end{thm}
\begin{proof}(Perles \cite{Perles}).
To induct on $\abs{P}$, note that for $\abs{P}=1$ there is nothing to be shown. Now suppose first that   $\fMA{P}\subseteq\set{\min (P),\ \max (P)}$.
Pick $x\in\min (P)$ and $y\in\max (P)$ such that $x \leq y$ to get $\fWidth{P\setminus\set{x,y}}<w$.
A chain partition $\alpha_1,\ldots,\alpha_{w-1}$ of $P\setminus\set{x,y}$, supplied by the induction hypothesis, together with $\alpha_w=\set{x,y}$ form a partition of $P$.

Now we assume that there is $A \in \fMA{P}- \set{\min (P),\ \max (P)}$. Thus $\min (P)\not\subseteq A\upsetc$ and $\max (P)\not\subseteq A\downsetc$. This yields $\abs{A\upsetc}<\abs{P}$ and $\abs{A\downsetc}<\abs{P}$. 
Again, the induction hypothesis allows us to partition $A\upsetc$ and $A\downsetc$ into $w$ chains $\beta_1,\ldots,\beta_w$ and $\gamma_1,\ldots,\gamma_w$, respectively. Since the antichain $A\subseteq A\upsetc\cap A\downsetc$ is maximum it must intersect with each of the chains $\beta_1,\ldots,\beta_w,\gamma_1,\ldots,\gamma_w$ at exactly one point. This allows to renumber the chains $\gamma_1,\ldots,\gamma_w$ so that $\beta_i\cap\gamma_i\neq\emptyset$. Consequently each $\beta_i\cup\gamma_i$ is a chain, so that the chains $\beta_1\cup\gamma_1,\ldots,\beta_w\cup\gamma_w$ form  a partition of $P$.
\end{proof}

Sometimes we relax the assumption that chains covering $P$ have to be disjoint and we say that $\alpha_1,\ldots,\alpha_w$ forms a \emph{chain covering} of a poset $\poset{P}$ if all of the $\alpha_i$'s are chains and $\alpha_1\cup\ldots\cup\alpha_n=P$. 
On the coloring side this corresponds to a \emph{multicoloring} of a poset \mbox{$\poset{P}=\brackets{P,\leq}$} defined to be a function $\chains:P\tto\powerp{\Gamma}=\power{\Gamma}\setminus\set{\emptyset}$ such that $\set{p\in P:\gamma\in\fC{}{p}}$ is a chain for any $\gamma\in\Gamma$. Note that any point $p\in P$ has to be colored by at least one color, i.e. $\fC{}{p}\neq\emptyset$. An example of a coloring and a multicoloring is presented on Figure \ref{F:Ex col and multicol}.

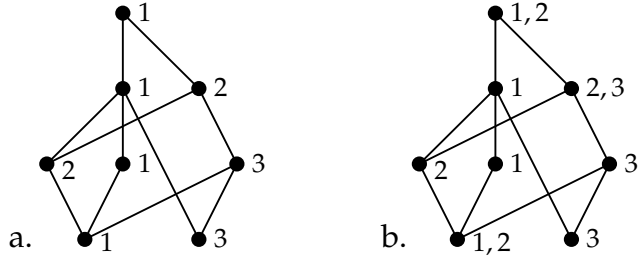
\begin{figure}[hbt]
\begin{footnotesize}
\centering\ifx\JPicScale\undefined\def\JPicScale{1}\fi
\psset{unit=\JPicScale mm}
\psset{linewidth=0.3,dotsep=1,hatchwidth=0.3,hatchsep=1.5,shadowsize=1,dimen=middle}
\psset{dotsize=0.7 2.5,dotscale=1 1,fillcolor=black}
\psset{arrowsize=1 2,arrowlength=1,arrowinset=0.25,tbarsize=0.7 5,bracketlength=0.15,rbracketlength=0.15}
\begin{pspicture}(0,0)(81,30)
\rput{0}(10,-1){\psellipse[linestyle=none,fillstyle=solid](0,0)(1,1)}
\psline[linewidth=0.25](5,9)(10,-1)
\rput{0}(5,9){\psellipse[linestyle=none,fillstyle=solid](0,0)(1,1)}
\rput{0}(15,9){\psellipse[linestyle=none,fillstyle=solid](0,0)(1,1)}
\rput{0}(25,-1){\psellipse[linestyle=none,fillstyle=solid](0,0)(1,1)}
\psline[linewidth=0.25](25,-1)(15,19)
\psline[linewidth=0.25](15,9)(10,-1)
\rput{0}(30,9){\psellipse[linestyle=none,fillstyle=solid](0,0)(1,1)}
\rput{0}(15,29){\psellipse[linestyle=none,fillstyle=solid](0,0)(1,1)}
\rput{0}(25,19){\psellipse[linestyle=none,fillstyle=solid](0,0)(1,1)}
\psline[linewidth=0.25](25,-1)(30,9)
\psline[linewidth=0.25](10,-1)(30,9)
\psline[linewidth=0.25](25,19)(30,9)
\psline[linewidth=0.25](15,9)(15,19)
\newrgbcolor{userLineColour}{0.6 0.6 1}
\rput[Bl](17,8){$1$}
\newrgbcolor{userLineColour}{0.6 0.6 1}
\rput[Bl](32,8){$3$}
\newrgbcolor{userLineColour}{0.6 0.6 1}
\rput[Bl](17,18){$1$}
\newrgbcolor{userLineColour}{0.6 0.6 1}
\rput[Bl](27,18){$2$}
\newrgbcolor{userLineColour}{0.6 0.6 1}
\rput[Bl](17,28){$1$}
\rput{0}(15,19){\psellipse[linestyle=none,fillstyle=solid](0,0)(1,1)}
\psline[linewidth=0.25](5,9)(15,19)
\psline[linewidth=0.25](5,9)(25,19)
\psline[linewidth=0.25](15,29)(25,19)
\psline[linewidth=0.25](15,29)(15,19)
\newrgbcolor{userLineColour}{0.6 0.6 1}
\rput[Bl](12,-2.5){$1$}
\newrgbcolor{userLineColour}{0.6 0.6 1}
\rput[Bl](27,-2){$3$}
\newrgbcolor{userLineColour}{0.6 0.6 1}
\rput[Bl](7,7){$2$}
\rput{0}(59,-1){\psellipse[linestyle=none,fillstyle=solid](0,0)(1,1)}
\psline[linewidth=0.25](54,9)(59,-1)
\rput{0}(54,9){\psellipse[linestyle=none,fillstyle=solid](0,0)(1,1)}
\rput{0}(64,9){\psellipse[linestyle=none,fillstyle=solid](0,0)(1,1)}
\rput{0}(74,-1){\psellipse[linestyle=none,fillstyle=solid](0,0)(1,1)}
\psline[linewidth=0.25](74,-1)(64,19)
\psline[linewidth=0.25](64,9)(59,-1)
\rput{0}(79,9){\psellipse[linestyle=none,fillstyle=solid](0,0)(1,1)}
\rput{0}(64,29){\psellipse[linestyle=none,fillstyle=solid](0,0)(1,1)}
\rput{0}(74,19){\psellipse[linestyle=none,fillstyle=solid](0,0)(1,1)}
\psline[linewidth=0.25](74,-1)(79,9)
\psline[linewidth=0.25](59,-1)(79,9)
\psline[linewidth=0.25](74,19)(79,9)
\psline[linewidth=0.25](64,9)(64,19)
\newrgbcolor{userLineColour}{0.6 0.6 1}
\rput[Bl](66,8){$1$}
\newrgbcolor{userLineColour}{0.6 0.6 1}
\rput[Bl](81,8){$3$}
\newrgbcolor{userLineColour}{0.6 0.6 1}
\rput[Bl](66,18){$1$}
\newrgbcolor{userLineColour}{0.6 0.6 1}
\rput[Bl](76,18){$2,3$}
\newrgbcolor{userLineColour}{0.6 0.6 1}
\rput[Bl](66,28){$1,2$}
\rput{0}(64,19){\psellipse[linestyle=none,fillstyle=solid](0,0)(1,1)}
\psline[linewidth=0.25](54,9)(64,19)
\psline[linewidth=0.25](54,9)(74,19)
\psline[linewidth=0.25](64,29)(74,19)
\psline[linewidth=0.25](64,29)(64,19)
\newrgbcolor{userLineColour}{0.6 0.6 1}
\rput[Bl](61,-2.5){$1,2$}
\newrgbcolor{userLineColour}{0.6 0.6 1}
\rput[Bl](76,-2){$3$}
\newrgbcolor{userLineColour}{0.6 0.6 1}
\rput[Bl](56,7){$2$}
\rput[Bl](0,-2){\begin{normalsize}a.\end{normalsize}}
\rput[Bl](49,-2){\begin{normalsize}b.\end{normalsize}}
\end{pspicture}
\end{footnotesize}
\caption{Example of (a) coloring and (b) multicoloring.}\label{F:Ex col and multicol}
\end{figure}

Almost all known algorithms that partition posets into chains  use matching in bipartite graph.

\begin{thm}[Feder, Motwani \cite{FederMotwani}]
There is an algorithm that partitions a poset $\poset{P}=\brackets{P,\leq}$ into $\fWidth{P}$ chains working in time $\fO{\frac{m\sqrt{n}}{\log n}}$, where $n=\abs{P}$ and $m\leq n^2$ is a number of pairs in $\leq$.
\end{thm}

%------------------------------------------------
\section{Interval Orders}
Since the on-line chain partitioning problem in its full generality seems to be hard enough, it is often considered for special
classes of posets. It turns out that the class of interval orders is
an interesting class for this problem.
\begin{defn}\label{definition interval}
A poset $\poset{P}$ is an interval order if there is a
 function $I:P \tto \R^2$ assigning to each point $x\in\poset{P}$ a
nondegenerate, closed interval $I(x)=[l_x,r_x]$ of the real line
$\mathbb{R}$ so that $x<y$ in $\poset{P}$ \mbox{if and only if} $r_x<l_y$
in $\mathbb{R}$. The function $I$ is called an interval
representation of the poset $\poset{P}$.
\end{defn}
The poset $\poset{Q}$ of Figure \ref{F:inter}.a 
is an interval order, as it has the interval
representation presented by Figure \ref{F:inter}.b. 

\begin{figure}[hbt]
\centering\ifx\JPicScale\undefined\def\JPicScale{1}\fi
\psset{unit=\JPicScale mm}
\psset{linewidth=0.3,dotsep=1,hatchwidth=0.3,hatchsep=1.5,shadowsize=1,dimen=middle}
\psset{dotsize=0.7 2.5,dotscale=1 1,fillcolor=black}
\psset{arrowsize=1 2,arrowlength=1,arrowinset=0.25,tbarsize=0.7 5,bracketlength=0.15,rbracketlength=0.15}
\begin{pspicture}(0,0)(116,23)
\rput{0}(11,1){\psellipse[linestyle=none,fillstyle=solid](0,0)(1,1)}
\psline[linewidth=0.25](6,11)(11,1)
\rput{0}(6,11){\psellipse[linestyle=none,fillstyle=solid](0,0)(1,1)}
\rput{0}(16,11){\psellipse[linestyle=none,fillstyle=solid](0,0)(1,1)}
\rput{0}(21,1){\psellipse[linestyle=none,fillstyle=solid](0,0)(1,1)}
\psline[linewidth=0.25](21,1)(16,11)
\psline[linewidth=0.25](16,11)(11,1)
\rput{0}(26,11){\psellipse[linestyle=none,fillstyle=solid](0,0)(1,1)}
\rput{0}(31,1){\psellipse[linestyle=none,fillstyle=solid](0,0)(1,1)}
\rput{0}(21,21){\psellipse[linestyle=none,fillstyle=solid](0,0)(1,1)}
\psline[linewidth=0.25](21,1)(26,11)
\psline[linewidth=0.25](31,1)(26,11)
\psline[linewidth=0.25](21,21)(26,11)
\psline[linewidth=0.25](16,11)(21,21)
\psline[linewidth=0.25](11,1)(26,11)
\newrgbcolor{userLineColour}{0.6 0.6 1}
\rput[Bl](13,0){$k$}
\newrgbcolor{userLineColour}{0.6 0.6 1}
\rput[Bl](23,0){$l$}
\newrgbcolor{userLineColour}{0.6 0.6 1}
\rput[Bl](33,0){$m$}
\newrgbcolor{userLineColour}{0.6 0.6 1}
\rput[Bl](8,10){$n$}
\newrgbcolor{userLineColour}{0.6 0.6 1}
\rput[Bl](18,10){$o$}
\newrgbcolor{userLineColour}{0.6 0.6 1}
\rput[Bl](28,10){$p$}
\newrgbcolor{userLineColour}{0.6 0.6 1}
\rput[Bl](23,20){$q$}
\psline[linewidth=0.1](72,13)(104,13)
\psline[linewidth=0.1](88,14)(92,14)
\psline[linewidth=0.1](88,18)(92,18)
\psline[linewidth=0.1](92,18)(92,14)
\psline[linewidth=0.1](72,9)(104,9)
\psline[linewidth=0.1](84,18)(84,14)
\psline[linewidth=0.1](64,14)(84,14)
\psline[linewidth=0.1](88,18)(88,14)
\psline[linewidth=0.1](64,18)(84,18)
\psline[linewidth=0.1](80,8)(96,8)
\psline[linewidth=0.1](80,4)(96,4)
\psline[linewidth=0.1](100,4)(116,4)
\psline[linewidth=0.1](100,8)(116,8)
\psline[linewidth=0.1](96,8)(96,4)
\psline[linewidth=0.1](100,8)(100,4)
\psline[linewidth=0.1](104,13)(104,9)
\psline[linewidth=0.1](80,8)(80,4)
\psline[linewidth=0.1](76,8)(76,4)
\psline[linewidth=0.1](56,8)(76,8)
\psline[linewidth=0.1](56,4)(76,4)
\psline[linewidth=0.1](72,13)(72,9)
\psline[linewidth=0.1](68,13)(68,9)
\psline[linewidth=0.1](60,9)(68,9)
\psline[linewidth=0.1](60,13)(68,13)
\psline[linewidth=0.1](64,18)(64,14)
\psline[linewidth=0.1](60,13)(60,9)
\psline[linewidth=0.1](56,8)(56,4)
\psline[linewidth=0.1](116,8)(116,4)
\newrgbcolor{userLineColour}{0.6 0.6 1}
\rput[Br](83,14.5){$m$}
\newrgbcolor{userLineColour}{0.6 0.6 1}
\rput[Br](67,9.5){$k$}
\newrgbcolor{userLineColour}{0.6 0.6 1}
\rput[Br](75,4.5){$l$}
\newrgbcolor{userLineColour}{0.6 0.6 1}
\rput[Br](103,9.5){$n$}
\newrgbcolor{userLineColour}{0.6 0.6 1}
\rput[Br](95,4.5){$o$}
\newrgbcolor{userLineColour}{0.6 0.6 1}
\rput[Br](91,14.5){$p$}
\newrgbcolor{userLineColour}{0.6 0.6 1}
\rput[Br](115,4.5){$q$}
\newrgbcolor{userLineColour}{0.6 0.6 1}
\rput[Bl](0,0){a.}
\newrgbcolor{userLineColour}{0.6 0.6 1}
\rput[Bl](50,0){b.}
\psline[linecolor=white](42,23)(63,23)
\end{pspicture}
\caption{(a) An interval order $\poset{Q}$ and (b) an interval representation of $\poset{Q}$.}\label{F:inter}
\end{figure}

It is worth noting that representation bears essentially more information. For example a poset $\poset{S}=\brackets{\set{d,e,f},\leq}$ can be correctly extended by $g$, as a Figure \ref{F:N int}.a shows.
But unfortunately if $\poset{S}$ is associated with the representation presented on Figure \ref{F:N int}.b this  extension can not be done.

\begin{figure}[hbt]
\centering\ifx\JPicScale\undefined\def\JPicScale{1}\fi
\psset{unit=\JPicScale mm}
\psset{linewidth=0.3,dotsep=1,hatchwidth=0.3,hatchsep=1.5,shadowsize=1,dimen=middle}
\psset{dotsize=0.7 2.5,dotscale=1 1,fillcolor=black}
\psset{arrowsize=1 2,arrowlength=1,arrowinset=0.25,tbarsize=0.7 5,bracketlength=0.15,rbracketlength=0.15}
\begin{pspicture}(0,0)(92,16)
\rput{0}(9,1){\psellipse[linewidth=0.25,fillstyle=solid](0,0)(1,1)}
\rput{0}(15,13){\psellipse[linewidth=0.25,fillstyle=solid](0,0)(1,1)}
\psline[linewidth=0.25](21,1)(15,13)
\psline[linewidth=0.25](15,13)(9,1)
\rput{0}(21,1){\psellipse[linewidth=0.25,fillstyle=solid](0,0)(1,1)}
\psline[linewidth=0.25,linestyle=dashed,dash=1 1](21,1)(27,13)
\newrgbcolor{userLineColour}{0.6 0.6 1}
\rput[Bl](11,0){$d$}
\newrgbcolor{userLineColour}{0.6 0.6 1}
\rput[Bl](23,0){$e$}
\newrgbcolor{userLineColour}{0.6 0.6 1}
\rput[Bl](17,12){$f$}
\newrgbcolor{userLineColour}{0.6 0.6 1}
\rput[Bl](29,12){$g$}
\psline[linewidth=0.1,linestyle=dashed,dash=1 1](74,11.5)(92,11.5)
\psline[linewidth=0.1,linestyle=dashed,dash=1 1](74,7.5)(92,7.5)
\psline[linewidth=0.1](74,6.5)(92,6.5)
\psline[linewidth=0.1](74,2.5)(92,2.5)
\psline[linewidth=0.1](92,6.5)(92,2.5)
\psline[linewidth=0.1,linestyle=dashed,dash=1 1](92,11.5)(92,7.5)
\psline[linewidth=0.1](74,6.5)(74,2.5)
\psline[linewidth=0.1](70,6.5)(70,2.5)
\psline[linewidth=0.1](52,6.5)(70,6.5)
\psline[linewidth=0.1](52,2.5)(70,2.5)
\psline[linewidth=0.1,linestyle=dashed,dash=1 1](74,11.5)(74,7.5)
\psline[linewidth=0.1](66,11.5)(66,7.5)
\psline[linewidth=0.1](48,7.5)(66,7.5)
\psline[linewidth=0.1](48,11.5)(66,11.5)
\psline[linewidth=0.1](48,11.5)(48,7.5)
\psline[linewidth=0.1](52,6.5)(52,2.5)
\newrgbcolor{userLineColour}{0.6 0.6 1}
\rput[Br](65,8){$d$}
\newrgbcolor{userLineColour}{0.6 0.6 1}
\rput[Br](69,3){$e$}
\newrgbcolor{userLineColour}{0.6 0.6 1}
\rput[Br](91,8){$g$}
\newrgbcolor{userLineColour}{0.6 0.6 1}
\rput[Br](91,3){$f$}
\newrgbcolor{userLineColour}{0.6 0.6 1}
\rput[Bl](0,0){a.}
\newrgbcolor{userLineColour}{0.6 0.6 1}
\rput[Bl](42,0){b.}
\rput{0}(27,13){\psellipse[linewidth=0.25,fillcolor=white,fillstyle=solid](0,0)(1,1)}
\rput(82.5,9.5){?}
\psline[linecolor=white](43,16)(64,16)
\end{pspicture}
\caption{(a) An interval order $\poset{S}$ (with (b) an interval representation of $\poset{S}$) and a vain attempt of an extension.}\label{F:N int}
\end{figure}
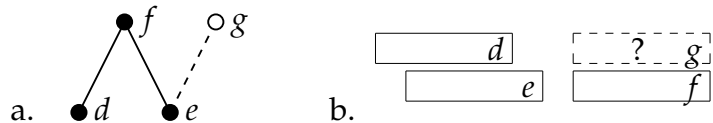

It may be complicated to use Definition \ref{definition interval} to check whether a poset is an interval order. Fortunately there are equivalent conditions that can be checked in a polynomial time. To state them we need the following definition.
\begin{defn}
A poset is $(\II)$-free if it has no elements
$a,b,c,d$ with: $a<b$,\ \ $c<d$,\ \ $a\parallel d$\ \ and\ \
$c\parallel b$.
\end{defn}
\noindent Loosely speaking in $(\II)$-free poset there is no subposet of \mbox{Figure \ref{F:2+2}}.
\begin{figure}[hbt]
\centering\input{chapter-01-pict-07}
\caption{The poset $\II$.}\label{F:2+2}
\end{figure}

\begin{thm}[Fishburn \cite{Fishburn}]\label{theorem fisburn}
Let $\poset{P}=\brackets{P,\leq}$ be a poset. Then the following statements are equivalent:
\begin{enumerate}
\item $\poset{P}$ is an interval order.
\item $\poset{P}$ is a $(\II)$-free poset.
\item For any $p, q \in P$ either $p\downseto \subseteq q\downseto$ or $p\downseto \supseteq q\downseto$, i.e. family of open downsets is linearly ordered with respect to inclusion.
\item For any $p, q \in P$ either $p\upseto \subseteq q\upseto$ or $p\upseto \supseteq q\upseto$, i.e. family of open upsets is linearly ordered with respect to inclusion.
\end{enumerate}
\end{thm}

Deleting the edge $k\prec p$ from poset $\poset{Q}$ of Figure \ref{F:inter} leads to a poset (see Figure \ref{F:noninter}) which has no interval representation, as $k,n,m,p$ now form a subposet of $\poset{R}$ isomorphic to $\II$.
\begin{figure}[hbt]
\centering\input{chapter-01-pict-08}
\caption{A noninterval order $\poset{R}$.}\label{F:noninter}
\end{figure}

\chapter{Chain Partitioning of Orders}\label{Ch:On-line Orders}

In this Chapter we deal with on-line chain partitioning of orders. One way to understand on-line structures is to incorporate into off-line ones a linear ordering of type $\omega\!$, as the following definition says.

\begin{defn}\label{D:online poset}
An \emph{on-line ordered set} (or \emph{on-line poset}) is a triple \mbox{$\poset{P}^{\ll}=\brackets{P,\leq,\ll}$}, where \mbox{$\poset{P}=\brackets{P,\leq}$} is a partially ordered set, $P=\set{p_1, p_2, \ldots}$ is countable and $\ll$ is a linear ordering of $P$ of type $\omega$ such that $p_1\ll p_2\ll p_3\ll\ldots$ 
The relation $\ll$ is called a \emph{presentation order of $\poset{P}$}. The set $P^{(i)}=\set{p_1,\ldots,p_i}$ is called an \emph{initial segment} of $\poset{P}^{\ll}$ and we let $\poset{P}^{(i)}=(P^{(i)},\rest{\leq}{^{P^{(i)}}})$. 
\end{defn}

Before we discuss on-line algorithms we only note that their output, i.e. number of chains used for a covering, is being compared with the number of chains needed to cover a poset in the off-line setting. Dilworth's Theorem \ref{T:Dilworth} suggests that this off-line number \linebreak of chains coincides with the width of a finite poset. Although \linebreak\mbox{Theorem \ref{T:Dilworth}} can be generalized for posets with arbitrary many elements (e.g. by using ultraproduct argument) we decided to present its generalization only in a countable setting. One reason for that is to model an on-line behavior described in Definition \ref{D:online poset}. Another one is to notice that in fact staying with minimal number of chains may require more information than presented covered initial segments of the poset give.

\begin{thm}\label{T:countouble partition}
Any countable poset of width $w$ can be partitioned into $w$ chains.
\end{thm}
\begin{proof}
Let $\brackets{P,\leq}$ be a countable poset with $P=\set{p_1,p_2,p_3,\ldots}$.
We induct on $n$ to define chains $\cc{1}{n},\ldots,\cc{w}{n}$ satisfying
\begin{enumerate}
\item $\cc{1}{n}\cup\ldots\cup\cc{w}{n}=\set{p_1,\ldots,p_n}$,
\item $\cc{i}{j}\subseteq\cc{i}{k}$ for all $j\leq k$ and $1\leq i\leq w$,
\item for infinitely many $m\geq n$ the chains $\cc{i}{n}$ can be extended to chains $\ccc{i}{m}\supseteq\cc{i}{n}$ covering the initial segment $P^{(m)}\!$, i.e. $\ccc{1}{m}\cup\ldots\cup\ccc{w}{m}=\set{p_1,\ldots,p_m}$.
\end{enumerate}
First we put $\cc{i}{0}=\emptyset$ for all $i\leq w$. Suppose that $n\geq 0$ and $\cc{1}{n},\ldots,\cc{w}{n}$ satisfy (1)--(3). From (3) we know that there are infinitely many $m$'s such that $\set{p_1,\ldots,p_n,p_{n+1},\ldots,p_m}$ can be partitioned into $\ccc{1}{m},\ldots,\ccc{w}{m}$ that extend $\cc{1}{n},\ldots,\cc{w}{n}$, respectively. Because for any fixed $m$ there are only finitely many (exactly $w$) of the $\ccc{i}{m}$'s, there must be at least one $l=1,\ldots,w$ such that the point $p_{n+1}$ has to occur in infinitely many extensions $\ccc{l}{m}$ of $\cc{l}{n}$. After fixing one of such $l$'s we extend $\cc{i}{n}$'s to $\cc{i}{n+1}$'s in the following way.
\begin{equation*}
\cc{i}{n+1}=
\begin{cases}
\cc{i}{n}, &\text{if}\ \ i\neq l,\\
\cc{l}{n}\cup\set{p_{n+1}}, &\text{if}\ \ i=l.
\end{cases}
\end{equation*}
Of course the new chains $\cc{1}{n+1}\!,\ldots,\cc{w}{n+1}$ satisfy conditions \mbox{(1)--(3)}.

Now it is easy to see that $\bigcup_{n\in\N}\cc{1}{n}\!,\ldots,\bigcup_{n\in\N}\cc{w}{n}$ form a covering of $P$.
\end{proof}

The proof of Theorem \ref{T:countouble partition} suggests how to choose a chain into which an incoming point can be incorporated. Unfortunately, this choice is not effective and can be hardly computed without the knowledge of entire poset $\poset{P}$. We will build partial chain partitionings for consecutive initial segments using algorithms of the following form.

\begin{defn} An \emph{on-line chain partitioning algorithm} is an algorithm which, for all $i$, creates a chain partitioning $\cc{1}{i},\ldots,\cc{k}{i}$ of an initial segment $P^{(i)}$ of an on-line poset $\poset{P}^{\ll}$ by putting first $\cc{1}{1}=\set{p_1}$. 
Then the algorithm as an input gets a chain partition $\cc{1}{n},\ldots,\cc{k}{n}$ of $P^{(n)}=\set{p_1,\ldots,p_n}$ and an extension of $\poset{P}^{(n)}$ to $\poset{P}^{(n+1)}$ to return a chain partitioning $\cc{1}{n+1},\ldots,\cc{k'}{n+1}$ of $P^{(n+1)}=\set{p_1,\ldots,p_n,p_{n+1}}$ which expands $\cc{1}{n},\ldots,\cc{k}{n}$.
This can be done by one of the following two ways:
\begin{itemize}
\item by adding $p_{n+1}$ to some chain $\cc{j}{n}$, i.e. $\cc{j}{n+1}=\cc{j}{n}\cup\set{p_{n+1}}$,
\item by creating a new chain $\cc{k+1}{n+1}=\set{p_{n+1}}$,
\end{itemize}
while other chains remain unchanged.
\end{defn} 

\begin{defn}
The \emph{value of the on-line chain partitioning problem}, $\fCP{w}$, is the least integer $k$, such that there is an on-line algorithm that never uses more than $k$ chains on posets of width at most $w$.
\end{defn}

From the technical point of view it is more convenient to look at on-line chain partitioning through \emph{on-line coloring}. An \emph{on-line coloring algorithm} recursively builds colorings $\chains^{_{(n)}}:P^{(n)}\tto\Gamma$ of consecutive initial segments $P^{\brackets{n}}\!$, where $\Gamma$ is a countable fixed set of colors. The coloring being created determines a chain partitioning by putting into one chain the points that get the same color.

\bigskip

On-line chain partitioning problem for posets of width at most $w$ can be viewed as the following two-person game. We call the players \emph{Algorithm} and \emph{Spoiler}. During each round:
\begin{itemize}
\item Spoiler introduces a new point with its comparability status to the previously presented points,
\item Algorithm colors this new point.
\end{itemize}
The aim of Algorithm is to use minimal number of colors. The goal of Spoiler is to force Algorithm to use as many colors as possible.

\bigskip

An example of an on-line coloring algorithm is \emph{First-Fit Algorithm}. Roughly speaking it colors the incoming point $p$ by the oldest color $\gamma$ for which points colored by $\gamma$ still form a chain.
If there is no such color a new one is taken.
To be more precise we identify $\Gamma$ with $\N$ and we say that a number $k$ is available for $p_n\in P^{(n)}$ if, after assigning $k$ to $p_n$, all points from $P^{(n)}$ colored by $k$ form a chain.
Each $p_n$ is colored by First-Fit Algorithm by the smallest available number. Although First-Fit Algorithm works pretty fast and it seems to be good enough, it can be cheated already on posets of width $2$.
\begin{thm}\label{Thm:FFA}
There is an on-line poset 
\mbox{$\poset{P}^{\ll}=\brackets{P,\leq,\ll}$} of width $2$ for which First-Fit Algorithm has to use infinitely many colors.
\end{thm}

\begin{proof} 
We will construct $\poset{P}^{\ll}$ and show that First-Fit Algorithm uses all natural numbers, when coloring it.
We describe \mbox{$\poset{P}^{\ll}$} by two disjoint chains $P=\alpha_0 \cup \alpha_1$. We start with $\alpha_1= \set{x_1}$ and \mbox{$\alpha_0 = \emptyset$}. Then, for each $m \in \N$ we add a group $G_m$ of $m$ consecutive linearly ordered points on the top of either $\alpha_0$ or $\alpha_1$, depending on whether $m$ is even or odd.
These new points, say $y_1, \ldots , y_m$, satisfy %
%\begin{align*}
%&y_m \ll \ldots \ll y_1,\\
%&y_m \prec \ldots \prec y_1,\\
%&y_k > G_1 \cup \ldots \cup G_{m-2} &\textrm{for}\ k=1, \ldots, m,\\
%&y_k \succ x_k &\textrm{for}\ k=1, \ldots, m-1,\\
%&y_k \parallel x_{k-1}, \ldots , x_1 &\textrm{for}\ k=1, \ldots, m,
%\end{align*}
%\note{
\mbox{$y_m\ll \ldots \ll y_1$} and
\begin{small}\begin{align*}
&y_k\ >\ G_1 \cup \ldots \cup G_{m-2}\ \cup\ \set{x_{m-1},x_{m-2},\ldots,x_k}\ \cup\ \set{y_m,y_{m-1},\ldots,y_{k+1}}\! ,\\
&y_k\ \, \parallel \, \ x_{k-1}, x_{k-2}, \ldots , x_1,
\end{align*}\end{small}%
for $k=1, \ldots, m$,
%}
where $x_{m-1} \prec \ldots \prec x_1$ is the group $G_{m-1}$ (added recently to the other chain). These relations are illustrated on \mbox{Figure\! \ref{F:FF}}.
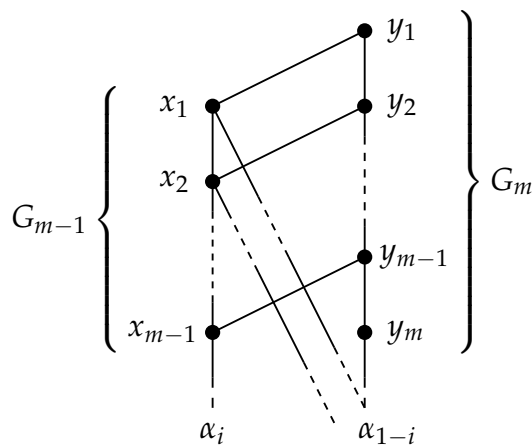
\begin{figure}[hbt]
\newcommand{\x}[1]{$x_{#1}$}
\newcommand{\y}[1]{$y_{#1}$}
\renewcommand{\a}{$\alpha_i$}
\renewcommand{\b}{$\alpha_{1-i}$}
\newcommand{\GG}{$\left.\begin{array}{c} \vspace{40mm} \end{array}\right\rbrace G_m$}
\newcommand{\GGG}{$G_{m-1} \left\lbrace\begin{array}{c}\vspace{30mm}\end{array}\right.$}
\centering\ifx\JPicScale\undefined\def\JPicScale{1}\fi
\psset{unit=\JPicScale mm}
\psset{linewidth=0.3,dotsep=1,hatchwidth=0.3,hatchsep=1.5,shadowsize=1,dimen=middle}
\psset{dotsize=0.7 2.5,dotscale=1 1,fillcolor=black}
\psset{arrowsize=1 2,arrowlength=1,arrowinset=0.25,tbarsize=0.7 5,bracketlength=0.15,rbracketlength=0.15}
\begin{pspicture}(0,0)(68,56)
\psline[linewidth=0.25](46,40)(46,54)
\rput[l](49,54){\y1}
\psline[linewidth=0.25](26,30)(26,44)
\psline[linewidth=0.25](46,54)(26,44)
\psline[linewidth=0.25](46,44)(26,34)
\psline[linewidth=0.25](46,24)(26,14)
\psline[linewidth=0.25](32,32)(26,44)
\psline[linewidth=0.25](28,30)(26,34)
\psline[linewidth=0.25](46,8)(46,28)
\psline[linewidth=0.25](26,8)(26,18)
\psline[linewidth=0.25,linestyle=dashed,dash=1 1](46,8)(46,4)
\psline[linewidth=0.25,linestyle=dashed,dash=1 1](26,8)(26,4)
\rput[l](49,44){\y2}
\rput[l](48,24){\y{m-1}}
\rput[l](49,14){\y{m}}
\rput[r](23,44){\x1}
\rput[r](23,34){\x2}
\rput[r](24,14){\x{m-1}}
\psline[linecolor=white](68,29)(68,25)
\psline[linecolor=white](0,18)(0,12)
\psline[linecolor=white](42,56)(36,56)
\rput[t](26,2){\a}
\rput[t](49,2){\b}
\rput{0}(46,14){\psellipse[linestyle=none,fillstyle=solid](0,0)(1,1)}
\rput{0}(26,14){\psellipse[linestyle=none,fillstyle=solid](0,0)(1,1)}
\rput{0}(26,34){\psellipse[linestyle=none,fillstyle=solid](0,0)(1,1)}
\rput{0}(46,44){\psellipse[linestyle=none,fillstyle=solid](0,0)(1,1)}
\rput{0}(26,44){\psellipse[linestyle=none,fillstyle=solid](0,0)(1,1)}
\rput{0}(46,54){\psellipse[linestyle=none,fillstyle=solid](0,0)(1,1)}
\psline[linewidth=0.25,linestyle=dashed,dash=1 1](46,40)(46,28)
\rput{0}(46,24){\psellipse[linestyle=none,fillstyle=solid](0,0)(1,1)}
\psline[linewidth=0.25,linestyle=dashed,dash=1 1](26,30)(26,18)
\psline[linewidth=0.25](42,12)(36,24)
\psline[linewidth=0.25](38,10)(32,22)
\psline[linewidth=0.25,linestyle=dashed,dash=1 1](32,32)(36,24)
\psline[linewidth=0.25,linestyle=dashed,dash=1 1](28,30)(32,22)
\psline[linewidth=0.25,linestyle=dashed,dash=1 1](42,12)(46,4)
\psline[linewidth=0.25,linestyle=dashed,dash=1 1](38,10)(42,2)
\rput[l](54,34){\GG}
\rput[r](18,29){\GGG}
\end{pspicture}
\caption{Top part of $G_1\cup\ldots\cup G_m$.}\label{F:FF}
\end{figure}

The inverse numbering in each $G_k$, consisting of \mbox{$u_k \prec \ldots \prec u_1$}, is intended to force the following behavior of First-Fit Algorithm:
\begin{equation}\label{Eq:InvFF}
 \fC{}{u_j}=j,\qquad\quad \textrm{for}\ \ j=1, \ldots, k.
\end{equation}
Indeed, for $m=1$ the group $G_1$ consists of one point $u_1=x_1$ which is obviously colored by $\fC{}{u_1}=1$. Now, if $G_m$ is $y_m \prec \ldots \prec y_1$ and $G_{m-1}$ is $x_{m-1} \prec \ldots \prec x_1$, then $y_j \parallel x_1, \ldots, x_{j-1}$ together with \ref{Eq:InvFF} says that colors $1, \ldots , k-1$ are not available for $y_j$. On the other hand $y_j$ is over all points in $G_1\cup \ldots \cup \big(G_{m-1}\setminus \set{x_1, \ldots , x_{j-1}}\!\big)$, so that at the moment $y_j$ has to be colored the color $j$ is available for $y_j$. Therefore $\fC{}{y_j}=\fMin{\N\setminus\fC{}{\set{x_1,\ldots,x_{j-1}}}}=j$, as required by our invariant (\ref{Eq:InvFF}).
\end{proof}

%\begin{thm}
%{\em The value of the on-line chain partitioning problem for posets of width at most $w$} is the largest integer $\fCP{w}$ so that there is a strategy of presenting points that forces any algorithm to use at least $\fCP{w}$ chains.
%\end{thm}
%\begin{proof}Included.\end{proof}

\bigskip

Despite of the total failure of First-Fit Algorithm there is a strategy for Algorithm that uses bounded (in terms of $w$) number of chains to cover on-line poset of width at most $w$. Such a strategy was presented already in 1981 by Kierstead, and no better one is known in general. On the other hand an argument of Endre \mbox{Szemer\'{e}di} (presented in \cite{Kierstead86}) proves the best known asymptotic lower bound.

\begin{thm}[Szemer\'{e}di \cite{Kierstead86}; Kierstead \cite{Kierstead}]\label{Th:Kie Sz}
\[
\binom{w+1}{2}\, \leq\, \fCP{w}\, \leq\, \frac{5^w-1}{4}.
\]
\end{thm}

Kierstead \cite{Kierstead} presented also a lower bound $4w-3\leq\fCP{w}$ which is better than $\binom{w+1}{2}$ for the first few values $w=2,3,4,5$.
In particular, this lower bound together with an algorithm of Felsner \cite{Felsner} gives the precise value for $\fCP{2}$.

\begin{thm}[Kierstead \cite{Kierstead}; Felsner \cite{Felsner}]
\[
\fCP{2}\ =\ 5.
\]
\end{thm}

The research carried out up to now puts $\fCP{3}$ between $9=\linebreak 4\cdot 3 -3$ and \mbox{$31=(5^3-1)/4$}. One aim of this chapter is to show a better upper bound:
\[
\fCP{3}\ \leq\ 16.
\]

%------------------------------------------------

\section{Reduction to Local Problem}\label{S:Main Reduction} %\subsection{Introduction} \label{SS:Intro}
In this section we reduce on-line chain partitioning problem to a special setting in which the number of points, that have some influence on the coloring, is bounded by $3w$, where $w$ is the width of the poset. Such localization will be done in three steps.

To simply our further analysis we assume that Algorithm, instead of partition the poset into disjoint chains, is allowed to cover this poset by a family of chains which are not necessarily disjoint. Obviously, to get a partition, one can simply choose one chain for each point. Our small, inessential change, can be expressed by multicoloring of a poset. In such a multicoloring each point can obtain several colors (but at least one) such that points colored by the same color form a chain. The reason for which we prefer multicoloring over coloring is that this new approach helps Algorithm remember more information for the future and keeps track about the past.

The main idea is that the multicoloring of a new point depends only on at most $3w$ points, where $w$ is the width of the poset. Loosely speaking, Algorithm chooses colors for the new point by analyzing ``the nearest neighborhood'' of the new point. In this way Algorithm plays new game, which is named a local game. After finding a strategy for Algorithm in this local game, we will reach our goal, provided the number of colors used in the created multicoloring is small enough and provided we can show that this local multicoloring can actually depend only on this local neighborhood of the incoming point. This has to be done in a very careful way, as otherwise the points colored by the same color would not form a chain. Subsection \ref{SS:main} includes the definitions that formalize presented intuitions.

Subsections \ref{SS:different} and \ref{SS:core} present two next reductions which finally lead to the following local game played by Spoiler and Algorithm on a structure $(L, T, \leq, \chains)$, where:
\begin{itemize}
\item $\brackets{L, T,\leq}$ is a regular bipartite poset, i.e.
\begin{itemize}
\item $(L\cup T,\leq)$ is a poset,
\item $L,T$ are disjoint antichains with $L\al T$,
\item $\abs{L}=\abs{T}=\fWidth{L\cup T}$;
\end{itemize}
\item moreover $\brackets{L, T,\leq}$ is a core, i.e. the set of edges of digraph $\brackets{L, T,\prec}$ is a sum of all perfect matchings between $L$ and $T$;
\item $\chains:L\cup T\tto \powerp{\Gamma}$ is the multicoloring, where $\Gamma$ is a finite set of colors fixed by Algorithm already in the first round.
\end{itemize}
In the following for a multicoloring $\chains:P\tto \powerp{\Gamma}$ and $X\subseteq P$ we will denote the set $\bigcup_{x\in X}\fC{}{x}$ simply by $\fC{}{X}$.

During the first round:
\begin{itemize}
\item Spoiler sets a natural number $w$ and then introduces two antichains $L,T\!$, each with $w$ elements such, that $L<T$.
\item Algorithm determines a finite set $\Gamma$ of colors that may be used in the entire game and then he colors each point $x\in L\cup T$ with some nonempty subset $\fC{}{x}$ of $\Gamma$ such that for each $\gamma\in\Gamma$ the points colored by $\gamma$ form a chain.
\end{itemize}
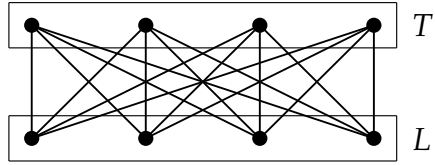
\begin{figure}[hbt]
%\begin{scriptsize}
\centering\ifx\JPicScale\undefined\def\JPicScale{1}\fi
\psset{unit=\JPicScale mm}
\psset{linewidth=0.3,dotsep=1,hatchwidth=0.3,hatchsep=1.5,shadowsize=1,dimen=middle}
\psset{dotsize=0.7 2.5,dotscale=1 1,fillcolor=black}
\psset{arrowsize=1 2,arrowlength=1,arrowinset=0.25,tbarsize=0.7 5,bracketlength=0.15,rbracketlength=0.15}
\begin{pspicture}(0,0)(53,22)
\psline[linewidth=0.25,fillcolor=white,fillstyle=solid](3,19)(3,4)
\psline[linewidth=0.25,fillcolor=white,fillstyle=solid](18,19)(18,4)
\psline[linewidth=0.25,fillcolor=white,fillstyle=solid](33,19)(33,4)
\psline[linewidth=0.25,fillcolor=white,fillstyle=solid](3,19)(18,4)
\psline[linewidth=0.25,fillcolor=white,fillstyle=solid](3,19)(33,4)
\psline[linewidth=0.25,fillcolor=white,fillstyle=solid](18,19)(3,4)
\psline[linewidth=0.25,fillcolor=white,fillstyle=solid](33,19)(3,4)
\psline[linewidth=0.25,fillcolor=white,fillstyle=solid](18,19)(33,4)
\psline[linewidth=0.25,fillcolor=white,fillstyle=solid](33,19)(18,4)
\rput[l](53,19){$T$}
\rput[l](53,4){$L$}
\psline[linewidth=0.25,fillcolor=white,fillstyle=solid](3,19)(48,4)
\psline[linewidth=0.25,fillcolor=white,fillstyle=solid](33,4)(48,19)
\psline[linewidth=0.25,fillcolor=white,fillstyle=solid](33,19)(48,4)
\psline[linewidth=0.25,fillcolor=white,fillstyle=solid](48,4)(48,19)
\psline[linewidth=0.15](51,1)(51,7)
\psline[linewidth=0.25,fillcolor=white,fillstyle=solid](48,19)(3,4)
\psline[linewidth=0.25,fillcolor=white,fillstyle=solid](18,19)(48,4)
\psline[linewidth=0.25,fillcolor=white,fillstyle=solid](48,19)(18,4)
\psline[linewidth=0.15](51,22)(0,22)
\psline[linewidth=0.15](51,16)(0,16)
\psline[linewidth=0.15](51,7)(0,7)
\psline[linewidth=0.15](51,1)(0,1)
\psline[linewidth=0.15](0,22)(0,16)
\psline[linewidth=0.15](0,7)(0,1)
\psline[linewidth=0.15](51,16)(51,22)
\rput{0}(3,19){\psellipse[linewidth=0.25,linestyle=none,fillstyle=solid](0,0)(1,1)}
\rput{0}(18,19){\psellipse[linewidth=0.25,linestyle=none,fillstyle=solid](0,0)(1,1)}
\rput{0}(33,19){\psellipse[linewidth=0.25,linestyle=none,fillstyle=solid](0,0)(1,1)}
\rput{0}(48,19){\psellipse[linewidth=0.25,linestyle=none,fillstyle=solid](0,0)(1,1)}
\rput{0}(3,4){\psellipse[linewidth=0.25,linestyle=none,fillstyle=solid](0,0)(1,1)}
\rput{0}(18,4){\psellipse[linewidth=0.25,linestyle=none,fillstyle=solid](0,0)(1,1)}
\rput{0}(33,4){\psellipse[linewidth=0.25,linestyle=none,fillstyle=solid](0,0)(1,1)}
\rput{0}(48,4){\psellipse[linewidth=0.25,linestyle=none,fillstyle=solid](0,0)(1,1)}
\end{pspicture}
%\end{scriptsize}
\caption{Poset introduced by Spoiler for $w=4$.}\label{F:LT}
\end{figure} 
During next rounds the structure $\brackets{L,T,\leq,\chains}$ from the previous round is transformed to a structure $\brackets{L^{\plus},T^{\plus},\leq,\chains^{\plus}}$ according to the following rules:
\begin{itemize}
\item Spoiler introduces $w$ new elements that form an antichain $M$ such that:
\begin{itemize}
\item the poset $\poset{B}'=\brackets{L\cup M\cup T,\leq}$ has width $w$,
\item $L\al M\al T$ in the lattice $\fMA{\poset{B}'}$,
\item both $\brackets{L, M,\leq}$ and $\brackets{M, T,\leq}$ are cores.
\end{itemize}
\item Algorithm colors each point $m\in M$ with a nonempty set of colors $\fC{\plus}{m}\subseteq\fC{}{L}\cap\fC{}{T}$ and keeps the old multicoloring on $L\cup T$, i.e. $\rest{\chains^{\plus}}{L\cup T}=\chains$.
In fact, in order for $\chains^{\plus}$ to be multicoloring, $\fC{\plus}{m}$ has to be a subset of $\fC{}{m\downseto}\cap\fC{}{m\upseto}$.
\item Finally, Spoiler redefines the levels $L, T$ to $L^{\plus}, T^{\plus}$ so that either $\brackets{L^{\plus},T^{\plus}}=\brackets{L,M}$ or $\brackets{L^{\plus},T^{\plus}}=\brackets{M,T}$.
This, after restricting to $L^{\plus}\cup T^{\plus}\!$, creates the new structure $\brackets{L^{\plus},T^{\plus},\leq,\chains^{\plus}}$.
\end{itemize}
Figure \ref{F:Example Local Disjoint Core} presents an example of Spoiler's move, for $w=4$, in some further round. First Spoiler introduces an antichain $M$ of $4$ new points. After Algorithm's move Spoiler chooses $(M,T)$ for $(L^{\plus},T^{\plus})$.
\begin{figure}[hbt]
%\begin{smnall}
\centering\ifx\JPicScale\undefined\def\JPicScale{1}\fi
\psset{unit=\JPicScale mm}
\psset{linewidth=0.3,dotsep=1,hatchwidth=0.3,hatchsep=1.5,shadowsize=1,dimen=middle}
\psset{dotsize=0.7 2.5,dotscale=1 1,fillcolor=black}
\psset{arrowsize=1 2,arrowlength=1,arrowinset=0.25,tbarsize=0.7 5,bracketlength=0.15,rbracketlength=0.15}
\begin{pspicture}(0,0)(119,30)
\psline[linewidth=0.25,fillcolor=white,fillstyle=solid](2,23)(2,13)
\psline[linewidth=0.25,fillcolor=white,fillstyle=solid](22,23)(22,13)
\psline[linewidth=0.25,fillcolor=white,fillstyle=solid](22,23)(32,13)
\psline[linewidth=0.25,fillcolor=white,fillstyle=solid](12,23)(32,13)
\psline[linewidth=0.25,fillcolor=white,fillstyle=solid](12,13)(12,23)
\psline[linewidth=0.25,fillcolor=white,fillstyle=solid](22,23)(12,13)
\psline[linewidth=0.25,fillcolor=white,fillstyle=solid](32,23)(32,13)
\rput[l](35,23){$T$}
\rput[l](35,13){$L$}
\psline[linewidth=0.25,fillcolor=white,fillstyle=solid](22,13)(12,23)
\psline[linewidth=0.25,fillcolor=white,fillstyle=solid](32,23)(22,13)
\psline[linewidth=0.1](34,15)(0,15)
\psline[linewidth=0.1](34,11)(0,11)
\psline[linewidth=0.1](0,15)(0,11)
\rput{0}(2,23){\psellipse[linewidth=0.25,linestyle=none,fillstyle=solid](0,0)(1,1)}
\rput{0}(12,13){\psellipse[linewidth=0.25,linestyle=none,fillstyle=solid](0,0)(1,1)}
\rput{0}(22,23){\psellipse[linewidth=0.25,linestyle=none,fillstyle=solid](0,0)(1,1)}
\rput{0}(32,23){\psellipse[linewidth=0.25,linestyle=none,fillstyle=solid](0,0)(1,1)}
\rput{0}(2,13){\psellipse[linewidth=0.25,linestyle=none,fillstyle=solid](0,0)(1,1)}
\rput{0}(12,23){\psellipse[linewidth=0.25,linestyle=none,fillstyle=solid](0,0)(1,1)}
\rput{0}(22,13){\psellipse[linewidth=0.25,linestyle=none,fillstyle=solid](0,0)(1,1)}
\rput{0}(32,13){\psellipse[linewidth=0.25,linestyle=none,fillstyle=solid](0,0)(1,1)}
\psline[linewidth=0.1](34,15)(34,11)
\psline[linewidth=0.1](34,21)(0,21)
\psline[linewidth=0.1](34,25)(0,25)
\psline[linewidth=0.1](34,25)(34,21)
\psline[linewidth=0.1](0,25)(0,21)
\psline[linewidth=0.25,fillcolor=white,fillstyle=solid](44,18)(44,8)
\psline[linewidth=0.25,fillcolor=white,fillstyle=solid](64,18)(64,8)
\psline[linewidth=0.25,fillcolor=white,fillstyle=solid](64,18)(74,8)
\psline[linewidth=0.25,fillcolor=white,fillstyle=solid](54,8)(54,18)
\psline[linewidth=0.25,fillcolor=white,fillstyle=solid](74,18)(74,8)
\rput[l](77,18){$M$}
\rput[l](77,8){$L$}
\psline[linewidth=0.25,fillcolor=white,fillstyle=solid](74,18)(64,8)
\psline[linewidth=0.1](76,10)(42,10)
\psline[linewidth=0.1](76,6)(42,6)
\psline[linewidth=0.1](42,10)(42,6)
\rput{0}(44,18){\psellipse[linewidth=0.25,linestyle=none,fillstyle=solid](0,0)(1,1)}
\rput{0}(54,8){\psellipse[linewidth=0.25,linestyle=none,fillstyle=solid](0,0)(1,1)}
\rput{0}(64,18){\psellipse[linewidth=0.25,linestyle=none,fillstyle=solid](0,0)(1,1)}
\rput{0}(74,18){\psellipse[linewidth=0.25,linestyle=none,fillstyle=solid](0,0)(1,1)}
\rput{0}(44,8){\psellipse[linewidth=0.25,linestyle=none,fillstyle=solid](0,0)(1,1)}
\rput{0}(54,18){\psellipse[linewidth=0.25,linestyle=none,fillstyle=solid](0,0)(1,1)}
\rput{0}(64,8){\psellipse[linewidth=0.25,linestyle=none,fillstyle=solid](0,0)(1,1)}
\rput{0}(74,8){\psellipse[linewidth=0.25,linestyle=none,fillstyle=solid](0,0)(1,1)}
\psline[linewidth=0.1](76,10)(76,6)
\psline[linewidth=0.1](76,16)(42,16)
\psline[linewidth=0.1](76,20)(42,20)
\psline[linewidth=0.1](76,20)(76,16)
\psline[linewidth=0.1](42,20)(42,16)
\psline[linewidth=0.25,fillcolor=white,fillstyle=solid](44,28)(44,18)
\psline[linewidth=0.25,fillcolor=white,fillstyle=solid](64,28)(74,18)
\psline[linewidth=0.25,fillcolor=white,fillstyle=solid](54,18)(54,28)
\psline[linewidth=0.25,fillcolor=white,fillstyle=solid](64,28)(54,18)
\psline[linewidth=0.25,fillcolor=white,fillstyle=solid](74,28)(74,18)
\rput[l](77,28){$T$}
\psline[linewidth=0.25,fillcolor=white,fillstyle=solid](64,18)(54,28)
\rput{0}(44,28){\psellipse[linewidth=0.25,linestyle=none,fillstyle=solid](0,0)(1,1)}
\rput{0}(64,28){\psellipse[linewidth=0.25,linestyle=none,fillstyle=solid](0,0)(1,1)}
\rput{0}(74,28){\psellipse[linewidth=0.25,linestyle=none,fillstyle=solid](0,0)(1,1)}
\rput{0}(54,28){\psellipse[linewidth=0.25,linestyle=none,fillstyle=solid](0,0)(1,1)}
\psline[linewidth=0.1](76,26)(42,26)
\psline[linewidth=0.1](76,30)(42,30)
\psline[linewidth=0.1](76,30)(76,26)
\psline[linewidth=0.1](42,30)(42,26)
\psline[linewidth=0.25,fillcolor=white,fillstyle=solid](64,18)(74,28)
\rput[l](119,18){$L^+$}
\rput{0}(86,18){\psellipse[linewidth=0.25,linestyle=none,fillstyle=solid](0,0)(1,1)}
\rput{0}(106,18){\psellipse[linewidth=0.25,linestyle=none,fillstyle=solid](0,0)(1,1)}
\rput{0}(116,18){\psellipse[linewidth=0.25,linestyle=none,fillstyle=solid](0,0)(1,1)}
\rput{0}(96,18){\psellipse[linewidth=0.25,linestyle=none,fillstyle=solid](0,0)(1,1)}
\psline[linewidth=0.1](118,16)(84,16)
\psline[linewidth=0.1](118,20)(84,20)
\psline[linewidth=0.1](118,20)(118,16)
\psline[linewidth=0.1](84,20)(84,16)
\psline[linewidth=0.25,fillcolor=white,fillstyle=solid](86,28)(86,18)
\psline[linewidth=0.25,fillcolor=white,fillstyle=solid](106,28)(116,18)
\psline[linewidth=0.25,fillcolor=white,fillstyle=solid](96,18)(96,28)
\psline[linewidth=0.25,fillcolor=white,fillstyle=solid](106,28)(96,18)
\psline[linewidth=0.25,fillcolor=white,fillstyle=solid](116,28)(116,18)
\rput[l](119,28){$T^+$}
\psline[linewidth=0.25,fillcolor=white,fillstyle=solid](106,18)(96,28)
\rput{0}(86,28){\psellipse[linewidth=0.25,linestyle=none,fillstyle=solid](0,0)(1,1)}
\rput{0}(106,28){\psellipse[linewidth=0.25,linestyle=none,fillstyle=solid](0,0)(1,1)}
\rput{0}(116,28){\psellipse[linewidth=0.25,linestyle=none,fillstyle=solid](0,0)(1,1)}
\rput{0}(96,28){\psellipse[linewidth=0.25,linestyle=none,fillstyle=solid](0,0)(1,1)}
\psline[linewidth=0.1](118,26)(84,26)
\psline[linewidth=0.1](118,30)(84,30)
\psline[linewidth=0.1](118,30)(118,26)
\psline[linewidth=0.1](84,30)(84,26)
\psline[linewidth=0.25,fillcolor=white,fillstyle=solid](106,18)(116,28)
\rput[B](38,0.5){\begin{footnotesize}putting an antichain $M$\end{footnotesize}}
\psline[linewidth=0.2,fillcolor=white,fillstyle=solid](19,-1)(57,-1)
\psline[linewidth=0.2,fillcolor=white,fillstyle=solid](56,0)(57.07,-1.07)
\psline[linewidth=0.2,fillcolor=white,fillstyle=solid](57.07,-0.93)(56,-2)
\rput(53,0){}
\rput[B](80,0.5){\begin{footnotesize}choosing new structure\end{footnotesize}}
\psline[linewidth=0.2,fillcolor=white,fillstyle=solid](61,-1)(99,-1)
\psline[linewidth=0.2,fillcolor=white,fillstyle=solid](98,0)(99.07,-1.07)
\psline[linewidth=0.2,fillcolor=white,fillstyle=solid](99.07,-0.94)(98,-2)
\rput(95,0){}
\end{pspicture}
%\end{small}
\caption{$ $}\label{F:Example Local Disjoint Core}
\end{figure}
 
The goal of Algorithm is to pick minimal number $\abs{\Gamma}$ of colors already during the first round so that he can play with these colors forever. If this number is too small Spoiler can accomplish his goal i.e. produce $\poset{B}'$ that cannot be colored by $\Gamma$ according to the described rules.
The smallest possible number of colors for width $w$ with which Algorithm can play forever will be denoted by $\fLCP{w}$.

Note that in the first round Algorithm actually sets an upper bound for the number of colors, say $k$, used later. This game is named local because the size of information necessary to describe the structure $\brackets{L,T,\leq,\chains}$ after every round is always bounded by $w, w, w^2,2^{2kw}$, respectively. This obviously stays in contrast to the continuously increasing size of posets in on-line chain partitioning problem.

Our effort in reducing the on-line chain partitioning problem
to the local game will pay off, as we will show in Corollary \ref{R:CP sum LCP} that $\fCP{w}\leq\fLCP{1} +\ldots+\fLCP{w}$. This in turn, allows to determine $\fCP{w}$ by considering only local games on posets of width at most $w$. Unfortunately, at present, we are able to determine $\fLCP{w}$ only for $w\leq 3$. This will be done in Section \ref{S:w<=3}, in a way that all known bounds for $\fCP{3}$ are substantially improved.

\subsection{Localizing the game} \label{SS:main} In this subsection we reduce the on-line chain partitioning problem to a more general version of a local game than the one described before. In this more general version we relax the assumptions that $L\cap T=\emptyset$ and that $\brackets{L, T,\leq}$ is a core. This allows Spoiler to perform a more relaxed move to present an antichain $M$ which is not necessarily disjoint with $L\cup T$ and does not have to form cores \mbox{$\brackets{L, M,\leq}$}, \mbox{$\brackets{M, T,\leq}$}. In particular Spoiler may actually introduce less than $w$ completely new points. However he has to discover the ``middle'' antichain $M$ of exactly $w$ points. For example, let $L$ and $T$ be like on Figure \ref{F:Example Local}.
\begin{figure}[hbt]
%\begin{smnall}
\centering\ifx\JPicScale\undefined\def\JPicScale{1}\fi
\psset{unit=\JPicScale mm}
\psset{linewidth=0.3,dotsep=1,hatchwidth=0.3,hatchsep=1.5,shadowsize=1,dimen=middle}
\psset{dotsize=0.7 2.5,dotscale=1 1,fillcolor=black}
\psset{arrowsize=1 2,arrowlength=1,arrowinset=0.25,tbarsize=0.7 5,bracketlength=0.15,rbracketlength=0.15}
\begin{pspicture}(0,0)(121.5,30)
\psline[linewidth=0.25,fillcolor=white,fillstyle=solid](22,23)(22,13)
\psline[linewidth=0.25,fillcolor=white,fillstyle=solid](22,23)(32,13)
\psline[linewidth=0.25,fillcolor=white,fillstyle=solid](12,23)(32,13)
\psline[linewidth=0.25,fillcolor=white,fillstyle=solid](12,13)(12,23)
\psline[linewidth=0.25,fillcolor=white,fillstyle=solid](22,23)(12,13)
\psline[linewidth=0.25,fillcolor=white,fillstyle=solid](32,23)(32,13)
\rput[l](35,23){$T$}
\rput[l](35.5,13){$L$}
\psline[linewidth=0.25,fillcolor=white,fillstyle=solid](22,13)(12,23)
\psline[linewidth=0.25,fillcolor=white,fillstyle=solid](32,23)(22,13)
\psline[linewidth=0.1](34.5,15.5)(10.5,15.5)
\psline[linewidth=0.1](34.5,10.5)(9.5,10.5)
\psline[linewidth=0.1](-0.5,20.5)(-0.5,15.5)
\rput{0}(12,13){\psellipse[linewidth=0.25,linestyle=none,fillstyle=solid](0,0)(1,1)}
\rput{0}(22,23){\psellipse[linewidth=0.25,linestyle=none,fillstyle=solid](0,0)(1,1)}
\rput{0}(32,23){\psellipse[linewidth=0.25,linestyle=none,fillstyle=solid](0,0)(1,1)}
\rput{0}(2,18){\psellipse[linewidth=0.25,linestyle=none,fillstyle=solid](0,0)(1,1)}
\rput{0}(12,23){\psellipse[linewidth=0.25,linestyle=none,fillstyle=solid](0,0)(1,1)}
\rput{0}(22,13){\psellipse[linewidth=0.25,linestyle=none,fillstyle=solid](0,0)(1,1)}
\rput{0}(32,13){\psellipse[linewidth=0.25,linestyle=none,fillstyle=solid](0,0)(1,1)}
\psline[linewidth=0.1](34.5,15.5)(34.5,10.5)
\psline[linewidth=0.1](34,21)(10,21)
\psline[linewidth=0.1](34,25)(10,25)
\psline[linewidth=0.1](34,25)(34,21)
\psline[linewidth=0.1](0,20)(0,16)
\psline[linewidth=0.25,fillcolor=white,fillstyle=solid](64,18)(64,8)
\psline[linewidth=0.25,fillcolor=white,fillstyle=solid](64,18)(74,8)
\psline[linewidth=0.25,fillcolor=white,fillstyle=solid](74,18)(74,8)
\rput[l](82,15.5){$M$}
\rput[l](81,8){$L$}
\psline[linewidth=0.25,fillcolor=white,fillstyle=solid](74,18)(64,8)
\psline[linewidth=0.1](79.5,11)(63,11)
\psline[linewidth=0.1](79.5,5)(61,5)
\psline[linewidth=0.1](41,21)(41,15)
\rput{0}(44,18){\psellipse[linewidth=0.25,linestyle=none,fillstyle=solid](0,0)(1,1)}
\rput{0}(64,18){\psellipse[linewidth=0.25,linestyle=none,fillstyle=solid](0,0)(1,1)}
\rput{0}(74,18){\psellipse[linewidth=0.25,linestyle=none,fillstyle=solid](0,0)(1,1)}
\rput{0}(54,13){\psellipse[linewidth=0.25,linestyle=none,fillstyle=solid](0,0)(1,1)}
\rput{0}(64,8){\psellipse[linewidth=0.25,linestyle=none,fillstyle=solid](0,0)(1,1)}
\rput{0}(74,8){\psellipse[linewidth=0.25,linestyle=none,fillstyle=solid](0,0)(1,1)}
\psline[linewidth=0.1](79.5,11)(79.5,5)
\psline[linewidth=0.1](81,15.5)(62.5,15.5)
\psline[linewidth=0.1](81,20.5)(62,20.5)
\psline[linewidth=0.1](81,20.5)(81,15.5)
\psline[linewidth=0.1](41.5,20.5)(41.5,15.5)
\psline[linewidth=0.25,fillcolor=white,fillstyle=solid](64,28)(74,18)
\psline[linewidth=0.25,fillcolor=white,fillstyle=solid](54,13)(54,28)
\psline[linewidth=0.25,fillcolor=white,fillstyle=solid](64,28)(54,13)
\psline[linewidth=0.25,fillcolor=white,fillstyle=solid](74,28)(74,18)
\rput[l](80,28){$T$}
\psline[linewidth=0.25,fillcolor=white,fillstyle=solid](64,18)(54,28)
\rput{0}(64,28){\psellipse[linewidth=0.25,linestyle=none,fillstyle=solid](0,0)(1,1)}
\rput{0}(74,28){\psellipse[linewidth=0.25,linestyle=none,fillstyle=solid](0,0)(1,1)}
\rput{0}(54,28){\psellipse[linewidth=0.25,linestyle=none,fillstyle=solid](0,0)(1,1)}
\psline[linewidth=0.1](78.5,26)(52,26)
\psline[linewidth=0.1](78.5,30)(52,30)
\psline[linewidth=0.1](78.5,30)(78.5,26)
\psline[linewidth=0.1](42,20)(42,16)
\psline[linewidth=0.25,fillcolor=white,fillstyle=solid](64,18)(64,28)
\rput[B](38,0.5){\begin{footnotesize}putting points $m_1, m_2$\end{footnotesize}}
\psline[linewidth=0.2,fillcolor=white,fillstyle=solid](19,-1)(57,-1)
\psline[linewidth=0.2,fillcolor=white,fillstyle=solid](56,0)(57.07,-1.07)
\psline[linewidth=0.2,fillcolor=white,fillstyle=solid](57.07,-0.93)(56,-2)
\rput(53,0){}
\rput[B](81,0.5){\begin{footnotesize}choosing new structure\end{footnotesize}}
\rput(97,0){}
\psline[linewidth=0.1](0,16)(4,16)
\psline[linewidth=0.1](0,20)(4,20)
\psline[linewidth=0.1](4,20)(10,25)
\psline[linewidth=0.1](4,16)(10,21)
\psline[linewidth=0.1](3.5,15.5)(9.5,10.5)
\psline[linewidth=0.1](3.5,15.5)(-0.5,15.5)
\psline[linewidth=0.1](-0.5,20.5)(4.5,20.5)
\psline[linewidth=0.1](10.5,15.5)(4.5,20.5)
\psline[linewidth=0.1](56.5,10.5)(51,10.5)
\psline[linewidth=0.1](45.5,15.5)(41.5,15.5)
\psline[linewidth=0.1](56,15.5)(52.5,15.5)
\psline[linewidth=0.1](46.5,20.5)(41.5,20.5)
\psline[linewidth=0.1](46.5,20.5)(52.5,15.5)
\psline[linewidth=0.1](45.5,15.5)(51,10.5)
\psline[linewidth=0.1](56,15.5)(62,20.5)
\psline[linewidth=0.1](56.5,10.5)(62.5,15.5)
\psline[linewidth=0.1](46,20)(52,30)
\psline[linewidth=0.1](46,16)(52,26)
\psline[linewidth=0.1](46,20)(42,20)
\psline[linewidth=0.1](46,16)(42,16)
\psline[linewidth=0.1](41,15)(45,15)
\psline[linewidth=0.1](45,15)(50.5,10)
\psline[linewidth=0.1](41,21)(47,21)
\psline[linewidth=0.1](50.5,10)(55,10)
\psline[linewidth=0.1](55,10)(61,5)
\psline[linewidth=0.1](47,21)(53,16)
\psline[linewidth=0.1](53,16)(57,16)
\psline[linewidth=0.1](57,16)(63,11)
\psline[linewidth=0.25,fillcolor=white,fillstyle=solid](108,28)(118,18)
\psline[linewidth=0.25,fillcolor=white,fillstyle=solid](98,18)(98,28)
\psline[linewidth=0.25,fillcolor=white,fillstyle=solid](108,28)(98,18)
\psline[linewidth=0.25,fillcolor=white,fillstyle=solid](118,28)(118,18)
\rput[l](121,28){$T^+$}
\rput[l](121.5,18){$L^+$}
\psline[linewidth=0.25,fillcolor=white,fillstyle=solid](108,18)(98,28)
\psline[linewidth=0.25,fillcolor=white,fillstyle=solid](108,28)(108,18)
\psline[linewidth=0.1](120.5,20.5)(96.5,20.5)
\psline[linewidth=0.1](120.5,15.5)(95.5,15.5)
\psline[linewidth=0.1](85.5,25.5)(85.5,20.5)
\rput{0}(98,18){\psellipse[linewidth=0.25,linestyle=none,fillstyle=solid](0,0)(1,1)}
\rput{0}(108,28){\psellipse[linewidth=0.25,linestyle=none,fillstyle=solid](0,0)(1,1)}
\rput{0}(118,28){\psellipse[linewidth=0.25,linestyle=none,fillstyle=solid](0,0)(1,1)}
\rput{0}(88,23){\psellipse[linewidth=0.25,linestyle=none,fillstyle=solid](0,0)(1,1)}
\rput{0}(98,28){\psellipse[linewidth=0.25,linestyle=none,fillstyle=solid](0,0)(1,1)}
\rput{0}(108,18){\psellipse[linewidth=0.25,linestyle=none,fillstyle=solid](0,0)(1,1)}
\rput{0}(118,18){\psellipse[linewidth=0.25,linestyle=none,fillstyle=solid](0,0)(1,1)}
\psline[linewidth=0.1](120.5,20.5)(120.5,15.5)
\psline[linewidth=0.1](120,26)(96,26)
\psline[linewidth=0.1](120,30)(96,30)
\psline[linewidth=0.1](120,30)(120,26)
\psline[linewidth=0.1](86,25)(86,21)
\psline[linewidth=0.1](86,21)(90,21)
\psline[linewidth=0.1](86,25)(90,25)
\psline[linewidth=0.1](90,25)(96,30)
\psline[linewidth=0.1](90,21)(96,26)
\psline[linewidth=0.1](89.5,20.5)(95.5,15.5)
\psline[linewidth=0.1](89.5,20.5)(85.5,20.5)
\psline[linewidth=0.1](85.5,25.5)(90.5,25.5)
\psline[linewidth=0.1](96.5,20.5)(90.5,25.5)
\rput(68,17.5){\begin{small}$m_1$\end{small}}
\rput(78,17.5){\begin{small}$m_2$\end{small}}
\psline[linewidth=0.2,fillcolor=white,fillstyle=solid](61,-1)(101,-1)
\psline[linewidth=0.2,fillcolor=white,fillstyle=solid](100,0)(101.07,-1.07)
\psline[linewidth=0.2,fillcolor=white,fillstyle=solid](101.07,-0.93)(100,-2)
\end{pspicture}
%\end{small}
\caption{$ $}\label{F:Example Local}
\end{figure} 
In spite of the fact that the poset $\brackets{L\cup T,\leq}$ has width $4$, Spoiler puts only $2$ new points $m_1$, $m_2$. However Spoiler indicates a $4$-element antichain $M$ which contains $m_1$ and $m_2$. After Algorithm colors $m_1$ and $m_2$ Spoiler chooses $(M,T)$ for $(L^{\plus},T^{\plus})$.

\begin{defn} \label{D:local game}
By a \emph{local on-line coloring game} we mean the following two-person game between Spoiler and Algorithm. During the first round:
\begin{itemize}
\item Spoiler sets a natural number $w$ and then introduces two antichains $L,T$, each with $w$ elements, such that $L<T$.
\item Algorithm determines a finite set $\Gamma$ of colors that may be used in the entire game and then he colors each point $x\in L\cup T$ with some nonempty subset $\fC{}{x}$ of $\Gamma$ such that for each $\gamma\in\Gamma$ the points colored by $\gamma$ form a chain.
\end{itemize}
%\begin{figure}[hbt]
%%\begin{scriptsize}
%\centering\input{chapter-02-pict-03}
%%\end{scriptsize}
%\caption{Poset introduced by Spoiler for width $4$.}\label{F:LT}
%\end{figure} 
The result of the first round (and as we will see of each other round) is a structure $\brackets{L,T,\leq,\chains}$, where:
\begin{itemize}
\item the poset $\brackets{L\cup T,\leq}$ has width $w$,
\item $L,T$ are two antichains of size $w$ such that $L\al T$,
\item $\chains:L\cup T\tto \powerp{\Gamma}$ is the multicoloring, i.e. for each color $\gamma\in\Gamma$ set of the form $C_{\gamma}:=\set{p\in L\cup T:\gamma\in\fC{}{p}}$ is a chain.
\end{itemize}
The structure $\brackets{L,T,\leq,\chains}$ introduced in the first round is called the \emph{initial board}.

During next rounds a \emph{board} $\brackets{L,T,\leq,\chains}$ from the previous round is transformed to a board $\brackets{L^{\plus},T^{\plus},\leq,\chains^{\plus}}$ according to the following rules:
\begin{itemize}
\item Spoiler introduces $v\leq w$ new points $m_1, m_2,\ldots, m_v$ and reveals a $w$-element subset $M$ of the already created poset $\poset{B}'=\brackets{L\cup T\cup\set{m_1,\ldots,m_v},\leq}$ so that:
\begin{itemize}
\item $\fWidth{\poset{B}'}=w$,
\item the revealed subset $M$ satisfies
\begin{itemize}
\item $m_1,\ldots,m_v\in M$,
\item $M$ is a maximum antichain in $\poset{B}'$, i.e. $M\in\fMA{\poset{B}'}$,
\item $L\al M\al T$ in the lattice $\fMA{\poset{B}'}$,
%\item $\fWidth{M}=w$.
\end{itemize}
\end{itemize}
\item Algorithm colors each point $m_i$ with a nonempty set of colors $\fC{\plus}{m_i}\subseteq\fC{}{L}\cap\fC{}{T}$ and keeps the old multicoloring on $L\cup T$, i.e. $\rest{\chains^{\plus}}{L\cup T}=\chains$.
\item Finally, Spoiler redefines the levels $L,T$ to $L^{\plus},T^{\plus}$ so that either $\brackets{L^{\plus},T^{\plus}}=\brackets{L,M}$ or $\brackets{L^{\plus},T^{\plus}}=\brackets{M,T}$.
This, after restricting to $L^{\plus}\cup T^{\plus}\!$, creates the new board $\brackets{L^{\plus},T^{\plus},\leq,\chains^{\plus}}$.
\end{itemize}
Once the width $w$ is set by Spoiler in the first round we refer to the rest of this game by \emph{$w$-local on-line coloring game}.

The goal of Algorithm is to pick minimal number $\abs{\Gamma}$ of colors already during the first round so that he can play with these colors forever.
\end{defn}

\begin{defn}
The \emph{value $\fLCP{w}$ of the $w$-local on-line coloring game} for posets (of width $w$) is the least integer $k$, such that there is strategy for Algorithm to play with $k$ colors forever.
\end{defn}

Our motivation for introducing local on-line coloring games lies in the following Theorem.

%\clearpage

\begin{thm}\label{L:local} 
Suppose that for each $v=1,\ldots,w$ there is an algorithm that can play forever in $v$-local on-line coloring game with $\fVal{v}$ colors. Then this algorithm can be used to build an on-line algorithm for chain partitioning using at most $\fVal{1}+\fVal{2}+\ldots+\fVal{w}$ chains.
\end{thm}

\begin{proof} 
To prove the theorem we built an algorithm \texttt{coloring}, which gets as an input a poset $\poset{P}^{\plus}=\brackets{P\cup\set{x},\leq}$ with coloring $\chains$ of $\poset{P}=\brackets{P,\leq}$ and returns coloring $\chains^{\plus}$ of $\poset{P}^{\plus}$ which expands $\chains$ by coloring $x$. \dd

The algorithm to be built will refer to the $v$-local algorithms secured in the assumption of the theorem. These local algorithms consist of two procedures: \ilc{} and \lc{}. \dd
The procedure \ilc{} returns multicoloring of the poset $\brackets{L,T,\leq}$ which Spoiler created in the first round, whereas \lc{} describes Algorithm's responses during next rounds i.e., \lc{}$\brackets{L,M,T,\leq}$ is a multicoloring of new points from the set $M\setminus\brackets{L\cup T}$. \dd
The local algorithm, when coloring poset of width $v$, uses a set $\Gamma_v$ of $\fVal{v}$ colors. \dd 
The sets $\Gamma_v$'s are supposed to be pairwise disjoint. \dd 

The local algorithm colors points $m_1,\ldots m_v$ localized in the middle of the poset $\poset{B}'$ and our coloring $\chains^{\plus}$ of $\poset{P}^{\plus}$ will refer to this multicoloring. \dd This reference can be done much simpler if the incoming point $x$ is neither minimal nor maximal in $\poset{P}^{\plus}$. \dd 
Obviously we can not assure such behavior of Spoiler, but instead we can artificially expand $P$ to $P\cup\bot\cup\top$ by new antichains $\bot=\set{b_1,\ldots,b_w}$ and $\top=\set{t_1,\ldots,t_w}$ of minimal and maximal elements. \dd
In fact during entire game we have
\[
b_i < p < t_j
\]
for all $b_i \in \bot$, $p\in P$, $t_j \in \top$. \dd
Obviously maintaining some coloring for $\brackets{P\cup\bot\cup\top,\leq}$ and restricting it to $P$ gives the required coloring for the original on-line poset presented by Spoiler. \dd

Our algorithm \texttt{coloring} maintains an auxiliary data structure $\str$ that, among other things, keeps the information about the poset $\brackets{P,\leq}$ already presented (by Spoiler) and about the coloring response of Algorithm. \dd
Our data structure $\str$ is described by the following properties.

\begin{inv} \label{I:main} Let 
\[
\str=\brackets{P,\leqslant,w,\bot,\top,P_0,\ldots,P_w,\ff{A}_0,\ldots,\ff{A}_w,\chains_1,\ldots,\chains_w,\chains},
\]
where
\begin{enumerateii}
\item $w=\fWidth{P}$ \label{I:main;C:width}
\item $\bot=\set{b_1,\ldots,b_w}$, $\top=\set{t_1,\ldots,t_w}$ and $b_i<p<t_j$ for all $p\in P$ \label{I:main;C:second}
\item $P_0,\ldots,P_w$ are subsets of $P\cup\bot\cup\top$, such that \label{I:main;C:first without w}\label{I:about Pi}
\begin{enumerate}
\item $\emptyset=P_0\subseteq P_1\subseteq\ldots\subseteq P_{w-1}\subseteq P_w = P\cup \bot\cup\top$ \label{I:about Pi subs}
\item $\fWidth{P_i}=i\quad$ \label{I:about Pi width}
\end{enumerate}
\item \label{I:main;C:Ai} $\ff{A}_0,\ldots,\ff{A}_w$ are families of antichains in $\poset{P}$, such that
\begin{enumerate}
\item every family $\ff{A}_i$ is a chain in the lattice $\fMA{P_i}$
\item the largest element of $\ff{A}_i$ is $\top\cap P_i=\set{t_1,\ldots,t_i}$ and the smallest element of $\ff{A}_i$ is $\bot\cap P_i=\set{b_1,\ldots,b_i}$ \label{I:main;C:tb}
\item $P_i=P_{i-1}\cup\bigcup\ff{A}_i$,\ \ \ for $i=1,\ldots,w$\label{I:main;C:Pi=Pi-1sumAi}
\end{enumerate}
\item $\chains_1,\ldots,\chains_w$ are partial multicolorings of $\poset{P}$, i.e. \label{I:main;C:ci}
\begin{enumerate}
\item $\chains_i:\bigcup\ff{A}_i\tto\powerp{\Gamma_i}$
\item for each $\gamma\in\Gamma_i$ the set $\set{p\in \bigcup\ff{A}_i: \gamma\in\fCi{}{i}{p}}$ is a chain\label{I:main;C:ci chain}
\end{enumerate}
\item $\chains$ is a global coloring of $\poset{P}$ consistent with the partial multicolorings $\chains_i$'s, i.e.: \label{I:main;C:chains}
\begin{enumerate}
\item $\chains:P\tto\bigcup_{i=1}^w\Gamma_i$ \label{I:main;C:chains;P:function}
\item for each $\gamma\in\bigcup_{i=1}^w\Gamma_i$ the set $\fC{-1}{\gamma}$ is a chain \label{I:main;C:be chain}
\item for each $p\in P$ there is $i\in\set{1,\ldots,w}$ such that $\fC{}{p}\in\fCi{}{i}{p}$ \label{I:main;C:inclusion of chains}
%\item If $P^-$ is the poset of the game from the previous round and $\str^- = \brackets{P^-,\leq,\ldots,\chains^-}$ is its auxiliary structure then $\chains^- = \rest{\chains}{P^-}$. \label{I:main;C:c-}
\end{enumerate}
\item for two antichains $\level{1}, \level{2}$, such that $\level{1}$ is an immediate predecessor of $\level{2}$ in the chain $\ff{A}_i$ of antichains, $(\level{1},\level{2},\leq,\rest{\chains_i}{\level{1}\cup \level{2}})$ is a board returned by some round of $i$-local on-line coloring game, i.e., one of the following holds \label{I:main;C:last} \label{I:main;C:local}
\begin{itemize}
\item $\level{1}\cup \level{2}$ with the multicoloring $\rest{\chains_i}{\level{1}\cup \level{2}}$ is a result of the procedure \ilc{} 
\item there is an antichain $\level{0}\al \level{1}$ in $\ff{A}_i$ such that $\brackets{\level{0}, \level{1}, \level{2}}$ with the multicoloring $\rest{\chains_i}{\level{0}\cup \level{1}\cup \level{2}}$ is a result of the procedure \lc{}
\item there is an antichain $\level{3}\ag \level{2}$ in $\ff{A}_i$ such that $\brackets{\level{1},\level{2},\level{3}}$ with the multicoloring $\rest{\chains_i}{\level{1}\cup \level{2}\cup \level{3}}$ is a result of the procedure \lc{}
\end{itemize}
\end{enumerateii}
\end{inv}

%\clearpage

%Our algorithm \texttt{coloring} keeps all the data in $\str$ in order to property maintain the required coloring $\chains$. \dd 
%
%
%Indeed, in view of \ref{I:main}.\ref{I:main;C:chains}, the function $\chains$ is a coloring of $\brackets{P,\leq}$ and uses at most $\abs{\bigcup_{i=1}^w \Gamma_i}=\fVal{1}+\fVal{2}+\ldots+\fVal{w}$ colors. \dd 
%

After Spoiler's move from $\poset{P}=\brackets{P,\leq}$ to $\poset{P}^{\plus}=\brackets{P\cup\set{x},\leq}$ our algorithm will modify the structure $\str$ to $\strp$, where %
\begin{small}%
\begin{align*}
\str &=\brackets{P,\ \leqslant,\ w,\ \bot,\ \top,\ P_0,\ \ldots,\ P_w,\ \ff{A}_0,\ \ldots,\ \ff{A}_w,\ \chains_0,\ \ldots,\ \chains_w,\ \chains}\ \ \textrm{for}\ \poset{P},\\
\strp&=\brackets{P^{\plus},\leqslant,w^{\plus},\bot^{\plus},\top^{\plus},P_0^{\plus},\ldots,P_{w^{\plus}}^{\plus},\ff{A}_0^{\plus},\ldots,\ff{A}_{w^{\plus}}^{\plus},\chains_0^{\plus},\ldots,\chains_{w^{\plus}}^{\plus},\chains^{\plus}}\ \textrm{for}\ \poset{P}^{\plus}\!.
\end{align*}%
\end{small}%
%
%
%\[
%\str=\brackets{P,\leqslant,w,\bot,\top,P_0,\ldots,P_w,\ff{A}_0,\ldots,\ff{A}_w,\chains_0,\ldots,\chains_w,\chains} \ \ \textrm{for}\ \poset{P}
%\]
%
%\vspace{-4mm}
%
%\begin{footnotesize}
%\begin{small}
%\[
%\strp=\brackets{P^{\plus},\leqslant,w^{\plus},\bot^{\plus},\top^{\plus},P_0^{\plus},\ldots,P_{w^{\plus}}^{\plus},\ff{A}_0^{\plus},\ldots,\ff{A}_{w^{\plus}}^{\plus},\chains_0^{\plus},\ldots,\chains_{w^{\plus}}^{\plus},\chains^{\plus}}\ \textrm{for}\ \poset{P}^{\plus}\!.
%\]
%\end{small}
%\end{footnotesize}
%
%
%\vspace{-2mm}

\medskip

The idea of the \texttt{coloring} algorithm is to identify several families $\ff{A}_1,\ldots,\ff{A}_w$ of antichains in $P\cup\bot\cup \top$ in a way that antichains appearing in these families cover $P\cup \bot\cup\top$, i.e.,
\[
P\cup\bot\cup \top=\bigcup\set{A: A\in\ff{A}_i\ \textrm{for some}\ i=1,\ldots,w}\!.
\]
For algorithm's purposes we also need the following partial coverings
\[
P_j=\bigcup\set{A: A\in\ff{A}_i\ \textrm{for some}\ i=1,\ldots,j}\!,
\]
so that we will have $P_w=P\cup\bot\cup\top$. 
In these coverings we allow intersections of the form $A\cap B$, for $A\in\ff{A}_i$, $B\in\ff{A}_j$ or even for $A,B\in\ff{A}_i$, to be nonempty. We do require however that each $\ff{A}_i$ is a chain with respect to $\aleq$. When the new point $x$ is presented by Spoiler we want to extend some of the $P_i$'s to $P_i\cup\set{x}$. This extension is doable only for those $P_i$'s width of which would not be increased. Actually we need more, namely we extend those $P_i$'s for which $P_i$ as well as $P_{i+1},\ldots,P_w$ preserve their width after adding $x$ to each of them. After identifying the smallest such index, say $i_0$, we know that $\fWidth{P_{i_0-1}\cup\set{x}}=\fWidth{P_{i_0-1}}+1$ and therefore $P_{i_0-1}\cup\set{x}$ has an antichain of size $i_0$. In fact each such antichain $A_0$ has to contain $x$, as $\fWidth{P_{i_0-1}}=i_0-1$. In the next step we want to include $A_0$ into $\ff{A}_{i_0}^{\plus}$ so that the new families $\ff{A}_i^{\plus}$ cover not only the old poset $P$ but also the new point $x$. However, $\ff{A}_{i_0}$ enlarged by $A_0$ need not be a chain any longer. Therefore we identify antichains $A_d,A_u\in\ff{A}_{i_0}$ that are as close to $x$ as possible and satisfy $x\in A_d\upseto\cap A_u\downseto$. Now, with the help of $A_d$ and $A_u$ we modify $A_0$ to $A_x=\brackets{A_0\vee A_d}\wedge A_u$, where $\vee$, $\wedge$ denote join and meet in the lattice $\fMA{P_{i_0}\cup\set{x}}$ of maximum antichains in $P_{i_0}\cup\set{x}$. This modification results in $A_d\al A_x\al A_u$ and allows to call the procedure \lc{} on the board $\brackets{A_d,A_x,A_u,\leq}$. This will extend the multicoloring of $A_d\cup A_u$ to $A_d\cup A_x\cup A_u$. In particular the point $x$ will get a set of colors previously used on both $A_d$ and $A_u$. One of the colors from this set will be used to color $x$.

This brief sketch of the algorithm lacks several arguments that everything will go correctly and our invariants on the data structure $\str$ will be kept. In particular during the game, especially at its beginning, the width of the poset may increase. \dd 
To keep \ref{I:main}.(\ref{I:main;C:second}) we may be forced to increase both $\bot$ and $\top$. \dd 
For a simpler description of our algorithm we assume now that elements, say $b_{w+1}, t_{w+1}$, that are forced to be added after $\poset{P}$ is expanded to $\poset{P}^{\plus}$, satisfy \mbox{$b_{w+1}<p<t_{w+1}$} for all elements $p\in P$ as well as for $p=x$. \dd 

Also, to simplify our notation, both in the description of the algorithm as well as in the forthcoming proofs, we use $X\downseto, X\downsetc, X\upseto, X\upsetc$ to denote the downsets and upsets of $X$ in the extended poset
\linebreak
$\brackets{P\cup\set{x}\cup\bot^{\plus}\cup\top^{\plus},\leq}$. \dd
Moreover we often overuse the symbol $\leq$ to express the order relation on various subsets $X$ of
$P\cup\set{x}\cup\bot^{\plus}\cup\top^{\plus}$ by writing simply $\leq$ instead of $\rest{\leq}{X}$. \dd 
Now we are ready to present our \texttt{coloring} algorithm that builds $\strp$ from $\str$ and the new point $x$ presented by Spoiler.

\begin{alg}{(\texttt{coloring} algorithm)} \label{A:main}%
\renewcommand*\thelstnumber{(A\oldstylenums{\the\value{lstnumber}})}%

%\vspace{-1mm}

\begin{lstlisting}[mathescape]
	$w^{\plus}$ := $\fWidth{P^{\plus}}$ (*@\label{A:first} @*)
	if $\fWidth{P^{\plus}}>\fWidth{P}$ then
		$\top^{\plus}$ := $\top \cup \set{t_{w+1}}$ (*@\label{A:first in if} @*)
		$\bot^{\plus}$ := $\bot \cup \set{b_{w+1}}$
		$P_{w+1}$ := $P\cup\bot^{\plus}\cup\top^{\plus}$
		$\ff{A}_{w+1}$ := $\set{\bot^{\plus},\top^{\plus}}$
		(*@$\chains_{w+1}$ is the multicoloring returned by \ilc{}$\,\brackets{\bot^{\!\plus}\!,\top^{\plus}\!, \leq}$@*) (*@\label{A:def.of cw+1} @*) (*@\label{A:last in if} @*)
	else	$\top^{\plus}$ := $\top$;     $\bot^{\plus}$ := $\bot$ (*@\label{A:end if} @*)(*@\vspace{4mm}@*)
	$i_0$ :=  $\min\set{i: \fWidth{P_j\cup\set{x}}=\fWidth{P_j}\ \textrm{for all}\ j=i,\ldots,w^{\plus}}$ (*@\label{A:def.of i0} @*)
	for $i$ := $0$ to $w^{\plus}$ do
		if $i < i_0$ then $P_i^{\plus}$ := $P_i$ else $P_i^{\plus}$ := $P_i\cup\set{x}$(*@\label{A:def.of Pi+}@*)
		if $i \neq i_0$ then $\ff{A}_i^{\plus}$ := $\ff{A}_i$
		if $i \neq i_0$ then $\chains^{\plus}_i$ := $\chains_i$ (*@\label{A:def.of ci}@*)
	$A_0$ := $\textrm{arbitrary maximum antichain in } P_{i_0-1}\cup\set{x}$ (*@\label{A:def.of A0} @*)
	$A_d$ := $\max_{\aleq\!}\set{A\in \ff{A}_{i_0}:\ x\in A\upseto}$ (*@\label{A:def.of Ad} @*)
	$A_u$ := $\min_{\aleq\!}\set{A\in \ff{A}_{i_0}:\ x\in A\downseto}$ (*@\label{A:def.of Au} @*)
	$A_x$ := $\brackets{A_0\vee A_d}\wedge A_u,$ (*@\hfill\it{where $\vee$ and $\wedge$ are operations in $\descMA\!\big(P_{i_0}^{\plus}\big)$}@*)(*@\label{A:def.of Ax}@*)
	$\ff{A}_{i_0}^{\plus}$ := $\ff{A}_{i_0}\cup\set{A_x}$ (*@\vspace{2mm}@*)
	(*@$\chains_{i_0}^{\plus}$ is the extension of $\chains_{i_0}$ by coloring the points from $A_x\setminus\brackets{A_d\cup A_u}$\\ \phantom{xxxxxi} by \lc{}$\,\big(A_d, A_x, A_u, \leq, \rest{\chains_{i_0}}{A_d\cup A_u}\big)$ @*) (*@\label{A:def.of ci0}@*)(*@\vspace{2mm}@*)
	$\fC{\plus}{p}$ := $\begin{cases}\fC{}{p},&\textrm{if}\ p\in P;\\\textrm{any member of}\ \fCi{\plus}{i_0}{x},&\textrm{if}\ p=x.\end{cases}$  (*@\label{A:def.of c+}\label{A:end}@*)
\end{lstlisting}
\end{alg}

Our first step is to analyze lines \ref{A:first}-\ref{A:end if} of the algorithm. First of all we need to argue that calling procedure \ilc{} in line \ref{A:def.of cw+1} is legal. This line is executed only if $\fWidth{P^{\plus}}>\fWidth{P}$. By our remark made in the second last paragraph before the statement of the algorithm we know that the poset $\brackets{\bot^{\plus}\cup\top^{\plus},\leq}$ satisfies \mbox{$\bot^{\plus}<\top^{\plus}$} so it can be considered as a correct first move of Spoiler in local on-line coloring game. Therefore \ilc{} returns, among other things, a multicoloring of $\bot^{\plus}\cup\top^{\plus}$ which is called $\chains_{w+1}$.

Independently of whether lines \ref{A:first in if}--\ref{A:last in if} are executed it is easy to observe that after line \ref{A:end if} almost all of our invariants are still satisfied. %\note{wersja alternatywna: for the structure $$\stro=\brackets{P,\leqslant,w^{\plus},\bot^{\plus},\top^{\plus},P_0,\ldots,P_{w^{\plus}},\ff{A}_0,\ldots,\ff{A}_{w^{\plus}},\chains_0,\ldots,\chains_{w^{\plus}},\chains}.$$}
More precisely

\begin{clm}\label{C:Sline} The structure %\note{wersja alternatywna: $\stro$} %
%\note{
\[
\brackets{P,\leqslant,w^{\plus},\bot^{\plus},\top^{\plus},P_0,\ldots,P_{w^{\plus}},\ff{A}_0,\ldots,\ff{A}_{w^{\plus}},\chains_0,\ldots,\chains_{w^{\plus}},\chains}
\]%} %
satisfies conditions \ref{I:main}.(\ref{I:main;C:second}--\ref{I:main;C:last}).\hfill$\square$
\end{clm}
The condition \ref{I:main}.(\ref{I:main;C:width}) fails in the structure of Claim \ref{C:Sline} only if $w^{\plus}>w$. However in this case we secured that $P\cup\bot^{\plus}\cup\top^{\plus}$ has the width $w^{\plus}=w+1$. Thus in both cases we are ready to incorporate the incoming point $x$ into $P\cup\bot^{\plus}\cup\top^{\plus}$ without changing width of $P\cup\bot^{\plus}\cup\top^{\plus}$.

%\note{wersja alternatywna: By Claim \ref{C:Sline},

%Now we will be based on invariants concern the structure of Claim \ref{C:Sline} instead of invariants concern $\str$.

Now we argue that the operations performed in lines \ref{A:def.of i0}, \ref{A:def.of Ad}, \ref{A:def.of Au}, \ref{A:def.of Ax}, \ref{A:def.of ci0} and \ref{A:def.of c+} are correctly defined.

For \ref{A:def.of i0} first note that
\begin{align*}
 \fWidth{P_{w^{\plus}}}\ &=\ \fWidth{P\cup\bot^{\plus}\cup\top^{\plus}}\\
&=\ \fWidth{P\cup\bot^{\plus}\cup\top^{\plus}\cup\set{x}}\\
&=\ \fWidth{P_{w^{\plus}}\cup\set{x}}.
\end{align*}
Moreover $\fWidth{P_0}=0<\fWidth{P_0\cup\set{x}}$ so that $i_0$ exists and we have $1\leq i_0\leq w^{\plus}$.

To see \ref{A:def.of Ad} and \ref{A:def.of Au} we need to argue that the sets
\[
 \set{A\in \ff{A}_{i_0}:\ x\in A\upseto }\quad \textrm{and}\quad \set{A\in \ff{A}_{i_0}:\ x\in A\downseto }
\]
--\ \ of which the largest and the smallest elements are supposed to be taken\ \ --\ \ are nonempty. But since \mbox{$\bot^{\!\plus}\!<\!x\!<\!\top^{\plus}$} we have $x\mspace{-2mu}\in\mspace{-2mu}\set{b_1,\ldots,b_{i_0}}\upseto$ and \mbox{$x\in\set{t_1,\ldots,t_{i_0}}\downseto$}. Moreover the antichains \linebreak $\set{b_1,\ldots,b_{i_0}}$ and $\set{t_1,\ldots,t_{i_0}}$ are guaranteed to lie in $\ff{A}_{i_0}$ by condition \ref{I:main}.(\ref{I:main;C:tb}) fullfiled by the structure of Claim \ref{C:Sline}. %\note{wersja alternatywna: fullfiled by $\stro$}.

Finally, for the correctness of line \ref{A:def.of Ax} we need to know that $A_d, A_u, A_0\in\fMA{P_{i_0}\cup\set{x}}$. Since $i_0$ is the smallest $i$ such that $\fWidth{P_j}=\fWidth{P_j\cup\set{x}}$ for all $j\geq i$, we know that
$$\fWidth{P_{i_0-1}\cup\set{x}}\ >\ \fWidth{P_{i_0-1}}\ =\ i_0-1,$$ and thus $\fWidth{P_{i_0-1}\cup\set{x}}=i_0$. \dd 
The set $A_0$ chosen in line \ref{A:def.of A0} is maximum antichain in $P_{i_0-1}\cup\set{x}$, and therefore $\abs{A_0}=i_0$. \dd 
Moreover, $P_{i_0-1}\cup\set{x}\subseteq P_{i_0}\cup\set{x}$ and consequently $A_0\subseteq P_{i_0}\cup\set{x}=P_{i_0}^{\plus}$. \dd 
To summarize, each antichain $A_d, A_u, A_0$ has $i_0$ elements and is contained in $P_{i_0}^{\plus}$. By the choice of $i_0$, we have $\width\!\big(P_{i_0}^{\plus}\big)=i_0$ so that $A_d,A_u,A_0\in\descMA\!\big(P_{i_0}^{\plus}\big)$, as required. \dd 

Before we proceed with the correctness of line \ref{A:def.of ci0} we need the following two claims which will be useful also when proving that our invariants are kept.

\begin{clm}\label{C:x in Ax}
 $x\in A_x$.
\end{clm}
\begin{proof}
Since $\fWidth{P_{i_0-1}}=i_0-1$ while $\fWidth{P_{i_0-1}\cup\set{x}}=i_0$, we know that the maximum antichain $A_0$ constructed in line \ref{A:def.of A0} must contain $x$. By \ref{A:def.of Ad} we know that $x$ dominates some point in $A_d$. As $A_d$ is an antichain, $x$ cannot be dominated by anything from $A_d$. This together with the fact that $x$ is incomparable with all other points in $A_0$ gives that $x$ is maximal in $A_0\cup A_d$. Recalling Observation \ref{O: max=sup min=inf} we know that $x\in\max\brackets{A_0\cup A_d}=A_0\vee A_d$. Analogously we argue that now $x$ is minimal in $\brackets{A_0\vee A_d}\cup A_u$ to get $x\in\min\brackets{\brackets{A_0\vee A_d}\cup A_u}=\brackets{A_0\vee A_d}\wedge A_u=A_x$.
\end{proof}

\begin{clm} \label{C:Ad<Ax<Au} The antichains $A_u$, $A_d$ and $A_x$ satisfy:
\begin{enumerate}
\item $A_u$ is the immediate successor of $A_d$ in the chain $\ff{A}_{i_0}$, \label{C:Ad<Ax<Au;E:AdAu}
\item $A_d\al A_x \al A_u$ in the lattice $\descMA\!\big(P_{i_0}^{\plus}\big)$. \label{C:Ad<Ax<Au;E:Ad<Ax<Au}
\end{enumerate}
\end{clm}
\begin{proof}

At the beginning, we will show that $A_u$ is the immediate successor of $A_d$ in $\ff{A}_{i_0}$. \dd 
Every $A\in\ff{A}_{i_0}\subseteq\fMA{P_{i_0}}$ contains $i_0$ elements, so that $A$ is a maximum antichain in $P_{i_0}^{\plus}=P_{i_0}\cup\set{x}$. \dd 
Thus $x\in A\downseto$ or $x\in A\upseto$, as otherwise $A\cup\set{x}$ would be an antichain of size $i_0+1$. \dd 
Now note that there are $a,b$ with $A_d\ni a<x<b \in A_u$. If $A_u\aleq A_d$ then there is $a'\in A_d$ with $b\leq a'$, so that $a<a'$, contradicting the fact that $A_d$ is an antichain.
Therefore $A_d\al A_u$. Moreover the choice of $A_d$ and $A_u$ tells us that for an antichain $A\in \ff{A}_{i_0}$ with $A_d\al A\al A_u$ we have $x\notin A\upseto$ and $x\notin A\downseto$. This means that $A\cup\set{x}$ is an antichain of size $i_0+1$ in the poset $P_{i_0}^{\plus}$ of width $i_0$. This contradiction shows (\ref{C:Ad<Ax<Au;E:AdAu}). To see (\ref{C:Ad<Ax<Au;E:Ad<Ax<Au}) note that
\[
A_d\ \aleq\ A_d\vee \brackets{A_0\wedge A_u}\ \aleq\ \brackets{A_0\vee A_d}\wedge A_u\ \aleq\ A_u,
\]
where the middle inequality is actually the equality in the distributive lattice $\descMA\!\big(P_{i_0}^{\plus}\big)$. Since $x\in A_x \setminus \brackets{A_d\cup A_u}$ the leftmost and the rightmost inequalities are strong.
\end{proof}

Now we are ready to show that line \ref{A:def.of ci0} makes sense. First of all note that Claim \ref{C:Ad<Ax<Au}.(\ref{C:Ad<Ax<Au;E:AdAu}) together with Invariant \ref{I:main}.(\ref{I:main;C:local}) guarantee that $(A_d,A_u,\leq,\rest{\chains_{i_0}}{A_d\cup A_u})$ is a board returned in some round of $i_0$-local on-line coloring game by the procedure \ilc{} or \lc{}. Trying to continue this $i_0$-local game the local coloring algorithm must be ready to respond to Spoiler's move that presents all points from the set $M_0:=A_x\setminus\brackets{A_d\cup A_u}$ contained in an antichain $M:=A_x$. Indeed:
\begin{itemize}
\item $i_0\ =\ \fWidth{A_x}\ \leq\ \fWidth{A_d\cup A_x\cup A_u}\leq$

\hfill$\leq\fWidth{P_{i_0}\cup\set{x}}\ =\ i_0$,

\item $A_x\in \fMA{A_d\cup A_x\cup A_u}$,
\item $A_d\al A_x\al A_u$ (by Claim \ref{C:Ad<Ax<Au}.(\ref{C:Ad<Ax<Au;E:Ad<Ax<Au})).
\end{itemize}
The response of the $i_0$-local algorithm is then a multicoloring of all points from the set $M_0$ by colors appearing on both $A_d$ and $A_u$. This multicoloring of $M_0$ together with $\chains_{i_0}$ that colors $\bigcup \ff{A}_{i_0}$ is supposed to be the multicoloring $\chains_{i_0}^{\plus}$ that now colors $M_0\cup\bigcup\ff{A}_{i_0}=\ff{A}_{i_0}^{\plus}$. The only thing that is missing is to show that the set $M_0$ (which gets possibly new color by \lc{}) does not intersect the set already colored by $\chains_{i_0}$, i.e., the set $\bigcup\ff{A}_{i_0}$.
\begin{clm}
$M_0\cap\bigcup\ff{A}_{i_0}=\emptyset$.
\end{clm}
\begin{proof}
We will show that if $p\in M$ and $p\in A\in\ff{A}_{i_0}$ then $p\in A_d\cup A_u$. Without loss of generality, we may assume that $M\al A$. Since $M$ lie between $A_d\al A_u$, we have $M\al A_u \aleq A$. Thus, $p\in M$ lies below some $q\in A_u$ and $q$ lies below some $r\in A$. Thus $p\leq q\leq r$, but $p$ and $r$ lie in the same antichain $A$, which yields $p=q=r \in A_u$.
\end{proof}

Finally we need to argue that using $\fCi{\plus}{i_0}{x}$ in line \ref{A:def.of c+} makes sense, i.e., that $x$ is in the domain of $\chains_{i_0}^{\plus}$. This will immediately follow from Claim \ref{C:x in Ax}.

Now, knowing that Algorithm \ref{A:main} is well defined, we proceed with proving that Invariants \ref{I:main} are kept when passing from $\str$ to $\strp$.

\begin{clm}\label{C:Pi-1 sub Pi and wPi=i}
$P_{i-1}^{\plus}\subseteq P_i^{\plus}$ and $\fWidth{P_i^{\plus}}=i$.
\end{clm} 

\begin{proof}
From the construction of the $P_i^{\plus}$'s in line \ref{A:def.of Pi+} we obtain that $P_i^{\plus}$ is either $P_i$ or $P_i\cup\set{x}$ and moreover $x\in P_{i+1}^{\plus}$ whenever $x\in P_i^{\plus}$. \dd 
By \ref{I:main}.(\ref{I:about Pi subs}) $P_{i-1}\subseteq P_i$ and therefore $P_{i-1}^{\plus}\subseteq P_i^{\plus}$. \dd 
The index $i_0$ defined in line \ref{A:def.of i0}, assures us that for $j\geq i_0$ the width of $P_j$ can not increase by adding $x$. \dd 
Thus $\fWidth{P_i^{\plus}}=\fWidth{P_i}=i$ holds for all $i$'s.
\end{proof}

\begin{clm} \label{C:A is a chain; Pi=Pi-1 cup bcupA}
$\ff{A}_0^{\plus},\ldots,\ff{A}_{w^{\plus}}^{\plus}$ are families of antichains in $\poset{P}^{\plus}$, such that
\begin{itemize}
\item every $\ff{A}_i^{\plus}$ is a chain in the lattice $\fMA{P_i^{\plus}}$,
\item $P_i^{\plus}=P_{i-1}^{\plus}\cup\bigcup\ff{A}_i^{\plus}\quad$ for $i=1,\ldots,w^{\plus}$.
\end{itemize}
\end{clm}
\begin{proof}
The only change done to the $\ff{A}_i$'s is $\ff{A}_{i_0}^{\plus}=\ff{A}_{i_0}\cup\set{A_x}$. \dd 
By Claim \ref{C:Ad<Ax<Au} we know that $A_x\in\descMA\big(P_{i_0}^{\plus}\big)$ and $A_x$ is inserted between two consecutive antichains $A_d$ and $A_u$ of $\ff{A}_{i_0}$. Therefore $\ff{A}_{i_0}^{\plus}$ is a chain in the lattice $\descMA\big(P_{i_0}^{\plus}\big)$. \dd 

To see the second item of our Claim note that for $i<i_0$ we have $P_i^{\plus}=P_i$ and $\ff{A}_i^{\plus}=\ff{A}_i$. Thus
\[
 P_i^{\plus}\ =\ P_i\ =\ P_{i-1}\cup\bigcup\ff{A}_i\ =\ P_{i-1}^{\plus}\cup\bigcup\ff{A}_i^{\plus}.
\]
Also for $i>i_0$ the situation is easy, as $P_i^{\plus}=P_i\cup\set{x}$ and $\ff{A}_i^{\plus}=\ff{A}_i$, so that
\[
 P_i^{\plus}\ =\ \set{x}\cup P_i\ =\ \set{x}\cup P_{i-1}\cup\bigcup\ff{A}_i\ =\ P_{i-1}^{\plus}\cup\bigcup\ff{A}_i^{\plus}.
\]

The only not completely trivial case is the borderline $i=i_0$. In this case however, we use Claim \ref{C:x in Ax} to get \dd 
\begin{equation}\label{E:limit.of Ax}
\set{x}\subseteq A_x \subseteq P_{i_0}\cup\set{x},
\end{equation} \dd 
and then \dd 
\begin{align*}
 P_{i_0}^{\plus}\ &=\ P_{i_0}\cup\set{x}&\textrm{by def. of $P_{i_0}^{\plus}$}\\
\ &=\ P_{i_0}\cup A_x&\textrm{by (\ref{E:limit.of Ax})}\\
\ &=\ P_{i_0-1}\cup\brackets{\bigcup\ff{A}_{i_0}}\cup A_x&\textrm{by \ref{I:main}.(\ref{I:main;C:Pi=Pi-1sumAi})}\\
\ &=\ P_{i_0-1}^{\plus}\cup\bigcup\ff{A}_{i_0}^{\plus}.&\textrm{by def. of $P_{i_0-1}^{\plus}$ and $\ff{A}_{i_0}^{\plus}$}
\end{align*}
\end{proof}

\begin{clm} Each $\chains_i^{\plus}:\bigcup\ff{A}_i^{\plus}\tto\powerp{\Gamma_i}$ is a multicoloring, i.e., for $\gamma\in\Gamma_i$ the set $\set{p\in \bigcup\ff{A}_i^{\plus}: \gamma\in\fCi{\plus}{i}{p}}$ is a chain.\label{C:ci}
\end{clm}

\begin{proof}
As for $i\neq i_0$ there is no change when passing from $\ff{A}_i$ to $\ff{A}_i^{\plus}$ and from $\chains_i$ to $\chains_i^{\plus}$, Invariant \ref{I:main}.(\ref{I:main;C:ci}) is kept and does the job.

Now we focus on $\ff{A}_{i_0}^{\plus}$ and $\chains_{i_0}^{\plus}$, where we need to show that $p \cmp q$ for all $p,q\in\bigcup \ff{A}_{i_0}^{\plus}$ such that $\fCi{\plus}{i_0}{p}\cap\fCi{\plus}{i_0}{q}\!\neq\!\emptyset$. As \mbox{$\bigcup\ff{A}_{i_0}^{\plus}\!=\!A_x\cup\bigcup\ff{A}_{i_0}$} we split our argument into cases. The first one is $p,q\in\bigcup\ff{A}_{i_0}$, when the comparability of $p$ and $q$ results from Invariant \ref{I:main}.(\ref{I:main;C:ci chain}) and the fact that $\chains_{i_0}^{\plus}\!|_{\bigcup\ff{A}_{i_0}}\!=\chains_{i_0}$. The second possibility, namely $p,q\in A_x$, is ruled out as the sets of colors assigned (by the procedure \lc{}) to the different points of the antichain $A_x$ have to be disjoint by the rules of the local on-line coloring game. Finally suppose that $p\in A_x\setminus\bigcup\ff{A}_{i_0}$ while $q\in \bigcup\ff{A}_{i_0}$. Pick $\gamma\in\fCi{\plus}{i_0}{p}\cap\fCi{\plus}{i_0}{q}$. The procedure \lc{}, while coloring $p$, has to pick for $\fCi{\plus}{i_0}{p}$ a subset of $\fCi{\plus}{i_0}{A_d}\cap\fCi{\plus}{i_0}{A_u}$. Therefore
\begin{equation} \label{E:gamma in cad cau}
\gamma\in\fCi{}{i_0}{a_d}\cap\fCi{}{i_0}{a_u}
\end{equation}
for some
\begin{equation}\label{E:ad<p<au}
A_d\ni a_d< p < a_u\in A_u.
\end{equation}
Now we switch to the coloring of $q$. We know that $\chains_{i_0}^+$ behaves on $\bigcup\ff{A}_{i_0}$ in the very same way as $\chains_{i_0}$ did. In particular $\gamma\in\fCi{}{i_0}{q}$. This together with (\ref{E:gamma in cad cau}) gives that $q, a_d, a_u$ are pairwise comparable. On the other hand the fact that $q\in \bigcup\ff{A}_{i_0}$ is witnessed by some $A_q\in\ff{A}_{i_0}$ with $q\in A_q$. By Claim \ref{C:Ad<Ax<Au}.(\ref{C:Ad<Ax<Au;E:AdAu}), the chain $\ff{A}_{i_0}$ has nothing between $A_d$ and $A_u$ so that either $A_q \aleq A_d$ or $A_u\aleq A_q$. In the first case ($A_q\aleq A_d$) the comparability of $q$ and $a_d$ gives $q\leq a_d$, while in the second one ($A_u\aleq A_q$) the comparability of $q$ and $a_u$ gives $a_u\leq q$. Combining this with (\ref{E:ad<p<au}) we get
\[
q\leq a_d\leq p \quad \textrm{or}\quad p\leq a_u \leq q,
\]
as required.
\end{proof}

With the help of Claims \ref{C:Pi-1 sub Pi and wPi=i}-\ref{C:ci} we see that \ref{I:main}.(\ref{I:main;C:width}--\ref{I:main;C:ci}) and (\ref{I:main;C:chains;P:function}) are kept when passing from $\str$ to $\strp$.

Also \ref{I:main}.(\ref{I:main;C:inclusion of chains}) for $\strp$ is obvious as in line \ref{A:def.of c+} $\fC{\plus}{x}$ is chosen to satisfy this condition. Moreover \ref{I:main}.(\ref{I:main;C:inclusion of chains}) for $\strp$ together with disjointness of the $\Gamma_i$'s gives that $\set{p\in P^{\plus}: \fC{\plus}{p}=\gamma}$ is a subset of the chain $\set{p\in\bigcup\ff{A}_i:\gamma\in\fCi{\plus}{i}{p}}$, whenever $\gamma\in\Gamma_i$. This establishes \ref{I:main}.(\ref{I:main;C:be chain}) for $\strp$.

Finally, for \ref{I:main}.(\ref{I:main;C:local}) it suffices to focus on $\level{1}, \level{2}\in\ff{A}_{i_0}^{\plus}$ with $\level{1}$ or $\level{2}$ being the only new antichain $A_x$. However in the chain $\ff{A}_{i_0}^{\plus}$ the antichains $A_d, A_x, A_u$ are neighbors, so that $\brackets{\level{1},\level{2}}$ is either $\brackets{A_d,A_x}$ or $\brackets{A_x,A_u}$. In both cases the triple $\brackets{A_d, A_x, A_u}$ together with \linebreak $\chains_{i_0}^{\plus}|_{A_d\cup A_x\cup A_u}$ witnesses \ref{I:main}.(\ref{I:main;C:local}).

\medskip

This shows that entire Invariant \ref{I:main} is kept while passing from $\str$ to $\strp$. In particular \ref{I:main}.(\ref{I:main;C:chains;P:function}) gives that $\poset{P}$ can be colored with
\(
%\abs{\bigcup_{v=1}^w \Gamma_v}=
\sum_{v=1}^w\abs{\Gamma_v}=\sum_{v=1}^w\fVal{v}
\)
colors and finishes our proof of \mbox{Theorem \ref{L:local}}.
\end{proof}

Immediately from Theorem \ref{L:local} we get the following

\begin{cor} \label{R:CP sum LCP}
In the on-line chain partitioning game Algorithm has a strategy that uses at most $\fLCP{1}+\fLCP{2}+\ldots+\fLCP{w}$ chains to partition posets of width $w$, i.e.
 \[
\fCP{w}\leq \sum_{v=1}^w \fLCP{v}.
 \]
\end{cor}

The local on-line coloring game described in Definition \ref{D:local game} was defined in a way general enough to prove the inequality of Corollary \ref{R:CP sum LCP}. However it is hard to determine $\fLCP{w}$ for such general local games. It would be much easier to do this for the game described at the beginning of Section \ref{S:Main Reduction}, i.e. for games satisfying two additional properties.

\begin{prp} \label{Prp:disjoint}
Antichains $L$, $M$, $T$ are pairwise disjoint.
\end{prp}

\begin{prp}\label{Prp:core}
$\brackets{L, M,\leq}$, $\brackets{M, T,\leq}$ and $\brackets{L, T,\leq}$ are cores.
\end{prp}

In the following two subsections we will consecutively show that imposing these two restrictions on local on-line coloring game actually does not change the value $\fLCP{w}$ for these modified games.

\subsection{Disjoint levels on the board}\label{SS:different}
In this section we will show that the value of the local on-line coloring game with the the additional
\begin{pr} \label{Pr:disjoint}
Antichains $L$, $M$, $T$ are pairwise disjoint during entire game.
\end{pr}
\noindent{}is bounded from below by $\fLCP{w}$.
Property \ref{Pr:disjoint}, together with the rules of the local game, gives that \mbox{$(L,M,\leq)$}, \mbox{$(M,T,\leq)$} and \mbox{$(L,T,\leq)$}
are regular bipartite posets in the sense of the following definition.
\begin{defn}
A \emph{regular bipartite poset} is a triple $(A,B,\leq)$ such that 
\begin{itemize}
\item $(A\cup B,\leq)$ is a poset,
\item $A,B$ are disjoint antichains with $A\al B$,
\item $\abs{A}=\abs{B}=\fWidth{A\cup B}$.
\end{itemize}
\end{defn}

\bigskip

Before we proceed with the proof we summarize all rules of the local game with Property \ref{Pr:disjoint}.

\begin{defn}\label{D:disjoint game}
By a \emph{disjoint game} we mean the following two-person game between Spoiler and Algorithm. During the first round:
\begin{itemize}
\item Spoiler sets a natural number $w$ and then introduces two antichains $L,T$, each with $w$ elements, such that $L<T$.
\item Algorithm determines a finite set $\Gamma$ of colors that may be used in the entire game and then he colors each point $x\in L\cup T$ with some nonempty subset $\fC{}{x}$ of $\Gamma$ such that for each $\gamma\in\Gamma$ the points colored by $\gamma$ form a chain.
\end{itemize}
During next rounds the board $\brackets{L,T,\leq,\chains}$ from the previous round is transformed to a board $\brackets{L^{\plus},T^{\plus},\leq,\chains^{\plus}}$ according to the following rules:
\begin{itemize}
\item Spoiler introduces $w$ new elements that form an antichain $M$ such that the poset $\poset{B}'=\brackets{L\cup M\cup T,\leq}$ has width $w$ and moreover $L\al M\al T$ in the lattice $\fMA{\poset{B}'}$.
\item Algorithm colors each point $m\in M$ with a nonempty finite set of colors $\fC{\plus}{m}\subseteq\fC{}{L}\cap\fC{}{T}$ and keeps the old multicoloring on $L\cup T$, i.e. $\rest{\chains^{\plus}}{L\cup T}=\chains$.
\item Finally, Spoiler redefines the levels $L, T$ to $L^{\plus}, T^{\plus}$ so that either $\brackets{L^{\plus},T^{\plus}}=\brackets{L,M}$ or $\brackets{L^{\plus},T^{\plus}}=\brackets{M,T}$.
This, after restricting to $L^{\plus}\cup T^{\plus}\!$, creates the new board $\brackets{L^{\plus},T^{\plus},\leq,\chains^{\plus}}$.\end{itemize}
Again, the goal of Algorithm is to pick minimal number $\abs{\Gamma}$ of colors already during the first round so that he can play with these colors forever.
\end{defn}

\bigskip

The following example of a local game with Property \ref{Pr:disjoint} will be further explored to illustrate that Property \ref{Pr:disjoint} is redundant for keeping the value of the game. Let a regular bipartite poset \mbox{$(L,T,\leq)$} be like on Figure \ref{F:Example Local 2}.
In order to satisfy Property \ref{Pr:disjoint} Spoiler must provide $4=\fWidth{L\cup T}$ new points, which form an antichain $M$. In our example after Algorithm colors all points of $M$, Spoiler chooses the subposet \mbox{$(M,T,\!\rest{\leq}{M\cup T})$} for the new regular bipartite poset \mbox{$(L^{\plus}\!,T^{\plus}\!,\leq)$}.
\begin{figure}[hbt]
\centering\ifx\JPicScale\undefined\def\JPicScale{1}\fi
\psset{unit=\JPicScale mm}
\psset{linewidth=0.3,dotsep=1,hatchwidth=0.3,hatchsep=1.5,shadowsize=1,dimen=middle}
\psset{dotsize=0.7 2.5,dotscale=1 1,fillcolor=black}
\psset{arrowsize=1 2,arrowlength=1,arrowinset=0.25,tbarsize=0.7 5,bracketlength=0.15,rbracketlength=0.15}
\begin{pspicture}(0,0)(119,30)
\psline[linewidth=0.25,fillcolor=white,fillstyle=solid](2,23)(2,13)
\psline[linewidth=0.25,fillcolor=white,fillstyle=solid](22,23)(22,13)
\psline[linewidth=0.25,fillcolor=white,fillstyle=solid](22,23)(32,13)
\psline[linewidth=0.25,fillcolor=white,fillstyle=solid](12,23)(32,13)
\psline[linewidth=0.25,fillcolor=white,fillstyle=solid](12,13)(12,23)
\psline[linewidth=0.25,fillcolor=white,fillstyle=solid](22,23)(12,13)
\psline[linewidth=0.25,fillcolor=white,fillstyle=solid](32,23)(32,13)
\rput[l](35,23){$T$}
\rput[l](35,13){$L$}
\psline[linewidth=0.25,fillcolor=white,fillstyle=solid](22,13)(12,23)
\psline[linewidth=0.25,fillcolor=white,fillstyle=solid](32,23)(22,13)
\psline[linewidth=0.1](34,15)(0,15)
\psline[linewidth=0.1](34,11)(0,11)
\psline[linewidth=0.1](0,15)(0,11)
\rput{0}(2,23){\psellipse[linewidth=0.25,linestyle=none,fillstyle=solid](0,0)(1,1)}
\rput{0}(12,13){\psellipse[linewidth=0.25,linestyle=none,fillstyle=solid](0,0)(1,1)}
\rput{0}(22,23){\psellipse[linewidth=0.25,linestyle=none,fillstyle=solid](0,0)(1,1)}
\rput{0}(32,23){\psellipse[linewidth=0.25,linestyle=none,fillstyle=solid](0,0)(1,1)}
\rput{0}(2,13){\psellipse[linewidth=0.25,linestyle=none,fillstyle=solid](0,0)(1,1)}
\rput{0}(12,23){\psellipse[linewidth=0.25,linestyle=none,fillstyle=solid](0,0)(1,1)}
\rput{0}(22,13){\psellipse[linewidth=0.25,linestyle=none,fillstyle=solid](0,0)(1,1)}
\rput{0}(32,13){\psellipse[linewidth=0.25,linestyle=none,fillstyle=solid](0,0)(1,1)}
\psline[linewidth=0.1](34,15)(34,11)
\psline[linewidth=0.1](34,21)(0,21)
\psline[linewidth=0.1](34,25)(0,25)
\psline[linewidth=0.1](34,25)(34,21)
\psline[linewidth=0.1](0,25)(0,21)
\psline[linewidth=0.25,fillcolor=white,fillstyle=solid](44,18)(44,8)
\psline[linewidth=0.25,fillcolor=white,fillstyle=solid](64,18)(64,8)
\psline[linewidth=0.25,fillcolor=white,fillstyle=solid](64,18)(74,8)
\psline[linewidth=0.25,fillcolor=white,fillstyle=solid](54,8)(54,18)
\psline[linewidth=0.25,fillcolor=white,fillstyle=solid](74,18)(74,8)
\rput[l](77,18){$M$}
\rput[l](77,8){$L$}
\psline[linewidth=0.25,fillcolor=white,fillstyle=solid](74,18)(64,8)
\psline[linewidth=0.1](76,10)(42,10)
\psline[linewidth=0.1](76,6)(42,6)
\psline[linewidth=0.1](42,10)(42,6)
\rput{0}(44,18){\psellipse[linewidth=0.25,linestyle=none,fillstyle=solid](0,0)(1,1)}
\rput{0}(54,8){\psellipse[linewidth=0.25,linestyle=none,fillstyle=solid](0,0)(1,1)}
\rput{0}(64,18){\psellipse[linewidth=0.25,linestyle=none,fillstyle=solid](0,0)(1,1)}
\rput{0}(74,18){\psellipse[linewidth=0.25,linestyle=none,fillstyle=solid](0,0)(1,1)}
\rput{0}(44,8){\psellipse[linewidth=0.25,linestyle=none,fillstyle=solid](0,0)(1,1)}
\rput{0}(54,18){\psellipse[linewidth=0.25,linestyle=none,fillstyle=solid](0,0)(1,1)}
\rput{0}(64,8){\psellipse[linewidth=0.25,linestyle=none,fillstyle=solid](0,0)(1,1)}
\rput{0}(74,8){\psellipse[linewidth=0.25,linestyle=none,fillstyle=solid](0,0)(1,1)}
\psline[linewidth=0.1](76,10)(76,6)
\psline[linewidth=0.1](76,16)(42,16)
\psline[linewidth=0.1](76,20)(42,20)
\psline[linewidth=0.1](76,20)(76,16)
\psline[linewidth=0.1](42,20)(42,16)
\psline[linewidth=0.25,fillcolor=white,fillstyle=solid](44,28)(44,18)
\psline[linewidth=0.25,fillcolor=white,fillstyle=solid](64,28)(74,18)
\psline[linewidth=0.25,fillcolor=white,fillstyle=solid](54,18)(54,28)
\psline[linewidth=0.25,fillcolor=white,fillstyle=solid](64,28)(54,18)
\psline[linewidth=0.25,fillcolor=white,fillstyle=solid](74,28)(74,18)
\rput[l](77,28){$T$}
\psline[linewidth=0.25,fillcolor=white,fillstyle=solid](64,18)(54,28)
\rput{0}(44,28){\psellipse[linewidth=0.25,linestyle=none,fillstyle=solid](0,0)(1,1)}
\rput{0}(64,28){\psellipse[linewidth=0.25,linestyle=none,fillstyle=solid](0,0)(1,1)}
\rput{0}(74,28){\psellipse[linewidth=0.25,linestyle=none,fillstyle=solid](0,0)(1,1)}
\rput{0}(54,28){\psellipse[linewidth=0.25,linestyle=none,fillstyle=solid](0,0)(1,1)}
\psline[linewidth=0.1](76,26)(42,26)
\psline[linewidth=0.1](76,30)(42,30)
\psline[linewidth=0.1](76,30)(76,26)
\psline[linewidth=0.1](42,30)(42,26)
\psline[linewidth=0.25,fillcolor=white,fillstyle=solid](64,18)(64,28)
\rput[l](119,18){$L^+$}
\rput{0}(86,18){\psellipse[linewidth=0.25,linestyle=none,fillstyle=solid](0,0)(1,1)}
\rput{0}(106,18){\psellipse[linewidth=0.25,linestyle=none,fillstyle=solid](0,0)(1,1)}
\rput{0}(116,18){\psellipse[linewidth=0.25,linestyle=none,fillstyle=solid](0,0)(1,1)}
\rput{0}(96,18){\psellipse[linewidth=0.25,linestyle=none,fillstyle=solid](0,0)(1,1)}
\psline[linewidth=0.1](118,16)(84,16)
\psline[linewidth=0.1](118,20)(84,20)
\psline[linewidth=0.1](118,20)(118,16)
\psline[linewidth=0.1](84,20)(84,16)
\psline[linewidth=0.25,fillcolor=white,fillstyle=solid](86,28)(86,18)
\psline[linewidth=0.25,fillcolor=white,fillstyle=solid](106,28)(116,18)
\psline[linewidth=0.25,fillcolor=white,fillstyle=solid](96,18)(96,28)
\psline[linewidth=0.25,fillcolor=white,fillstyle=solid](106,28)(96,18)
\psline[linewidth=0.25,fillcolor=white,fillstyle=solid](116,28)(116,18)
\rput[l](119,28){$T^+$}
\psline[linewidth=0.25,fillcolor=white,fillstyle=solid](106,18)(96,28)
\rput{0}(86,28){\psellipse[linewidth=0.25,linestyle=none,fillstyle=solid](0,0)(1,1)}
\rput{0}(106,28){\psellipse[linewidth=0.25,linestyle=none,fillstyle=solid](0,0)(1,1)}
\rput{0}(116,28){\psellipse[linewidth=0.25,linestyle=none,fillstyle=solid](0,0)(1,1)}
\rput{0}(96,28){\psellipse[linewidth=0.25,linestyle=none,fillstyle=solid](0,0)(1,1)}
\psline[linewidth=0.1](118,26)(84,26)
\psline[linewidth=0.1](118,30)(84,30)
\psline[linewidth=0.1](118,30)(118,26)
\psline[linewidth=0.1](84,30)(84,26)
\psline[linewidth=0.25,fillcolor=white,fillstyle=solid](106,18)(106,28)
\rput[B](38,0.5){\begin{footnotesize}putting an antichain $M$\end{footnotesize}}
\psline[linewidth=0.2,fillcolor=white,fillstyle=solid](19,-1)(57,-1)
\psline[linewidth=0.2,fillcolor=white,fillstyle=solid](56,0)(57.07,-1.07)
\psline[linewidth=0.2,fillcolor=white,fillstyle=solid](57.07,-0.93)(56,-2)
\rput(53,0){}
\rput[B](80,0.5){\begin{footnotesize}choosing new structure\end{footnotesize}}
\psline[linewidth=0.2,fillcolor=white,fillstyle=solid](61,-1)(99,-1)
\psline[linewidth=0.2,fillcolor=white,fillstyle=solid](98,0)(99.07,-1.07)
\psline[linewidth=0.2,fillcolor=white,fillstyle=solid](99.07,-0.94)(98,-2)
\rput(95,0){}
\end{pspicture}
\caption{$ $}\label{F:Example Local 2}
\end{figure}

\bigskip

Our aim is to show how Algorithm's strategy for the local game with Property \ref{Pr:disjoint} can be transformed into a strategy for Algorithm playing an unrestricted game of Definition \ref{D:local game}. Such a transformation has to be done on-line, i.e. each round has to be transformed separately. In particular there is nothing to be changed in the first round. However, in all other rounds some changes are needed. \linebreak Roughly speaking in such a transformation:
\begin{itemize}
\item We start with an unrestricted board $\brackets{L,T, \leq \chains}$ and an antichain $M$ not necessarily disjoint with $L\cup T$.
\item We transform $L, M, T$ into three pairwise disjoint antichains $\ol{L}, \ol{M}, \ol{T}$ of size $w$ and together with $\chainso$ for $\ol{L} \cup \ol{T}$ the board $\brackets{\ol{L}, \ol{T}, \leq, \chainso}$ is presented to Algorithm working in the disjoint game to get the multicoloring $\chainso^{\plus}$ of $\ol{M}$.
\item We use the multicoloring $\chainso^{\plus}$ of $\ol{M}$ to produce a multicoloring $\chains^{\plus}$ of $M \setminus (L \cup T)$.
\end{itemize}

For a better understanding of this idea we look into the example presented on Figure \ref{F:Example Local} (page \pageref{F:Example Local}). The poset $\brackets{\set{p_1,\ldots, p_7},\leq}$ with antichains $L=\set{p_1,p_5,p_6,p_7}$ and $T=\set{p_1,p_2,p_3,p_4}$ is shown on Figure \ref{F:Example Local tr1}.a.
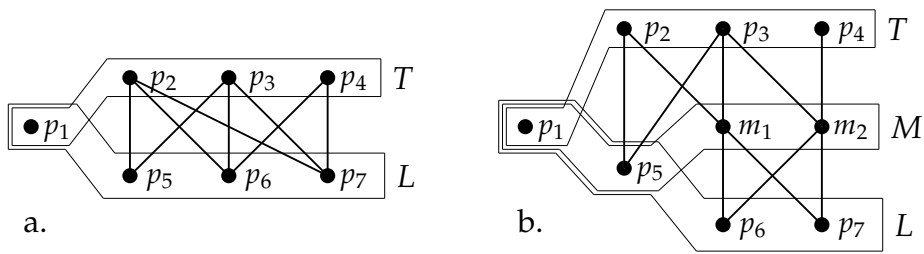
\begin{figure}[hbt]
\newcommand{\p}[1]{\begin{small}$p_{#1}$\end{small}}
\newcommand{\m}[1]{\begin{small}$m_{#1}$\end{small}}
\newcommand{\T}{$T$}
\renewcommand{\L}{$L$}
\newcommand{\M}{$M$}
\centering\ifx\JPicScale\undefined\def\JPicScale{1}\fi
\psset{unit=\JPicScale mm}
\psset{linewidth=0.3,dotsep=1,hatchwidth=0.3,hatchsep=1.5,shadowsize=1,dimen=middle}
\psset{dotsize=0.7 2.5,dotscale=1 1,fillcolor=black}
\psset{arrowsize=1 2,arrowlength=1,arrowinset=0.25,tbarsize=0.7 5,bracketlength=0.15,rbracketlength=0.15}
\begin{pspicture}(0,0)(120,31.5)
\psline[linewidth=0.25,fillcolor=white,fillstyle=solid](29,22.5)(29,9.5)
\psline[linewidth=0.25,fillcolor=white,fillstyle=solid](29,22.5)(42,9.5)
\psline[linewidth=0.25,fillcolor=white,fillstyle=solid](16,22.5)(42,9.5)
\psline[linewidth=0.25,fillcolor=white,fillstyle=solid](16,9.5)(16,22.5)
\psline[linewidth=0.25,fillcolor=white,fillstyle=solid](29,22.5)(16,9.5)
\psline[linewidth=0.25,fillcolor=white,fillstyle=solid](42,22.5)(42,9.5)
\rput[Bl](50.5,21){\T}
\rput[Bl](51,8){\L}
\psline[linewidth=0.25,fillcolor=white,fillstyle=solid](29,9.5)(16,22.5)
\psline[linewidth=0.25,fillcolor=white,fillstyle=solid](42,22.5)(29,9.5)
\psline[linewidth=0.1](49.5,12.5)(14,12.5)
\psline[linewidth=0.1](49.5,6.5)(12.5,6.5)
\psline[linewidth=0.1](0,19)(0,13)
\rput{0}(16,9.5){\psellipse[linewidth=0.25,linestyle=none,fillstyle=solid](0,0)(1,1)}
\rput{0}(42,22.5){\psellipse[linewidth=0.25,linestyle=none,fillstyle=solid](0,0)(1,1)}
\rput{0}(29,22.5){\psellipse[linewidth=0.25,linestyle=none,fillstyle=solid](0,0)(1,1)}
\rput{0}(3,16){\psellipse[linewidth=0.25,linestyle=none,fillstyle=solid](0,0)(1,1)}
\rput{0}(16,22.5){\psellipse[linewidth=0.25,linestyle=none,fillstyle=solid](0,0)(1,1)}
\rput{0}(42,9.5){\psellipse[linewidth=0.25,linestyle=none,fillstyle=solid](0,0)(1,1)}
\rput{0}(29,9.5){\psellipse[linewidth=0.25,linestyle=none,fillstyle=solid](0,0)(1,1)}
\psline[linewidth=0.1](49.5,12.5)(49.5,6.5)
\psline[linewidth=0.1](49,20)(13,20)
\psline[linewidth=0.1](49,25)(13,25)
\psline[linewidth=0.1](49,25)(49,20)
\psline[linewidth=0.1](0.5,18.5)(0.5,13.5)
\psline[linewidth=0.25,fillcolor=white,fillstyle=solid](94,16)(94,3)
\psline[linewidth=0.25,fillcolor=white,fillstyle=solid](94,16)(107,3)
\psline[linewidth=0.25,fillcolor=white,fillstyle=solid](107,16)(107,3)
\psline[linewidth=0.25,fillcolor=white,fillstyle=solid](107,16)(94,3)
\psline[linewidth=0.1](115,6.5)(91.5,6.5)
\psline[linewidth=0.1](115,-0.5)(90.5,-0.5)
\psline[linewidth=0.1](65,19)(65,13)
\rput{0}(68,16){\psellipse[linewidth=0.25,linestyle=none,fillstyle=solid](0,0)(1,1)}
\rput{0}(94,16){\psellipse[linewidth=0.25,linestyle=none,fillstyle=solid](0,0)(1,1)}
\rput{0}(107,16){\psellipse[linewidth=0.25,linestyle=none,fillstyle=solid](0,0)(1,1)}
\rput{0}(81,10.5){\psellipse[linewidth=0.25,linestyle=none,fillstyle=solid](0,0)(1,1)}
\rput{0}(94,3){\psellipse[linewidth=0.25,linestyle=none,fillstyle=solid](0,0)(1,1)}
\rput{0}(107,3){\psellipse[linewidth=0.25,linestyle=none,fillstyle=solid](0,0)(1,1)}
\psline[linewidth=0.1](115,6.5)(115,-0.5)
\psline[linewidth=0.1](114.5,13)(91.5,13)
\psline[linewidth=0.1](114.5,19)(90.5,19)
\psline[linewidth=0.1](114.5,19)(114.5,13)
\psline[linewidth=0.1](65.5,18.5)(65.5,13.5)
\psline[linewidth=0.25,fillcolor=white,fillstyle=solid](94,29)(107,16)
\psline[linewidth=0.25,fillcolor=white,fillstyle=solid](81,10.5)(81,29)
\psline[linewidth=0.25,fillcolor=white,fillstyle=solid](94,29)(81,10.5)
\psline[linewidth=0.25,fillcolor=white,fillstyle=solid](107,29)(107,16)
\psline[linewidth=0.25,fillcolor=white,fillstyle=solid](94,16)(81,29)
\rput{0}(107,29){\psellipse[linewidth=0.25,linestyle=none,fillstyle=solid](0,0)(1,1)}
\rput{0}(94,29){\psellipse[linewidth=0.25,linestyle=none,fillstyle=solid](0,0)(1,1)}
\rput{0}(81,29){\psellipse[linewidth=0.25,linestyle=none,fillstyle=solid](0,0)(1,1)}
\psline[linewidth=0.1](114,26.5)(79,26.5)
\psline[linewidth=0.1](114,31.5)(78.5,31.5)
\psline[linewidth=0.1](114,31.5)(114,26.5)
\psline[linewidth=0.1](64.5,19.5)(64.5,12.5)
\psline[linewidth=0.25,fillcolor=white,fillstyle=solid](94,16)(94,29)
\rput[B](99,0){}
\psline[linewidth=0.1](0.5,13.5)(8,13.5)
\psline[linewidth=0.1](0.5,18.5)(8,18.5)
\psline[linewidth=0.1](8,18.5)(13,25)
\psline[linewidth=0.1](8,13.5)(13,20)
\psline[linewidth=0.1](7.5,13)(12.5,6.5)
\psline[linewidth=0.1](7.5,13)(0,13)
\psline[linewidth=0.1](0,19)(9,19)
\psline[linewidth=0.1](14,12.5)(9,19)
\psline[linewidth=0.1](85.5,14)(79.5,14)
\psline[linewidth=0.1](73.5,13.5)(65.5,13.5)
\psline[linewidth=0.1](85.5,7.5)(77.5,7.5)
\psline[linewidth=0.1](73,18.5)(65.5,18.5)
\psline[linewidth=0.1](73,18.5)(78.5,31.5)
\psline[linewidth=0.1](73.5,13.5)(79,26.5)
\psline[linewidth=0.1](84.5,13.5)(90.5,19)
\psline[linewidth=0.1](85.5,7.5)(91.5,13)
\psline[linewidth=0.1](79.5,14)(74.5,19.5)
\psline[linewidth=0.1](77,7)(72,12.5)
\psline[linewidth=0.1](74.5,19.5)(64.5,19.5)
\psline[linewidth=0.1](72,12.5)(64.5,12.5)
\psline[linewidth=0.1](65,13)(72.5,13)
\psline[linewidth=0.1](72.5,13)(77.5,7.5)
\psline[linewidth=0.1](65,19)(74,19)
\psline[linewidth=0.1](77,7)(84.5,7)
\psline[linewidth=0.1](84.5,7)(90.5,-0.5)
\psline[linewidth=0.1](74,19)(79,13.5)
\psline[linewidth=0.1](79,13.5)(84.5,13.5)
\psline[linewidth=0.1](85.5,14)(91.5,6.5)
\rput[Bl](96,2){\p6}
\rput[Bl](108.5,15){\m2}
\rput[Bl](96,15){\m1}
\rput[Bl](109,2){\p7}
\rput[Bl](82.5,9.5){\p5}
\rput[Bl](109,28){\p4}
\rput[Bl](96.5,28){\p3}
\rput[Bl](83.5,28){\p2}
\rput[Bl](69.5,15){\p1}
\rput[Bl](31,8.5){\p6}
\rput[Bl](43.5,8.5){\p7}
\rput[Bl](18,8.5){\p5}
\rput[Bl](43.5,21.5){\p4}
\rput[Bl](31.5,21.5){\p3}
\rput[Bl](18.5,21.5){\p2}
\rput[Bl](4.5,15){\p1}
\rput[Bl](115.5,27.5){\T}
\rput[Bl](116.5,1.5){\L}
\rput[Bl](116,14.5){\M}
\psline[linewidth=0.1,linecolor=white](120,24.5)(120,18.5)
\rput[Bl](2,2){a.}
\rput[Bl](67,2){b.}
\end{pspicture}
\caption{Example of Spoiler's move, when Property \ref{Pr:disjoint} isn't in force.}\label{F:Example Local tr1}
\end{figure}
Spoiler introduces $2$ new points $m_1, m_2$ and reveals an antichain $M=\set{p_1, p_5, m_1, m_2}$ of the new poset \linebreak $\brackets{\set{p_1,\ldots, p_7,m_1,m_2},\leq}$. For $q \in L \cup M \cup T$ we know that $q$ may belong to more than one level $A \in \set{L, M, T}$. In our example $p_1\in L\cap M\cap T$ and $p_5\in L\cap M$. The idea is to inflate $p_1$ to a $3$-element chain $p_1^L < p_1^M < p_1^T$ and to inflate $p_5$ to a $2$-element chain $p_5^L < p_5^M$, as presented by Figure \ref{F:Example Local tr2}. The points that in $L \cup M \cup T$ belong to exactly one of antichains $L$, $M$ or $T$ get only an upper index of this antichain. Now, $\ol{L}, \ol{M}, \ol{T}$ are the antichains of points with the corresponding upper indices. The multicoloring $\chainso$ of $\ol{L} \cup \ol{T}$ is set so that $\fCo{}{q^A} = \fC{}{q}$. After presenting the board $\brackets{\ol{L}, \ol{T}, \leq, \chainso}$ together with an antichain $\ol{M}$ that is disjoint with $\ol{L} \cup \ol{T}$, Algorithm returns the multicoloring $\chainso^{\plus}$ of $\ol{M}=\set{p_1^M,p_5^M,m_1^M,m_2^M}$. 
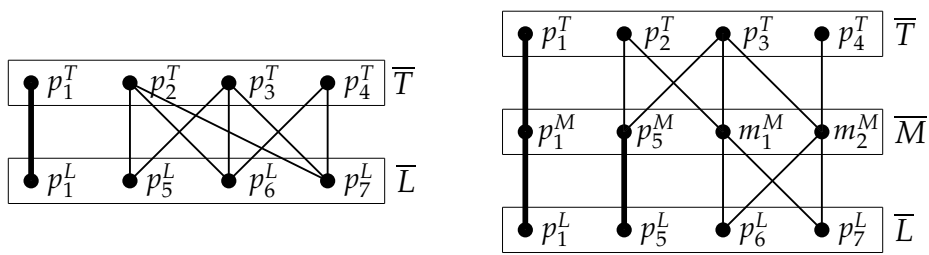
\begin{figure}[hbt]
\newcommand{\p}[2]{\begin{small}$p_{#1}^{#2}$\end{small}}
\newcommand{\m}[2]{\begin{small}$m_{#1}^{#2}$\end{small}}
\newcommand{\T}{$\ol{T}$}
\renewcommand{\L}{$\ol{L}$}
\newcommand{\M}{$\ol{M}$}
\centering\ifx\JPicScale\undefined\def\JPicScale{1}\fi
\psset{unit=\JPicScale mm}
\psset{linewidth=0.3,dotsep=1,hatchwidth=0.3,hatchsep=1.5,shadowsize=1,dimen=middle}
\psset{dotsize=0.7 2.5,dotscale=1 1,fillcolor=black}
\psset{arrowsize=1 2,arrowlength=1,arrowinset=0.25,tbarsize=0.7 5,bracketlength=0.15,rbracketlength=0.15}
\begin{pspicture}(0,0)(120,32)
\psline[linewidth=0.25,fillcolor=white,fillstyle=solid](29,22.5)(29,9.5)
\psline[linewidth=0.25,fillcolor=white,fillstyle=solid](29,22.5)(42,9.5)
\psline[linewidth=0.25,fillcolor=white,fillstyle=solid](16,22.5)(42,9.5)
\psline[linewidth=0.25,fillcolor=white,fillstyle=solid](16,9.5)(16,22.5)
\psline[linewidth=0.25,fillcolor=white,fillstyle=solid](29,22.5)(16,9.5)
\psline[linewidth=0.25,fillcolor=white,fillstyle=solid](42,22.5)(42,9.5)
\rput[Bl](50.5,21){\T}
\rput[Bl](51,8){\L}
\psline[linewidth=0.25,fillcolor=white,fillstyle=solid](29,9.5)(16,22.5)
\psline[linewidth=0.25,fillcolor=white,fillstyle=solid](42,22.5)(29,9.5)
\psline[linewidth=0.1](49.5,12.5)(0,12.5)
\psline[linewidth=0.1](49.5,6.5)(0,6.5)
\psline[linewidth=0.1](0,12.5)(0,6.5)
\rput{0}(16,9.5){\psellipse[linewidth=0.25,linestyle=none,fillstyle=solid](0,0)(1,1)}
\rput{0}(42,22.5){\psellipse[linewidth=0.25,linestyle=none,fillstyle=solid](0,0)(1,1)}
\rput{0}(29,22.5){\psellipse[linewidth=0.25,linestyle=none,fillstyle=solid](0,0)(1,1)}
\rput{0}(3,9.5){\psellipse[linewidth=0.25,linestyle=none,fillstyle=solid](0,0)(1,1)}
\rput{0}(16,22.5){\psellipse[linewidth=0.25,linestyle=none,fillstyle=solid](0,0)(1,1)}
\rput{0}(42,9.5){\psellipse[linewidth=0.25,linestyle=none,fillstyle=solid](0,0)(1,1)}
\rput{0}(29,9.5){\psellipse[linewidth=0.25,linestyle=none,fillstyle=solid](0,0)(1,1)}
\psline[linewidth=0.1](49.5,12.5)(49.5,6.5)
\psline[linewidth=0.1](49.5,19.5)(0,19.5)
\psline[linewidth=0.1](49.5,25.5)(0,25.5)
\psline[linewidth=0.1](49.5,25.5)(49.5,19.5)
\psline[linewidth=0.25,fillcolor=white,fillstyle=solid](94,16)(94,3)
\psline[linewidth=0.25,fillcolor=white,fillstyle=solid](94,16)(107,3)
\psline[linewidth=0.25,fillcolor=white,fillstyle=solid](107,16)(107,3)
\psline[linewidth=0.25,fillcolor=white,fillstyle=solid](107,16)(94,3)
\psline[linewidth=0.1](115,6)(65,6)
\psline[linewidth=0.1](115,0)(65,0)
\psline[linewidth=0.1](65,19)(65,13)
\rput{0}(68,16){\psellipse[linewidth=0.25,linestyle=none,fillstyle=solid](0,0)(1,1)}
\rput{0}(94,16){\psellipse[linewidth=0.25,linestyle=none,fillstyle=solid](0,0)(1,1)}
\rput{0}(107,16){\psellipse[linewidth=0.25,linestyle=none,fillstyle=solid](0,0)(1,1)}
\rput{0}(81,16){\psellipse[linewidth=0.25,linestyle=none,fillstyle=solid](0,0)(1,1)}
\rput{0}(94,3){\psellipse[linewidth=0.25,linestyle=none,fillstyle=solid](0,0)(1,1)}
\rput{0}(107,3){\psellipse[linewidth=0.25,linestyle=none,fillstyle=solid](0,0)(1,1)}
\psline[linewidth=0.1](115,6)(115,0)
\psline[linewidth=0.1](115,13)(65,13)
\psline[linewidth=0.1](115,19)(65,19)
\psline[linewidth=0.1](115,19)(115,13)
\psline[linewidth=0.1](65,32)(65,26)
\psline[linewidth=0.25,fillcolor=white,fillstyle=solid](94,29)(107,16)
\psline[linewidth=0.25,fillcolor=white,fillstyle=solid](81,16)(81,29)
\psline[linewidth=0.25,fillcolor=white,fillstyle=solid](94,29)(81,16)
\psline[linewidth=0.25,fillcolor=white,fillstyle=solid](107,29)(107,16)
\psline[linewidth=0.25,fillcolor=white,fillstyle=solid](94,16)(81,29)
\rput{0}(107,29){\psellipse[linewidth=0.25,linestyle=none,fillstyle=solid](0,0)(1,1)}
\rput{0}(94,29){\psellipse[linewidth=0.25,linestyle=none,fillstyle=solid](0,0)(1,1)}
\rput{0}(81,29){\psellipse[linewidth=0.25,linestyle=none,fillstyle=solid](0,0)(1,1)}
\psline[linewidth=0.1](115,26)(65,26)
\psline[linewidth=0.1](115,32)(65,32)
\psline[linewidth=0.1](115,32)(115,26)
\psline[linewidth=0.1](65,6)(65,0)
\psline[linewidth=0.25,fillcolor=white,fillstyle=solid](94,16)(94,29)
\rput[B](99,0){}
\rput[Bl](96,2){\p6L}
\rput[Bl](108.5,15){\m2M}
\rput[Bl](96,15){\m1M}
\rput[Bl](109,2){\p7L}
\rput[Bl](83,2){\p5L}
\rput[Bl](109,28){\p4T}
\rput[Bl](96.5,28){\p3T}
\rput[Bl](83.5,28){\p2T}
\rput[Bl](69.5,15){\p1M}
\rput[Bl](31.5,8.5){\p6L}
\rput[Bl](44,8.5){\p7L}
\rput[Bl](18,8.5){\p5L}
\rput[Bl](44,21.5){\p4T}
\rput[Bl](31.5,21.5){\p3T}
\rput[Bl](18.5,21.5){\p2T}
\rput[Bl](5,21.5){\p1T}
\rput[Bl](116.5,27.5){\T}
\rput[Bl](116.5,1.5){\L}
\rput[Bl](116.5,14.5){\M}
\psline[linewidth=0.1,linecolor=white](120,24.5)(120,18.5)
\rput{0}(3,22.5){\psellipse[linewidth=0.25,linestyle=none,fillstyle=solid](0,0)(1,1)}
\psline[linewidth=0.75,fillcolor=white,fillstyle=solid](3,9.5)(3,22.5)
\psline[linewidth=0.1](0,25.5)(0,19.5)
\rput[Bl](5,8.5){\p1L}
\rput{0}(68,29){\psellipse[linewidth=0.25,linestyle=none,fillstyle=solid](0,0)(1,1)}
\rput{0}(68,3){\psellipse[linewidth=0.25,linestyle=none,fillstyle=solid](0,0)(1,1)}
\rput{0}(81,3){\psellipse[linewidth=0.25,linestyle=none,fillstyle=solid](0,0)(1,1)}
\psline[linewidth=0.75,fillcolor=white,fillstyle=solid](81,16)(81,3)
\psline[linewidth=0.75,fillcolor=white,fillstyle=solid](68,29)(68,16)
\psline[linewidth=0.75,fillcolor=white,fillstyle=solid](68,16)(68,3)
\rput[Bl](70,28){\p1T}
\rput[Bl](70,2){\p1L}
\rput[Bl](82.5,15){\p5M}
\end{pspicture}
\caption{Posets with inflating points $p_1$ and $p_5$. }\label{F:Example Local tr2}
\end{figure} 
In order to obtain a multicoloring $\chains^{\plus}$ of $m_1$ and $m_2$ we take the multicoloring $\chainso^{\plus}$ of $m_1^M$ and $m_2^M$, respectively. The values $\fCo{\plus} {p_1^M}$ and $\fCo{\plus}{p_5^M}$ are of no help, as $\fC{\plus}{p_1}$ and $\fC{\plus}{p_5}$ can be already deduced from $\fC{}{p_1}$ and $\fC{}{p_5}$.

\bigskip

Before describing the above transformation in a formal way we need some preparation. First of all we need to describe how a two-level structure $\brackets{L, T, \leq}$ or a three-level structure $\brackets{L, M, T, \leq}$ can be turn into $\brackets{\ol{L}, \ol{T}, \leq}$ or $\brackets{\ol{L}, \ol{M}, \ol{T},\leq}$, respectively. We describe how to proceed in the three-level case, as it should be obvious what to do in the two-level case. Recall that
\begin{itemize}
\item the poset $\brackets{L\cup M\cup T,\leq}$ has width $w$,
\item $L,M,T$ are antichains of size $w$,
\item $L\al M\al T$ in the lattice $\fMA{L\cup M\cup T,\leq}$.
\end{itemize}
The antichains $L, M, T$ are transformed into
\begin{align*}
\ol{L}&\ =\ \{\ l^L\ : l\in L\ \},\\
\ol{M}&\ =\ \{m^M\!:m\in M\},\\
\ol{T}&\ =\ \{\ t^T\ : t\in T\ \},
\end{align*}
so that the sets $\ol{L}, \ol{M}, \ol{T}$ are now pairwise disjoint. The ordering of $\ol{L} \cup \ol{M} \cup \ol{T}$ is defined by
\[
x^A \leq y^B\quad \textrm{iff}\quad x \leq y\ \ \textrm{and}\ \ A\aleq B.
\]

In the following we will switch between two games:
\begin{itemize}
 \item an \emph{unrestricted game} described in Definition \ref{D:local game}. Their players are to be called \emph{Unrestricted Spoiler} and \emph{Unrestricted Algorithm}.
\item a \emph{disjoint game} described in Definition \ref{D:disjoint game}. Their players are to be called \emph{Disjoint Spoiler} and \emph{Disjoint Algorithm}.
\end{itemize}
To color the initial poset $\brackets{L, T, \leq}$ of the first round, note that since $L \cap T = \emptyset$, passing from $(L, T, \leq)$ to $\brackets{\ol{L}, \ol{T}, \leq}$ relies only on renaming the elements of $L\cup T$ without actually duplicating them. Then $\brackets{\ol{L}, \ol{T}, \leq}$ is presented to Disjoint Algorithm, who returns a colored board $\brackets{\ol{L}, \ol{T}, \leq, \chainso}$. The multicoloring Unrestricted Algorithm has to return for $(L, T, \leq)$ can be produced from $\chainso$ in the obvious way, namely $\fC{}{x} = \fCo{}{x^A}$. For the further performance of Unrestricted Algorithm not only $\brackets{L, T, \leq, \chains}$ is stored but actually, in each round from now on, a structure $\brackets{L, T, \leq, \chains, \chainso}$ is maintained.
To be more precise $\brackets{L, T, \leq, \chains}$ is a board for the unrestricted game while $\chainso$ is a multicoloring of $\ol{L} \cup \ol{T}$ returned by Disjoint Algorithm and then restricted by Disjoint Spoiler to the two levels of his choice. Note that after the first round we have
\begin{equation*}
 \chainso\mspace{-1.5mu}\big(x^A\big)\ =\ \fC{}{x}\qquad\textrm{for}\ \ A \in \set{L, T}.
\end{equation*}
This equality can not be maintained in the future but we will keep the following inclusion as an invariant.
\begin{equation}\label{E:DG inv}
 \chainso\mspace{-1.5mu}\big(x^A\big)\ \subseteq\ \fC{}{x}\qquad\textrm{for}\ \ A \in \set{L, T}.
\end{equation}
To describe how to proceed in the consecutive rounds of the unrestricted game, suppose that Unrestricted Spoiler presents the middle antichain $M$ to Unrestricted Algorithm who maintains $\brackets{L, T, \leq, \chains, \chainso}$. Unrestricted Algorithm produces $\brackets{\ol{L}, \ol{M}, \ol{T}, \leq}$ and together with the multicoloring $\chainso$ of $\ol{L} \cup \ol{T}$ presents it to Disjoint Algorithm. He extends $\chainso$ to a multicoloring $\chainso^{\plus}$ of entire $\ol{L} \cup \ol{M} \cup \ol{T}$ and makes Disjoint Spoiler to set $\big(\ol{L}^{\plus}, \ol{T}^{\plus}\big)$ to be either $\brackets{\ol{L}, \ol{M}}$ or $\brackets{\ol{M}, \ol{T}}$ depending whether $\brackets{L, M}$ or $\brackets{M, T}$ was chosen by Unrestricted Spoiler for $\brackets{L^{\plus}, T^{\plus}}$. Now the board $\big(\ol{L}^{\plus}, \ol{T}^{\plus}, \leq, \chainso^{\plus}\big)$, with $\chainso^{\plus}$ restricted to $\ol{L}^{\plus}, \ol{T}^{\plus}$, is returned to Unrestricted Algorithm. All Unrestricted Algorithm has to do now is to extend $\chains$ to a multicoloring $\chains^{\plus}$ of $M\setminus\brackets{L \cup T}$. We will show that putting $\fC{\plus}{m} = \fCo{\plus}{m^M}$ fulfills the requirement of Definition \ref{D:local game}.
To show that described transformation works we need to show the following three properties.
\begin{enumerate}
 \item presenting $\big(\ol{L}, \ol{M}, \ol{T}, \leq, \chainso\big)$ to Disjoint Algorithm is legal, i.e.
\begin{enumerate}
\item $\ol{M}$ is an antichain of width $w$,\label{Enu:DG to do - olM is an ant. of width w}
\item $\ol{L}, \ol{M}, \ol{T}$ are pairwise disjoint,\label{Enu:DG to do - olL,olM,olT disjoint}
\item $\fWidth{\ol{L} \cup \ol{M} \cup \ol{T}}=w$,
\item $\ol{L}\al \ol{M} \al \ol{T}$ holds in $\fMA{\ol{L}\cup \ol{M} \cup \ol{T},\leq}$,\label{Enu:DG to do - olL al olM alolT}
\end{enumerate}
\item the multicoloring $\chains^{\plus}$ is legal, i.e.\label{Enu:DG to do - c+ is legal}
\begin{enumerate}
 \item for each color $\gamma$ the set $\set{p \in L\cup M \cup T : \gamma \in \fC{\plus}{p}}$ is a chain,
\item $\fC{\plus}{m} \subseteq \fC{}{L}\cap\fC{}{T}$ for $m \in M \setminus\brackets{L \cup T}$,\label{Enu:DG to do - c+ sub sum c(L) inter sum c(T)}
\end{enumerate}
\item the invariant (\ref{E:DG inv}) is kept, i.e. $\fC{\plus}{x} \supseteq \fCo{\plus}{x^A}$ for each antichain $A=L,M,T$ and each $x \in A$.\label{Enu:DG to do - inv}
\end{enumerate}
The properties (\ref{Enu:DG to do - olM is an ant. of width w}) and (\ref{Enu:DG to do - olL,olM,olT disjoint}) and the inequality $\fWidth{\ol{L} \cup \ol{M} \cup \ol{T}} \geq w$ are obvious. The converse inequality as well as (\ref{Enu:DG to do - olL al olM alolT}) will follow from the next two claims.

\begin{clm}\label{C:X<Y}
$\fWidth{\ol{L}\cup\ol{M}\cup\ol{T}}\leq w$.
\end{clm}
\begin{proof}
Since $\fWidth{L\cup M \cup T }=w$, Dilworth's Theorem \ref{T:Dilworth} allows us to partition $L\cup M \cup T $ into $w$ chains, say $C_1,\ldots, C_w$. To cover $\ol{L}\cup \ol{M} \cup \ol{T}$ by $w$ chains we let $$\ol{C}_i\ =\ \set{x^A : x \in C_i \cap A\ \textrm{and}\ A \in \set{L,M,T}}.$$ Obviously $\ol{C}_1 \cup \ldots \cup \ol{C}_w = \ol{L} \cup \ol{M} \cup \ol{T}$. To see that two points $x^A\!, y^B$ from $\ol{C}_i$ are comparable, first note that since $x,y$ are taken from a chain $C_i$ without loss of generality we may assume that $x\leq y$. Thus, if $x^A \not\leq y^B$ then $A \ag B$. This together with $x \in A$, $y\in B$ and $x \leq y$ gives $x=y$. Consequently $x^A \geq x^B = y^B$. 
\end{proof}

Since $\ol{L}, \ol{M}, \ol{T}$ have the same width $w$ as the entire poset \mbox{$\ol{L} \cup \ol{M} \cup \ol{T}$}, we know that $\ol{L}, \ol{M}, \ol{T} \in \fMA{\ol{L} \cup \ol{M} \cup \ol{T}, \leq}$. Since $L \al M \al T$, property (\ref{Enu:DG to do - olL al olM alolT}) follows directly from the following claim. 

\begin{clm}
For $A,B \in \set{L,M,T}$ with $A\aleq B$ we have $\ol{A}\aleq \ol{B}$.
\end{clm}
\begin{proof}
To see that an element $x^A$ of $\ol{A}$ is below some $y^B \in \ol{B}$, simply choose $y\in B$ with $x \leq y$, witnessing that $x \in A$ and $A \aleq B$.
\end{proof}

To show (\ref{Enu:DG to do - c+ is legal}) we first prove a condition that is stronger than (\ref{Enu:DG to do - c+ sub sum c(L) inter sum c(T)}).

\begin{clm}\label{C:lmt}
If $m\in M \setminus\brackets{L\cup T}$ and $\gamma\in\fC{\plus}{m}$ then there 
are $l\in L$ and $t\in T$ such that $\gamma\in \fC{}{l}\cap\fC{}{t}$ and $l\leq m\leq t$.
\end{clm}
\begin{proof}
Since $\fC{\plus}{m}$ was set by Unrestricted Algorithm to $\fC{\plus}{m^M\mspace{-1mu}}$ we have $\gamma\in\fCo{\plus}{m^M}$. On the other hand the rules of the disjoint game allow Disjoint Algorithm to use for $\fCo{\plus}{m^M}$ a subset of $\fCo{}{\ol{L}}\cap\fCo{}{\ol{T}}$. Therefore $\gamma\in \fCo{}{l^L}\cap \fCo{}{t^T}$ for some $l\in L$ and $t\in T$. Moreover, since $\set{x^A\in \ol{L} \cup \ol{M} \cup \ol{T} : \gamma\in \fCo{\plus}{x^A}}$ is a chain we know that $l^L < m^M < t^T$. These inequalities imply that $l\leq m\leq t$. Finally, using our invariant (\ref{E:DG inv}) for $\chainso$ and $\chains$ we have $\chainso\mspace{-1.5mu}\big(l^L\big)\subseteq \fC{}{l}$ and $\chainso\mspace{-1.5mu}\big(t^T\big)\subseteq \fC{}{t}$, so that $\gamma \in \fC{}{l} \cap \fC{}{t}$.
\end{proof}

From the next claim we obtain that for every color $\gamma$ the points in $L\cup M\cup T$ colored by $\gamma$ form a chain.

\begin{clm}
Every two points $p,q\in L\cup M\cup T$ with $\fC{\plus}{p}\cap\fC{\plus}{q}\neq\emptyset$ are comparable.
\end{clm}
\begin{proof}
Obviously if $p,q\in L\cup T$ then $\rest{\chains^{\plus}}{L\cup T} = \chains$ together with the chain condition for $\chains$ does the job. Now suppose $p,q \in \linebreak M \setminus\brackets{L\cup T}$. Since $\set{m^M: m\in M}$ is an antichain the sets of the form $\fCo{\plus}{m^M}$ are pairwise disjoint. As for $m\in M \setminus\brackets{L\cup T}$ we have $\fC{\plus}{m} = \fCo{\plus}{m^M}$, we know that $\fC{\plus}{p}\cap \fC{\plus}{q}= \emptyset$ unless $p=q$. Finally suppose that $p\in L\cup T$ while $q \in M \setminus\brackets{L\cup T}$ and pick $\gamma \in \fC{\plus}{q}$. Claim \ref{C:lmt} supplies us with a three-element chain \mbox{$L\ni l< q < t\in T$} such that $\gamma \in \fC{}{l}\cap \fC{}{t}\cap \fC{\plus}{q}$. This color $\gamma$ cannot be used on any other element $p$ of $L\cup T$ unless $p=l$ or $p=t$.
\end{proof}

Before we proceed with the proof of (\ref{Enu:DG to do - inv}) a few words clarifying the situation may be needed to understand why in (\ref{Enu:DG to do - inv}), and therefore in (\ref{E:DG inv}), only the inclusion can be kept. This is because the Disjoint Algorithm can sometimes get two copies of a point $x\in L\cup M\cup T$, say $x^L \in \ol{L}$ and $x^M\in \ol{M}$. As Disjoint Algorithm has no choice for $\fCo{\plus}{x^L}$, but to set it to $\fCo{}{x^L}$, he has a choice for $\fCo{\plus}{x^M}$. We will see in the proof of the next claim that his freedom is restricted to $\fCo{\plus}{x^M}\subseteq \fCo{\plus}{x^L}$. However if Unrestricted Spoiler returns $(L^{\plus}, T^{\plus})=(M,T)$ forcing $\big(\ol{L}^{\plus}, \ol{T}^{\plus}\big) =\big(\ol{M}, \ol{T}\big)$, only the inclusion $\fCo{\plus}{x^M}\subseteq\fC{\plus}{x}$ can survive. 

\begin{clm}
$\fCo{\plus}{x^A}\subseteq\fC{\plus}{x}$ for all points $x^A$ of $\ol{L}\cup\ol{M}\cup\ol{T}$. 
\end{clm}
\begin{proof} 
Again, our argument splits into three cases:
\begin{itemize}
 \item $x^A\in \ol{L}\cup \ol{T}$,
\item $x^A\in \ol{M}$ and $x\in M \setminus\brackets{L\cup T}$ (in particular $A=M)$,
\item $x^A\in \ol{M}$ and $x\in M\cap \brackets{L\cup T}$ (in particular $A=M$).
\end{itemize}
In the first case, since $\rest{\chains^{\plus}}{L\cup T}=\chains$ and $\rest{\chainso^{\plus}}{\ol{L}\cup\ol{T}}=\chainso$ the invariant (\ref{E:DG inv}) does the job.
Also the second case is easy, as then $\fC{\plus}{x}= \fCo{\plus}{x^M}$. Actually the only possibility when the equality between $\fC{\plus}{x}$ and $\fCo{\plus}{x^A}$ cannot be kept arises in the third case. In this setting without loss of generality we may assume that $x\in M\cap L$. The rules of the disjoint game guarantee that $\fCo{\plus}{x^M}\subseteq \fCo{}{\ol{L}}\cap\fCo{}{\ol{T}}$. Actually we have much more, namely $\fCo{\plus}{x^M}\subseteq \fCo{}{x^L}= \fCo{\plus}{x^L}$, as otherwise $\fCo{\plus}{x^M}\cap \fCo{\plus}{l^L}\neq \emptyset$ for some $l^L\in \ol{L}\setminus \set{x^L}$. This however cannot happen as $l,x \in L$ are incomparable, so that the incomparable points $l^L$ and $x^M$ cannot share a common color. On the other hand we have $\fCo{\plus}{x^L}\subseteq \fC{\plus}{x}$ as $x^L\in \ol{L}$ is treated in the first case. Summing up we get $\fCo{\plus}{x^M}\subseteq \fCo{\plus}{x^L}\subseteq \fC{\plus}{x}$, as required. 
\end{proof}

%\clearpage

\subsection{Cores in the board}\label{SS:core} 
In this section we show that the value of the local game with both Properties \ref{Pr:disjoint} and \ref{Pr:core} is bounded from below by $\fLCP{w}$. We have already seen (in Subsection \ref{SS:different}) that restricting Spoiler's power by imposing Property \ref{Pr:disjoint} does not help Algorithm. Thus, all we have to do is to show that an additional restriction (Property \ref{Pr:core}) put on the top of the disjoint game does not help Algorithm as well. This new game is to be called a \coredisjoint{} game.

Now, all that has to be done, is to transform each round of the disjoint game according to the following rules:
\begin{itemize}
\item Starting with a regular bipartite poset $\brackets{L,T,\leq}$ and an antichain $M$ which is disjoint with $L\cup T$ we have to produce a regular bipartite poset $\brackets{L,T,\leq'}$ which is a core (in particular $\rest{\leq'}{L\cup T}$ is contained in $\rest{\leq}{L\cup T} $\ ) and to modify $\leq$ on entire $L\cup M\cup T$ to get subrelation $\leq'$ of $\leq$ so that both $\brackets{L, M, \leq'}$ and $\brackets{M, T, \leq'}$ are cores.
\item This modified structure $\brackets{L,M,T,\leq'}$ is presented to \CoreDisjoint{} Algorithm to color $M$.
\item After this multicoloring is done, Disjoint Spoiler decides to choose $(L^{\plus},T^{\plus})$ to be either $(L,M)$ or $(M,T)$. 
\end{itemize}
The proper choice for $\leq'$ is essential here. Indeed, it may seem that taking $\leq'$ such that the corresponding $\prec '$ determines perfect matching in both $(L,M)$ and $(M,T)$ would be the best one, as it makes it extremely easy for Algorithm to color $M$ along this matching. For example, suppose $(L,T,\leq)$ is the poset presented on Figure \ref{F:Example omission}.a
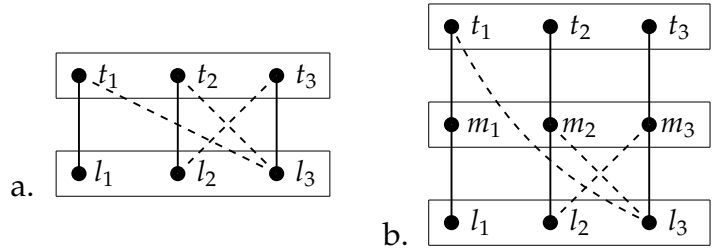
\begin{figure}[hbt]
\renewcommand{\l}[1]{\begin{small}$l_{#1}$\end{small}}
\newcommand{\m}[1]{\begin{small}$m_{#1}$\end{small}}
\renewcommand{\t}[1]{\begin{small}$t_{#1}$\end{small}}
\newcommand{\T}{$T$}
\renewcommand{\L}{$L$}
\newcommand{\M}{$M$}
\centering\ifx\JPicScale\undefined\def\JPicScale{1}\fi
\psset{unit=\JPicScale mm}
\psset{linewidth=0.3,dotsep=1,hatchwidth=0.3,hatchsep=1.5,shadowsize=1,dimen=middle}
\psset{dotsize=0.7 2.5,dotscale=1 1,fillcolor=black}
\psset{arrowsize=1 2,arrowlength=1,arrowinset=0.25,tbarsize=0.7 5,bracketlength=0.15,rbracketlength=0.15}
\begin{pspicture}(0,0)(92,32)
\psline[linewidth=0.25,fillcolor=white,fillstyle=solid](22,22.5)(22,9.5)
\psline[linewidth=0.25,fillcolor=white,fillstyle=solid](9,9.5)(9,22.5)
\rput[Bl](0,6.5){a.}
\rput{0}(9,9.5){\psellipse[linewidth=0.25,linestyle=none,fillstyle=solid](0,0)(1,1)}
\rput{0}(22,22.5){\psellipse[linewidth=0.25,linestyle=none,fillstyle=solid](0,0)(1,1)}
\rput{0}(9,22.5){\psellipse[linewidth=0.25,linestyle=none,fillstyle=solid](0,0)(1,1)}
\rput{0}(22,9.5){\psellipse[linewidth=0.25,linestyle=none,fillstyle=solid](0,0)(1,1)}
\psline[linewidth=0.25,fillcolor=white,fillstyle=solid](58,16)(58,3)
\psline[linewidth=0.25,fillcolor=white,fillstyle=solid](71,16)(71,3)
\rput{0}(58,16){\psellipse[linewidth=0.25,linestyle=none,fillstyle=solid](0,0)(1,1)}
\rput{0}(71,16){\psellipse[linewidth=0.25,linestyle=none,fillstyle=solid](0,0)(1,1)}
\rput{0}(58,3){\psellipse[linewidth=0.25,linestyle=none,fillstyle=solid](0,0)(1,1)}
\rput{0}(71,3){\psellipse[linewidth=0.25,linestyle=none,fillstyle=solid](0,0)(1,1)}
\psline[linewidth=0.25,fillcolor=white,fillstyle=solid](71,29)(71,16)
\rput{0}(71,29){\psellipse[linewidth=0.25,linestyle=none,fillstyle=solid](0,0)(1,1)}
\rput{0}(58,29){\psellipse[linewidth=0.25,linestyle=none,fillstyle=solid](0,0)(1,1)}
\psline[linewidth=0.25,fillcolor=white,fillstyle=solid](58,16)(58,29)
\rput[Bl](72.5,15){\m2}
\rput[Bl](60,15){\m1}
\rput[Bl](24.5,8.5){\l2}
\rput[Bl](11,8.5){\l1}
\rput[Bl](24.5,21.5){\t2}
\rput[Bl](11.5,21.5){\t1}
\rput[Bl](60.5,28){\t1}
\rput[Bl](73.5,28){\t2}
\rput[Bl](60.5,2){\l1}
\rput[Bl](73.5,2){\l2}
\rput[Bl](49,0){b.}
\psline[linewidth=0.25,fillcolor=white,fillstyle=solid](35,22.5)(35,9.5)
\psline[linewidth=0.25,linestyle=dashed,dash=1 1,fillcolor=white,fillstyle=solid](35,9.5)(22,22.5)
\rput{0}(35,22.5){\psellipse[linewidth=0.25,linestyle=none,fillstyle=solid](0,0)(1,1)}
\rput{0}(35,9.5){\psellipse[linewidth=0.25,linestyle=none,fillstyle=solid](0,0)(1,1)}
\psline[linewidth=0.25,linestyle=dashed,dash=1 1,fillcolor=white,fillstyle=solid](35,9.5)(9,22.5)
\psline[linewidth=0.25,fillcolor=white,fillstyle=solid](84,16)(84,3)
\rput{0}(84,16){\psellipse[linewidth=0.25,linestyle=none,fillstyle=solid](0,0)(1,1)}
\rput{0}(84,3){\psellipse[linewidth=0.25,linestyle=none,fillstyle=solid](0,0)(1,1)}
\psline[linewidth=0.25,fillcolor=white,fillstyle=solid](84,29)(84,16)
\rput{0}(84,29){\psellipse[linewidth=0.25,linestyle=none,fillstyle=solid](0,0)(1,1)}
\rput[Bl](85.5,15){\m3}
\psline[linewidth=0.25,linestyle=dashed,dash=1 1,fillcolor=white,fillstyle=solid](84,3)(71,16)
\rput[Bl](37.5,8.5){\l3}
\rput[Bl](37.5,21.5){\t3}
\rput[Bl](86.5,28){\t3}
\rput[Bl](86.5,2){\l3}
\psline[linewidth=0.1](42,25.5)(6,25.5)
\psline[linewidth=0.1](42,19.5)(6,19.5)
\psline[linewidth=0.1](6,25.5)(6,19.5)
\psline[linewidth=0.1](42,25.5)(42,19.5)
\psline[linewidth=0.1](42,12.5)(6,12.5)
\psline[linewidth=0.1](42,6.5)(6,6.5)
\psline[linewidth=0.1](6,12.5)(6,6.5)
\psline[linewidth=0.1](42,12.5)(42,6.5)
\psline[linewidth=0.1](92,32)(55,32)
\psline[linewidth=0.1](92,26)(55,26)
\psline[linewidth=0.1](92,32)(92,26)
\psline[linewidth=0.1](92,19)(55,19)
\psline[linewidth=0.1](92,13)(55,13)
\psline[linewidth=0.1](92,19)(92,13)
\psline[linewidth=0.1](92,6)(55,6)
\psline[linewidth=0.1](92,0)(55,0)
\psline[linewidth=0.1](92,6)(92,0)
\psline[linewidth=0.1](55,32)(55,26)
\psline[linewidth=0.1](55,19)(55,13)
\psline[linewidth=0.1](55,6)(55,0)
\psline[linewidth=0.25,linestyle=dashed,dash=1 1,fillcolor=white,fillstyle=solid](35,22.5)(22,9.5)
\psline[linewidth=0.25,linestyle=dashed,dash=1 1,fillcolor=white,fillstyle=solid](71,3)(84,16)
\psbezier[linewidth=0.25,linestyle=dashed,dash=1 1](58,29)(64,17)(71,9)(84,3)
\end{pspicture}
\caption{(a) After Algorithm deleted the edges $\brackets{l_2, t_3}$, $\brackets{l_3, t_1}$, $\brackets{l_3,t_2}$, Spoiler presents $m_1,m_2,m_3$ with solid and doted lines. (b) Doted lines have to be removed to keep comparabilities between $L$ and $T$ as the solid lines on (a) show.}\label{F:Example omission}
\end{figure} with both solid and doted lines, while $\leq'$ consists only of solid lines. Presenting $M=\set{m_1, m_2, m_3}$ with comparabilities described on Figure \ref{F:Example omission}.b by both solid and doted lines has to be transformed into $\leq'$ on $L\cup M\cup T$ e.g. by removing doted lines. This however may lead into future troubles, as it may happen that the further rounds cannot be transformed any more with keeping the same width (see Figure \ref{F:Example del}).
Therefore this will only show that the value of the disjoint game on width $w$ is bounded by the value of the \coredisjoint{} game with some width $w'\geq w$, which is not very helpful.
\begin{figure}[hbt]
\renewcommand{\l}[1]{\begin{small}$l_{#1}^{\plus}$\end{small}}
\newcommand{\m}[1]{\begin{small}$m_{#1}^{\plus}$\end{small}}
\renewcommand{\t}[1]{\begin{small}$t_{#1}^{\plus}$\end{small}}
\newcommand{\T}{$T^{\plus}$}
\renewcommand{\L}{$L^{\plus}$}
\newcommand{\M}{$M^{\plus}$}
\centering\ifx\JPicScale\undefined\def\JPicScale{1}\fi
\psset{unit=\JPicScale mm}
\psset{linewidth=0.3,dotsep=1,hatchwidth=0.3,hatchsep=1.5,shadowsize=1,dimen=middle}
\psset{dotsize=0.7 2.5,dotscale=1 1,fillcolor=black}
\psset{arrowsize=1 2,arrowlength=1,arrowinset=0.25,tbarsize=0.7 5,bracketlength=0.15,rbracketlength=0.15}
\begin{pspicture}(0,0)(92,32)
\psline[linewidth=0.25,fillcolor=white,fillstyle=solid](22,22.5)(22,9.5)
\psline[linewidth=0.25,fillcolor=white,fillstyle=solid](9,9.5)(9,22.5)
\rput[Bl](0,6.5){a.}
\rput{0}(9,9.5){\psellipse[linewidth=0.25,linestyle=none,fillstyle=solid](0,0)(1,1)}
\rput{0}(22,22.5){\psellipse[linewidth=0.25,linestyle=none,fillstyle=solid](0,0)(1,1)}
\rput{0}(9,22.5){\psellipse[linewidth=0.25,linestyle=none,fillstyle=solid](0,0)(1,1)}
\rput{0}(22,9.5){\psellipse[linewidth=0.25,linestyle=none,fillstyle=solid](0,0)(1,1)}
\psline[linewidth=0.25,fillcolor=white,fillstyle=solid](58,16)(58,3)
\psline[linewidth=0.25,linestyle=dashed,dash=1 1,fillcolor=white,fillstyle=solid](71,16)(84,3)
\rput{0}(58,16){\psellipse[linewidth=0.25,linestyle=none,fillstyle=solid](0,0)(1,1)}
\rput{0}(71,16){\psellipse[linewidth=0.25,linestyle=none,fillstyle=solid](0,0)(1,1)}
\rput{0}(58,3){\psellipse[linewidth=0.25,linestyle=none,fillstyle=solid](0,0)(1,1)}
\rput{0}(71,3){\psellipse[linewidth=0.25,linestyle=none,fillstyle=solid](0,0)(1,1)}
\psline[linewidth=0.25,fillcolor=white,fillstyle=solid](71,29)(71,16)
\rput{0}(71,29){\psellipse[linewidth=0.25,linestyle=none,fillstyle=solid](0,0)(1,1)}
\rput{0}(58,29){\psellipse[linewidth=0.25,linestyle=none,fillstyle=solid](0,0)(1,1)}
\psline[linewidth=0.25,fillcolor=white,fillstyle=solid](58,16)(58,29)
\rput[Bl](72.5,15){\m2}
\rput[Bl](60,15){\m1}
\rput[Bl](24.5,8.5){\l2}
\rput[Bl](11,8.5){\l1}
\rput[Bl](24.5,21.5){\t2}
\rput[Bl](11.5,21.5){\t1}
\rput[Bl](60.5,28){\t1}
\rput[Bl](73.5,28){\t2}
\rput[Bl](60.5,2){\l1}
\rput[Bl](73.5,2){\l2}
\rput[Bl](49,0){b.}
\psline[linewidth=0.25,fillcolor=white,fillstyle=solid](35,22.5)(35,9.5)
\psline[linewidth=0.25,linestyle=dashed,dash=1 1,fillcolor=white,fillstyle=solid](35,9.5)(22,22.5)
\rput{0}(35,22.5){\psellipse[linewidth=0.25,linestyle=none,fillstyle=solid](0,0)(1,1)}
\rput{0}(35,9.5){\psellipse[linewidth=0.25,linestyle=none,fillstyle=solid](0,0)(1,1)}
\psline[linewidth=0.25,linestyle=dashed,dash=1 1,fillcolor=white,fillstyle=solid](84,16)(71,3)
\rput{0}(84,16){\psellipse[linewidth=0.25,linestyle=none,fillstyle=solid](0,0)(1,1)}
\rput{0}(84,3){\psellipse[linewidth=0.25,linestyle=none,fillstyle=solid](0,0)(1,1)}
\psline[linewidth=0.25,fillcolor=white,fillstyle=solid](84,29)(84,16)
\rput{0}(84,29){\psellipse[linewidth=0.25,linestyle=none,fillstyle=solid](0,0)(1,1)}
\rput[Bl](85.5,15){\m3}
\rput[Bl](37.5,8.5){\l3}
\rput[Bl](37.5,21.5){\t3}
\rput[Bl](86.5,28){\t3}
\rput[Bl](86.5,2){\l3}
\psline[linewidth=0.1](42,25.5)(6,25.5)
\psline[linewidth=0.1](42,19.5)(6,19.5)
\psline[linewidth=0.1](6,25.5)(6,19.5)
\psline[linewidth=0.1](42,25.5)(42,19.5)
\psline[linewidth=0.1](42,12.5)(6,12.5)
\psline[linewidth=0.1](42,6.5)(6,6.5)
\psline[linewidth=0.1](6,12.5)(6,6.5)
\psline[linewidth=0.1](42,12.5)(42,6.5)
\psline[linewidth=0.1](92,32)(55,32)
\psline[linewidth=0.1](92,26)(55,26)
\psline[linewidth=0.1](92,32)(92,26)
\psline[linewidth=0.1](92,19)(55,19)
\psline[linewidth=0.1](92,13)(55,13)
\psline[linewidth=0.1](92,19)(92,13)
\psline[linewidth=0.1](92,6)(55,6)
\psline[linewidth=0.1](92,0)(55,0)
\psline[linewidth=0.1](92,6)(92,0)
\psline[linewidth=0.1](55,32)(55,26)
\psline[linewidth=0.1](55,19)(55,13)
\psline[linewidth=0.1](55,6)(55,0)
\psline[linewidth=0.25,linestyle=dashed,dash=1 1,fillcolor=white,fillstyle=solid](35,22.5)(22,9.5)
\psbezier[linewidth=0.25](84,29)(80,19)(80,13)(84,3)
\psbezier[linewidth=0.25](71,29)(67,19)(67,13)(71,3)
\end{pspicture}
\caption{Spoiler chooses $(L^{\plus},T^{\plus})=(L,M)$ for the next round (a) and presents $m_1^{\plus},m_2^{\plus},m_3^{\plus}$. Because \mbox{$t_2^{\plus}\parallel l_3^{\plus}$} and \mbox{$t_3^{\plus}\parallel l_2^{\plus}$} there cannot be chains \mbox{$l_3^{\plus}<m_2^{\plus}<t_2^{\plus}$} and \mbox{$l_2^{\plus}<m_3^{\plus}<t_3^{\plus}$}. In our example there is an antichain of $5$ elements \mbox{$l_1^{\plus},l_2^{\plus},l_3^{\plus},m_2^{\plus},m_3^{\plus}$}.}\label{F:Example del}
\end{figure}

However, if $\leq '$ is chosen to be not too small (as in \textit{single} perfect matching cases) there is enough room for a transformation that keeps the width not only at the very moment, but also could keep it in an unpredictable future. It appears that taking $\prec '$ for two levels, e.g. $(L,T)$, to be the sum of \textit{all} perfect matchings between $L$ and $T$ will suffice. This leads to the following definition.
\begin{defn}
Let $\brackets{A,B,\leq}$ be a regular bipartite poset and $\descPM = \fPM{A,B,\leq}$ be the family of all perfect matchings in the bipartite graph $\brackets{A,B,\prec}$.
By the \emph{core} of $\brackets{A,B,\leq}$ we mean the triple $\brackets{A,B,\fCore{A}{B}{\leq}}$ such that $\fCore{A}{B}{\leq}=\set{(x,x):x\in A\cup B}\cup\mspace{2mu} \bigcup \mspace{-2mu}\descPM$.
\end{defn}
It should be obvious that the core $\brackets{A,B,\fCore{A}{B}{\leq}}$ of a poset \linebreak $\brackets{A,B,\leq}$ is a poset again. Sometimes we will write $\cleq$ for $\fCore{A}{B}{\leq}$ if both $A$ and $B$ are clear from the context. Also we will refer to the partial order $\cleq$ itself to be the core of the order $\leq$. 

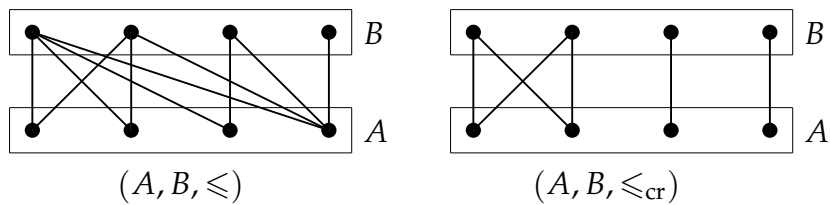
\begin{figure}[hbt]
\centering\ifx\JPicScale\undefined\def\JPicScale{1}\fi
\psset{unit=\JPicScale mm}
\psset{linewidth=0.3,dotsep=1,hatchwidth=0.3,hatchsep=1.5,shadowsize=1,dimen=middle}
\psset{dotsize=0.7 2.5,dotscale=1 1,fillcolor=black}
\psset{arrowsize=1 2,arrowlength=1,arrowinset=0.25,tbarsize=0.7 5,bracketlength=0.15,rbracketlength=0.15}
\begin{pspicture}(0,0)(107,24.5)
\psline[linewidth=0.25,fillcolor=white,fillstyle=solid](3,8.5)(3,21.5)
\rput[B](22.5,0){$\brackets{A, B,\leq}$}
\psline[linewidth=0.25,fillcolor=white,fillstyle=solid](16,8.5)(3,21.5)
\rput{0}(3,8.5){\psellipse[linewidth=0.25,linestyle=none,fillstyle=solid](0,0)(1,1)}
\rput{0}(16,21.5){\psellipse[linewidth=0.25,linestyle=none,fillstyle=solid](0,0)(1,1)}
\rput{0}(3,21.5){\psellipse[linewidth=0.25,linestyle=none,fillstyle=solid](0,0)(1,1)}
\rput{0}(16,8.5){\psellipse[linewidth=0.25,linestyle=none,fillstyle=solid](0,0)(1,1)}
\psline[linewidth=0.25,fillcolor=white,fillstyle=solid](3,8.5)(16,21.5)
\psline[linewidth=0.1](45,24.5)(0,24.5)
\psline[linewidth=0.1](45,18.5)(0,18.5)
\psline[linewidth=0.1](0,24.5)(0,18.5)
\psline[linewidth=0.1](45,24.5)(45,18.5)
\psline[linewidth=0.1](45,11.5)(0,11.5)
\psline[linewidth=0.1](45,5.5)(0,5.5)
\psline[linewidth=0.1](0,11.5)(0,5.5)
\psline[linewidth=0.1](45,11.5)(45,5.5)
\rput[Bl](46.5,7){$A$}
\rput[Bl](46.5,20){$B$}
\psline[linewidth=0.25,fillcolor=white,fillstyle=solid](29,8.5)(29,21.5)
\psline[linewidth=0.25,fillcolor=white,fillstyle=solid](42,8.5)(29,21.5)
\rput{0}(29,8.5){\psellipse[linewidth=0.25,linestyle=none,fillstyle=solid](0,0)(1,1)}
\rput{0}(42,21.5){\psellipse[linewidth=0.25,linestyle=none,fillstyle=solid](0,0)(1,1)}
\rput{0}(29,21.5){\psellipse[linewidth=0.25,linestyle=none,fillstyle=solid](0,0)(1,1)}
\rput{0}(42,8.5){\psellipse[linewidth=0.25,linestyle=none,fillstyle=solid](0,0)(1,1)}
\psline[linewidth=0.25,fillcolor=white,fillstyle=solid](16,21.5)(42,8.5)
\psline[linewidth=0.25,fillcolor=white,fillstyle=solid](42,8)(42,21.5)
\psline[linewidth=0.25,fillcolor=white,fillstyle=solid](3,21.5)(29,8.5)
\psline[linewidth=0.25,fillcolor=white,fillstyle=solid](16,8.5)(16,21.5)
\psline[linewidth=0.25,fillcolor=white,fillstyle=solid](61,8.5)(61,21.5)
\psline[linewidth=0.25,fillcolor=white,fillstyle=solid](74,8.5)(61,21.5)
\rput{0}(61,8.5){\psellipse[linewidth=0.25,linestyle=none,fillstyle=solid](0,0)(1,1)}
\rput{0}(74,21.5){\psellipse[linewidth=0.25,linestyle=none,fillstyle=solid](0,0)(1,1)}
\rput{0}(61,21.5){\psellipse[linewidth=0.25,linestyle=none,fillstyle=solid](0,0)(1,1)}
\rput{0}(74,8.5){\psellipse[linewidth=0.25,linestyle=none,fillstyle=solid](0,0)(1,1)}
\psline[linewidth=0.25,fillcolor=white,fillstyle=solid](61,8.5)(74,21.5)
\psline[linewidth=0.1](103,24.5)(58,24.5)
\psline[linewidth=0.1](103,18.5)(58,18.5)
\psline[linewidth=0.1](58,24.5)(58,18.5)
\psline[linewidth=0.1](103,24.5)(103,18.5)
\psline[linewidth=0.1](103,11.5)(58,11.5)
\psline[linewidth=0.1](103,5.5)(58,5.5)
\psline[linewidth=0.1](58,11.5)(58,5.5)
\psline[linewidth=0.1](103,11.5)(103,5.5)
\psline[linewidth=0.25,fillcolor=white,fillstyle=solid](87,8.5)(87,21.5)
\rput{0}(87,8.5){\psellipse[linewidth=0.25,linestyle=none,fillstyle=solid](0,0)(1,1)}
\rput{0}(100,21.5){\psellipse[linewidth=0.25,linestyle=none,fillstyle=solid](0,0)(1,1)}
\rput{0}(87,21.5){\psellipse[linewidth=0.25,linestyle=none,fillstyle=solid](0,0)(1,1)}
\rput{0}(100,8.5){\psellipse[linewidth=0.25,linestyle=none,fillstyle=solid](0,0)(1,1)}
\psline[linewidth=0.25,fillcolor=white,fillstyle=solid](100,8)(100,21.5)
\psline[linewidth=0.25,fillcolor=white,fillstyle=solid](74,8.5)(74,21.5)
\rput[B](78.5,0){$\brackets{A, B,\cleq}$}
\psline[linewidth=0.25,fillcolor=white,fillstyle=solid](3,21.5)(42,8.5)
\rput[Bl](104.5,20){$B$}
\rput[Bl](104.5,7){$A$}
\psline[linecolor=white](107,17.5)(107,14)
\end{pspicture}
\caption{Example of a core.}\label{F:Example core}
\end{figure} 

As each perfect matching from $\fPM{A,B,\leq}$ can be viewed as an increasing bijection from $A$ to $B$ we immediately get the following observation.

\begin{obs}\label{O:old def of core}
Let $(A,B,\leq)$ be a regular bipartite poset. Then
$p\cleq q$ iff $p=q$ or there is a bijection $f:A\tto B$ satisfying:
\begin{enumerate}
\item $a < f\!\brackets{a}$ for all $a\in A$,
\item $f\!\brackets{p}=q$.\hfill$\square$
\end{enumerate}
\end{obs}

Since together with a poset $\brackets{P,\rel{S}}$ we will consider posets of the form $(P,\rel{R})$ with $\rel{R}\subseteq\rel{S}$ we need the following observation.

\begin{obs}\label{O:R al S}
Let $(P,\rel{R})$ and $(P,\rel{S})$ be the posets such that
\begin{itemize}
\item $\rel{R}\subseteq\rel{S}$,
\item $\fWidth{P,\rel{R}}=\fWidth{P,\rel{S}}$.
\end{itemize}
Then for all $A,B\in \fMA{P,\rel{S}}$ with $A\aleq_{\rel{S}} B$ we have $A\aleq_{\rel{R}} B$.
\end{obs}
\begin{proof}
Since $A\aleq_{\rel{S}} B$, we know that, in the order $\rel{S}$, no point of $A$ can lie strictly above some point of $B$. Now, if $A\anleq_{\rel{R}} B$ then there is $a\in A$ with $a\parallel_{\rel{R}} b$ for all $b\in B$. This however leads to an antichain $B\cup\set{a}$ of size $\fWidth{P,\rel{R}}+1$.
\end{proof}

Directly from the definition we know that the order $\cleq$ is contained in the original order $\leq$. An important feature of the core is that containing less comparable pairs it still keeps the width of the poset.

\begin{obs}\label{C: w core = w}
If $(A,B,\leq)$ is a regular bipartite poset of width $w$ then its core $(A,B,\cleq)$ is a regular bipartite poset of width $w$, as well.
\end{obs}
\begin{proof}
The fact that $(A\cup B,\cleq)$ is a poset is obvious, so we will show that $\fWidth{A\cup B,\leq}=\fWidth{A\cup B,\cleq}$.
Since each \mbox{$\cleq$-comparable} pair is also \mbox{$\leq$-comparable}, every antichain in \linebreak
\mbox{$(A\cup B,\leq )$} is also an antichain in $\brackets{A\cup B,\cleq}$. Thus, $w\leq \linebreak \fWidth{A\cup B,\cleq}$.
By Dilworth's Theorem \ref{T:Dilworth} we know that there is a chain partition of $(A\cup B,\leq )$ into $w$ chains $C_1,C_2,\ldots,C_w$. Both antichains $A$ and $B$ contain $w$ elements, so that each $C_i$ intersects both $A$ and $B$, in fact $\abs{A\cap C_i}=\abs{B\cap C_i}=1$. Moreover for each $a\in A$ there is $C^a\in\set{C_1,\ldots,C_w}$ such that $\set{a}= A\cap C^a$. Now let $f:A\tto B$ sends each $a\in A$ to the unique element of $B\cap C^a$. This means that $C^a=\set{a, f(a)}$ and actually each $C_i$ is of the form $\set{a,f(a)}$ for some $a\in A$. Thus $f$ serves as an uniform witness for $a\cleq f(a)$ with $a$ ranging over $A$. This means that the $C_i$'s are also chains in $\cleq$, so that $\fWidth{A\cup B,\cleq}\leq w$, as required.

Moreover, the fact that $A\al B$ in $\brackets{A,B,\cleq}$ follows from Observation \ref{O:R al S}.
\end{proof}

The core $\fCore{A}{B}{\rel{R}}$ is contained in a partial order $\rel{R}$. Moreover, \linebreak the operation of taking a core is idempotent and monotone, i.e. \linebreak $\fCore{A}{B}{\fCore{A}{B}{\rel{R}}}=\fCore{A}{B}{\rel{R}}$ and $\rel{R}\subseteq \rel{S}$ implies $\fCore{A}{B}{\rel{R}}\subseteq \fCore{A}{B}{\rel{S}}$. This leads to the following definition.

\begin{defn}
A poset $\brackets{A,B,\rel{R}}$ is a \emph{core} if it is a regular bipartite poset and $\fCore{A}{B}{\rel{R}}=\rel{R}$.
\end{defn}
As an example illustrating the notion of cores we list all cores of width at most $3$.
\begin{exm}\label{C:char}
Let $\brackets{A, B,\leq}$ be a core of width $w\leq 3$.
Then $\brackets{A, B, \leq}$ is isomorphic to one of the posets listed by Figure \ref{F:Pii}.
\end{exm}
\begin{figure}[hbt]
%\begin{scriptsize}
\centering\input{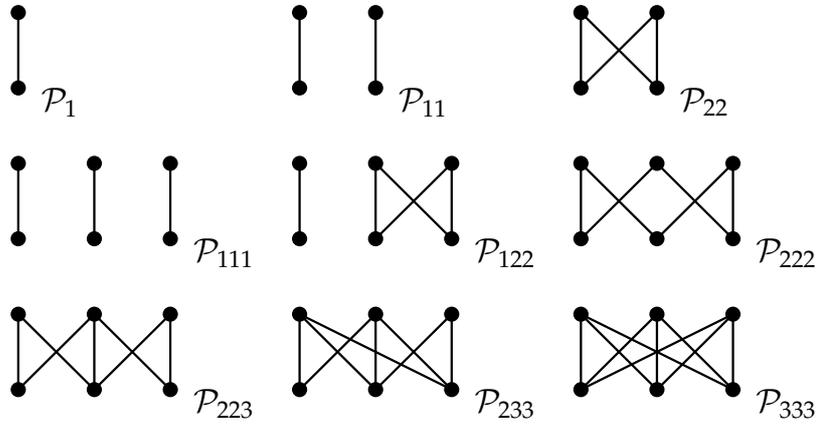}
%\end{scriptsize}
\caption{Complete list of cores of width at most $3$.}\label{F:Pii}
\end{figure} 
\begin{proof}
One can easily check that posets $\poset{P}_1$--$\poset{P}_{333}$ can be covered by $1$, $1$, $2$, $1$, $2$, $2$, $3$, $4$ and $3$ perfect matchings, respectively. This shows that these posets are cores. To see that there is no other one we will use the notion of degree in the bipartite digraph $\brackets{A,B,\prec}$. More formally $\delta(a)=\abs{a\upseto}$ for $a\in A$ and $\delta(b)=\abs{b\downseto}$ if $b\in B$. We also define $\Delta$ to be the multiset $\delta(A)=\set{\delta(a):a\in A}$. In particular for the posets $\poset{P}_1$--$\poset{P}_{333}$ these multisets $\Delta$ are $\set{1}$, $\set{1,1}$, $\set{2,2}$, $\set{1,1,1}$, $\set{1,2,2}$, $\set{2,2,2}$, $\set{2,3,2}$, $\set{2,3,3}$, $\set{3,3,3}$ respectively%
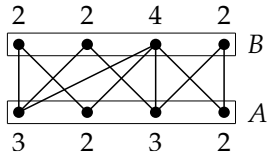
\footnote{
The reader should be warned here that the multisets $\delta(A)$ and $\delta(B)=\set{\delta(b):b\in B}$ that coincides for cores of width at most $3$, do not have to coincide in general. For example
\begin{center}
\noindent%%Created by jPicEdt 1.4.1_03: mixed JPIC-XML/LaTeX format
\ifx\JPicScale\undefined\def\JPicScale{1}\fi
\psset{unit=\JPicScale mm}
\psset{linewidth=0.3,dotsep=1,hatchwidth=0.3,hatchsep=1.5,shadowsize=1,dimen=middle}
\psset{dotsize=0.7 2.5,dotscale=1 1,fillcolor=black}
\psset{arrowsize=1 2,arrowlength=1,arrowinset=0.25,tbarsize=0.7 5,bracketlength=0.15,rbracketlength=0.15}
\begin{pspicture}(0,0)(32,19.75)
\rput{0}(2,14.25){\psellipse[linestyle=none,fillstyle=solid](0,0)(0.75,0.75)}
\psline[linewidth=0.2](2,14.25)(2,5.25)
\rput{0}(2,5.25){\psellipse[linestyle=none,fillstyle=solid](0,0)(0.75,0.75)}
\rput{0}(11,14.25){\psellipse[linestyle=none,fillstyle=solid](0,0)(0.75,0.75)}
\rput{0}(11,5.25){\psellipse[linestyle=none,fillstyle=solid](0,0)(0.75,0.75)}
\psline[linewidth=0.2](11,5.25)(2,14.25)
\rput[b](2,17){$2$}
\rput{0}(20,14.25){\psellipse[linestyle=none,fillstyle=solid](0,0)(0.75,0.75)}
\psline[linewidth=0.2](20,14.25)(11,5.25)
\rput{0}(20,5.25){\psellipse[linestyle=none,fillstyle=solid](0,0)(0.75,0.75)}
\rput{0}(29,14.25){\psellipse[linestyle=none,fillstyle=solid](0,0)(0.75,0.75)}
\rput{0}(29,5.25){\psellipse[linestyle=none,fillstyle=solid](0,0)(0.75,0.75)}
\psline[linewidth=0.2](29,5.25)(20,14.25)
\psline[linewidth=0.2](29,5.25)(29,14.25)
\psline[linewidth=0.2](20,5.25)(29,14.25)
\psline[linewidth=0.2](20,5.25)(11,14.25)
\psline[linewidth=0.2](2,5.25)(11,14.25)
\psline[linewidth=0.2](20,5.5)(20,14.25)
\psline[linewidth=0.2](2,5.25)(20,14.25)
\psline[linewidth=0.2,linecolor=white](0,19.75)(31,19.75)
\rput[b](11,17){$2$}
\rput[b](29,17){$2$}
\rput[b](20,17){$4$}
\rput[t](2,2.5){$3$}
\rput[t](20,2.5){$3$}
\rput[t](11,2.5){$2$}
\rput[t](29,2.5){$2$}
\psline[linewidth=0.2,linecolor=white](0,0)(31,0)
\psline[linewidth=0.1](30.5,6.75)(0.5,6.75)
\psline[linewidth=0.1](30.5,3.75)(0.5,3.75)
\psline[linewidth=0.1](30.5,12.75)(0.5,12.75)
\psline[linewidth=0.1](30.5,15.75)(0.5,15.75)
\psline[linewidth=0.1](30.5,15.75)(30.5,12.75)
\psline[linewidth=0.1](30.5,6.75)(30.5,3.75)
\psline[linewidth=0.1](0.5,15.75)(0.5,12.75)
\psline[linewidth=0.1](0.5,6.75)(0.5,3.75)
\rput[l](32,5.25){$A$}
\rput[l](32,14.25){$B$}
\end{pspicture}

\end{center}
 $\delta(A)=\set{3,2,3,2}$ and $\delta(B)=\set{2,2,4,2}$ are different.
}. %
Before we proceed with our classification we observe that:
\begin{texteqn}\label{TE:forb1}
In a core $(A,B,\leq)$ there is no configuration $A\ni a,a'<b\in B$ with $\delta(a)=1$.
\end{texteqn}
Indeed, to witness that $a'<b$ is in the core we need a bijection \mbox{$f:A\tto B$} that sends $a'$ to $b$. But then $f(a)$ has to be bigger than $a$ which is impossible as the only element $b$ in $a\upseto$ has been already taken \mbox{by $a'$}.\\
Analogously we have
\begin{texteqn}\label{TE:forb2}
In a core $(A,B,\leq)$ there is no configuration $A\ni a<b,b'\in B$ with $\delta(b)=1$.
\end{texteqn}

Now we are ready for our classification. Since $\delta(a)\leq w\leq 3$ the only possibilities for $\Delta$ are
\begin{itemize}
\item $\set{1}$ for $w=1$;
\item $\set{1,1}$, $\set{1,2}$, $\set{2,2}$ for $w=2$;
\item $\set{1,1,1}$, $\set{1,1,2}$, $\set{1,1,3}$, $\set{1,2,2}$, $\set{1,2,3}$, $\set{1,3,3}$, $\set{2,2,2}$, $\set{2,2,3}$, $\set{2,3,3}$, $\set{3,3,3}$ for $w=3$.
\end{itemize}
Now we proceed by cases, with an additional listing of elements $A=\set{a_1,\ldots,a_w}$ and $B=\set{b_1,\ldots,b_w}$. Moreover we assume that the multiset $\Delta$ is listed exactly in the order $\delta(a_1),\ldots,\delta(a_w)$. %Next we consider every possibilities.

\begin{description}
\item[$\set{1}$] This leads to the poset $\poset{P}_1$.
\item[$\set{1,1}$] This leads to the poset $\poset{P}_{11}$, as otherwise we will have a configuration forbidden in (\ref{TE:forb1}).
\item[$\set{1,2}$] Without loss of generality let $b_1$ be the only neighbor of $a_1$. Then $a_2$ has $b_1$ as a neighbor as well. This produces a configuration forbidden in (\ref{TE:forb1}) so that $\set{1,2}$ cannot be realized by a core.
\item[$\set{2,2}$] This leads to the poset $\poset{P}_{22}$.
\item[$\set{1,1,1}$] This leads to the poset $\poset{P}_{111}$, as otherwise some $b_i$ will have two neighbors of degree $1$ which is forbidden by (\ref{TE:forb1}).
\item[$\set{1,1,2} and \set{1,1,3}$] Since $\delta(a_1)=\delta(a_2)=1$, by (\ref{TE:forb1}) the points $a_1$ and $a_2$ cannot share a neighbor with any other point. Without loss of generality $a_1<b_1$ and $a_2<b_2$. But now there is no room for neighbors of $a_3$.
\item[$\set{1,2,2}$] Again, by (\ref{TE:forb1}), $a_1$ has its ``private'' neighbor, say $b_1$. To realize $\delta(a_2)=\delta(a_3)=2$ we need to have $a_2,a_3<b_2,b_3$ which gives $\poset{P}_{122}$.
\item[$\set{1,2,3} and \set{1,3,3}$] These multisets are excluded, as again by (\ref{TE:forb1}), $a_3$ taking all of the $b_i$'s leaves no room for a ``private'' neighbor for $a_1$.
\item[$\set{2,2,2}$] First we argue that the upsets $a_1\upseto$, $a_2\upseto$, $a_3\upseto$ have to be pairwise different. Indeed, suppose e.g. $a_1\upseto=a_2\upseto=\set{b_1,b_2}$. Then $a_3<b_3$ as otherwise $\set{a_1,a_2,a_3,b_3}$ would be a $4$-element antichain. Now to realize $\delta(a_3)=2$ without loss of generality we may assume that $a_3<b_2$ so that we are in the situation presented by Figure \ref{F:not core}.
\begin{figure}[hbt]
%\begin{scriptsize}
\centering\input{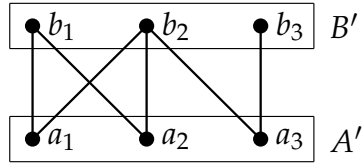}
%\end{scriptsize}
\caption{Poset which is not a core.}\label{F:not core}
\end{figure} 
But then $a_3<b_3,b_2$ is a situation forbidden by (\ref{TE:forb2}). Thus $a_1\upseto$, $a_2\upseto$, $a_3\upseto$ are pairwise different and then the poset $\poset{P}_{222}$ appears.
\item[$\set{2,2,3}$] By the same token as in case $\set{2,2,2}$ the sets $a_1\upseto$, $a_2\upseto$, $a_3\upseto$ are pairwise different. The only way to realize degrees $2,2,3$ on $a_1,a_3,a_2$ is by a poset isomorphic to $\poset{P}_{223}$.
\item[$\set{2,3,3}$] Poset $\poset{P}_{233}$ is the only way to realize this multiset.
\item[$\set{3,3,3}$] Poset $\poset{P}_{333}$ is the only way to realize this multiset.\hfill\qedsymbol
\end{description}
\renewcommand{\qedsymbol}{}
\end{proof}

Now we are ready to reduce the local game with Property \ref{Pr:disjoint} to the game which in addition satisfies the following property.

\begin{pr}\label{Pr:core}
$\brackets{L, M,\leq}$, $\brackets{M, T,\leq}$ and $\brackets{L, T,\leq}$ are cores during entire game.
\end{pr}
For example, let $\brackets{L, T, \leq}$ be a core presented on Figure \ref{F:Example Local core}.
In order to fulfill Property \ref{Pr:disjoint}, Spoiler provides $4$ new points which form an antichain $M$. Moreover in order to satisfy Property \ref{Pr:core} he has to do it in a way that $\brackets{L, M,\leq}$ and $\brackets{M, T,\leq}$ are cores. After Algorithm colors $M$, Spoiler redefines $(L^{\plus}\!,T^{\plus})$ to $(M,T)$.

\begin{figure}[hbt]
\centering\ifx\JPicScale\undefined\def\JPicScale{1}\fi
\psset{unit=\JPicScale mm}
\psset{linewidth=0.3,dotsep=1,hatchwidth=0.3,hatchsep=1.5,shadowsize=1,dimen=middle}
\psset{dotsize=0.7 2.5,dotscale=1 1,fillcolor=black}
\psset{arrowsize=1 2,arrowlength=1,arrowinset=0.25,tbarsize=0.7 5,bracketlength=0.15,rbracketlength=0.15}
\begin{pspicture}(0,0)(119,30)
\psline[linewidth=0.25,fillcolor=white,fillstyle=solid](2,23)(2,13)
\psline[linewidth=0.25,fillcolor=white,fillstyle=solid](22,23)(22,13)
\psline[linewidth=0.25,fillcolor=white,fillstyle=solid](22,23)(32,13)
\psline[linewidth=0.25,fillcolor=white,fillstyle=solid](12,23)(32,13)
\psline[linewidth=0.25,fillcolor=white,fillstyle=solid](12,13)(12,23)
\psline[linewidth=0.25,fillcolor=white,fillstyle=solid](22,23)(12,13)
\psline[linewidth=0.25,fillcolor=white,fillstyle=solid](32,23)(32,13)
\rput[l](35,23){$T$}
\rput[l](35,13){$L$}
\psline[linewidth=0.25,fillcolor=white,fillstyle=solid](22,13)(12,23)
\psline[linewidth=0.25,fillcolor=white,fillstyle=solid](32,23)(22,13)
\psline[linewidth=0.1](34,15)(0,15)
\psline[linewidth=0.1](34,11)(0,11)
\psline[linewidth=0.1](0,15)(0,11)
\rput{0}(2,23){\psellipse[linewidth=0.25,linestyle=none,fillstyle=solid](0,0)(1,1)}
\rput{0}(12,13){\psellipse[linewidth=0.25,linestyle=none,fillstyle=solid](0,0)(1,1)}
\rput{0}(22,23){\psellipse[linewidth=0.25,linestyle=none,fillstyle=solid](0,0)(1,1)}
\rput{0}(32,23){\psellipse[linewidth=0.25,linestyle=none,fillstyle=solid](0,0)(1,1)}
\rput{0}(2,13){\psellipse[linewidth=0.25,linestyle=none,fillstyle=solid](0,0)(1,1)}
\rput{0}(12,23){\psellipse[linewidth=0.25,linestyle=none,fillstyle=solid](0,0)(1,1)}
\rput{0}(22,13){\psellipse[linewidth=0.25,linestyle=none,fillstyle=solid](0,0)(1,1)}
\rput{0}(32,13){\psellipse[linewidth=0.25,linestyle=none,fillstyle=solid](0,0)(1,1)}
\psline[linewidth=0.1](34,15)(34,11)
\psline[linewidth=0.1](34,21)(0,21)
\psline[linewidth=0.1](34,25)(0,25)
\psline[linewidth=0.1](34,25)(34,21)
\psline[linewidth=0.1](0,25)(0,21)
\psline[linewidth=0.25,fillcolor=white,fillstyle=solid](44,18)(44,8)
\psline[linewidth=0.25,fillcolor=white,fillstyle=solid](64,18)(64,8)
\psline[linewidth=0.25,fillcolor=white,fillstyle=solid](64,18)(74,8)
\psline[linewidth=0.25,fillcolor=white,fillstyle=solid](54,8)(54,18)
\psline[linewidth=0.25,fillcolor=white,fillstyle=solid](74,18)(74,8)
\rput[l](77,18){$M$}
\rput[l](77,8){$L$}
\psline[linewidth=0.25,fillcolor=white,fillstyle=solid](74,18)(64,8)
\psline[linewidth=0.1](76,10)(42,10)
\psline[linewidth=0.1](76,6)(42,6)
\psline[linewidth=0.1](42,10)(42,6)
\rput{0}(44,18){\psellipse[linewidth=0.25,linestyle=none,fillstyle=solid](0,0)(1,1)}
\rput{0}(54,8){\psellipse[linewidth=0.25,linestyle=none,fillstyle=solid](0,0)(1,1)}
\rput{0}(64,18){\psellipse[linewidth=0.25,linestyle=none,fillstyle=solid](0,0)(1,1)}
\rput{0}(74,18){\psellipse[linewidth=0.25,linestyle=none,fillstyle=solid](0,0)(1,1)}
\rput{0}(44,8){\psellipse[linewidth=0.25,linestyle=none,fillstyle=solid](0,0)(1,1)}
\rput{0}(54,18){\psellipse[linewidth=0.25,linestyle=none,fillstyle=solid](0,0)(1,1)}
\rput{0}(64,8){\psellipse[linewidth=0.25,linestyle=none,fillstyle=solid](0,0)(1,1)}
\rput{0}(74,8){\psellipse[linewidth=0.25,linestyle=none,fillstyle=solid](0,0)(1,1)}
\psline[linewidth=0.1](76,10)(76,6)
\psline[linewidth=0.1](76,16)(42,16)
\psline[linewidth=0.1](76,20)(42,20)
\psline[linewidth=0.1](76,20)(76,16)
\psline[linewidth=0.1](42,20)(42,16)
\psline[linewidth=0.25,fillcolor=white,fillstyle=solid](44,28)(44,18)
\psline[linewidth=0.25,fillcolor=white,fillstyle=solid](74,28)(74,18)
\psline[linewidth=0.25,fillcolor=white,fillstyle=solid](54,18)(54,28)
\psline[linewidth=0.25,fillcolor=white,fillstyle=solid](64,28)(54,18)
\psline[linewidth=0.25,fillcolor=white,fillstyle=solid](64,28)(64,18)
\rput[l](77,28){$T$}
\psline[linewidth=0.25,fillcolor=white,fillstyle=solid](64,18)(54,28)
\rput{0}(44,28){\psellipse[linewidth=0.25,linestyle=none,fillstyle=solid](0,0)(1,1)}
\rput{0}(64,28){\psellipse[linewidth=0.25,linestyle=none,fillstyle=solid](0,0)(1,1)}
\rput{0}(74,28){\psellipse[linewidth=0.25,linestyle=none,fillstyle=solid](0,0)(1,1)}
\rput{0}(54,28){\psellipse[linewidth=0.25,linestyle=none,fillstyle=solid](0,0)(1,1)}
\psline[linewidth=0.1](76,26)(42,26)
\psline[linewidth=0.1](76,30)(42,30)
\psline[linewidth=0.1](76,30)(76,26)
\psline[linewidth=0.1](42,30)(42,26)
\rput[l](119,18){$L^+$}
\rput{0}(86,18){\psellipse[linewidth=0.25,linestyle=none,fillstyle=solid](0,0)(1,1)}
\rput{0}(106,18){\psellipse[linewidth=0.25,linestyle=none,fillstyle=solid](0,0)(1,1)}
\rput{0}(116,18){\psellipse[linewidth=0.25,linestyle=none,fillstyle=solid](0,0)(1,1)}
\rput{0}(96,18){\psellipse[linewidth=0.25,linestyle=none,fillstyle=solid](0,0)(1,1)}
\psline[linewidth=0.1](118,16)(84,16)
\psline[linewidth=0.1](118,20)(84,20)
\psline[linewidth=0.1](118,20)(118,16)
\psline[linewidth=0.1](84,20)(84,16)
\psline[linewidth=0.25,fillcolor=white,fillstyle=solid](86,28)(86,18)
\psline[linewidth=0.25,fillcolor=white,fillstyle=solid](96,18)(96,28)
\psline[linewidth=0.25,fillcolor=white,fillstyle=solid](106,28)(96,18)
\psline[linewidth=0.25,fillcolor=white,fillstyle=solid](116,28)(116,18)
\rput[l](119,28){$T^+$}
\psline[linewidth=0.25,fillcolor=white,fillstyle=solid](106,18)(96,28)
\rput{0}(86,28){\psellipse[linewidth=0.25,linestyle=none,fillstyle=solid](0,0)(1,1)}
\rput{0}(106,28){\psellipse[linewidth=0.25,linestyle=none,fillstyle=solid](0,0)(1,1)}
\rput{0}(116,28){\psellipse[linewidth=0.25,linestyle=none,fillstyle=solid](0,0)(1,1)}
\rput{0}(96,28){\psellipse[linewidth=0.25,linestyle=none,fillstyle=solid](0,0)(1,1)}
\psline[linewidth=0.1](118,26)(84,26)
\psline[linewidth=0.1](118,30)(84,30)
\psline[linewidth=0.1](118,30)(118,26)
\psline[linewidth=0.1](84,30)(84,26)
\psline[linewidth=0.25,fillcolor=white,fillstyle=solid](106,18)(106,28)
\rput[B](38,0.5){\begin{footnotesize}putting an antichain $M$\end{footnotesize}}
\psline[linewidth=0.2,fillcolor=white,fillstyle=solid](19,-1)(57,-1)
\psline[linewidth=0.2,fillcolor=white,fillstyle=solid](56,0)(57.07,-1.07)
\psline[linewidth=0.2,fillcolor=white,fillstyle=solid](57.07,-0.93)(56,-2)
\rput(53,0){}
\rput[B](80,0.5){\begin{footnotesize}choosing new structure\end{footnotesize}}
\psline[linewidth=0.2,fillcolor=white,fillstyle=solid](61,-1)(99,-1)
\psline[linewidth=0.2,fillcolor=white,fillstyle=solid](98,0)(99.07,-1.07)
\psline[linewidth=0.2,fillcolor=white,fillstyle=solid](99.07,-0.94)(98,-2)
\rput(95,0){}
\end{pspicture}
\caption{$ $}\label{F:Example Local core}
\end{figure}

\medskip

The local game with Properties \ref{Pr:disjoint} and \ref{Pr:core} is described by the following rules.

\begin{defn}\label{D:core disjoint game}
By a \emph{\coredisjoint{} game} we mean the following two-person game between Spoiler and Algorithm. During the first round:
\begin{itemize}
\item Spoiler sets a natural number $w$ and then introduces two disjoint antichains $L,T$, each with $w$ elements, such that $L<T$.
\item Algorithm determines a finite set $\Gamma$ of colors that may be used in the entire game and then he colors each point $x\in L\cup T$ with some nonempty subset $\fC{}{x}$ of $\Gamma$ such that for each $\gamma\in\Gamma$ the points colored by $\gamma$ form a chain.
\end{itemize}
During next rounds the board $\brackets{L,T,\leq,\chains}$ from the previous round is transformed to a board $\brackets{L^{\plus},T^{\plus},\leq,\chains^{\plus}}$ according to the following rules:
\begin{itemize}
\item Spoiler introduces $w$ new elements that form an antichain $M$ such that the poset $\poset{B}'=\brackets{L\cup M\cup T,\leq}$ has width $w$ and $L\al M\al T$ in the lattice $\fMA{\poset{B}'}$. Moreover, $(L,M,\leq)$ and $(M,T,\leq)$ have to be cores.
\item Algorithm colors each point $m\in M$ with a nonempty set of colors $\fC{\plus}{m}\subseteq\fC{}{L}\cap\fC{}{T}$ and keeps the old multicoloring on $L\cup T$, i.e. $\rest{\chains^{\plus}}{L\cup T}=\chains$.
\item Finally, Spoiler redefines the levels $L, T$ to $L^{\plus}, T^{\plus}$ so that either $(L^{\plus},T^{\plus})=(L,M)$ or $(L^{\plus},T^{\plus})=(M,T)$.
This, after restricting to $L^{\plus}\cup T^{\plus}\!$, creates the new board $\brackets{L^{\plus},T^{\plus},\leq,\chains^{\plus}}$.
\end{itemize}
Again, the goal of Algorithm is to pick minimal number $\abs{\Gamma}$ of colors already during the first round so that he can play with these colors forever.
\end{defn}

Now we will present a reduction of the disjoint game to the \coredisjoint{} game. In this reduction we will switch between two games:
\begin{itemize}
\item a disjoint game described in Definition \ref{D:disjoint game}. Their players are to be called \emph{Disjoint Spoiler} and \emph{Disjoint Algorithm}.
\item a \coredisjoint{} game described in Definition \ref{D:core disjoint game}. Their players are to be called \emph{\CoreDisjoint{} Spoiler} and \emph{\CoreDisjoint{} Algorithm}. 
\end{itemize}
We will copy multicoloring from the \coredisjoint{} game to the disjoint game. After the first move of Spoiler the poset $\brackets{L, T,\leq}$ is a core. Thus a multicoloring $\chains$ of $L\cup T$ returned by \CoreDisjoint{} Algorithm is also a correct multicoloring in the disjoint game. For the further performance of Disjoint Algorithm in each round from now on, a structure $\brackets{L,T,\leq,\chains}$ is a board for the disjoint game, where the multicoloring $\chains$ is exactly the one returned by \CoreDisjoint{} Algorithm on $(L\cup T,\cleq)$.

Now we focus on a round which is not the first one. Let $\brackets{L,T,\leq,\chains}$ be a board returned by the previous round. According to Definition \ref{D:disjoint game} let $M$ be an antichain introduced by Disjoint Spoiler.
We transform $\poset{B}'=\brackets{\,L\cup M\cup T,\; \leq\,}$ to 
\(\ol{\poset{B}}\ =\ \brackets{\,L\cup M\cup T,\ \fDCore{L}{M}{T}{\leq}\,}
\),
where $\fDCore{L}{M}{T}{\leq}=\fCore{L}{T}{\leq}\,\cup\,\fCore{L}{M}{\leq}\,\cup\,\fCore{M}{T}{\leq}$.
Restricting the order $\leq$ of $\poset{B}'$ to $\fDCore{L}{M}{T}{\leq}$ was made carefully enough, so that on any two levels $A,B \in \set{L,M,T}$ the ordering $\fDCore{L}{M}{T}{\leq}$ is a core of $(A,B,\leq)$. In particular the restriction of $\ol{\poset{B}}$ to $L\cup T$ gives $\brackets{L, T, \fCore{L}{T}{\leq}}$.
All we need to show is that $\ol{\poset{B}}$ is a correct input that can be presented to \CoreDisjoint{} Algorithm.
Since we already know that $\chains$ is a multicoloring of $(L, T,\cleq)$, this amounts in showing that 
\begin{enumerate}
\item $\ol{\poset{B}}$ is a partially ordered set,\label{P:olB}
\item $\fWidth{\ol{\poset{B}}}=w$,\label{P:width}
\item $L \al M \al T$ in the lattice $\fMA{\ol{\poset{B}}}$.\label{P:fMA}
\end{enumerate}
To see (\ref{P:olB}) note that the only nontrivial part is to show that $\fDCore{L}{M}{T}{\leq}$ is transitive. Due to the fact that each $3$-element chain intersect all $3$ levels $L \al M \al T$ it suffices to show the following claim.

\begin{clm}\label{C:core tr}
If \mbox{$\brackets{l,m}\in\fCore{L}{M}{\leq}$} and \mbox{$\brackets{m,t}\in\fCore{M}{T}{\leq}$} then 
$\brackets{l,t}\in\fCore{L}{T}{\leq}$.
\end{clm}
%\begin{figure}[hbt]
%\renewcommand{\l}{\begin{small}$l$\end{small}}
%\newcommand{\m}{\begin{small}$m$\end{small}}
%\renewcommand{\t}{\begin{small}$t$\end{small}}
%\centering\input{chapter-02-pict-17}
%\caption{Cores and Matchings in a poset $\brackets{L\cup M\cup T,\leq}$.}\label{F:matching}
%\end{figure} 
\begin{proof}
Since for $l=m$ or $m=t$ there is nothing to be proved, we may assume that $l\in L$, $m\in M$ and $t\in T$. Now we start with a perfect matching $\fMatch{L}{M}\!\subseteq\fCore{L}{M}{\leq}$ in the digraph $\brackets{L,M,\prec}$ which contains the edge $\brackets{l,m}$ and a perfect matching $\fMatch{M}{T}\!\subseteq\fCore{M}{T}{\leq}$ in the digraph $\brackets{M,T,\prec}$ which contains $\brackets{m,t}$ to construct a perfect matching $\fMatch{L}{T}\!\subseteq\fCore{L}{T}{\leq}$ which contains $\brackets{l,t}$. 
Such a matching $\fMatch{L}{T}$ can be constructed e.g. by a simple superposition of the matchings $\fMatch{L}{M}$ and $\fMatch{M}{T}$, i.e.:
\begin{multline*}
(l',t')\in\fMatch{L}{T}\quad\textrm{iff}\\
(l',m')\in\fMatch{L}{M}\ \,\textrm{and}\,\ (m',t')\in\fMatch{M}{T}\ \ \textrm{for some}\,\ m'\in M.
\end{multline*}
Obviously $\fMatch{L}{T}$, as a superposition of two bijections, is a perfect matching in the digraph $(L,T,\prec)$.
Because the core contains all perfect matchings, we get $\brackets{l,t}\in\fMatch{L}{T}\subseteq\fCore{L}{T}{\leq}$.
\end{proof}

To see (\ref{P:width}) note first that $L,M,T$ are $w$-element antichains in $\ol{\poset{B}}$. Thus we get $\fWidth{\ol{\poset{B}}}\geq w$.
On the other hand, from Claim \ref{C: w core = w} we know that $\fWidth{L\cup M,\fCore{L}{M}{\leq}}=w=\fWidth{M\cup T,\fCore{M}{T}{\leq}}$ so that Dilworth's Theorem \ref{T:Dilworth} supplies us with coverings of \linebreak \mbox{$\brackets{L\cup M,\,\fCore{L}{M}{\leq}}$} and \mbox{$\brackets{M\cup T,\,\fCore{M}{T}{\leq}}$} by $w$ chains $D_1,\ldots, D_w$ and $U_1,\ldots, U_w$, respectively. As in the Perles' proof of Dilworth's Theorem we note that the maximum antichain $M$ has to meet each of the $D_i$'s and $U_i$'s at exactly one point. Therefore after renumbering $U_i$'s we may arrange that $\abs{ D_{i} \cap U_{i}}=1$. This shows that $\ol{\poset{B}}$ can be covered by $w$ chains $D_1\cup U_1,\ldots,D_w\cup U_w$, so that $\fWidth{\ol{\poset{B}}}\leq w$.

\medskip

Finally, the property (\ref{P:fMA}) follows directly from Observation \ref{O:R al S}.

\bigskip

After establishing that $\ol{\poset{B}}$ is a correct input that can be presented to \CoreDisjoint{} Algorithm, he colors each point $m\in M$ by a nonempty set of colors $\fC{\plus}{m}\subseteq\fC{}{L}\cap\fC{}{T}$.
This has to be done in a way that, together with $\rest{\chains^{\plus}}{L\cup T}=\chains$, the sets $C_{\gamma}=\set{p\in L\cup M\cup T:\gamma\in\fC{\plus}{p}}$
are chains in $\ol{\poset{B}}$ for all $\gamma \in \Gamma$. Obviously each $C_{\gamma}$ is also a chain in $(L\cup M\cup T,\leq)$, as $\leq$ has more comparable pairs than $\fDCore{L}{M}{T}{\leq}$ does.

\bigskip

This finishes our reduction of the disjoint game to the \coredisjoint{} game and shows that the value of the \coredisjoint{} game bounds from above the value of the disjoint game. Together with the reduction of Subsection \ref{SS:different} we get that the value of the \coredisjoint{} game on posets of width $w$ bounds from above the value $\fLCP{w}$ of an unrestricted game. As the converse inequality is obvious, we finally get the following theorem.

\begin{thm}\label{T:CDG}
The value of the \coredisjoint{} game on posets of width $w$ is $\fLCP{w}$.
\end{thm}

\section{Lower Bound of Local Problem} \label{S:lower bounds}

In this section we focus on a lower bound of the value of the local on-line coloring game. This will be done by analyzing how \CoreDisjoint{} Algorithm can use colors from $\fC{}{L}\cap \fC{}{T}$ to color $M$. First note that the multicoloring $\chains$ determines a covering of the vertices of the poset by chains $C_{\gamma}= \set{p: \gamma\in \fC{}{p}}$. However, as we will see in the next lemma, such a multicoloring actually determines a covering of all edges of the digraph $(L,T,\prec)$ corresponding to the core $(L,T,\cleq)$, e.g. for each edge $l\prec t$ we have $\fC{}{l}\cap \fC{}{t} \neq \emptyset$. This in particular means, that for each $\gamma\in\fC{}{l}\cap \fC{}{t}$ the chain $C_{\gamma}$ is exactly $\set{l,t}$ and therefore $\gamma$ cannot be used on any other point. This motivates the following definition.
\begin{defn}
Let $\brackets{L,T,\leq,\chains}$ be a board. For $l\in L$ and $t\in T$ with $l\prec t$ the set $\fC{}{l}\cap \fC{}{t}$ is called the set of \emph{private colors} of the edge $(l,t)$ of the digraph $\brackets{L,T,\prec}$.
\end{defn}
Obviously if $\gamma \in \fC{}{L}\cap \fC{}{T}$ then $C_{\gamma}$ is a two element chain and therefore $\gamma$ is a private color for some edge $l\prec t$. The converse implication does not hold in general, but it does hold for cores.
\begin{lem}\label{L:private}
Let $\brackets{L,T,\leq,\chains}$ be a board for which \CoreDisjoint{} Algorithm has a strategy to play forever. Then 
any edge of the digraph $\brackets{L,T,\prec}$ has a private color.
\end{lem}
\begin{proof}
To the contrary suppose that there is an edge $\brackets{l_1,t_1}$ in the digraph $\brackets{L,T,\prec}$ such that $\fC{}{l_1}\cap\fC{}{t_1}=\emptyset$. We will present a possible move of \CoreDisjoint{} Spoiler for which \CoreDisjoint{} Algorithm has no correct response.

The poset $\brackets{L, T, \leq}$ is a core, thus we can put the edge $\brackets{l_1,t_1}$ into a matching $\set{\brackets{l_1,t_1},\ldots,\brackets{l_w,t_w}}$ in the digraph $\brackets{L,T,\prec}$ (see Figure \ref{F:matching and move}.a for an example).
\begin{figure}[hbt]
\renewcommand{\l}[1]{\begin{footnotesize}$l_{#1}$\end{footnotesize}}
\newcommand{\m}[1]{\begin{footnotesize}$m_{#1}$\end{footnotesize}}
\renewcommand{\t}[1]{\begin{footnotesize}$t_{#1}$\end{footnotesize}}
\centering\ifx\JPicScale\undefined\def\JPicScale{1}\fi
\psset{unit=\JPicScale mm}
\psset{linewidth=0.3,dotsep=1,hatchwidth=0.3,hatchsep=1.5,shadowsize=1,dimen=middle}
\psset{dotsize=0.7 2.5,dotscale=1 1,fillcolor=black}
\psset{arrowsize=1 2,arrowlength=1,arrowinset=0.25,tbarsize=0.7 5,bracketlength=0.15,rbracketlength=0.15}
\begin{pspicture}(0,0)(113.5,28)
\psline[linewidth=0.75,fillcolor=white,fillstyle=solid](9,19.5)(9,8)
\psline[linewidth=0.25,fillcolor=white,fillstyle=solid](32,19.5)(32,8)
\psline[linewidth=0.25,fillcolor=white,fillstyle=solid](32,19.5)(43.5,8)
\psline[linewidth=0.75,fillcolor=white,fillstyle=solid](20.5,19.5)(43.5,8)
\psline[linewidth=0.25,fillcolor=white,fillstyle=solid](20.5,8)(20.5,19.5)
\psline[linewidth=0.75,fillcolor=white,fillstyle=solid](32,19.5)(20.5,8)
\psline[linewidth=0.25,fillcolor=white,fillstyle=solid](43.5,19.5)(43.5,8)
\rput[l](51.5,19.5){$T$}
\rput[l](51.5,8){$L$}
\psline[linewidth=0.25,fillcolor=white,fillstyle=solid](9,8)(20.5,19.5)
\psline[linewidth=0.75,fillcolor=white,fillstyle=solid](43.5,19.5)(32,8)
\psline[linewidth=0.1](49.5,10.5)(6.5,10.5)
\psline[linewidth=0.1](49.5,5.5)(6.5,5.5)
\psline[linewidth=0.1](6.5,10.5)(6.5,5.5)
\rput{0}(9,19.5){\psellipse[linewidth=0.25,linestyle=none,fillstyle=solid](0,0)(1,1)}
\rput{0}(20.5,8){\psellipse[linewidth=0.25,linestyle=none,fillstyle=solid](0,0)(1,1)}
\rput{0}(32,19.5){\psellipse[linewidth=0.25,linestyle=none,fillstyle=solid](0,0)(1,1)}
\rput{0}(43.5,19.5){\psellipse[linewidth=0.25,linestyle=none,fillstyle=solid](0,0)(1,1)}
\rput{0}(9,8){\psellipse[linewidth=0.25,linestyle=none,fillstyle=solid](0,0)(1,1)}
\rput{0}(20.5,19.5){\psellipse[linewidth=0.25,linestyle=none,fillstyle=solid](0,0)(1,1)}
\rput{0}(32,8){\psellipse[linewidth=0.25,linestyle=none,fillstyle=solid](0,0)(1,1)}
\rput{0}(43.5,8){\psellipse[linewidth=0.25,linestyle=none,fillstyle=solid](0,0)(1,1)}
\psline[linewidth=0.1](49.5,10.5)(49.5,5.5)
\psline[linewidth=0.1](49.5,17)(6.5,17)
\psline[linewidth=0.1](49.5,22)(6.5,22)
\psline[linewidth=0.1](49.5,22)(49.5,17)
\psline[linewidth=0.1](6.5,22)(6.5,17)
\psline[linewidth=0.75,fillcolor=white,fillstyle=solid](68.5,14)(68.5,2.5)
\psline[linewidth=0.75,fillcolor=white,fillstyle=solid](91.5,14)(103,2.5)
\psline[linewidth=0.75,fillcolor=white,fillstyle=solid](80,2.5)(80,14)
\rput[l](111.5,14){$M$}
\rput[l](111.5,2.5){$L$}
\psline[linewidth=0.75,fillcolor=white,fillstyle=solid](103,14)(91.5,2.5)
\psline[linewidth=0.1](109.5,5)(66,5)
\psline[linewidth=0.1](109.5,0)(66,0)
\psline[linewidth=0.1](66,5)(66,0)
\rput{0}(68.5,14){\psellipse[linewidth=0.25,linestyle=none,fillstyle=solid](0,0)(1,1)}
\rput{0}(80,2.5){\psellipse[linewidth=0.25,linestyle=none,fillstyle=solid](0,0)(1,1)}
\rput{0}(91.5,14){\psellipse[linewidth=0.25,linestyle=none,fillstyle=solid](0,0)(1,1)}
\rput{0}(103,14){\psellipse[linewidth=0.25,linestyle=none,fillstyle=solid](0,0)(1,1)}
\rput{0}(68.5,2.5){\psellipse[linewidth=0.25,linestyle=none,fillstyle=solid](0,0)(1,1)}
\rput{0}(80,14){\psellipse[linewidth=0.25,linestyle=none,fillstyle=solid](0,0)(1,1)}
\rput{0}(91.5,2.5){\psellipse[linewidth=0.25,linestyle=none,fillstyle=solid](0,0)(1,1)}
\rput{0}(103,2.5){\psellipse[linewidth=0.25,linestyle=none,fillstyle=solid](0,0)(1,1)}
\psline[linewidth=0.1](109.5,5)(109.5,0)
\psline[linewidth=0.1](109.5,11.5)(66,11.5)
\psline[linewidth=0.1](109.5,16.5)(66,16.5)
\psline[linewidth=0.1](109.5,16.5)(109.5,11.5)
\psline[linewidth=0.1](66,16.5)(66,11.5)
\psline[linewidth=0.75,fillcolor=white,fillstyle=solid](68.5,25.5)(68.5,14)
\psline[linewidth=0.75,fillcolor=white,fillstyle=solid](103,25.5)(103,14)
\psline[linewidth=0.75,fillcolor=white,fillstyle=solid](91.5,25.5)(80,14)
\rput[l](111.5,25.5){$T$}
\psline[linewidth=0.75,fillcolor=white,fillstyle=solid](91.5,14)(80,25.5)
\rput{0}(68.5,25.5){\psellipse[linewidth=0.25,linestyle=none,fillstyle=solid](0,0)(1,1)}
\rput{0}(91.5,25.5){\psellipse[linewidth=0.25,linestyle=none,fillstyle=solid](0,0)(1,1)}
\rput{0}(103,25.5){\psellipse[linewidth=0.25,linestyle=none,fillstyle=solid](0,0)(1,1)}
\rput{0}(80,25.5){\psellipse[linewidth=0.25,linestyle=none,fillstyle=solid](0,0)(1,1)}
\psline[linewidth=0.1](109.5,23)(66,23)
\psline[linewidth=0.1](109.5,28)(66,28)
\psline[linewidth=0.1](109.5,28)(109.5,23)
\psline[linewidth=0.1](66,28)(66,23)
\rput[Bl](11,18.5){\t2}
\rput[Bl](23.5,18.5){\t1}
\rput[Bl](34,18.5){\t3}
\rput[Bl](45.5,18.5){\t4}
\psline[linewidth=0.25,fillcolor=white,fillstyle=solid](20.5,8)(9,19.5)
\rput[Bl](11,7){\l2}
\rput[Bl](45.5,7){\l1}
\rput[Bl](34.5,7){\l4}
\rput[Bl](23,7){\l3}
\rput[Bl](70.5,24.5){\t2}
\rput[Bl](82.5,24.5){\t1}
\rput[Bl](93.5,24.5){\t3}
\rput[Bl](105,24.5){\t4}
\rput[Bl](70.5,1.5){\l2}
\rput[Bl](105,1.5){\l1}
\rput[Bl](94,1.5){\l4}
\rput[Bl](82,1.5){\l3}
\psline[linewidth=0.25,fillcolor=white,fillstyle=solid](68.5,25.5)(80,2.5)
\psline[linewidth=0.25,fillcolor=white,fillstyle=solid](80,25.5)(68.5,2.5)
\psbezier[linewidth=0.25](80,25.5)(77.5,17.5)(77.5,10.5)(80,2.5)
\psline[linewidth=0.25,fillcolor=white,fillstyle=solid](91.5,25.5)(103,2.5)
\psbezier[linewidth=0.25](91.5,25.5)(89,17.5)(89,10.5)(91.5,2.5)
\psbezier[linewidth=0.25](103,25.5)(100.5,17.5)(100.5,10.5)(103,2.5)
\rput[Bl](70,13){\m2}
\rput[Bl](81.5,13){\m3}
\rput[Bl](93.5,13){\m1}
\rput[Bl](104.5,13){\m2}
\rput[Bl](0,5.5){a.}
\rput[Bl](59.5,0){b.}
\psline[linewidth=0.1,linecolor=white](113.5,12)(113.5,7)
\end{pspicture}
\caption{Spoiler's move based on a matching in $\brackets{L,T,\prec}$.}\label{F:matching and move}
\end{figure}
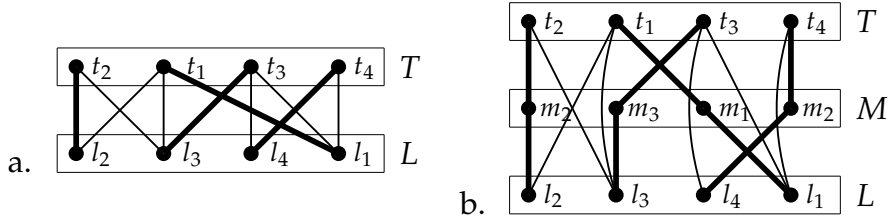
Spoiler introduces $w$ points $m_1,\ldots,m_w$ such that each point $m_i$ lies only between $l_i$ and $t_i$ (see Figure \ref{F:matching and move}.b). The constructed poset $\poset{B}^{\plus}$ has width $w$ as $\set{l_1,m_1,t_1},\ldots,\set{l_w,m_w,t_w}$ form a partition of $\poset{B}^{\plus}$ into $w$ chains. The antichain $M\mspace{-2mu}=\mspace{-2mu}\set{m_1,\ldots,m_w}$ is disjoint with $L\cup T$ and moreover $L\al M\al T$ in the lattice $\fMA{\poset{B}^{\plus}}$. Each of $\brackets{L, M,\leq}$ and $\brackets{M, T,\leq}$ consists of a single matching, thus they are cores. This certifies that $M$ could be a move of \CoreDisjoint{} Spoiler. Now \CoreDisjoint{} Algorithm falls into troubles, as $m_1$ can be colored only by colors used on both $m_1\downseto=\set{l_1}$ and $m_1\upseto=\set{t_1}$, while
 there is no such color, as $\fC{}{l_1}\cap\fC{}{t_1}=\emptyset$.
\end{proof}

\begin{cor}\label{Co:lval>=w2}
$\fLCP{w}\geq w^2$.
\end{cor}
\begin{proof}
Lemma \ref{L:private} tells us that $\fLCP{w}$ is bounded from below by the maximal number of edges in the digraph $(L,T,\prec)$. However, already at the very beginning there are $w^2$ such edges in the complete bipartite digraph.
\end{proof}

Corollary \ref{Co:lval>=w2} supplies us with the lower bounds of $1$, $4$ and $9$ for $\fLCP{w}$ with $w=1,2,3$ respectively. Obviously \mbox{$\fLCP{1}=1$}.
\linebreak In the next Section we will see that $\fLCP{2}=4$. However \linebreak\mbox{$\fLCP{3}\geq 11$} as the following two Observations show. 
\begin{obs}\label{O:lval3>=10}
$\fLCP{3}\geq 10$.
\end{obs}
\begin{proof}
Suppose that \CoreDisjoint{} Algorithm uses only $9$ colors on the initial board $(L,T,\leq, \chains)$ shown on Figure \ref{F:LCP3 > 9}. Due to Lemma \ref{L:private} each of the $9$ edges has a private color, so that each edge has only one private color: $\abs{\fC{}{l_i}\cap \fC{}{t_j}}=1$ for $i,j=1,2,3$. Now, \CoreDisjoint{} Spoiler introduces the middle level $\set{m_1,m_2,m_3}$ as shown by both thin and thick edges on Figure \ref{F:LCP3 > 9}.b. 
Again by Lemma \ref{L:private} all edges between $L$ and $M$ as well as between $M$ and $T$ must have some private color.
Focusing on $7$ thick edges $\brackets{l_2,m_1}$, $\brackets{l_2,m_2}$, $\brackets{l_3,m_1}$, $\brackets{l_3,m_2}$, $\brackets{m_3,t_1}$, $\brackets{m_3,t_2}$, $\brackets{m_3,t_3}$ we learn that $l_1$ is incomparable with some endpoint of each of these edges.
Therefore these thick edges can have private color only from the set $\fC{}{L}\setminus\fC{}{l_1}=\linebreak\fC{}{l_2}\cup \fC{}{l_3}$. Unfortunately $\abs{\fC{}{l_i}}=3$ (for $i=2,3$) and there is not enough colors to be spread over the $7$ edges.
\end{proof}

\begin{figure}[hbt]
\renewcommand{\l}[1]{\begin{footnotesize}$l_{#1}$\end{footnotesize}}
\newcommand{\m}[1]{\begin{footnotesize}$m_{#1}$\end{footnotesize}}
\renewcommand{\t}[1]{\begin{footnotesize}$t_{#1}$\end{footnotesize}}
\centering\ifx\JPicScale\undefined\def\JPicScale{1}\fi
\psset{unit=\JPicScale mm}
\psset{linewidth=0.3,dotsep=1,hatchwidth=0.3,hatchsep=1.5,shadowsize=1,dimen=middle}
\psset{dotsize=0.7 2.5,dotscale=1 1,fillcolor=black}
\psset{arrowsize=1 2,arrowlength=1,arrowinset=0.25,tbarsize=0.7 5,bracketlength=0.15,rbracketlength=0.15}
\begin{pspicture}(0,0)(96.5,28)
\psline[linewidth=0.25,fillcolor=white,fillstyle=solid](9,19.5)(9,8)
\psline[linewidth=0.25,fillcolor=white,fillstyle=solid](32,19.5)(32,8)
\psline[linewidth=0.25,fillcolor=white,fillstyle=solid](20.5,8)(20.5,19.5)
\psline[linewidth=0.25,fillcolor=white,fillstyle=solid](32,19.5)(20.5,8)
\rput[l](40.5,19.5){$T$}
\rput[l](40.5,8){$L$}
\psline[linewidth=0.25,fillcolor=white,fillstyle=solid](9,8)(20.5,19.5)
\psline[linewidth=0.1](38.5,10.5)(6.5,10.5)
\psline[linewidth=0.1](38.5,5.5)(6.5,5.5)
\psline[linewidth=0.1](6.5,10.5)(6.5,5.5)
\rput{0}(9,19.5){\psellipse[linewidth=0.25,linestyle=none,fillstyle=solid](0,0)(1,1)}
\rput{0}(20.5,8){\psellipse[linewidth=0.25,linestyle=none,fillstyle=solid](0,0)(1,1)}
\rput{0}(32,19.5){\psellipse[linewidth=0.25,linestyle=none,fillstyle=solid](0,0)(1,1)}
\rput{0}(9,8){\psellipse[linewidth=0.25,linestyle=none,fillstyle=solid](0,0)(1,1)}
\rput{0}(20.5,19.5){\psellipse[linewidth=0.25,linestyle=none,fillstyle=solid](0,0)(1,1)}
\rput{0}(32,8){\psellipse[linewidth=0.25,linestyle=none,fillstyle=solid](0,0)(1,1)}
\psline[linewidth=0.1](38.5,10.5)(38.5,5.5)
\psline[linewidth=0.1](38.5,17)(6.5,17)
\psline[linewidth=0.1](38.5,22)(6.5,22)
\psline[linewidth=0.1](38.5,22)(38.5,17)
\psline[linewidth=0.1](6.5,22)(6.5,17)
\psline[linewidth=0.75,fillcolor=white,fillstyle=solid](61,14)(84,2.5)
\psline[linewidth=0.75,fillcolor=white,fillstyle=solid](72.5,2.5)(72.5,14)
\rput[l](93,14){$M$}
\rput[l](93,2.5){$L$}
\psline[linewidth=0.1](90.5,5)(58.5,5)
\psline[linewidth=0.1](90.5,0)(58.5,0)
\psline[linewidth=0.1](58.5,5)(58.5,0)
\rput{0}(61,14){\psellipse[linewidth=0.25,linestyle=none,fillstyle=solid](0,0)(1,1)}
\rput{0}(72.5,2.5){\psellipse[linewidth=0.25,linestyle=none,fillstyle=solid](0,0)(1,1)}
\rput{0}(84,14){\psellipse[linewidth=0.25,linestyle=none,fillstyle=solid](0,0)(1,1)}
\rput{0}(61,2.5){\psellipse[linewidth=0.25,linestyle=none,fillstyle=solid](0,0)(1,1)}
\rput{0}(72.5,14){\psellipse[linewidth=0.25,linestyle=none,fillstyle=solid](0,0)(1,1)}
\rput{0}(84,2.5){\psellipse[linewidth=0.25,linestyle=none,fillstyle=solid](0,0)(1,1)}
\psline[linewidth=0.1](90.5,5)(90.5,0)
\psline[linewidth=0.1](91,11.5)(58.5,11.5)
\psline[linewidth=0.1](91,16.5)(58.5,16.5)
\psline[linewidth=0.1](91,16.5)(91,11.5)
\psline[linewidth=0.1](58.5,16.5)(58.5,11.5)
\psline[linewidth=0.75,fillcolor=white,fillstyle=solid](84,14)(61,25.5)
\psline[linewidth=0.75,fillcolor=white,fillstyle=solid](84,25.5)(84,14)
\rput[l](93,25.5){$T$}
\psline[linewidth=0.75,fillcolor=white,fillstyle=solid](84,14)(72.5,25.5)
\rput{0}(61,25.5){\psellipse[linewidth=0.25,linestyle=none,fillstyle=solid](0,0)(1,1)}
\rput{0}(84,25.5){\psellipse[linewidth=0.25,linestyle=none,fillstyle=solid](0,0)(1,1)}
\rput{0}(72.5,25.5){\psellipse[linewidth=0.25,linestyle=none,fillstyle=solid](0,0)(1,1)}
\psline[linewidth=0.1](90.5,23)(58.5,23)
\psline[linewidth=0.1](90.5,28)(58.5,28)
\psline[linewidth=0.1](90.5,28)(90.5,23)
\psline[linewidth=0.1](58.5,28)(58.5,23)
\rput[Bl](12,18.5){\t1}
\rput[Bl](22.5,18.5){\t2}
\rput[Bl](34,18.5){\t3}
\psline[linewidth=0.25,fillcolor=white,fillstyle=solid](20.5,8)(9,19.5)
\rput[Bl](12,7){\l1}
\rput[Bl](34.5,7){\l3}
\rput[Bl](23,7){\l2}
\rput[Bl](63.5,25){\t1}
\rput[Bl](75,24.5){\t2}
\rput[Bl](86,24.5){\t3}
\rput[Bl](63,1.5){\l1}
\rput[Bl](86,1.5){\l3}
\rput[Bl](74.5,1.5){\l2}
\psline[linewidth=0.25,fillcolor=white,fillstyle=solid](61,25.5)(72.5,14)
\psline[linewidth=0.25,fillcolor=white,fillstyle=solid](72.5,25.5)(61,14)
\rput[Bl](63.5,13.5){\m1}
\rput[Bl](74.5,13){\m2}
\rput[Bl](85.5,13){\m3}
\rput[Bl](0,5.5){a.}
\rput[Bl](52,0){b.}
\psline[linewidth=0.1,linecolor=white](96.5,12)(96.5,7)
\psline[linewidth=0.25,fillcolor=white,fillstyle=solid](20.5,19.5)(32,8)
\psline[linewidth=0.25,fillcolor=white,fillstyle=solid](32,19.5)(9,8)
\psline[linewidth=0.25,fillcolor=white,fillstyle=solid](9,19.5)(32,8)
\psline[linewidth=0.25,fillcolor=white,fillstyle=solid](84,25.5)(72.5,14)
\psline[linewidth=0.25,fillcolor=white,fillstyle=solid](72.5,25.5)(72.5,14)
\psline[linewidth=0.25,fillcolor=white,fillstyle=solid](61,25.5)(61,14)
\psline[linewidth=0.25,fillcolor=white,fillstyle=solid](84,14)(84,2.5)
\psline[linewidth=0.75,fillcolor=white,fillstyle=solid](72.5,14)(84,2.5)
\psline[linewidth=0.75,fillcolor=white,fillstyle=solid](61,14)(72.5,2.5)
\psline[linewidth=0.25,fillcolor=white,fillstyle=solid](84,14)(72.5,2.5)
\psline[linewidth=0.25,fillcolor=white,fillstyle=solid](72.5,14)(61,2.5)
\psline[linewidth=0.25,fillcolor=white,fillstyle=solid](61,14)(61,2.5)
\end{pspicture}
\caption{$ $}\label{F:LCP3 > 9}
\end{figure}

\begin{obs}\label{O:locval>=11}
$\fLCP{3}\geq 11$.
\end{obs}
\begin{proof}
With Observation \ref{O:lval3>=10} \CoreDisjoint{} Algorithm already knows that he has to start with at least $10$ colors on the initial board.
Thus suppose that \CoreDisjoint{} Algorithm uses exactly $10$ colors on the initial board $(L,T,\leq, \chains)$ presented on Figure \ref{F:LCP3 > 9}.a.
Due to Lemma \ref{L:private} each of the $9$ edges has a private color, so that all edges beside one, say $(l_1,t_3)$, have exactly one private color: $\abs{\fC{}{l_i}\cap \fC{}{t_j}}=1$ for $(l_i,t_j)\neq (l_1,t_3)$. The edge $(l_1,t_3)$ has two private colors, i.e. \linebreak $\abs{\fC{}{l_1}\cap \fC{}{t_3}}=2$.
If \CoreDisjoint{} Spoiler introduces the middle level \linebreak $\set{m_1,m_2,m_3}$ as shown on Figure \ref{F:LCP3 > 9}.b,
\CoreDisjoint{} Algorithm can not correctly color points $m_1,m_2,m_3$ by the very same token as in the proof of Observation \ref{O:lval3>=10}.
\end{proof}

\section{Orders of Width at most $3$} \label{S:w<=3}

Being equipped with Corollary \ref{R:CP sum LCP} and Theorem \ref{T:CDG} one can work on bounding (from above) the value $\fCP{w}$ of the on-line chain partitioning problem.
In this section we will do it for $w\leq3$ by determining that

\begin{itemize}
\item $\fLCP{1} = 1$,
\item $\fLCP{2}\leq 4$,
\item $\fLCP{3}\leq 11$.
\end{itemize}

The first item is obvious. Before proving the other two note that the second bound of $4$ is optimal as otherwise $\fCP{2}\leq \fLCP{1}+\fLCP{2}<5$, contrary to Kierstead \cite{Kierstead} and Felsner \cite{Felsner} works that show $\fCP{2}=5$. Also in view of Observation \ref{O:locval>=11} the third bound is optimal.
For showing these $3$ bounds we are going to present a strategy for \CoreDisjoint{} Algorithm playing with $1$, $4$ or $11$ colors respectively. 
This strategy will keep the following conditions as an invariant.
\begin{inv}\label{I:width3}
If $(L,T,\leq, \chains)$ is a board during the \coredisjoint{} game then:
\begin{enumerate}
\item each edge $l\prec t$ in the digraph $(L,T,\prec)$ has at least one private color,
\item if $L<T$ and $w=3$ then there are two disjoint edges (i.e. with no common endpoints) each having at least two private colors.
\end{enumerate}
\end{inv}
Imposing the first condition of Invariant \ref{I:width3} is clear in view of Lemma \ref{L:private} to ensure that \CoreDisjoint{} Algorithm can play forever. The second condition is partially motivated by Observation \ref{O:locval>=11}.
Surprisingly we will see that it is enough to keep only these two conditions.

To keep Invariant \ref{I:width3} after the first round Core Disjoin Algorithm has to color (up to the permutation of vertices and colors) the initial board as presented on Figure \ref{F: first move}. Two thick edges are those that have two private colors.

\begin{figure}[hbt]
\renewcommand{\a}[1]{\begin{footnotesize}$#1$\end{footnotesize}}
\renewcommand{\b}[2]{\begin{footnotesize}$#1,#2$\end{footnotesize}}
\renewcommand{\c}[3]{\begin{footnotesize}$#1,#2,#3$\end{footnotesize}}
\renewcommand{\d}[4]{\begin{footnotesize}$#1,#2,#3,#4$\end{footnotesize}}
\newcommand{\w}[1]{$w=#1$}
\centering\input{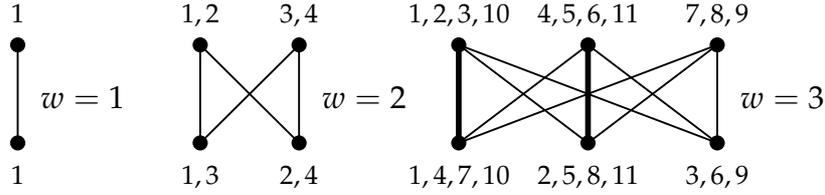}
\caption{Algorithm's move in the first round.}\label{F: first move}
\end{figure}

In any other round Algorithm receives poset $\poset{B}^{\plus}=\brackets{L\cup M\cup T,\leq}$ of width $w\leq 3$.
Moreover, the antichains $L=\set{l_1,\ldots,l_w}$, $M=\set{m_1,\ldots,m_w}$, $T=\set{t_1,\ldots,t_w}$ are pairwise disjoint and $L\al M\al T$ in the lattice $\fMA{\poset{B}^{\plus}}$.
Property \ref{Pr:core} says that $\brackets{L, M,\leq}$ and $\brackets{M, T, \leq}$ are cores, thus by Example \ref{C:char}, we know that each of $\brackets{L, M,\leq}$ and $\brackets{M, T, \leq}$ is isomorphic to one of the posets $\poset{P}_1,\ldots,\poset{P}_{333}$ listed on Figure \ref{F:Pii}.

Together with $\poset{B}^{\plus}$ Algorithm knows a multicoloring $\chains$ of $L\cup T$. Invariant \ref{I:width3} allows to pick a private color $\gamma^i_j\in\fC{}{l_j}\cap\fC{}{t_i}$ for each edge $(l_j, t_i)$ in $(L, T, \prec)$. We know that all of the $\gamma^i_j$'s are pairwise different. The collection of all such selected $\gamma^i_j$'s is to be denoted by $\Gamma_{\!0}$. The set $\Gamma_{\!0}$ is changing during the game, however in case $w=3$ there are always at least two additional colors in $\Gamma\setminus\Gamma_{\!0}$.

The strategy for \CoreDisjoint{} Algorithm will be presented by cases. In most of them the multicoloring $\chains^{\plus}$ of the middle level uses only colors from $\Gamma_{\!0}$ and satisfy: 
\begin{description}
\item[\cnone] The sets of the form $\fC{\plus}{m}$, with $m\in M$, are nonempty and pairwise disjoint.
\item[\cntwo] If $\gamma^i_j\in \fC{\plus}{m}$ then $l_j<m<t_i$.
\item[\cninv] Each edge of either $(L,M,\prec)$ or $(M,T,\prec)$ has a private color from $\Gamma_{\!0}$.
\end{description}
The conditions \cnone{} and \cntwo{} guarantee that each $\chains^{\plus}$, that extends $\chains$ and satisfies $\fC{\plus}{M}\subseteq\Gamma_{\!0}$, is allowed by the rules of the \coredisjoint{} game.
The third condition proves Invariant \ref{I:width3}.(1).

The cases are organized first by the width of the poset and then by isomorphism types $(\poset{P}_{\Delta},\poset{P}_{\nabla})$ of $(L,M,\leq)$ and $(M,T,\leq)$, meaning that $(L,M,\leq)\iso \poset{P}_{\Delta}$ and $(M,T,\leq)\iso \poset{P}_{\nabla}$, where $\Delta$, $\nabla$ are appropriate multisets of degrees defined in Example \ref{C:char}. An obvious exchange of $L$ with $T$ together with going to the dual poset of $(L,M,T,\leq)$ allows us to describe only one of two cases $(\poset{P}_{\Delta},\poset{P}_{\nabla})$ and $(\poset{P}_{\nabla},\poset{P}_{\Delta})$.

\casewidth{1} (See Figure \ref{F: Alg: P1 P1}.) 
In this case the only possibility is \linebreak $\brackets{L, M,\leq}\iso\poset{P}_{1}$ and $\brackets{M, T,\leq}\iso\poset{P}_{1}$.
Spoiler introduces only one point $m_1$ between $l_1$ and $t_1$. Algorithm colors $\fC{\plus}{m_1}=\set{\gamma_1^1}$.
\begin{figure}[hbt]
\renewcommand{\l}[1]{$l_{#1}$}
\newcommand{\m}[1]{$m_{#1}$}
\renewcommand{\t}[1]{$t_{#1}$}
\newcommand{\g}[2]{\begin{gammasize}$\gamma_{#1}^{#2}$\end{gammasize}}
\centering\input{chapter-02-pict-1_1}
\caption{$\brackets{L, M,\leq}\iso\poset{P}_{1}$ and $\brackets{M, T,\leq}\iso\poset{P}_{1}$.}\label{F: Alg: P1 P1}
\end{figure}

\casewidth{2}

%%%%%%%%%%%%%%%%%%%%%%%%%%%%%%%%%%%%%%%%%%%%%%%%%
\begin{case}{11}{11}{Figure \ref{F: Alg: P11 P11}}
Without loss of generality we may assume that $l_1<m_1<t_1$ and $l_2<m_2<t_2$. The colors $\gamma_1^1\in\fC{}{l_1}\cap\fC{}{t_1}$ and $\gamma_2^2\in\fC{}{l_2}\cap\fC{}{t_2}$ are reused so that $\fC{\plus}{m_1}=\set{\gamma_1^1}$ and $\fC{\plus}{m_2}=\set{\gamma_2^2}$ satisfy \cnall.
\end{case}

\ec

\begin{figure}[hbt]
\renewcommand{\l}[1]{$l_{#1}$}
\newcommand{\m}[1]{\$m_{#1}$}
\renewcommand{\t}[1]{$t_{#1}$}
\newcommand{\g}[2]{\begin{gammasize}$\gamma_{#1}^{#2}$\end{gammasize}}
\centering\input{chapter-02-pict-11_11}
\caption{$\brackets{L, M,\leq}\iso\poset{P}_{11}$ and $\brackets{M, T,\leq}\iso\poset{P}_{11}$.}\label{F: Alg: P11 P11}
\end{figure}

%%%%%%%%%%%%%%%%%%%%%%%%%%%%%%%%%%%%%%%%%%%%%%%%%
\begin{case}{11}{22}{Figure \ref{F: Alg: P11 P22}}
As $M<T$ we have had $L<T$, so that $\Gamma_0=\set{\gamma_1^1, \gamma_1^2, \gamma_2^1, \gamma_2^2}$. Obviously $\fC{\plus}{m_1}=\set{\gamma_1^1,\gamma_1^2}$ and $\fC{\plus}{m_2}=\set{\gamma_2^1,\gamma_2^2}$ satisfy \cnall.
\end{case}

\ec

\begin{figure}[hbt]
\renewcommand{\l}[1]{$l_{#1}$}
\renewcommand{\t}[1]{$t_{#1}$}
\newcommand{\m}{\begin{gammasize}$\gamma_1^1,\gamma_1^2$\end{gammasize}}
\newcommand{\mm}{\begin{gammasize}$\gamma_2^1,\gamma_2^2$\end{gammasize}}
\centering\input{chapter-02-pict-11_22}
\caption{$\brackets{L, M,\leq}\iso\poset{P}_{11}$ and $\brackets{M, T,\leq}\iso\poset{P}_{22}$.}\label{F: Alg: P11 P22}
\end{figure}

%%%%%%%%%%%%%%%%%%%%%%%%%%%%%%%%%%%%%%%%%%%%%%%%%
\begin{case}{22}{22}{Figure \ref{F: Alg: P22 P22}}
Now $L<M<T$, but again $\Gamma_0=\set{\gamma_1^1, \gamma_1^2, \gamma_2^1, \gamma_2^2}$. One can easily check that $\fC{\plus}{m_1}=\set{\gamma_2^1,\gamma_1^2}$ and $\fC{\plus}{m_2}=\set{\gamma_1^1,\gamma_2^2}$ satisfy \cnall.
\end{case}

\ec

\begin{figure}[hbt]
\renewcommand{\l}[1]{$l_{#1}$}
\renewcommand{\t}[1]{$t_{#1}$}
\newcommand{\m}{\begin{gammasize}$\gamma_2^1,\gamma_1^2$\end{gammasize}}
\newcommand{\mm}{\begin{gammasize}$\gamma_1^1,\gamma_2^2$\end{gammasize}}
\centering\input{chapter-02-pict-22_22}
\caption{$\brackets{L, M,\leq}\iso\poset{P}_{22}$ and $\brackets{M, T,\leq}\iso\poset{P}_{22}$.}\label{F: Alg: P22 P22}
\end{figure}

\casewidth{3}
Before proceeding with $w=3$ we need a tool that ensures us that Invariant \ref{I:width3}.(2) is kept.
This invariant allows Algorithm to use two extra colors that are not in $\Gamma_{\!0}$. These two extra colors, say $\alpha$ and $\beta$, are kept as long as $L<T$. This is to help Algorithm to manage situations where $9$ or even $10$ colors do not suffice. We have seen such situations in the proofs of Observation \ref{O:lval3>=10} and \ref{O:locval>=11}, see Figure \ref{F:LCP3 > 9}. Note that if \mbox{$(L,T,\leq) \not\mspace{-5.5mu}\iso \poset{P}_{333}$} then $M$ cannot be inserted to get $\brackets{L,M,\leq}$ or $\brackets{M,T,\leq}$ be isomorphic to $\poset{P}_{333}$. Indeed, if for example \mbox{$\brackets{M,T,\leq}\iso \poset{P}_{333}$} then since in $\brackets{L,T,\prec}$ there is at least one perfect matching, we have $L<T$. Therefore once $\poset{P}_{333}$ does not occur as $(L,T,\leq)$, it can not occur any later. On the other hand, if $\brackets{L,T,\leq}\iso\poset{P}_{333}$ then Invariant \ref{I:width3}.(2) guarantees that $\abs{\fC{}{L}\cap \fC{}{T}}=11$.
In particular $\abs{\Gamma_{\!0}}=9$. After inserting level $M$, to keep Invariant \ref{I:width3}.(1) it obviously suffices to fulfill \cnall\ by using only colors from $\Gamma_{0}$. But if $(L,M,\leq)$ or $(M,T,\leq)$ is $\poset{P}_{333}$, the troubles may arrive later on. Therefore Algorithm has to keep these extra colors $\alpha$, $\beta$ for the future use.
In many cases, when $L<T$, we will actually present a multicoloring $\chainso$ of $M$ by $\Gamma_{\!0}$ and then enrich $\chainso$ to $\chains^{\plus}$ that uses all $11$ colors.

\begin{clm}\label{C:P333}
Suppose that a level $M$ is added to the board \mbox{$(L,T,\leq, \chains)$} of width $3$ in a way that $L<M$ or $M<T$. Then a multicoloring \mbox{$\chainso : M \tto \powerp{\Gamma_{\!0}}$} that satisfies \cnall{} can be enriched to a correct multicoloring \mbox{$\chains^{\plus} : M \tto \powerp{\Gamma}$} such that
\begin{itemize}
\item $\chains^{\plus}$ satisfies \cnall,
\item there are $6$ different points $l_1,l_2,m_1,m_2,t_1,t_2$ with \mbox{$l_1 \prec m_1 \prec t_1$} and \mbox{$l_2 \prec m_2 \prec t_2$} and moreover \mbox{$\abs{\fC{\plus}{l_1}\cap\fC{\plus}{m_1}} =$} \linebreak \mbox{$\abs{\fC{\plus}{m_1}\cap\fC{\plus}{t_1}}\mspace{-1mu} =\mspace{-1mu} \abs{\fC{\plus}{l_2}\cap\fC{\plus}{m_2}}\mspace{-1mu} =\mspace{-1mu} \abs{\fC{\plus}{m_2}\cap\fC{\plus}{t_2}}\mspace{-1mu}=\!2$}.
\end{itemize}
\end{clm}
\begin{proof}
Without loss of generality we may assume that $L<M$. By \cninv{} for $\chainso$, we already know that each edge in $(L, M, \prec)$ has a private color from $\Gamma_{\!0}$. To secure additional private colors $\alpha$, $\beta$ for two disjoint chains $l_1 \prec m_1 \prec t_1$ and $l_2 \prec m_2 \prec t_2$ note first that $M \al T$ gives $L<T$. Thus Invariant \ref{I:width3} gives that $\abs{\fC{}{L}\cap \fC{}{T}}=11$ and there are two edges, say $l_1 \prec t_1$ and $l_2 \prec t_2$ with $\alpha \in\fC{}{l_1}\cap\fC{}{t_1}$ and $\beta \in \fC{}{l_2}\cap\fC{}{t_2}$, where $\set{\alpha, \beta}=\fC{}{L}\cap \fC{}{T}\setminus \Gamma_{\!0}$. To get the required points $m_1, m_2\in M$ take any
matching $\fMatch{M}{T}$ in the digraph $(M,T,\prec)$ and let $m_1, m_2$ be such that $(m_1, t_1),(m_2,t_2)\in \fMatch{M}{T}$. Putting $\fC{\plus}{m_1}=\fCo{}{m_1}\cup \set{\alpha}$, 
$\fC{\plus}{m_2}=\fCo{}{m_2}\cup \set{\beta}$ and $\fC{\plus}{x}= \fCo{}{x}$ for $x \neq m_1, m_2$ we see that the second item of our Claim is fulfilled. Since $\chains^{\plus}$ uses colors from $\Gamma_{\!0}$ in the very same way as $\chainso$ did, we get $\cnall$ for 
$\chains^{\plus}$. To see that $\chains^{\plus}$ is a correct 
coloring it remains to note that the points colored by $\alpha$ or $\beta$ form the chains $l_1 \prec m_1 \prec t_1$ and $l_2 \prec m_2 \prec t_2$, respectively.
\end{proof}
By the conditions imposed on the enrichment $\chains^{\plus}$ of $\chainso$ described in Claim \ref{C:P333}, we know that $\chains^{\plus}$ satisfies Invariant \ref{I:width3}.

Now we are ready to proceed with all $12$ cases for $w=3$.

%%%%%%%%%%%%%%%%%%%%%%%%%%%%%%%%%%%%%%%%%%%%%%%%%

\begin{case}{111}{\nabla}{See Figure \ref{F: Alg: P111 P233} for $\poset{P}_{\nabla}=\poset{P}_{233}$}
Independently what $\nabla$ is, we can arrange that $l_1<m_1$, $l_2<m_2$ and $l_3<m_3$. What does depend on $\nabla$ is the set $\Gamma_0$. Note however that if $m_k<t_i$ then $l_k<t_i$ and $\Gamma_0$ contains $\gamma_k^i$. To see that the sets $\fC{\plus}{m_k}=\set{\gamma_k^i:m_k<t_i}$ are nonempty, note that otherwise $\set{m_k,t_1,t_2,t_3}$ would be a $4$-element antichain. Since the $\fC{\plus}{m_k}$'s consist of colors $\gamma_k^i$ with different subscripts $k$, they are pairwise disjoint. This shows \cnone{}. Moreover $l_k<m_k$ together with $\gamma_k^i\in\fC{\plus}{m_k}$ gives $l_k<m_k<t_i$ as required in \cntwo. To see \cninv, note that each ``lower'' edge (i.e. in $(L, M,\leq)\iso\poset{P}_{111}$) is of the form $l_k\prec m_k$ and each color in $\fC{\plus}{m_k}$ is its private one. Moreover each ``upper'' edge $m_k\prec t_i$ has $\gamma_k^i$ as a private color. If $\brackets{M, T,\leq}\iso\poset{P}_{333}$ then, to secure Invariant \ref{I:width3}.(2), we simply invoke Claim \ref{C:P333}.
\end{case}

\ec

%\pagebreak

\FloatBarrier 

\begin{figure}[hbt]
\renewcommand{\l}[1]{$l_{#1}$}
\renewcommand{\t}[1]{$t_{#1}$}
\newcommand{\m}{\begin{gammasize}$\gamma_1^1,\gamma_1^2,\gamma_1^3$\end{gammasize}}
\newcommand{\mm}{\begin{gammasize}$\gamma_2^1,\gamma_2^2,\gamma_2^3$\end{gammasize}}
\newcommand{\mmm}{\begin{gammasize}$\gamma_3^2,\gamma_3^3$\end{gammasize}}
\centering\input{chapter-02-pict-111_233}
\caption{$\brackets{L, M,\leq}\iso\poset{P}_{111}$ and $\brackets{M, T,\leq}\iso\poset{P}_{233}$.}\label{F: Alg: P111 P233}
\end{figure}

\FloatBarrier

%%%%%%%%%%%%%%%%%%%%%%%%%%%%%%%%%%%%%%%%%%%%%%%%%
\begin{case}{122}{\nabla}{See Figure \ref{F: Alg: P122 P122} for $\poset{P}_{\nabla}=\poset{P}_{122}$ and Figure \ref{F: Alg: P122 P233} for $\poset{P}_{\nabla}=\poset{P}_{233}$}
Again, independently what $\nabla$ is, we arrange that $l_1<m_1$ and $l_2,l_3<m_2,m_3$, as presented on Figure \ref{F: Alg: P122}. Moreover, without loss of generality we may assume that $m_1<t_1$, $m_2<t_2$ and $m_3<t_3$.
Now our proof splits into two cases depending on whether $m_2$ and $m_3$ have a common neighbor in $T$.
\begin{figure}[hbt]
\newcommand{\m}[1]{$m_{#1}$}
\renewcommand{\l}[1]{$l_{#1}$}
\centering\input{chapter-02-pict-122}
\caption{$\brackets{L, M,\leq}\iso\poset{P}_{122}$.}\label{F: Alg: P122}
\end{figure}

First, suppose that $m_2,m_3$ have no common neighbor in $T$.
Analyzing $6$ cores of width $3$ (see Figure \ref{F:Pii}) this can happen only if either $\poset{P}_{\nabla} = \poset{P}_{111}$ or $\poset{P}_{\nabla} = \poset{P}_{122}$ and $\poset{P}_{\nabla}$ is 
related to $\poset{P}_{\Delta}$ as
shown on Figure \ref{F: Alg: P122 P122}. For $(M,T,\leq)\iso \poset{P}_{111}$ we are actually in
a case dual to \nc{111}{\nabla}. 
If $\brackets{M, T,\leq}\iso \poset{P}_{122}$ then \mbox{$m_1,m_2<t_1,t_2$}, so that 
the set $\Gamma_{\!0}$ of available colors is \mbox{$\set{\gamma_1^1,\gamma_1^2,\gamma_2^1,\gamma_2^2,\gamma_2^3,\gamma_3^1,\gamma_3^2,\gamma_3^3}$}. Obviously the multicoloring $\chains^{\plus}\!$, given by $\fC{\plus}{m_1}=\set{\gamma_1^1,\gamma_1^2}$, $\fC{\plus}{m_2}=\set{\gamma_2^1,\gamma_3^2}$ and $\fC{\plus}{m_3}=\set{\gamma_2^3,\gamma_3^3}$, satisfies \cnall.
\begin{figure}[hbt]
\renewcommand{\l}[1]{$l_{#1}$}
\renewcommand{\t}[1]{$t_{#1}$}
\newcommand{\m}{\begin{gammasize}$\gamma_1^1,\gamma_1^2$\end{gammasize}}
\newcommand{\mm}{\begin{gammasize}$\gamma_2^1,\gamma_3^2$\end{gammasize}}
\newcommand{\mmm}{\begin{gammasize}$\gamma_2^3,\gamma_3^3$\end{gammasize}}
\centering\input{chapter-02-pict-122_122}
\caption{$\brackets{L, M,\leq}\iso\poset{P}_{122}$ and $\brackets{M, T,\leq}\iso\poset{P}_{122}$.}\label{F: Alg: P122 P122}
\end{figure}

Now suppose that $m_2$ and $m_3$ have a common neighbor, say \mbox{$t_s\in T$}. As usual $\Gamma_{\!0}$ depends on $\nabla\!$. But note that whenever $m_1<t_i$, we have $l_1<t_i$ so that $\Gamma_0$ contains $\gamma_1^i$. 
Moreover if $m_2<t_i$ or $m_3<t_i$ then $l_2,l_3<t_i$ and $\Gamma_0$ contains both $\gamma_2^i$ and $\gamma_3^i$. 
All Algorithm has to do, is to follow the behavior in Case \nc{111}{\nabla} with exchanging the role of $\gamma^s_2$ and $\gamma^s_3$:
\begin{align*}
\fC{\plus}{m_1}&=\big\{\gamma_1^i:m_1<t_i\big\},\\
\fC{\plus}{m_2}&=\brackets{\big\{\gamma_2^i:m_2<t_i\big\}\setminus\big\{\gamma_2^s\big\}}\cup\big\{\gamma_3^s\big\},\\
\fC{\plus}{m_3}&=\brackets{\big\{\gamma_3^i:m_3<t_i\big\}\setminus\big\{\gamma_3^s\big\}}\cup\big\{\gamma_2^s\big\}.
\end{align*}
\begin{figure}[hbt]
\renewcommand{\l}[1]{$l_{#1}$}
\renewcommand{\t}[1]{$t_{#1}$}
\newcommand{\n}{\begin{gammasize}$\gamma_1^1,\gamma_1^2,\gamma_1^3$\end{gammasize}}
\newcommand{\nn}{\begin{gammasize}$\gamma_2^1,\gamma_2^2,\gamma_3^3$\end{gammasize}}
\newcommand{\nnn}{\begin{gammasize}$\gamma_3^2,\gamma_2^3$\end{gammasize}}
\newcommand{\m}{\begin{gammasize}$\gamma_1^1,\gamma_1^2$\end{gammasize}}
\newcommand{\mm}{\begin{gammasize}$\gamma_3^1,\gamma_2^2,\gamma_2^3$\end{gammasize}}
\newcommand{\mmm}{\begin{gammasize}$\gamma_2^1,\gamma_3^2,\gamma_3^3$\end{gammasize}}
\centering\input{chapter-02-pict-122_233}
\caption{$\brackets{L, M,\leq}\iso\poset{P}_{122}$ and $\brackets{M, T,\leq}\iso\poset{P}_{233}$ (for $s=1$ or $3$ respectively).}\label{F: Alg: P122 P233}
\end{figure}
To see that the set $\fC{\plus}{m_1}$ is nonempty, note that otherwise $m_1\upseto=\emptyset$ so that $\set{m_1,t_1,t_2,t_3}$ would be a $4$-element antichain. Since the sets $\fC{\plus}{m_k}\setminus\{\gamma_2^s,\gamma_3^s\}\subseteq\{\gamma_k^i:m_k<t_i\}$ consist of colors with different subscripts $k$ and $\gamma_2^s\in\fC{\plus}{m_3}\setminus\fC{\plus}{m_2}$, $\gamma_3^s\in\fC{\plus}{m_2}\setminus\fC{\plus}{m_3}$, the sets $\fC{\plus}{m_k}$ are pairwise disjoint. This proves \cnone. To see \cntwo, note that $\gamma_k^i\in\fC{\plus}{m_k}\setminus\{\gamma_3^s,\gamma_2^s\}$ together with $l_k<m_k$ gives $l_k<m_k<t_i$.
Moreover, for $\gamma_k^i$ being one of $\gamma_2^s$ or $\gamma_3^s$ there are chains $l_2<m_3<t_s$ and $l_3<m_2<t_s$ of color $\gamma_2^s$ and $\gamma_3^s$, respectively.
To see \cninv, note that each ``lower'' edge (i.e. in $(L, M,\leq)\iso\poset{P}_{122}$) is of the form $l_k\prec m_k$ or $l_2\prec m_3$, $l_3\prec m_2$. However
\begin{align*}
\fC{\plus}{l_1}\, \cap\, \fC{\plus}{m_1}\ \ni\ &\ \gamma_1^1,\\
\fC{\plus}{l_2}\, \cap\, \fC{\plus}{m_3}\ \ni\ &\ \gamma_2^s,\\
\fC{\plus}{l_3}\, \cap\, \fC{\plus}{m_2}\ \ni\ &\ \gamma_3^s.
\end{align*}
For the edges $l_2\prec m_2$ and $l_3\prec m_3$ note first that a matching connecting $m_3$ with $t_s$ has to connect $m_2$ with $t_{i_2}\neq t_s$. 
Moreover a matching connecting $m_2$ with $t_s$ has to connect $m_3$ with $t_{i_3}\neq t_s$. 
This gives $m_2 < t_{i_2}\neq t_s$ and $m_3< t_{i_3}\neq t_s$ so that 
\begin{align*}
\fC{\plus}{l_2}\, \cap\, \fC{\plus}{m_2}\ \ni\ &\ \gamma_2^{i_2},\\
\fC{\plus}{l_3}\, \cap\, \fC{\plus}{m_3}\ \ni\ &\ \gamma_3^{i_3}.
\end{align*}
Each ``upper'' edge $m_k\prec t_i$, except $m_2\prec t_s$, $m_3\prec t_s$, has $\gamma_k^i$ as a private color, while $m_2\prec t_s$, $m_3\prec t_s$ are colored by $\gamma_3^s$, $\gamma_2^s$, respectively. 
As usual, if $\brackets{M, T,\leq}\iso\poset{P}_{333}$ then, to secure Invariant \ref{I:width3}.(2), we simply invoke Claim \ref{C:P333}.
\end{case}

\ec

\bigskip

Before considering cases in which $\brackets{L, M,\leq}\iso\poset{P}_{222}$, note that up to renumbering, we may arrange that $l_1<m_1,m_2$; $l_2<m_1,m_3$ and $l_3<m_2,m_3$ (see Figure \ref{F: Alg: P222}). Moreover we may exclude $\poset{P}_{\nabla}$ being $\poset{P}_{111}$ or $\poset{P}_{122}$, as they are already covered by their duals \nc{111}{222} and \nc{122}{222}, respectively. For all other $\poset{P}_{\nabla}$'s we have $L<T$, so that $\Gamma_0$ has all nine $\gamma_j^i$'s.

\FloatBarrier 

\begin{figure}[hbt]
\newcommand{\m}[1]{$m_{#1}$}
\renewcommand{\l}[1]{$l_{#1}$}
\centering\input{chapter-02-pict-222}
\caption{$\brackets{L, M,\leq}\iso\poset{P}_{222}$.}\label{F: Alg: P222}
\end{figure}

\FloatBarrier

%%%%%%%%%%%%%%%%%%%%%%%%%%%%%%%%%%%%%%%%%%%%%%%%%
\begin{case}{222}{222}{Figure \ref{F: Alg: P222 P222}}
Again, up to renumbering of the $t_i$'s we have $m_1<t_1,t_2$; $m_2<t_1,t_3$ and $m_3<t_2,t_3$. Putting $\fC{\plus}{m_1}=\set{\gamma_1^1,\gamma_2^2}$, $\fC{\plus}{m_2}=\set{\gamma_1^3,\gamma_3^1}$ and $\fC{\plus}{m_3}=\set{\gamma_3^2,\gamma_2^3}$ the conditions \cnall\ are fulfilled.
\end{case}

\ec

\FloatBarrier 

\begin{figure}[hbt]
\renewcommand{\l}[1]{$l_{#1}$}
\renewcommand{\t}[1]{$t_{#1}$}
\newcommand{\m}{\begin{gammasize}$\gamma_1^1,\gamma_2^2$\end{gammasize}}
\newcommand{\mm}{\begin{gammasize}$\gamma_1^3,\gamma_3^1$\end{gammasize}}
\newcommand{\mmm}{\begin{gammasize}$\gamma_3^2,\gamma_2^3$\end{gammasize}}
\centering\input{chapter-02-pict-222_222}
\caption{$\brackets{L, M,\leq}\iso\poset{P}_{222}$ and $\brackets{M, T,\leq}\iso\poset{P}_{222}$.}\label{F: Alg: P222 P222}
\end{figure}

\FloatBarrier

%%%%%%%%%%%%%%%%%%%%%%%%%%%%%%%%%%%%%%%%%%%%%%%%%
\begin{case}{222}{223}{Figure \ref{F: Alg: P222 P223}}
We may assume that $m_1<t_1,t_2$; $m_2<T$ and $m_3<t_2,t_3$. One can easily check that the sets $\fC{\plus}{m_1}=\set{\gamma_1^1,\gamma_2^2}$, $\fC{\plus}{m_2}=\set{\gamma_1^3,\gamma_1^2,\gamma_3^1}$ and $\fC{\plus}{m_3}=\set{\gamma_3^2,\gamma_2^3}$ satisfy the conditions \cnall{}.
\end{case}

\ec

\FloatBarrier 

\begin{figure}[hbt]
\renewcommand{\l}[1]{$l_{#1}$}
\renewcommand{\t}[1]{$t_{#1}$}
\newcommand{\m}{\begin{gammasize}$\gamma_1^1,\gamma_2^2$\end{gammasize}}
\newcommand{\mm}{\begin{gammasize}$\gamma_1^3,\gamma_1^2,\gamma_3^1$\end{gammasize}}
\newcommand{\mmm}{\begin{gammasize}$\gamma_3^2,\gamma_2^3$\end{gammasize}}
\centering\input{chapter-02-pict-222_223}
\caption{$\brackets{L, M,\leq}\iso\poset{P}_{222}$ and $\brackets{M, T,\leq}\iso\poset{P}_{223}$.}\label{F: Alg: P222 P223}
\end{figure}

\FloatBarrier

%%%%%%%%%%%%%%%%%%%%%%%%%%%%%%%%%%%%%%%%%%%%%%%%%
\begin{case}{222}{233}{Figure \ref{F: Alg: P222 P233}}
We may assume that $m_1,m_2<T$ and $m_3<t_2,t_3$. Obviously the sets $\fC{\plus}{m_1}=\set{\gamma_1^1,\gamma_1^3,\gamma_2^2}$, $\fC{\plus}{m_2}=\set{\gamma_1^2,\gamma_3^1,\gamma_3^3}$ and $\fC{\plus}{m_3}=\set{\gamma_2^3,\gamma_3^2}$ satisfy \cnall.
\end{case}

\ec

\FloatBarrier 

\begin{figure}[hbt]
\renewcommand{\l}[1]{$l_{#1}$}
\renewcommand{\t}[1]{$t_{#1}$}
\newcommand{\m}{\begin{gammasize}$\gamma_1^1,\gamma_1^3,\gamma_2^2$\end{gammasize}}
\newcommand{\mm}{\begin{gammasize}$\gamma_1^2,\gamma_3^1,\gamma_3^3$\end{gammasize}}
\newcommand{\mmm}{\begin{gammasize}$\gamma_2^3,\gamma_3^2$\end{gammasize}}
\centering\input{chapter-02-pict-222_233}
\caption{$\brackets{L, M,\leq}\iso\poset{P}_{222}$ and $\brackets{M, T,\leq}\iso\poset{P}_{233}$.}\label{F: Alg: P222 P233}
\end{figure}

\FloatBarrier 

%%%%%%%%%%%%%%%%%%%%%%%%%%%%%%%%%%%%%%%%%%%%%%%%%
\begin{case}{222}{333}{Figure \ref{F: Alg: P222 P333}}
Now we have $M<T$. The sets $\fCo{}{m_1}=\set{\gamma_1^2,\gamma_1^3,\gamma_2^1}$, $\fCo{}{m_2}=\set{\gamma_1^1,\gamma_3^2,\gamma_3^3}$ and $\fCo{}{m_3}=\set{\gamma_2^2,\gamma_2^3,\gamma_3^1}$ satisfy \cnall.
Claim \ref{C:P333} allows us to expand $\chainso$ to $\chains^{\plus}$ which satisfies Invariant \ref{I:width3}.
\end{case}

\ec

\FloatBarrier 

\begin{figure}[hbt]
\renewcommand{\l}[1]{$l_{#1}$}
\renewcommand{\t}[1]{$t_{#1}$}
\newcommand{\m}{\begin{gammasize}$\gamma_1^2,\gamma_1^3,\gamma_2^1$\end{gammasize}}
\newcommand{\mm}{\begin{gammasize}$\gamma_1^1,\gamma_3^2,\gamma_3^3$\end{gammasize}}
\newcommand{\mmm}{\begin{gammasize}$\gamma_2^2,\gamma_2^3,\gamma_3^1$\end{gammasize}}
\centering\input{chapter-02-pict-222_333}
\caption{$\brackets{L, M,\leq}\iso\poset{P}_{222}$ and $\brackets{M, T,\leq}\iso\poset{P}_{333}$.}\label{F: Alg: P222 P333}
\end{figure}

\FloatBarrier 

%\pagebreak

\bigskip

Being done with $\poset{P}_{\Delta}= \poset{P}_{222}$ we switch to $\poset{P}_{\Delta}=\poset{P}_{223}$. 
Without loss of generality we may assume that $l_1,l_2<m_1$; $L<m_2$ and $l_2,l_3<m_3$ as shown on Figure \ref{F: Alg: P223}.
As in Cases \nc{222}{\nabla} we may exclude $\poset{P}_{\nabla}=\poset{P}_{111},\poset{P}_{122}$ so that $L<T$ and in consequence $\abs{\Gamma_0}=9$.

\FloatBarrier 

\begin{figure}[hbt]
\newcommand{\m}[1]{$m_{#1}$}
\renewcommand{\l}[1]{$l_{#1}$}
\centering\input{chapter-02-pict-223}
\caption{$\brackets{L, M,\leq}\iso\poset{P}_{223}$.}\label{F: Alg: P223}
\end{figure}

\FloatBarrier

%%%%%%%%%%%%%%%%%%%%%%%%%%%%%%%%%%%%%%%%%%%%%%%%%
\begin{cas}{223}{223}{Figure \ref{F: Alg: P223 P223}}
There is only one point in $M$, say $m$, that has three neighbors in $T$. There are essentially two cases: $m=m_2$ and $m=m_1$.
\begin{enumerate}
\item If $m=m_2$ then $L<m_2<T$ and up to renumbering of $T$ we may assume that $m_1<t_1,t_2$ and $m_3<t_2,t_3$. The sets $\fC{\plus}{m_1}=\set{\gamma_1^1,\gamma_2^2}$, $\fC{\plus}{m_2}=\set{\gamma_1^2,\gamma_2^1,\gamma_3^3}$ and $\fC{\plus}{m_3}=\set{\gamma_2^3,\gamma_3^2}$ satisfy \cnall\ (see Figure \ref{F: Alg: P223 P223}.a).
\item If $m=m_1$, we may arrange $m_2<t_1, t_3$ and $m_3<t_2,t_3$. The sets $\fC{\plus}{m_1}=\set{\gamma_1^2,\gamma_1^3,\gamma_2^1}$, $\fC{\plus}{m_2}= \set{\gamma_1^1,\gamma_2^3,\gamma_3^1}$ and $\fC{\plus}{m_3}=\set{\gamma_2^2,\gamma_3^3}$ satisfy \cnall\ (see Figure \ref{F: Alg: P223 P223}.b).\hfill$\square$
\end{enumerate}
\end{cas}

\ec

\FloatBarrier 

\begin{figure}[hbt]
\renewcommand{\l}[1]{$l_{#1}$}
\renewcommand{\t}[1]{$t_{#1}$}
\newcommand{\m}{\begin{gammasize}$\gamma_1^1,\gamma_2^2$\end{gammasize}}
\newcommand{\mm}{\begin{gammasize}$\gamma_1^2,\gamma_2^1,\gamma_3^3$\end{gammasize}}
\newcommand{\mmm}{\begin{gammasize}$\gamma_2^3,\gamma_3^2$\end{gammasize}}
\newcommand{\n}{\begin{gammasize}$\gamma_1^2,\gamma_1^3,\gamma_2^1$\end{gammasize}}
\newcommand{\nn}{\begin{gammasize}$\gamma_1^1,\gamma_2^3,\gamma_3^1$\end{gammasize}}
\newcommand{\nnn}{\begin{gammasize}$\gamma_2^2,\gamma_3^3$\end{gammasize}}
\centering\input{chapter-02-pict-223_223}
\caption{$\brackets{L, M,\leq}\iso\poset{P}_{223}$ and $\brackets{M, T,\leq}\iso\poset{P}_{223}$.}\label{F: Alg: P223 P223}
\end{figure}

\FloatBarrier

%%%%%%%%%%%%%%%%%%%%%%%%%%%%%%%%%%%%%%%%%%%%%%%%%
\begin{cas}{223}{233}{Figures \ref{F: Alg: P223 P233 a} and \ref{F: Alg: P223 P233 b}}
There is only one point in $M$, say $m$, that has two neighbors in $T$. There are essentially two cases: $m=m_3$ and $m=m_2$.
\begin{enumerate}
\item For $m=m_3$ we may arrange $m_3<t_2,t_3$ and $m_1,m_2<T$. \label{Case:223 233 a} The sets $\fC{\plus}{m_1}=\set{\gamma_1^2,\gamma_1^3,\gamma_2^1}$, $\fC{\plus}{m_2}= \set{\gamma_1^1,\gamma_2^3,\gamma_3^2}$ and $\fC{\plus}{m_3}=\set{\gamma_2^2,\gamma_3^3}$ satisfy \cnall\ (see Figure \ref{F: Alg: P223 P233 a}).
\begin{figure}[hbt]
\renewcommand{\l}[1]{$l_{#1}$}
\renewcommand{\t}[1]{$t_{#1}$}
\newcommand{\m}{\begin{gammasize}$\gamma_1^2,\gamma_1^3,\gamma_2^1$\end{gammasize}}
\newcommand{\mm}{\begin{gammasize}$\gamma_1^1,\gamma_2^3,\gamma_3^2$\end{gammasize}}
\newcommand{\mmm}{\begin{gammasize}$\gamma_2^2,\gamma_3^3$\end{gammasize}}
\centering\input{chapter-02-pict-223_233_a}
\caption{$\brackets{L, M,\leq}\iso\poset{P}_{223}$, $\brackets{M, T,\leq}\iso\poset{P}_{233}$ and $m=m_3$.}\label{F: Alg: P223 P233 a}
\end{figure}
\item If $m_2$ has two neighbors in $T$, say $t_1,t_3$, then $m_1,m_3<T$.\label{Case:223 233 b}\linebreak
Recall that \mbox{$L<T$} gives not only \mbox{$\abs{\Gamma_0}=9$} but also 
\linebreak\mbox{$\abs{\fC{}{L}\cap \fC{}{T}}=11$}, by Invariant \ref{I:width3}.(2). In fact Algorithm must use both these additional colors \mbox{$\alpha, \beta\in\fC{}{L}\cap \fC{}{T}-\Gamma_{\!0}$}.
Invariant \ref{I:width3}.(2) tells us that there is no point $x$ in $L\cup T$ with $\alpha, \beta \in \fC{}{x}$. In particular $t_1$ or $t_3$ must be colored by one of $\alpha,\beta$. We may assume that $\alpha\in\fC{}{t_1}$. There are two cases depending on whether $\alpha\in\fC{}{l_3}$.
\begin{enumerate}
\item If $\alpha\notin\fC{}{l_3}$ then Algorithm colors $\fC{\plus}{m_1}=\set{\gamma_1^3,\gamma_2^2,\alpha}$, $\fC{\plus}{m_2}= \set{\gamma_1^1,\gamma_2^3,\gamma_3^1}$ and $\fC{\plus}{m_3}=\set{\gamma_2^1,\gamma_3^2,\gamma_3^3}$ (see Figure \ref{F: Alg: P223 P233 b}.a).\label{ain}
\item If $\alpha\in\fC{}{l_3}$ then Algorithm colors $\fC{\plus}{m_1}=\set{\gamma_1^2,\gamma_1^3,\gamma_2^1}$, $\fC{\plus}{m_2}= \set{\gamma_1^1,\gamma_2^3,\gamma_3^1}$ and $\fC{\plus}{m_3}=\set{\gamma_2^2,\gamma_3^3,\alpha}$ (see Figure \ref{F: Alg: P223 P233 b}.b).\label{anotin}
\end{enumerate}
It is easy to see that in both cases (\ref{ain}) and (\ref{anotin}) the multicoloring $\chains^{\plus}$ satisfies all conditions imposed by the rules of the \coredisjoint{} game as well as keeps Invariant \ref{I:width3}.\hfill$\square$
\end{enumerate}
\end{cas}

\ec

\FloatBarrier 

\begin{figure}[hbt]
\renewcommand{\a}{\begin{gammasize}$\ldots,\alpha$\end{gammasize}}
\newcommand{\na}{\begin{gammasize}$\ldots,\not\!\alpha$\end{gammasize}}
\renewcommand{\l}[1]{$l_{#1}$}
\renewcommand{\t}[1]{$t_{#1}$}
\newcommand{\m}{\begin{gammasize}$\gamma_1^3,\gamma_2^2,\alpha$\end{gammasize}}
\newcommand{\mm}{\begin{gammasize}$\gamma_1^1,\gamma_2^3,\gamma_3^1$\end{gammasize}}
\newcommand{\mmm}{\begin{gammasize}$\gamma_2^1,\gamma_3^2,\gamma_3^3$\end{gammasize}}
\newcommand{\n}{\begin{gammasize}$\gamma_1^2,\gamma_1^3,\gamma_2^1$\end{gammasize}}
\newcommand{\nn}{\begin{gammasize}$\gamma_1^1,\gamma_2^3,\gamma_3^1$\end{gammasize}}
\newcommand{\nnn}{\begin{gammasize}$\gamma_2^2,\gamma_3^3,\alpha$\end{gammasize}}
\centering\input{chapter-02-pict-223_233_b}
\caption{$\brackets{L, M,\leq}\iso\poset{P}_{223}$, $\brackets{M, T,\leq}\iso\poset{P}_{233}$ and $m=m_2$. Dotted lines indicate edges that have no private color in $\Gamma_{\!0}$.}\label{F: Alg: P223 P233 b}
\end{figure}

\FloatBarrier 

%\pagebreak

%%%%%%%%%%%%%%%%%%%%%%%%%%%%%%%%%%%%%%%%%%%%%%%%%
\begin{case}{223}{333}{Figure \ref{F: Alg: P223 P333}}
As in Case \nc{222}{333} we have $M<T$. The very same coloring that was used there, namely
$\fCo{}{m_1}=\set{\gamma_1^2,\gamma_1^3,\gamma_2^1}$, $\fCo{}{m_2}=\set{\gamma_1^1,\gamma_2^2,\gamma_3^3}$ and $\fCo{}{m_3}=\set{\gamma_2^3,\gamma_3^1,\gamma_3^2}$ satisfies \cnall.
Claim \ref{C:P333} allows us to expand $\chainso$ to $\chains^{\plus}$ which satisfies Invariant \ref{I:width3}.(2).
\end{case}

\ec

\FloatBarrier 

\begin{figure}[hbt]
\renewcommand{\l}[1]{$l_{#1}$}
\renewcommand{\t}[1]{$t_{#1}$}
\newcommand{\m}{\begin{gammasize}$\gamma_1^2,\gamma_1^3,\gamma_2^1$\end{gammasize}}
\newcommand{\mm}{\begin{gammasize}$\gamma_1^1,\gamma_2^2,\gamma_3^3$\end{gammasize}}
\newcommand{\mmm}{\begin{gammasize}$\gamma_2^3,\gamma_3^1,\gamma_3^2$\end{gammasize}}
\centering\input{chapter-02-pict-223_333}
\caption{$\brackets{L, M,\leq}\iso\poset{P}_{223}$ and $\brackets{M, T,\leq}\iso\poset{P}_{333}$.}\label{F: Alg: P223 P333}
\end{figure}

\FloatBarrier

%\bigskip

%\begin{figure}[hbt]
%\newcommand{\m}[1]{$m_{#1}$}
%\renewcommand{\l}[1]{$l_{#1}$}
%\centering\input{chapter-02-pict-233}
%\caption{$\brackets{L\cup M,\leq}\iso\poset{P}_{233}$.}\label{F: Alg: P233}
%\end{figure}

%%%%%%%%%%%%%%%%%%%%%%%%%%%%%%%%%%%%%%%%%%%%%%%%%
\begin{cas}{233}{233}{Figure \ref{F: Alg: P233 P233}}
Again we have $L<T$, so that $\abs{\Gamma_0}=9$.
We may assume that $l_1,l_2<M$ and $l_3<m_2,m_3$. Only one point from $M$, say $m$, has two neighbors in $T$. 
There are essentially two cases: $m=m_1$ and $m=m_3$.
\begin{enumerate}
\item If $m=m_1$ then, without loss of generality, we may assume that the only missing edge in $(M,T,\prec)$ is $(m_1, t_3)$. The sets $\fC{\plus}{m_1}=\set{\gamma_1^1,\gamma_2^2}$, $\fC{\plus}{m_2}= \set{\gamma_1^3,\gamma_2^3,\gamma_3^1,\gamma_3^2}$ and $\fC{\plus}{m_3}=\set{\gamma_1^2,\gamma_2^1,\gamma_3^3}$ (see Figure \ref{F: Alg: P233 P233}.a) satisfy \cnall.
\item If $m=m_3$ then again without loss of generality we may assume that the only missing edge in $(M,T,\prec)$ is $(m_3, t_1)$. \label{Case:233 233 b} In this case Algorithm has to use additional colors \mbox{$\alpha,\beta\notin\Gamma_{\!0}$} which are available due to Invariant \ref{I:width3}.(2) and \mbox{$L<T$}.
Moreover we know by Invariant \ref{I:width3}.(2) that there is no point $x\in L\cup T$ with $\alpha, \beta \in \fC{}{x}$.
Since $\alpha , \beta \in \fC{}{L} \cap \fC{}{T}$ we know that at least one point of degree $3$ from $L$ and one point of degree $3$ from $T$ must be colored by either $\alpha$ or $\beta$.
Without loss of generality we may assume that $\fC{}{l_2} \cap \set{\alpha, \beta} \neq \emptyset \neq \fC{}{t_2} \cap \set{\alpha, \beta}$.
Since $L<m_2<T$, any color can be used for $m_2$. Thus Algorithm puts $\alpha$ and $\beta$ into $\fC{\plus}{m_2}$. This gives that both edges $l_2 \prec m_2$ and $m_2 \prec t_2$ have private colors in $\set{\alpha, \beta}$. Removing these two edges gives rise to posets isomorphic to $\poset{P}_{223}$, namely $(L, M, \leq)$ without $l_2 \prec m_2$ and $(M,T, \leq)$ without $m_2 \prec t_2$. Distributing colors from $\Gamma_{\!0}$ on these two copies of $\poset{P}_{223}$ follows Case \nc{223}{223} by exchanging $l_1$ with $l_2$, $t_1$ with $t_2$ and $m_2$ with $m_3$, see Figures \ref{F: Alg: P223 P223}.b.\!\! and \ref{F: Alg: P233 P233}.b. This guarantees that each edge has a private color and therefore \cnall{} are fulfilled. Moreover points that got $\alpha$ or $\beta$ form chains going through $m_2$.\hfill$\square$
\end{enumerate}
\end{cas}

\ec

%\pagebreak

\FloatBarrier 

\begin{figure}[hbt]
\renewcommand{\a}{\begin{gammasize}$\ldots,\alpha\ \textrm{or}\ \beta$\end{gammasize}}
\newcommand{\na}{\begin{gammasize}$\ldots,\not\!\alpha$\end{gammasize}}
\renewcommand{\l}[1]{$l_{#1}$}
\renewcommand{\t}[1]{$t_{#1}$}
\newcommand{\m}{\begin{gammasize}$\gamma_1^1,\gamma_2^2$\end{gammasize}}
\newcommand{\mm}{\begin{gammasize}$\gamma_1^3,\gamma_2^3,\gamma_3^1,\gamma_3^2$\end{gammasize}}
\newcommand{\mmm}{\begin{gammasize}$\gamma_1^2,\gamma_2^1,\gamma_3^3$\end{gammasize}}
\newcommand{\n}{\begin{gammasize}$\gamma_1^2,\gamma_2^1,\gamma_2^3$\end{gammasize}}
\newcommand{\nn}{\begin{gammasize}$\gamma_1^1,\gamma_3^3,\alpha,\beta$\end{gammasize}}
\newcommand{\nnn}{\begin{gammasize}$\gamma_1^3,\gamma_2^2,\gamma_3^2$\end{gammasize}}
\centering\input{chapter-02-pict-233_233}
\caption{$\brackets{L, M,\leq}\iso\poset{P}_{233}$ and $\brackets{M, T,\leq}\iso\poset{P}_{233}$. Dotted lines indicate edges that have no private color in $\Gamma_{\!0}$.}\label{F: Alg: P233 P233}
\end{figure}

\FloatBarrier

%%%%%%%%%%%%%%%%%%%%%%%%%%%%%%%%%%%%%%%%%%%%%%%%%

\begin{case}{233}{333}{Figure \ref{F: Alg: P233 P333}} We may assume that $l_1,l_2<M$ and $l_3<m_2,m_3$.
Again we have here $L<T$, so that $\abs{\Gamma_0}=9$ and Algorithm has two additional colors $\alpha,\beta\notin\Gamma_{\!0}$. 
As previously, by Invariant \ref{I:width3}.(2) we know that there is no $x \in L \cup T$ with $\alpha, \beta \in \fC{}{x}$. Since $\alpha, \beta \in \fC{}{L} \cap \fC{}{T}$, at least one point of degree $3$ from $L$ must be colored by either $\alpha$ or $\beta$. Without loss of generality we may assume that $\alpha \in \fC{}{l_2}$. Putting $\fC{\plus}{m_1}=\set{\gamma_1^3,\gamma_2^1,\gamma_2^2}$, $\fC{\plus}{m_2}= \set{\gamma_1^2,\gamma_3^1,\gamma_3^3,\alpha}$ and $\fC{\plus}{m_3}=\set{\gamma_1^1,\gamma_2^3,\gamma_3^2,\beta}$ we fulfill \cnall{} as well as the requirement that the points colored by $\alpha$ (or $\beta$) form a chain. Since \mbox{$(M,T,\leq)\iso \poset{P}_{333}$} we have to check Invariant \ref{I:width3}.(2) on \mbox{$(M,T,\leq)$}. Each of its $9$ edges has exactly one color from $\Gamma_{\!0}$. Moreover $\alpha$ and $\beta$ were used on $4$ different points.
\end{case}

\ec

%\pagebreak

\FloatBarrier

\begin{figure}[hbt]
\renewcommand{\l}[1]{$l_{#1}$}
\renewcommand{\t}[1]{$t_{#1}$}
\renewcommand{\a}{\begin{scriptsize}$\ldots,\alpha$\end{scriptsize}}
\newcommand{\m}{\begin{gammasize}$\gamma_1^3,\gamma_2^1,\gamma_2^2$\end{gammasize}}
\newcommand{\mm}{\begin{gammasize}$\gamma_1^2,\gamma_3^1,\gamma_3^3,\alpha$\end{gammasize}}
\newcommand{\mmm}{\begin{gammasize}$\gamma_1^1,\gamma_2^3,\gamma_3^2,\beta$\end{gammasize}}
\centering\input{chapter-02-pict-233_333}
\caption{$\brackets{L, M,\leq}\iso\poset{P}_{233}$ and $\brackets{M, T,\leq}\iso\poset{P}_{333}$. %
Dotted line indicates edge that has no private color in~$\Gamma_{\!0}$.%
}\label{F: Alg: P233 P333}
\end{figure}

\FloatBarrier

%%%%%%%%%%%%%%%%%%%%%%%%%%%%%%%%%%%%%%%%%%%%%%%%%
\begin{case}{333}{333}{Figure \ref{F: Alg: P333 P333}}
We have $L<M<T$, so that $\abs{\Gamma_{\!0}}=9$. Algorithm initially colors $\fCo{}{m_1}=\set{\gamma_1^1,\gamma_2^2,\gamma_3^3}$, $\fCo{}{m_2}=\set{\gamma_1^2,\gamma_2^3,\gamma_3^1}$ and $\fCo{}{m_3}=\set{\gamma_1^3,\gamma_2^1,\gamma_3^2}$.
Obviously $\chainso$ satisfies \cnall.
Claim \ref{C:P333} allows Algorithm to expand $\chainso$ to a correct multicoloring $\chains^{\plus}$ which satisfies Invariant \ref{I:width3}.
\end{case}

\ec

%\pagebreak

\FloatBarrier

\begin{figure}[hbt]
\renewcommand{\l}[1]{$l_{#1}$}
\renewcommand{\t}[1]{$t_{#1}$}
\newcommand{\m}{\begin{gammasize}$\gamma_1^1,\gamma_2^2,\gamma_3^3$\end{gammasize}}
\newcommand{\mm}{\begin{gammasize}$\gamma_1^2,\gamma_2^3,\gamma_3^1$\end{gammasize}}
\newcommand{\mmm}{\begin{gammasize}$\gamma_1^3,\gamma_2^1,\gamma_3^2$\end{gammasize}}
\centering\input{chapter-02-pict-333_333}
\caption{$\brackets{L, M,\leq}\iso\poset{P}_{333}$ and $\brackets{M, T,\leq}\iso\poset{P}_{333}$.}\label{F: Alg: P333 P333}
\end{figure}

\FloatBarrier

%%%%%%%%%%%%%%%%%%%%%%%%%%%%%%%%%%%%%%%%%%%%%%%%%

\bigskip

This finishes our analysis of all possible cases in the \coredisjoint{} game. Recalling Theorem \ref{T:CDG} and Observations \ref{O:lval3>=10}, \ref{O:locval>=11} from Section \ref{S:lower bounds} we get the following theorem.
\begin{thm}\label{Thm:123}
$\fLCP{1} =1,\ \ \fLCP{2} = 4,\ \ \fLCP{3} = 11$.
\end{thm}

\bigskip

As a corollary from Theorem \ref{Thm:123} and Corollary \ref{R:CP sum LCP} we have: 
\begin{thm}\label{T:Bosek}
$\fCP{3} \leq 16$.
\end{thm}

\chapter{Chain Partitioning of  Up-growing Interval Orders}\label{Ch:On-line Interval}
In \cite{Felsner} Felsner considered a well motivated and interesting variant of the on-line chain partitioning problem. In this variant on-line posets are supposed to be presented in a special way. As we see in Theorem \ref{Th:Felsner}, this restriction for Spoiler allows Algorithm to perform much better than in general case, when only an exponential upper bound is known.

\begin{defn}
An \emph{on-line up-growing ordered set} (or \emph{on-line up-growing poset}) is the on-line ordered set $\poset{P}^{\ll}=\brackets{P,\leq,\ll}$ such that  each point $p_n\in P$ is maximal at the moment of its arrival, i.e. $p_n$ is maximal in the initial segment $\set{p_1,\ldots,p_n}\subseteq P$.
\end{defn}

As before, we define the \emph{value of the on-line chain partitioning problem for up-growing posets} to be the least integer $k$ such that there is an on-line algorithm that never uses more than $k$ chains on posets of width $w$.

\begin{thm}[Felsner \cite{Felsner}]\label{Th:Felsner} The value of the on-line chain partitioning problem for up-growing posets of width $w$ is $\binom{w+1}{2}$.
\end{thm}

In this chapter we are interested in on-line chain partitioning of up-growing \textit{interval orders}. These interval orders, as we have already seen in the Introduction, can sometimes be used to model scheduling of jobs with predefined time intervals in which they have to be executed.

\begin{defn}
An \emph{on-line interval order} is the on-line ordered set $\poset{P}^{\ll}=\brackets{P,\leq,\ll}$ such that the poset $\brackets{P,\leq}$ is an interval order.
\end{defn}

For interval orders the exponential bound $(5^w-1)/4$ of Theorem \ref{Th:Kie Sz} can be narrowed to a linear one, as Kiearstead and Trotter had shown.

\begin{thm}[Kierstead, Troter \cite{KiersteadTrotter}]\label{Thm:KiersteadTrotter}
The value of the on-line chain partitioning problem for interval orders of width $w$ is
\(
3w-2
\).
\end{thm}

Our goal is to consider on-line posets with both restrictions: we assume they are interval orders and they are presented in an up-growing way.

\begin{defn}
An \emph{on-line up-growing interval order}  is the on-line ordered set $\poset{P}^{\ll}=\brackets{P,\leq,\ll}$ such that
\begin{itemize}
\item the poset $\brackets{P,\leq}$ is an interval order,
\item  each point $p_n\in P$ is maximal at the moment of its arrival, i.e. $p_n$ is maximal in the initial segment $\set{p_1,\ldots,p_n}\subseteq P$.
\end{itemize}
\end{defn}

\begin{defn}
The \emph{value of the on-line chain partitioning problem for up-growing interval orders}, $\fCPUI{w}$, is the least integer $k$, such that there is an on-line algorithm that never uses more than $k$ chains on interval orders of width $w$.
\end{defn}

On-line chain partitioning of up-growing interval orders can also be viewed as the following two-person game. \emph{Spoiler} builds an interval order in a consecutive way and \emph{Algorithm} responds by partitioning it into chains. During each round:
\begin{itemize}
\item Spoiler introduces a new point $p$ with its comparability status to the previously presented points. This has to be done  in such a way that $p$ is maximal at the moment of its arrival.
\item Algorithm adds $p$ to some chain.
\end{itemize}
The goal of Algorithm is to use minimal number of chains. The goal of Spoiler is to force Algorithm to use as many chains as possible.

\bigskip

In Section \ref{section_Lower_bound} we present a strategy for Spoiler forcing $2w-1$ chains on interval orders of width $w$ (see Theorem \ref{twierdzenie lowerbound}). Moreover in Section \ref{section_Upper_bound} we present an on-line algorithm using $2w-1$ chains on interval orders of width $w$ (see Theorem \ref{twierdzenie upper}). In consequence we prove the following result.

\begin{thm}
The value of the on-line chain partitioning problem for up-growing interval orders of width $w$ is
\(
2w-1
\).
\end{thm}

%This paper is devoted to determine the analogous value of the on-line chain partitioning of up-growing interval orders denoted by $\fCPUI{w}$. We prove that
%Instead of referring directly to the definition of an interval order there are settings in which we use the following characterization theorem, a proof of which can be found in \cite{Trotter}.

\section{Lower bound}\label{section_Lower_bound}
\begin{thm}\label{twierdzenie lowerbound}
There is no on-line algorithm for chain partitioning of up-growing interval orders using less than $2w-1$ chains to cover posets of width $w$, i.e., $2w-1\leqslant \fCPUI{w}$.
\end{thm}

To prove Theorem \ref{twierdzenie lowerbound} one should provide a strategy for Spoiler building a poset of width $w$ and forcing Algorithm to use at least $2w-1$ chains. For the convenience, we will consider points presented by Spoiler as intervals representing these points. Such interval representation of an interval order gives too much information to Algorithm and therefore we relax the rules how this representation is given: before introducing a new interval Spoiler may change the position of right endpoints of intervals as long as after this change the corresponding poset  is the same. This means that Spoiler may  change the position of the right endpoint $r_a$ of an interval $I(a)=[l_a,r_a]$ whenever the modified interval $I'(a)=[l_a,r_a']$ does intersect the same intervals as $I(a)$ did before the change. This can be expressed by saying that for every interval $I(x)=[l_x,r_x]$ we have
\[
r_a\ <\ l_x\quad\ \textrm{iff}\ \quad r'_a\ <\ l_x.
\]
This relaxation is both necessary and allowed because Spoiler is not obliged to present a representation of the interval order and such modified family of intervals still represents the same interval order.

Our proof of Theorem \ref{twierdzenie lowerbound} heavily relies on Spoilers strategies described in the following claim.

\begin{clm} Let $M$ be an antichain and let $2\leqslant v\leqslant \abs{M}$. There is a strategy $S(v,M)$ for Spoiler to build, in an up-growing way, an extension $M\cup Q$ of $M$ such that
\begin{description}
\item[\inwidth] $Q$ has width $v$ and has $v$ maximal intervals.\\ Moreover $\fWidth{M\cup Q}=\abs{M}$.
\item[\inchains] Algorithm is forced to use at least $2v-2$ chains in $Q$, among them at least $v-1$ chains not used in $M$.
\item[\inq] For some minimal interval $q$ in $Q$ we have $q>M$.
\item[\insep] All intervals in $M$ that are not covered by chains used in $Q$ have the same right endpoint $r$ while the other intervals in $M$ have right endpoints to the left of $r$.
\end{description}
\end{clm}
\begin{proof}
Inducting on $v$ we construct recursively a strategy $S(v,M)$ for Spoiler satisfying all four conditions \inall{}. First, we provide a strategy $S(2,M)$ for an arbitrary antichain $M$ of width $w$:

Without loss of generality let $M$ be covered by the chains $1,\ldots,w$. First, Spoiler makes equal all right endpoints of intervals from $M$. Then Spoiler puts a new interval $x$ to the right of all the intervals of $M$. Algorithm will either decide to use one of the chains covering $M$ or covers $x$ with a new chain $w+1$.
In the first case we may assume that $x$ is covered by the chain $1$ containing interval $m\in M$.
Then Spoiler decreases the right endpoint of the interval $m$ and introduces a new interval $y$ only above $m$ (see the left part of Figure \ref{rys strategia s(2,m)}).
In the second case where $x$ is covered by a new chain $w+1$
Spoiler puts $y$ above each interval of $M$ (see the right part of Figure \ref{rys strategia s(2,m)}). No matter which chain Algorithm uses for $y$, the invariants are satisfied.
%\begin{center}
\begin{figure}[hbt]
\begin{small}
\psset{xunit=1mm,yunit=1mm,runit=1mm}
\psset{linewidth=0.3,dotsep=1,hatchwidth=0.3,hatchsep=1.5,shadowsize=1}
\psset{dotsize=0.7 2.5,dotscale=1 1,fillcolor=black}
\psset{arrowsize=1 2,arrowlength=1,arrowinset=0.25,tbarsize=0.7 5,bracketlength=0.15,rbracketlength=0.15}
\begin{pspicture}(0,0)(119,43)
\psline[linewidth=0.1](0,31)(33,31)
\psline[linewidth=0.1](2,35)(33,35)
\psline[linewidth=0.1](3,2)(24,2)
\psline[linewidth=0.1](2,6)(24,6)
\psline[linewidth=0.1](1,7)(33,7)
\psline[linewidth=0.1](2,11)(33,11)
\psline[linewidth=0.1](1,12)(33,12)
\psline[linewidth=0.1](1,16)(33,16)
\psline[linewidth=0.1](0,30)(33,30)
\psline[linewidth=0.1](2,26)(33,26)
\psline[linewidth=0.1](28,6)(56,6)
\psline[linewidth=0.1](28,2)(56,2)
\psline[linewidth=0.1](37,7)(56,7)
\psline[linewidth=0.1](37,11)(56,11)
\psline[linewidth=0.1](37,11)(37,7)
\psline[linewidth=0.1](33,11)(33,7)
\psline[linewidth=0.1](24,6)(24,2)
\psline[linewidth=0.1](28,6)(28,2)
\psline[linewidth=0.1](33,35)(33,31)
\psline[linewidth=0.1](56,6)(56,2)
\psline[linewidth=0.1](33,16)(33,12)
\psline[linewidth=0.1](33,30)(33,26)
\rput(13,19){$.$}
\rput(13,21){$.$}
\rput(13,23){$.$}
\psline[linewidth=0.5,linestyle=dashed,dash=1 1](6,0)(6,37)
\rput[t](6,42){$M$}
\rput(10,43){}
\psline[linewidth=0.1](56,11)(56,7)
\psline[linewidth=0.1](65,31)(93,31)
\psline[linewidth=0.1](67,35)(93,35)
\psline[linewidth=0.1](68,2)(93,2)
\psline[linewidth=0.1](67,6)(93,6)
\psline[linewidth=0.1](66,7)(93,7)
\psline[linewidth=0.1](67,11)(93,11)
\psline[linewidth=0.1](66,12)(93,12)
\psline[linewidth=0.1](66,16)(93,16)
\psline[linewidth=0.1](65,30)(93,30)
\psline[linewidth=0.1](67,26)(93,26)
\psline[linewidth=0.1](98,6)(119,6)
\psline[linewidth=0.1](98,2)(119,2)
\psline[linewidth=0.1](98,7)(119,7)
\psline[linewidth=0.1](98,11)(119,11)
\psline[linewidth=0.1](98,11)(98,7)
\psline[linewidth=0.1](93,11)(93,7)
\psline[linewidth=0.1](93,6)(93,2)
\psline[linewidth=0.1](98,6)(98,2)
\psline[linewidth=0.1](93,35)(93,31)
\psline[linewidth=0.1](119,6)(119,2)
\psline[linewidth=0.1](93,16)(93,12)
\psline[linewidth=0.1](93,30)(93,26)
\rput(79,19){$.$}
\rput(79,21){$.$}
\rput(79,23){$.$}
\psline[linewidth=0.5,linestyle=dashed,dash=1 1](72,0)(72,37)
\rput[t](72,42){$M$}
\rput(75,43){}
\psline[linewidth=0.1](119,11)(119,7)
\rput(13,33){$w$}
\rput(13,28){$w-1$}
\rput(13,14){$3$}
\rput(13,9){$2$}
\rput(13,4){$1$}
\rput(79,4){$1$}
\rput(79,9){$2$}
\rput(79,14){$3$}
\rput(79,28){$w-1$}
\rput(79,33){$w$}
\rput(43.5,9){$1$}
\rput(43.5,4){$w+1$}
\rput(106,9){$w+1$}
\rput(106.5,4){$?$}
\rput(54,3){$y$}
\rput(54,8){$x$}
\rput(21,3){$m$}
\rput(117,8){$x$}
\rput(117,3){$y$}
\end{pspicture}
\end{small}
\caption{The strategy $S(2,M)$.}\label{rys strategia s(2,m)}
\end{figure}
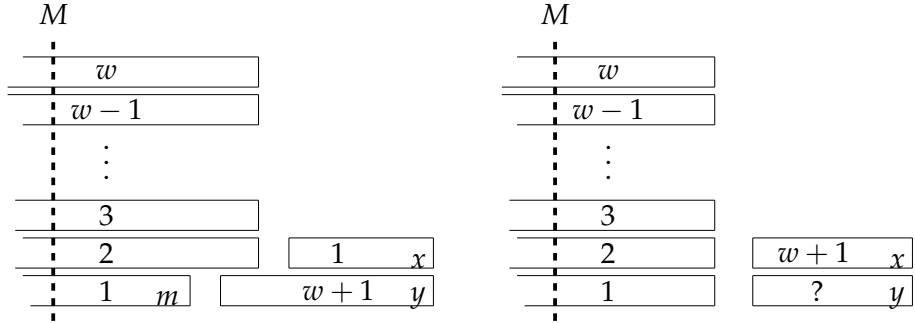
%\end{center}

The induction step proceeds from $v$ to $v+1$ for $v<w$ and  an arbitrary antichain $M$ of width $w$. First, Spoiler extends $M$ by $Q'$ calling recursively the procedure $S(v,M)$. Thus, Algorithm uses at least $2v-2$ chains to cover $Q'$ and at least $v-1$ among them are not used in $M$. Let $A$ be the antichain of all maximal points in $Q'$. \linebreak By \inwidth{} for $Q'$ we have $\abs{A}=v$. Now, Spoiler runs $S(v,A)$ producing an extension $Q''$ of $A$. According to the induction hypothesis Algorithm is forced to use in $Q''$ at least $v-1$ chains not used in $A$. All these together yield that
\begin{texteqn}\label{ET:QQ}
at least $(v-1)+v=2v-1$ chains are used in $Q'\cup Q''$
\end{texteqn}
and among them at least $v-1$ chains are not used in $M$. Let $N$ be the subset of intervals of $M$ whose chains are used in $Q'\cup Q''$; $M_1$ be the subset of $M$ whose chains are used in $Q'$ and $M_2:=M-M_1$. Note that $M_1\subseteq N$. Now, there are two cases:

\smallskip

\noindent\textbf{Case 1.} Algorithm used at least $v$ chains from $M$ in $Q'\cup Q''$, i.e. $\abs{N}\geqslant v$.

In this case Spoiler introduces an interval $x$ above all intervals in $N$ and incomparable to other intervals.
\begin{figure}[hbt]
\unitlength 1mm
\begin{picture}(115,66)(0,0)
\linethickness{0.1mm}
\put(1,44){\line(1,0){53}}
\linethickness{0.1mm}
\put(3,48){\line(1,0){51}}
\linethickness{0.1mm}
\put(4,8){\line(1,0){21}}
\linethickness{0.1mm}
\put(3,12){\line(1,0){22}}
\linethickness{0.1mm}
\put(3,56){\line(1,0){51}}
\linethickness{0.1mm}
\put(1,60){\line(1,0){53}}
\linethickness{0.1mm}
\put(1,20){\line(1,0){34}}
\linethickness{0.1mm}
\put(2,24){\line(1,0){33}}
\linethickness{0.1mm}
\put(3,26){\line(1,0){42}}
\linethickness{0.1mm}
\put(2,30){\line(1,0){43}}
\linethickness{0.1mm}
\put(1,42){\line(1,0){44}}
\linethickness{0.1mm}
\put(4,38){\line(1,0){41}}
\linethickness{0.1mm}
\put(29,13){\line(1,0){43}}
\linethickness{0.1mm}
\put(29,9){\line(1,0){45}}
\linethickness{0.1mm}
\put(39,21){\line(1,0){36}}
\linethickness{0.1mm}
\put(39,25){\line(1,0){38}}
\linethickness{0.1mm}
\put(61,6){\line(1,0){12}}
\linethickness{0.1mm}
\put(61,2){\line(1,0){14}}
\linethickness{0.1mm}
\put(39,21){\line(0,1){4}}
\linethickness{0.1mm}
\put(35,20){\line(0,1){4}}
\linethickness{0.1mm}
\put(25,8){\line(0,1){4}}
\linethickness{0.1mm}
\put(29,9){\line(0,1){4}}
\linethickness{0.1mm}
\put(54,56){\line(0,1){4}}
\linethickness{0.1mm}
\put(54,44){\line(0,1){4}}
\linethickness{0.1mm}
\put(61,2){\line(0,1){4}}
\linethickness{0.1mm}
\put(45,26){\line(0,1){4}}
\linethickness{0.1mm}
\put(45,38){\line(0,1){4}}
\put(48,50){\makebox(0,0)[cc]{$.$}}

\linethickness{0.5mm}
\multiput(18,6)(0,2){10}{\line(0,1){1}}
\put(48,52){\makebox(0,0)[cc]{$.$}}

\put(48,54){\makebox(0,0)[cc]{$.$}}

\linethickness{0.5mm}
\multiput(65,0)(0,2){15}{\line(0,1){1}}
\put(23,14){\makebox(0,0)[cc]{$.$}}

\put(23,16){\makebox(0,0)[cc]{$.$}}

\put(23,18){\makebox(0,0)[cc]{$.$}}

\put(40,32){\makebox(0,0)[cc]{$.$}}

\put(40,34){\makebox(0,0)[cc]{$.$}}

\put(40,36){\makebox(0,0)[cc]{$.$}}

\linethickness{0.1mm}
\put(98,5){\line(1,0){17}}
\linethickness{0.1mm}
\put(99,9){\line(1,0){16}}
\linethickness{0.1mm}
\put(100,18){\line(1,0){15}}
\linethickness{0.1mm}
\put(101,22){\line(1,0){14}}
\linethickness{0.1mm}
\put(100,23){\line(1,0){15}}
\linethickness{0.1mm}
\put(98,27){\line(1,0){17}}
\linethickness{0.1mm}
\put(115,5){\line(0,1){4}}
\linethickness{0.1mm}
\put(115,18){\line(0,1){4}}
\linethickness{0.1mm}
\put(115,23){\line(0,1){4}}
\put(110,12){\makebox(0,0)[cc]{$.$}}

\put(110,14){\makebox(0,0)[cc]{$.$}}

\put(110,16){\makebox(0,0)[cc]{$.$}}

\linethickness{0.5mm}
\multiput(105,0)(0,2){15}{\line(0,1){1}}
\linethickness{0.5mm}
\multiput(65,29)(1.95,0){21}{\line(1,0){0.98}}
\linethickness{0.5mm}
\multiput(65,0)(1.95,0){21}{\line(1,0){0.98}}
\linethickness{0.1mm}
\put(49,35){\line(1,0){66}}
\put(85,12){\makebox(0,0)[cc]{$Q'\cup Q''$}}

\linethickness{0.5mm}
\multiput(8,6)(0,1.96){29}{\line(0,1){0.98}}
\put(8,66){\makebox(0,0)[tc]{$M$}}

\put(11,49){\makebox(0,0)[cc]{}}

\put(25,66){\makebox(0,0)[tc]{$M_2$}}

\linethickness{0.1mm}
\put(49,39){\line(1,0){66}}
\linethickness{0.1mm}
\put(115,35){\line(0,1){4}}
\linethickness{0.1mm}
\put(49,35){\line(0,1){4}}
\put(83,36){\makebox(0,0)[bc]{$x$}}

\put(70,2){\makebox(0,0)[bc]{$q$}}

\linethickness{0.5mm}
\multiput(13,6)(0,2){19}{\line(0,1){1}}
\linethickness{0.5mm}
\multiput(25,25)(0,2){19}{\line(0,1){1}}
\put(50,15){\makebox(0,0)[cc]{$.$}}

\put(50,17){\makebox(0,0)[cc]{$.$}}

\put(50,19){\makebox(0,0)[cc]{$.$}}

\put(13,5){\makebox(0,0)[tc]{$N$}}

\put(19,5){\makebox(0,0)[tc]{$M_1$}}

\end{picture}
\caption{The strategy $S(v+1,M)$ in case 1.}\label{F: Str S(v+1,M)}
\end{figure}
To do that Spoiler first rearranges the right endpoints of intervals in $M$ so that each interval in $N$ ends before any interval in $M-N$ does (see Figure \ref{F: Str S(v+1,M)}). To see that such rearrangement is possible, first note that by \insep{} for $Q'$, intervals from $M_1$ are already enough to the left. Moreover, by \insep{} for $Q'$, the right endpoints of intervals from $M_2$ coincide. Thus, Spoiler may shorten intervals in $M_2\cap N$ to end them to the left of right endpoints of intervals $M_2-N$.

Now, knowing that such an $x$ can be introduced by Spoiler we put $Q:=Q'\cup Q''\cup\set{x}$ and we argue that $M$ extended with $Q$ satisfies conditions \inall{}.
\begin{description}
\item[ad \inwidth{}] $M\cup Q'$ has width $w$, $Q'$ has width $v$ and $v$ maximal points, forming the set $A$,  which is the base for $Q''$. This together with the fact that $A\cup Q''$ has width $v$ implies that $\fWidth{M\cup Q'\cup Q''}=w$, $\fWidth{Q'\cup Q''}=v$ and therefore $\fWidth{Q'\cup Q''\cup\set{x}}=v+1$. Obviously, \mbox{$Q'\mspace{-1mu}\cup\mspace{-1mu} Q''\mspace{-1mu}\cup\mspace{-1mu}\set{x}$} has $v+1$ maximal points: all $v$ maximal points of $Q''$ and the point $x$. It remains to prove that
$$\fWidth{M\cup Q'\cup Q''\cup\set{x}}=w.$$ The only region where $x$ could increase the width to become greater than $w$ is where $x\parallel M-N$. But since $q$ is above \mbox{$M-N$} the size of such an antichain is bounded by $\abs{M-N}+1+(v-1)\leqslant (w-v)+1+(v-1)=w$.
\item[ad \inchains{}] Algorithm has to cover $x$ by a completely new chain, i.e. one not used in $M\cup Q'\cup Q''$. Thus, there are at least
\linebreak $v+(v-1)+1=2(v+1)-2$ chains used in $Q$ and \linebreak $(v-1)+1=(v+1)-1$ chains in $Q$ and not used in $M$.
\item[ad \inq{}] The point witnessing \inq{} for $Q'$ witnesses \inq{} also for $Q$.
\item[ad \insep{}] The rearrangement of right endpoints of intervals in $M$ made by Spoiler before introducing $x$, guarantees that \insep{} still holds.
\end{description}
\begin{figure}[hbt]
\unitlength 1mm
\begin{picture}(115,62)(0,0)
\linethickness{0.1mm}
\put(1,44){\line(1,0){44}}
\linethickness{0.1mm}
\put(3,48){\line(1,0){42}}
\linethickness{0.1mm}
\put(4,8){\line(1,0){21}}
\linethickness{0.1mm}
\put(3,12){\line(1,0){22}}
\linethickness{0.1mm}
\put(3,56){\line(1,0){42}}
\linethickness{0.1mm}
\put(1,60){\line(1,0){44}}
\linethickness{0.1mm}
\put(1,20){\line(1,0){34}}
\linethickness{0.1mm}
\put(2,24){\line(1,0){33}}
\linethickness{0.1mm}
\put(3,26){\line(1,0){42}}
\linethickness{0.1mm}
\put(2,30){\line(1,0){43}}
\linethickness{0.1mm}
\put(1,42){\line(1,0){44}}
\linethickness{0.1mm}
\put(4,38){\line(1,0){41}}
\linethickness{0.1mm}
\put(29,13){\line(1,0){43}}
\linethickness{0.1mm}
\put(29,9){\line(1,0){45}}
\linethickness{0.1mm}
\put(39,21){\line(1,0){36}}
\linethickness{0.1mm}
\put(39,25){\line(1,0){38}}
\linethickness{0.1mm}
\put(61,6){\line(1,0){12}}
\linethickness{0.1mm}
\put(61,2){\line(1,0){14}}
\linethickness{0.1mm}
\put(39,21){\line(0,1){4}}
\linethickness{0.1mm}
\put(35,20){\line(0,1){4}}
\linethickness{0.1mm}
\put(25,8){\line(0,1){4}}
\linethickness{0.1mm}
\put(29,9){\line(0,1){4}}
\linethickness{0.1mm}
\put(45,56){\line(0,1){4}}
\linethickness{0.1mm}
\put(45,44){\line(0,1){4}}
\linethickness{0.1mm}
\put(61,2){\line(0,1){4}}
\linethickness{0.1mm}
\put(45,26){\line(0,1){4}}
\linethickness{0.1mm}
\put(45,38){\line(0,1){4}}
\put(40,50){\makebox(0,0)[cc]{$.$}}

\put(40,52){\makebox(0,0)[cc]{$.$}}

\put(40,54){\makebox(0,0)[cc]{$.$}}

\linethickness{0.5mm}
\multiput(65,0)(0,2){15}{\line(0,1){1}}
\put(23,14){\makebox(0,0)[cc]{$.$}}

\put(23,16){\makebox(0,0)[cc]{$.$}}

\put(23,18){\makebox(0,0)[cc]{$.$}}

\put(40,32){\makebox(0,0)[cc]{$.$}}

\put(40,34){\makebox(0,0)[cc]{$.$}}

\put(40,36){\makebox(0,0)[cc]{$.$}}

\linethickness{0.1mm}
\put(98,5){\line(1,0){17}}
\linethickness{0.1mm}
\put(99,9){\line(1,0){16}}
\linethickness{0.1mm}
\put(100,18){\line(1,0){15}}
\linethickness{0.1mm}
\put(101,22){\line(1,0){14}}
\linethickness{0.1mm}
\put(100,23){\line(1,0){15}}
\linethickness{0.1mm}
\put(98,27){\line(1,0){17}}
\linethickness{0.1mm}
\put(115,5){\line(0,1){4}}
\linethickness{0.1mm}
\put(115,18){\line(0,1){4}}
\linethickness{0.1mm}
\put(115,23){\line(0,1){4}}
\put(110,12){\makebox(0,0)[cc]{$.$}}

\put(110,14){\makebox(0,0)[cc]{$.$}}

\put(110,16){\makebox(0,0)[cc]{$.$}}

\linethickness{0.5mm}
\multiput(105,0)(0,2){15}{\line(0,1){1}}
\linethickness{0.5mm}
\multiput(65,29)(1.95,0){21}{\line(1,0){0.98}}
\linethickness{0.5mm}
\multiput(65,0)(1.95,0){21}{\line(1,0){0.98}}
\linethickness{0.1mm}
\put(49,35){\line(1,0){66}}
\put(85,12){\makebox(0,0)[cc]{$Q'\cup Q''$}}

\linethickness{0.5mm}
\multiput(8,6)(0,1.96){29}{\line(0,1){0.98}}
\put(8,5){\makebox(0,0)[tc]{$M$}}

\put(11,49){\makebox(0,0)[cc]{}}

\linethickness{0.1mm}
\put(49,39){\line(1,0){66}}
\linethickness{0.1mm}
\put(115,35){\line(0,1){4}}
\linethickness{0.1mm}
\put(49,35){\line(0,1){4}}
\put(83,36){\makebox(0,0)[bc]{$x$}}

\put(70,2){\makebox(0,0)[bc]{$q$}}

\linethickness{0.5mm}
\multiput(15,6)(0,2){19}{\line(0,1){1}}
\put(50,15){\makebox(0,0)[cc]{$.$}}

\put(50,17){\makebox(0,0)[cc]{$.$}}

\put(50,19){\makebox(0,0)[cc]{$.$}}

\put(15,5){\makebox(0,0)[tc]{$N$}}

\end{picture}
\caption{The strategy $S(v+1,M)$ in case 2.}\label{F: case2 : S(v+1,M)}
\end{figure}
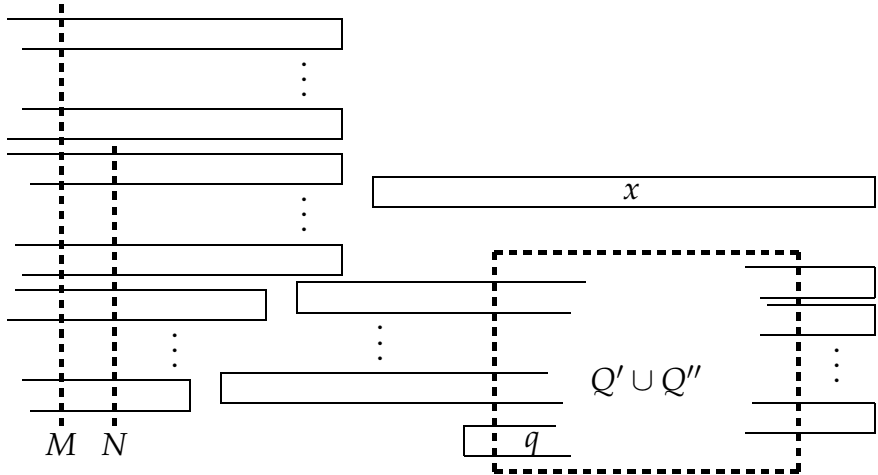

\medskip

\noindent\textbf{Case 2.} Algorithm used at most $v-1$ chains from $M$ to cover \mbox{$Q'\cup Q''$}.

By (\ref{ET:QQ}), this means that at least $v$ chains not used in $M$ are used in \mbox{$Q'\cup Q''$}. In this case, Spoiler puts new point $x$ above all intervals in $M$ and only above them (see Figure \ref{F: case2 : S(v+1,M)}). Put $Q:=Q'\cup Q''\cup\set{x}$. Invariants \inwidth{} and \inq{} may be proved in a very similar way as in \mbox{Case 1}. To prove \inchains{}, observe that while covering $x$, Algorithm may use a completely new chain or a chain used in $M$ and not used in $Q$. In both cases we get $2v-1+1=2(v+1)-2$ chains in $Q'\cup Q''\cup\set{x}$ and already knowing that there are at least $v$ chains used in $Q'\cup Q''$ not from $M$, we are done. The condition \insep{} is trivially kept.
\end{proof}

Now, to prove Theorem \ref{twierdzenie lowerbound}, Spoiler starts with an antichain $M$ with $w$ points. Next he applies strategy $S(w,M)$ to build $Q$ that extends $M$ and forces Algorithm to use at least $w-1$ chains not used in $M$. This shows that at least $2w-1$ chains have to be used, so that $\fCPUI{w}\geqslant 2w-1$.

\section{Upper bound}\label{section_Upper_bound}
\begin{thm}\label{twierdzenie upper}
There is an on-line algorithm for chain partitioning of up-growing interval orders that uses at most $2w-1$ chains, where $w$ is the width of the poset, i.e., $\fCPUI{w}\leqslant 2w-1$.
\end{thm}
\begin{proof}
Our algorithm maintains an auxiliary data structure $\str$ that depends on an up-growing interval order presented as an input $\poset{P}=(P,\leqslant)$ as well as on the already built covering of $\poset{P}$. When $\poset{P}$ expands to $\poset{P}^{\plus}=(P\cup\set{x},\leqslant)$ by a new maximal point $x$ then our algorithm modifies $\str$ to get a new structure $\strp$ for $\poset{P^{\plus}}$. The chain for $x$ can be easily read from $\strp$. Put%
\[
\str\ =\ (P,\leqslant,w,L_1,\ldots,L_w,\alpha_1,\ldots,\alpha_w,\beta_1,\ldots,\beta_w),
\]%
where $w=\fWidth{P}$ and:
\begin{description}
\item[\jnlevels{}] $L_1,\dots,L_w$ are high antichains in $P$ such that \mbox{$L_1\upsetc\subseteq\ldots\subseteq L_w\upsetc$} and $\abs{L_i}=i$,
\item[\jnchains{}] $\alpha_1,\ldots,\alpha_w,\beta_1,\ldots,\beta_w$ forms a chain partition of $P$,
\item[\jnrel{}] $\alpha_i\subseteq L_i\downsetc -L_{i-1}\upsetc$\, and\, $\beta_i\subseteq L_i\downseto$,
\end{description}
where for a further simplification we put additionally $L_0=\emptyset$.

For an on-line algorithm it is important that chains generated to cover $P^{\plus}$ expand those for $P$. This will be secured by the condition that $\alpha_i's$ and $\beta_i's$ grow in time as shown in lines \ref{algorytm CaseA alfai+=alfai}, \ref{algorytm CaseA betai+=betai}, \ref{algorytm CaseB if dla alfai+}, \ref{algorytm CaseB if dla beta+} of the algorithm we construct. Antichains $L_1,\ldots,L_w$ may be seen as levels of the poset $P$. Two consecutive levels $L_{i-1}$, $L_i$ determine our chains $\alpha_i$ and $\beta_i$ as described in \jnrel{}.

Before describing our algorithm we note that all segments of the form $X\downsetc, X\downseto, X\upsetc, X\upseto$ are considered in $\poset{P^{\plus}}$. Whenever we need to refer to an upset in $P$ we write $X\upsetc\cap P$, while for $X\subseteq P$ downsets $X\downsetc$ are the same in $P$ and $P^{\plus}$.

\clearpage

\begin{alg}{$\ $} \label{A:Int.main}
\renewcommand*\thelstnumber{(A\oldstylenums{\the\value{lstnumber}})}
\begin{lstlisting}[mathescape]
	if $\fWidth{P^{\plus}}>\fWidth{P}$ then (*@\textbf{Case A}@*) else (*@\textbf{Case B}@*)	(*@\label{alg if w(P)<w(P+) then Case A} @*)(*@\vspace{3.5mm}@*)
	(*@\textbf{Case A}@*)
	$w^{\plus}:=w+1$
	for $i:=1$ to $w$ do
		$L_i^{\plus}:=L_i$						(*@\label{algorytm CaseA Li+=Li} @*)
		$\alpha_i^{\plus}:=\alpha_i$					(*@\label{algorytm CaseA alfai+=alfai} @*)
		$\beta_i^{\plus}:=\beta_i$					(*@\label{algorytm CaseA betai+=betai} @*)
	$L_{w+1}^{\plus}:=L_w\cup\set{x}$					(*@\label{algorytm CaseA L+w+1 = Lw suma x} @*)
	$\alpha_{w+1}^{\plus}:=\set{x}$					(*@\label{algorytm CaseA alfaw+1+=x} @*)
	$\beta_{w+1}^{\plus}:=\emptyset$					(*@\label{algorytm CaseA betaw+1+=pusty} @*)(*@\vspace{3.5mm}@*)
	(*@\textbf{Case B}@*) (*@\hfill\it{if $\fWidth{P}=\fWidth{P^{\plus}}$ then $x\in L_w\upseto$} @*)
	$w^{\plus}:=w$
	$i_0 := \descMin\set{i \in \N : x \in L_i\upseto}$		(*@\label{algorytm CaseB definicja i0} @*)
	for $i:=1$ to $w$ do
		$L_i^{\plus}:=\fHMA{L_i\upsetc}$ 				(*@\label{algorytm CaseB Li+=hmaLi} @*)
		if $i\neq i_0$ then $\alpha_i^{\plus}:=\alpha_i$ else $\alpha_i^{\plus}:=\beta_i\cup\set{x}$ (*@\label{algorytm CaseB if dla alfai+}@*)
		if $i\neq i_0$ then $\beta_i^{\plus}:=\beta_i$ else $\beta_i^{\plus}:=\alpha_i$ (*@\label{algorytm CaseB if dla beta+} @*)
\end{lstlisting}
\end{alg}

First of all our algorithm checks whether the new point $x$ enlarges the width of the poset. If it does then the algorithm may use two new chains (see \ref{algorytm CaseA alfaw+1+=x},\ref{algorytm CaseA betaw+1+=pusty}). We prove that our algorithm upgrading $\str$ to $\strp=(P^{\plus},\leqslant,w^{\plus},L_1^{\plus},\ldots,L_{w^{\plus}}^{\plus},\alpha_1^{\plus},\ldots,\alpha_{w^{\plus}}^{\plus},\beta_1^{\plus},\ldots,\beta_{w^{\plus}}^{\plus})$ keeps the properties \jnall{} invariant. The proof of this splits into two parts corresponding to cases A and B, respectively.

\bigskip

\noindent\textbf{Case A: } In this setting, according to \ref{alg if w(P)<w(P+) then Case A}, we have $\fWidth{P^{\plus}}>\fWidth{P}$. Since one point may increase the width of the poset at most by 1, we have
\[
\fWidth{P^{\plus}}=\fWidth{P}+1.
\]
An additional new level $L_{w+1}^{\plus}$ is defined by \ref{algorytm CaseA L+w+1 = Lw suma x}. The other levels are unchanged, see \ref{algorytm CaseA Li+=Li}. The new point $x$ is covered by a new chain $\alpha_{w+1}^{\plus}$, see \ref{algorytm CaseA alfaw+1+=x}. Chain $\beta_{w+1}^{\plus}$ is defined as an empty set, see \ref{algorytm CaseA betaw+1+=pusty}.

By \jnlevels{} we know that $L_w=\fHMA{P}$. Since point $x$ increased the
width, we know that $x$ is incomparable with some maximum antichain
$A$ in $P$. Of course, we have $L_w=\fHMA{P}\subseteq A\upsetc$. Thus,
from $x\notin L_w\upseto$ and from the fact that $x$ is maximal in
$P^{\plus}$ (since presented poset must be up-growing) we get that $L_{w+1}^{\plus}=L_w\cup\set{x}$ is a high antichain in $P^{\plus}$.

Obviously, $L_{i}^{+}\upsetc\subseteq L_{i+1}^{+}\upsetc$ for $i=1,\ldots,w$. Combining this and $x\notin L_w\upseto$ we obtain $x\notin L_i\upseto$ for all $i=1,\ldots,w$. Thus,
\[
L_i^{\plus}\upsetc=L_i\upsetc=L_i\upsetc\cap P.
\]
Since the $L_i$'s are high in $P$ (see \jnlevels{} for $\str$) we immediately get that the $L_i^{\plus}$'s are high in $P^{\plus}$.

The condition describing the cardinalities of $L_i^{\plus}$'s as well as these concerning sets $\alpha_i^{\plus}$ and $\beta_i^{\plus}$ trivially follow from those for $\str$.

\bigskip

\noindent\textbf{Case B: } By \ref{alg if w(P)<w(P+) then Case A} we have $\fWidth{P^{\plus}}=\fWidth{P}$. Thus, there is a point in $L_w$ comparable with $x$. Since $x$ is maximal in $P^{\plus}$ we get that $x\in L_w\upseto$. Now, we know that $i_0$ in line \ref{algorytm CaseB definicja i0} is well-defined as the set under $\min$ function is not empty.
\begin{clm}\label{C:L1Lw} $L_1^{\plus},\dots,L_w^{\plus}$ are high antichains in $P^{\plus}$. Moreover,\linebreak  \mbox{$L_1^{\plus}\upsetc\subseteq\ldots\subseteq L_w^{\plus}\upsetc$} and $\abs{L_i^{\plus}}=i$.
\end{clm}
\begin{proof}
Since $L_i^{\plus}\!=\!\descHMA(L_i\upsetc)$ we get that the levels \mbox{$L_1^{\plus},\ldots,L_w^{\plus}$} of the poset $P^{\plus}$ are high antichains in $P^{\plus}\!$. To prove that
\mbox{$L_1^{\plus}\upsetc\subseteq\ldots\subseteq L_w^{\plus}\upsetc$} we observe that
\begin{align*}
L_i^{\plus}&\ =\ \fHMA{L_i\upsetc}\qquad&\textrm{by \ref{algorytm CaseB Li+=hmaLi}}\\
&\ \subseteq\ L_i\upsetc&\\
&\ \subseteq\ L_{i+1}\upsetc.&\textrm{by \jnlevels{} for $\str$}
\end{align*}
By \jnlevels{} we already know that $L_i^{\plus}$ is high in $P^{\plus}$. But $L_i^{\plus}\subseteq L_{i+1}\upsetc$ so that $L_i^{\plus}$ is high in $L_{i+1}\upsetc$. Applying Observation \ref{O:HM} to the poset \mbox{$(L_{i+1}\upsetc,\leq)$} we get that $L_i^{\plus}\subseteq\fHMA{L_{i+1}\upsetc}\upsetc=L_{i+1}^{\plus}\upsetc$.

To prove that $\abs{L_i^{\plus}}=\abs{L_i}$ (i.e. \jnlevels{} for $\strp$) we first consider the width of $L_i\upsetc$. By \jnlevels{} for $\str$ we know that $L_i$ is high in $P$. Thus for any antichain $A\subseteq L_i\upsetc$ we have $A=L_i$ or $\abs{A-\set{x}}<\abs{L_i}$. This implies that $\fWidth{L_i\upsetc}=\abs{L_i}$, and consequently
\begin{align*}
\abs{L_i^{\plus}}&\ =\ \abs{\fHMA{L_i\upsetc}}\qquad&\textrm{by \ref{algorytm CaseB Li+=hmaLi}}\\
&\ =\ \fWidth{L_i\upsetc}&\\
&\ =\ \abs{L_i}&\\
&\ =\ i.&\textrm{by \jnlevels{} for $\str$}
\end{align*}
\end{proof}
\begin{clm}$\alpha_1^{\plus},\ldots,\alpha_w^{\plus},\beta_1^{\plus},\ldots,\beta_w^{\plus}$ forms a chain partition of $P^{\plus}$.
\end{clm}
\begin{proof}
After updating $L_i$'s to $L_i^{\plus}$'s algorithm defines the sets\\
$\alpha_1^{\plus},\ldots,\alpha_w^{\plus},\beta_1^{\plus},\ldots,\beta_w^{\plus}$. It turns out (see \ref{algorytm CaseB if dla alfai+},\ref{algorytm CaseB if dla beta+}) that the only chains modified are those at $i_0$:
\begin{align*}
\alpha_{i_0}^{\plus}&\ =\ \beta_{i_0}\cup\set{x},\\
\beta_{i_0}^{\plus}&\ =\ \alpha_{i_0}.
\end{align*}
The Claim is obvious except the fact that $\alpha_{i_0}^{\plus}$ is a chain. Thus, we need to prove that \begin{equation}\label{b<x dla dow b nalezacego do Bi0}
b<x,\qquad\textrm{for all $b\in\beta_{i_0}$}.
\end{equation}
By \jnrel{} for $\str$ we have
\begin{equation}\label{betai0 zawiera sie w Li0 stozek dolny}
\beta_{i_0}\subseteq L_{i_0}\downseto.
\end{equation}
Let $b\in\beta_{i_0}$. By (\ref{betai0 zawiera sie w Li0 stozek
  dolny}) we obtain an $l\in L_{i_0}$ such that $b<l$. On the other hand since $x\in L_{i_0}\upseto$ (see \ref{algorytm CaseB definicja i0}) we get $l'\in L_{i_0}$ such that $l'<x$. If $b\parallel x$ then points $b,l,l',x$ would form a $\II$ configuration, which is forbidden in an interval order (see Figure \ref{rys 2+2}). This proves (\ref{b<x dla dow b nalezacego do Bi0}).
\end{proof}
\begin{figure}[hbt]
\input{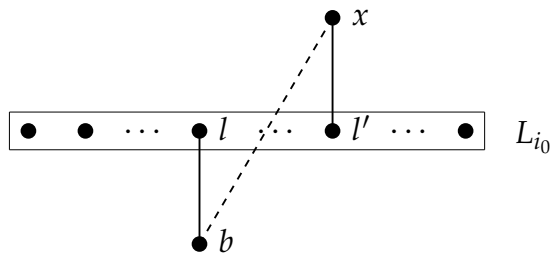}
\caption{$\brackets{2\!+\!2}$-free condition yields $b<x$.}\label{rys 2+2}
\end{figure}
As we can see the idea of the algorithm is to ensure that the chain $\beta_i$ is ready to able to cover any point over the level $L_i$ (we may say that $\beta_i$ secures $L_i$). Thus, when a new point $x$ arrives over $L_i$ then chain $\beta_i$ may be used to cover it. But the algorithm must choose such $i$ that chain $\beta_i^{\plus}=\alpha_i$ (see \ref{algorytm CaseB if dla beta+}) will be able to secure the new level $L_i^{\plus}$. It will turn out that the choice of $i_0$ as in \ref{algorytm CaseB definicja i0} is a good one.

Before proving \jnrel{} for $\strp$ we describe how the open upsets of old ($L_i$'s) and new ($L_i^{\plus}$'s) levels are related.
\begin{clm}\label{C:diagram}
\begin{small}
\[
\begin{array}{ccccccccccc}
%L_1\upsetc\subseteq\ldots\subseteq L_{i_0-1}\upsetc\subseteq L_{i_0-1}\upsetc\cup\set{x}\subseteq L_{i_0}\upsetc\subseteq L_{i_0+1}\upsetc\subseteq\ldots\subseteq L_w\upsetc
L_1\upsetc\!&\!\subseteq\ldots\subseteq\!&\! L_{i_0-1}\upsetc\!&\!\subseteq\!&\! L_{i_0-1}\upsetc\cup\set{x}\!&\!\subseteq\!&\! L_{i_0}\upsetc\!&\!\subseteq\!&\! L_{i_0+1}\upsetc\!&\!\subseteq\ldots\subseteq\!&\! L_w\upsetc\\
\shortparallel\!&\!\!&\!\shortparallel\!&\!\!&\!\shortparallel\!&\!\!&\!\!&\!\!&\!\cup\!&\!\!&\!\cup\\
L_1^{+}\upsetc\!&\!\subseteq\ldots\subseteq\!&\!
L_{i_0-1}^{+}\upsetc\!&\!\subseteq\!&\!
L_{i_0}^{+}\upsetc\!&\!\!&\!\!&\!\subseteq\!&\!
L_{i_0+1}^{+}\upsetc\!&\!\subseteq\ldots\subseteq\!&\! L_w^{+}\upsetc
\end{array}
\]
\end{small}
\end{clm}
To be consistent with the order in $\poset{P}^{\plus}\!$, all inclusions of the claim are presented upside down in Figure  \ref{rys diagram}.
\begin{figure}[hbt]
\ifx\JPicScale\undefined\def\JPicScale{1}\fi
\psset{unit=\JPicScale mm}
\psset{linewidth=0.3,dotsep=1,hatchwidth=0.3,hatchsep=1.5,shadowsize=1,dimen=middle}
\psset{dotsize=0.7 2.5,dotscale=1 1,fillcolor=black}
\psset{arrowsize=1 2,arrowlength=1,arrowinset=0.25,tbarsize=0.7 5,bracketlength=0.15,rbracketlength=0.15}
\begin{pspicture}(0,0)(115,82)
\rput[l](9.5,80){$L_1^+\upsetc = L_1\upsetc$}
\rput[l](19.5,70){$L_2^+\upsetc = L_2\upsetc$}
\psline[linewidth=0.25](17,70)(7,80)
\psline[linewidth=0.25,linestyle=dashed,dash=1 1](32,55)(17,70)
\rput[r](29.5,55){$L_{i_0-1}^+\upsetc = L_{i_0-1}\upsetc$}
\rput[l](44.5,45){$L_{i_0}^+\upsetc = L_{i_0-1}\upsetc\cup\set{x}=\left(L_{i_0-1}\cup\set{x}\right)\upsetc$}
\psline[linewidth=0.25](42,45)(32,55)
\psline[linewidth=0.25](42,45)(52,55)
\rput[l](54.5,55){$\set{x}$}
\psline[linewidth=0.25](42,25)(32,35)
\psline[linewidth=0.25,linestyle=dashed,dash=1 1](57,10)(42,25)
\psline[linewidth=0.25](42,25)(52,35)
\psline[linewidth=0.25](52,35)(42,45)
\psline[linewidth=0.25](32,35)(42,45)
\rput[r](29.5,35){$L_{i_0}\upsetc$}
\rput[l](54.5,35){$L_{i_0+1}^+\upsetc$}
\rput[r](39.5,25){$L_{i_0+1}\upsetc$}
\psline[linewidth=0.25,linestyle=dashed,dash=1 1](67,20)(52,35)
\psline[linewidth=0.25](57,10)(67,20)
\rput[r](54.5,10){$L_{w-1}\upsetc$}
\rput[l](69.5,20){$L_{w-1}^+\upsetc$}
\psline[linewidth=0.25](67,0)(57,10)
\psline[linewidth=0.25](67,0)(77,10)
\psline[linewidth=0.25](77,10)(67,20)
\rput[l](79.5,10){$L_w^+\upsetc$}
\rput[r](64.5,0){$L_w\upsetc$}
\rput[bl](75,36.5){\begin{footnotesize}(in fact,  $L_{i_0}^+=L_{i_0-1}\cup\set{x}$)\end{footnotesize}}
\rput{0}(7,80){\psellipse[fillstyle=solid](0,0)(0.88,0.88)}
\rput{0}(17,70){\psellipse[fillstyle=solid](0,0)(0.88,0.88)}
\rput{0}(32,55){\psellipse[fillstyle=solid](0,0)(0.88,0.88)}
\rput{0}(52,55){\psellipse[fillstyle=solid](0,0)(0.88,0.88)}
\rput{0}(42,45){\psellipse[fillstyle=solid](0,0)(0.88,0.88)}
\rput{0}(32,35){\psellipse[fillstyle=solid](0,0)(0.88,0.88)}
\rput{0}(52,35){\psellipse[fillstyle=solid](0,0)(0.88,0.88)}
\rput{0}(42,25){\psellipse[fillstyle=solid](0,0)(0.88,0.88)}
\rput{0}(67,20){\psellipse[fillstyle=solid](0,0)(0.88,0.88)}
\rput{0}(57,10){\psellipse[fillstyle=solid](0,0)(0.88,0.88)}
\rput{0}(67,0){\psellipse[fillstyle=solid](0,0)(0.88,0.88)}
\rput{0}(77,10){\psellipse[fillstyle=solid](0,0)(0.88,0.88)}
\psline[linecolor=white](115,61)(115,50)
\psline[linecolor=white](0,48)(0,40)
\psline[linecolor=white](74.5,82)(86,82)
\end{pspicture}
\caption{$L_i\upsetc$ and $L_i^{\plus}\upsetc$.}\label{rys diagram}
\end{figure}
\begin{proof}
From \jnlevels{} for $\str$ it follows that $L_1\upsetc\subseteq\ldots\subseteq L_w\upsetc$. Moreover, we have already proved in Claim \ref{C:L1Lw} that $L_1^{\plus}\upsetc\subseteq\ldots\subseteq L_w^{\plus}\upsetc$. By \ref{algorytm CaseB definicja i0} we also know that $x\in L_{i_0}\upseto$ and $x\notin L_{i_0-1}\upseto$.

Directly from the definition of $L_i^{\plus}$(see \ref{algorytm CaseB Li+=hmaLi}) we obtain that $L_i^{\plus}\upsetc\subseteq L_i\upsetc$ for $i=1,\ldots,w$. The inclusions $L_i^{\plus}\upsetc\supseteq L_i\upsetc$ for $i\leqslant i_0-1$ do not influence our considerations but they are really helpful to understand the dynamic construction of levels of posets $P$ and $P^{\plus}$. 
To see them, note that the choice of $i_0$ in \ref{algorytm CaseB definicja i0} secures that nothing has changed in $L_{i_0-1}\upsetc$.
%Since $x\notin L_{i_0-1}\upseto$ and \jnlevels{} for $\str$ we get that $x\notin L_i\upsetc\subseteq L_{i_0-1}\upsetc$ for $i\leqslant i_0-1$. Thus, $L_i\upsetc=L_i\upsetc\cap P$ for $i\leqslant i_0-1$. Now, using the fact that $L_i$ are high in $\poset{P}$ (\jnlevels{} for $\str$) we obtain that for $i\leqslant i_0-1$\[L_i^{\plus}=\fHMA{L_i\upsetc}=\fHMA{L_i\upsetc\cap P}=L_i.\]
To get all other relations in the diagram all we need to prove is:
\begin{equation}\label{Li0+=Li0-1 suma x}
L_{i_0}^{\plus}=L_{i_0-1}\cup\set{x}.
\end{equation}
Since $L_{i_0}^{\plus}=\fHMA{L_{i_0}\upsetc}$ it suffices to show that $L_{i_0-1}\cup\set{x}\subseteq L_{i_0}\upsetc$ and that $L_{i_0-1}\cup\set{x}$ is high in $P^{\plus}$. The inclusion we get from \jnlevels{} for $\str$ and \ref{algorytm CaseB definicja i0}. The second statement follows from the fact that $L_{i_0-1}$ is high in $P$ (see \jnlevels{} for $\str$) and that $x$ is maximal in $P^{\plus}$.
\end{proof}
Before proving \jnrel{} for $\strp$ note that
\begin{equation}\label{prawda o stozkach dolnych L}
L_i\downsetc\subseteq L_i^{\plus}\downsetc.
\end{equation}
This easily follows from $L^{\plus}_i\upsetc\subseteq L_i\upsetc$ by Observation \ref{O:another def order ant} applied to the maximum antichains $L_i$ and $L_i^{\plus}$ of the poset \mbox{$\brackets{L_i\upsetc,\leq}$}.
%
%\bigskip
%\begin{rem} Let A and B be maximum antichains in a subposet $Q$ of $\poset{R}$. Then $A\subseteq B\upsetc_{\poset{R}}$ iff $B\subseteq A\downsetc_{\poset{R}}$.
%\end{rem}
%\begin{proof}
%Since $A$ and $B$ are maximum antichains in $Q$ we know that
%\begin{equation}\label{lemat o odwracaniu stozkow podzial}
%Q\subseteq A\upseto_{\poset{R}}\cup A\downsetc_{\poset{R}}\qquad\textrm{and}\qquad Q\subseteq B\upseto_{\poset{R}}\cup B\downsetc_{\poset{R}}.
%\end{equation}
%From $A\subseteq B\upsetc_{\poset{R}}$ we get $A\upseto_{\poset{R}}\subseteq B\upseto_{\poset{R}}$. This together with (\ref{lemat o odwracaniu stozkow podzial}) yields
%\[
%B\subseteq B\downsetc_{\poset{R}}\cap Q=Q-B\upseto_{\poset{R}}\subseteq Q-A\upseto_{\poset{R}}\subseteq A\downsetc_{\poset{R}}.
%\]
%The converse follows by symmetry.
%\end{proof}
%
\begin{clm}$\alpha_i^{\plus}\subseteq L_i^{\plus}\downsetc \setminus L_{i-1}^{\plus}\upsetc$\, and\, $\beta_i^{\plus}\subseteq L_i^{\plus}\downseto$.
\end{clm}
\begin{proof}
First we prove our Claim for $i\neq i_0$ where we have $\alpha_i^{\plus}=\alpha_i$ and $\beta_i^{\plus}=\beta_i$ (see \ref{algorytm CaseB if dla alfai+},\ref{algorytm CaseB if dla beta+}). By \jnrel{} for $\str$ we get
\[
\alpha_i^{\plus}=\alpha_i\subseteq(L_i\downsetc-L_{i-1}\upsetc)\cap P\subseteq L_i\downsetc - L_{i-1}\upsetc,\quad\textrm{for $i\neq i_0$},
\]
and
\[
\beta_i^{\plus}=\beta_i\subseteq L_i\downseto\cap P = L_i\downseto,\quad\textrm{for $i\neq i_0$}.
\]
Now $L_i^{\plus}\upsetc\subseteq L_i\upsetc$ (see Claim \ref{C:diagram}) and $L_i\downsetc\subseteq L_i^{\plus}\downsetc$ (see (\ref{prawda o stozkach dolnych L})) yield
\[
\alpha_i^{\plus}\subseteq L_i\downsetc - L_{i-1}\upsetc\subseteq L_i^{\plus}\downsetc-L_{i-1}^{\plus}\upsetc
\]
and
\[
\beta_i^{\plus}\subseteq L_i\downseto\subseteq L_i^{\plus}\downseto,
\]
which finishes the proof for $i\neq i_0$.

Now we consider the chain $\beta_{i_0}^{\plus}=\alpha_{i_0}$ (see \ref{algorytm CaseB if dla beta+}). By \jnrel{} for $\str$ we have $\alpha_{i_0}\subseteq L_{i_0}\downsetc-L_{i_0-1}\upsetc$ and therefore
\begin{equation}\label{cos o betai0+}
\beta^{\plus}_{i_0}=\alpha_{i_0}\subseteq L_{i_0}\downsetc\subseteq L^{\plus}_{i_0}\downsetc.
\end{equation}
Again, using \jnrel{} for $\str$ we get that $\beta_{i_0}^{\plus}\cap L_{i_0-1}\upsetc=\emptyset$. Of course, \mbox{$x\notin\alpha_{i_0}=\beta_{i_0}^{\plus}$}. Thus $\beta_{i_0}^{\plus}$ is disjoint with $L_{i_0}^{\plus}=L_{i_0-1}\cup\set{x}$. Combining this with (\ref{cos o betai0+}) we obtain \jnrel{} for $\beta^{\plus}_{i_0}$.

It remains to prove \jnrel{} for $\alpha_{i_0}^{\plus}=\beta_{i_0}\cup\set{x}$. For the old part $\beta_{i_0}$ of $\alpha_{i_0}^{\plus}$ we have
\begin{align*}
\beta_{i_0}&\ \subseteq\ L_{i_0}\downseto\cap P\qquad&\textrm{by \jnrel{}}\ \textrm{for}\ \str\\
&\ =\ L_{i_0}\downsetc-L_{i_0}\upsetc&\textrm{by}\ x\in L_{i_0}\mspace{-2mu}\upseto\ \textrm{(see \ref{algorytm CaseB definicja i0})}\\
&\ \subseteq\ L_{i_0}\downsetc-L_{i_0-1}\upsetc.&\textrm{by \jnlevels{} for}\ \str
\end{align*}
This together with $L_{i_0-1}^{\plus}\upsetc\subseteq L_{i_0-1}\upsetc$ and $L_{i_0}\downsetc\subseteq L_{i_0}^{\plus}\downsetc$ gives
\[
\beta_{i_0}\ \subseteq\  L_{i_0}\downsetc-L_{i_0-1}\upsetc\ \subseteq\ L^{\plus}_{i_0}\downsetc-L^{\plus}_{i_0-1}\upsetc.
\]
It remains to show that also $x\in L_{i_0}^{\plus}\downsetc-L_{i_0-1}^{\plus}\upsetc$. But this follows from (\ref{Li0+=Li0-1 suma x}), the choice of $i_0$ in \ref{algorytm CaseB definicja i0} and the fact that $L^{\plus}_{i_0-1}\upsetc\subseteq L_{i_0-1}\upsetc$.
\end{proof}

\bigskip

We have just shown that both in Case A and B the conditions \jnall{} are kept invariant by the algorithm when expanding $\str$ to $\strp$. In particular \jnchains{} tells us that in every moment the poset is covered by at most $2w$ chains, namely $\alpha_1,\ldots,\alpha_w,\beta_1,\ldots,\beta_w$.

A careful inspection of the algorithm allows us to eliminate $\beta_1$. Indeed, everytime algorithm uses $\beta_1$ (Case B: $i_0=1$) to cover a new point $x$, it could use chain $\alpha_1$ as well. This immediately follows from the fact that in this setting $x\in L_1\upsetc$ and $\alpha_1\subseteq L_1\downsetc$.
\end{proof}

%\clearpage

%------------------------------------------------------------------------------
\section{Related problems}

A related line of research deals with on-line chain partitioning algorithms which get as an input an interval representation instead of pure interval order.
This variant of on-line chain partitioning problem of interval orders models scheduling of jobs which has to be executed in prescribed time frames.
Its possible application is scheduling in real-time systems mentioned in the Introduction.
The algorithm of Kierstead and Trotter \cite{KiersteadTrotter}  
that they used to prove Theorem \ref{Thm:KiersteadTrotter} gives an obvious upper bound of $3w-2$ for the value of on-line chain partitioning problem for interval orders presented with representation.
On the other hand, Chrobak with \'{S}lusarek proved the very same lower bound.
\begin{thm}[\begin{small}Chrobak, \'{S}lusarek \cite{ChrobakSlusarek}; %\'{S}lusarek \cite{Slusarek};
Kierstead, Trotter \cite{KiersteadTrotter}\end{small}] The value of the on-line chain partitioning problem for interval orders of width $w$  presented with representation is \mbox{$3w-2$}. \end{thm}
In the up-growing case Broniek showed an on-line algorithm \linebreak which is as good as the off-line solution.
\begin{thm}[Broniek \cite{Broniek}]\label{Thm:Broniek}The value of the on-line chain partitioning problem for up-growing interval orders of width $w$ presented with representation is $w$. \end{thm}

For interval orders the exponential bound $(5^w-1)/4$ of Theorem \ref{Th:Kie Sz} has been narrowed to the linear 
bound $3w-2$ of Theorem \ref{Thm:KiersteadTrotter}.
This gives hope that First-Fit Algorithm, which sometimes (see Theorem \ref{Thm:FFA}) has to use infinitely many chains in general case, will behave better in the case of interval orders.
Indeed, it appears that First-Fit Algorithm uses bounded (in terms of $w$) number of chains to partition an on-line interval order of width at most  $w$.
Actually Kierstead  \cite{Kierstead88} showed a linear upper bound of $40\cdot w$.
Later on the number of chains used by First-Fit Algorithm was put between $4.45 \cdot w$ \cite{ChrobakSlusarek} and $25.72\cdot w$ \cite{KiersteadQin}. Then Pemmaraju, Raman and Varadarajan \cite{PemRamVar} lowered the upper bound to $10\cdot w$. Further improvement, to $8\cdot w$, was announced by Graham Brightwell, Henry Kierstead and William Trotter \cite{KiersteadSl} and another one, to $8w-3$, was published by Narayanaswamy and Babu in \cite{Nara}.
On the other hand Henry Kierstead and William Trotter \cite{KiersteadSl} announced an improvement of the lower bound from $4.45 \cdot w$ to $4.99 \cdot w$ and noticed that their technique cannot be applied to go up to $5 \cdot w$.

\medskip

Also a substantial progress has been done recently for semi-orders, i.e., interval orders that can be represented by intervals of the same length.
Micek, in his Ph.D. thesis \cite{Micek}, observed that the upper bound of $2w-1$ provided by greedy algorithm cannot be decreased. Thus we have:

\begin{thm}\label{Thm:semi}
The value of on-line chain partitioning problem for semi-orders of width $w$ presented without representation is $2w-1$.
\end{thm}

Much more involved is a proof of Felsner, Kloch, Matecki and Micek that shows the following theorem.
\begin{thm}[Felsner, Kloch, Matecki and Micek \cite{FKMM}] The value of the on-line chain partitioning problem for semi-orders of width $w$ presented in an up-growing way is $\frac{1+\sqrt{5}}{2}w$. \end{thm}
The problem is not completely solved when semi-orders are presented together with their representation by unit length intervals.
Obviously in the up-growing case the result of Theorem \ref{Thm:Broniek} transfers to unit length intervals so that the value of the on-line chain partitioning problem for up-growing semi-orders of width $w$ presented with representation is $w$.
However all we know in a not necessarily up-growing case is that $\ffloor{3w/2}$ and $2w-1$ are the lower and upper bounds for the on-line chain partitioning problem for semi-orders of width $w$ presented with representation. Again the lower bound was shown in \cite{Micek}, while the upper one follows from Theorem \ref{Thm:semi}.

\nocite{*}
\bibliographystyle{alpha}
\bibliography{phd-thesis}

\begin{thebibliography}{FKMM07}

\bibitem[BBM07]{Bosek}
Patrick Baier, Bart{\l}omiej Bosek, and Piotr Micek.
\newblock On-line chain partitioning of up-growing interval orders.
\newblock {\em Order}, 24(1):1--13, 2007.

\bibitem[BM05]{BosekMicek2}
Bart{\l}omiej Bosek and Piotr Micek.
\newblock Variants of online chain partition problem of posets.
\newblock In {\em Proceedings of the Second Workshop on Computational Logic and
  Applications (CLA 2004)}, volume 140 of {\em Electron. Notes Theor. Comput.
  Sci.}, pages 3--13 (electronic), Amsterdam, 2005. Elsevier.

\bibitem[Bro05]{Broniek}
Przemys{\l}aw Broniek.
\newblock On-line chain partitioning as a model for real-time scheduling.
\newblock In {\em Proceedings of the Second Workshop on Computational Logic and
  Applications (CLA 2004)}, volume 140 of {\em Electron. Notes Theor. Comput.
  Sci.}, pages 15--29 (electronic), Amsterdam, 2005. Elsevier.

\bibitem[C{\'S}88]{ChrobakSlusarek}
Marek Chrobak and Maciej {\'S}lusarek.
\newblock On some packing problem related to dynamic storage allocation.
\newblock {\em RAIRO Inform. Th{\'{e}}or. Appl.}, 22(4):487--499, 1988.

\bibitem[Dil50]{Dilworth}
Robert~P. Dilworth.
\newblock A decomposition theorem for partially ordered sets.
\newblock {\em Ann. of Math.}, 51(2):161--166, 1950.

\bibitem[Dil60]{Dilworth2}
Robert~P. Dilworth.
\newblock Some combinatorial problems on partially ordered sets.
\newblock In {\em Proc. Sympos. Appl. Math.}, volume~10, pages 85--90, 1960.

\bibitem[Fel97]{Felsner}
Stefan Felsner.
\newblock On-line chain partitions of orders.
\newblock {\em Theoret. Comput. Sci.}, 175(2):283--292, 1997.

\bibitem[Fis70]{Fishburn}
Peter~C. Fishburn.
\newblock Intransitive indifference with unequal indifference intervals.
\newblock {\em J. Mathematical Psychology}, 7:144--149, 1970.

\bibitem[FKMM07]{FKMM}
Stefan Felsner, Kamil Kloch, Grzegorz Matecki, and Piotr Micek.
\newblock On-line chain partitioning of upgrowing orders: The case of
  2-di\-men\-sional orders and semi-orders, 2007.
\newblock \\{\tt http://arxiv.org/abs/0704.1829v1}.

\bibitem[FM91]{FederMotwani}
Tom{\'{a}}s Feder and Rajeev Motwani.
\newblock Clique partitions, graph compression and speeding-up algorithms.
\newblock In {\em Proceedings of the twenty-third annual ACM symposium on
  Theory of computing (STOC '91)}, pages 123--133, New York, NY, 1991. ACM.

\bibitem[Kie81]{Kierstead}
Henry~A. Kierstead.
\newblock An effective version of {D}ilworth's theorem.
\newblock {\em Trans. Amer. Math. Soc.}, 268(1):63--77, 1981.

\bibitem[Kie86]{Kierstead86}
Henry~A. Kierstead.
\newblock Recursive ordered sets.
\newblock In {\em Combinatorics and ordered sets (Arcata, Calif., 1985)},
  volume~57 of {\em Contemp. Math.}, pages 75--102. Amer. Math. Soc.,
  Providence, RI, 1986.

\bibitem[Kie88]{Kierstead88}
Henry~A. Kierstead.
\newblock The linearity of first-fit coloring of interval graphs.
\newblock {\em SIAM J. Discrete Math.}, 1(4):526--530, 1988.

\bibitem[Kie08]{KiersteadSl}
Henry~A. Kierstead.
\newblock On-line partitioning.
\newblock {\em {\rm Slides from conference} New directions in algorithms,
  combinatorics, and optimization}, 2008.
\newblock \\{\tt
  http://www.aco.gatech.edu/conference/archive/acokierstead.pdf}.

\bibitem[KQ95]{KiersteadQin}
Henry~A. Kierstead and Jun Qin.
\newblock Coloring interval graphs with {F}irst-{F}it.
\newblock {\em Discrete Math.}, 144(1-3):47--57, 1995.

\bibitem[KT81]{KiersteadTrotter}
Henry~A. Kierstead and William~T. Trotter.
\newblock An extremal problem in recursive combinatorics.
\newblock In {\em Proceedings of the Twelfth Southeastern Conference on
  Combinatorics, Graph Theory and Computing, Vol. II (Baton Rouge, La., 1981)},
  volume~33, pages 143--153, 1981.

\bibitem[Mic08]{Micek}
Piotr Micek.
\newblock {\em On-line chain partitioning of semi-orders}.
\newblock PhD thesis, Jagiellonian University, 2008.

\bibitem[NS08]{Nara}
N.~S. Narayanaswamy and R.~{Subhash Babu}.
\newblock A note on first-fit coloring of interval graphs.
\newblock {\em Order}, 25(1):49--53, 2008.

\bibitem[Per63]{Perles}
Micha Perles.
\newblock A proof of dilworth’s decomposition theorem for partially ordered
  sets.
\newblock {\em Israel Journal of Mathematics}, 1:105--107, 1963.

\bibitem[PRV04]{PemRamVar}
Sriram~V. Pemmaraju, Rajiv Raman, and Kasturi~R. Varadarajan.
\newblock Buffer minimization using max-coloring.
\newblock In {\em Proceedings of the Fifteenth Annual ACM-SIAM Symposium on
  Discrete Algorithms (SODA 2004)}, pages 562--571 (electronic), New York,
  2004. ACM.

\bibitem[{\'S}lu89]{Slusarek}
Maciej {\'S}lusarek.
\newblock A coloring algorithm for interval graphs.
\newblock In {\em Mathematical Foundations of Computer Science 1989
  (Porabka-Kozubnik, 1989)}, volume 379 of {\em Lecture Notes in Comput. Sci.},
  pages 471--480. Springer, Berlin, 1989.

\bibitem[Tro]{TrotterSl}
William~T. Trotter.
\newblock First {F}it {C}oloring of {I}nterval {G}raphs.
\newblock \\{\tt http://www.math.gatech.edu/\~{}trotter/slides/P-FirstFit.ppt}.

\bibitem[Tro92]{Trotter}
William~T. Trotter.
\newblock {\em Combinatorics and partially ordered sets: Dimension theory}.
\newblock Johns Hopkins Series in the Mathematical Sciences. Johns Hopkins
  University Press, Baltimore, MD, 1992.

\end{thebibliography}

\end{document}